\documentclass[
11pt, 
english, 
singlespacing, 
headsepline, 
]{MastersDoctoralThesis} 

\usepackage[utf8]{inputenc} 
\usepackage[T1]{fontenc} 

\usepackage{mathpazo} 
\usepackage[export]{adjustbox}
\usepackage{mdframed}
\usepackage{tcolorbox}
\usepackage{collectbox}
\usepackage{fancybox}
\usepackage{booktabs} 
\usepackage{xcolor}
\usepackage{color}
\usepackage{minitoc}
\usepackage[linesnumbered, ruled]{algorithm2e} 
\usepackage{multicol}
\usepackage{tabularx}
\usepackage{amssymb}
\usepackage[tracking]{microtype}
\usepackage{pifont}
\usepackage{varwidth}
\usepackage{mathpartir}
\usepackage{makecell}
\usepackage{mathtools}
\usepackage{doi}
\usepackage{rotating}
\usepackage{dsfont}
\usepackage{mathtools}
\usepackage{amsmath,amssymb,amsthm}

\newtheorem{hypothesis}{Induction Hypothesis}
\newtheorem{definition}{Definition}
\newtheorem{lemma}{Lemma}
\newtheorem{basecase}{Base Case}
\theoremstyle{plain}

\newcommand\narrowstyle{\SetTracking{encoding=*}{-50}\lsstyle}

\newcommand\normalstyle{\SetTracking{encoding=*}{0}\lsstyle}
\definecolor{BLACK}{rgb}{0.0, 0.0, 0.0}

\usepackage{tikz}
\usetikzlibrary{tikzmark}
\usetikzlibrary{positioning}
\usetikzlibrary{shapes.callouts}

\tikzset{
  level/.style   = { ultra thick, blue },
  connect/.style = { dashed, red },
  notice/.style  = { draw, rectangle callout, callout relative pointer={#1} },
  label/.style   = { text width=2cm }
}

\usepackage{listings}
\usepackage{amsmath}
\usepackage{textcomp}
\usepackage{tabularx}
\usepackage{subfig}
\usepackage{multicol}
\usepackage{booktabs}
\usepackage{float}

\DeclareMathOperator*{\argmin}{argmin}

\newfloat{lstfloat}{htbp}{lop}
\floatname{lstfloat}{Listing}

\newcommand\crule[3][black]{\textcolor{#1}{\rule{#2}{#3}}}

\definecolor{kwcolor11}{rgb}{0,0,0.75}

\makeatletter
\newcommand{\nosemic}{\renewcommand{\@endalgocfline}{\relax}}
\newcommand{\dosemic}{\renewcommand{\@endalgocfline}{\algocf@endline}}
\let\oldnl\nl
\newcommand{\nonl}{\renewcommand{\nl}{\let\nl\oldnl}}
\makeatother

\usepackage[autostyle=true]{csquotes} 

\thesistitle{Memory-Efficient Object-Oriented Programming on GPUs} 
\supervisor{Prof. Dr. Hidehiko \textsc{Masuhara}} 
\examiner{} 
\degree{Doctor of Philosophy} 
\author{Matthias \textsc{Springer}} 
\addresses{} 
\studentnum{15D*****} 
\subject{Computer Science} 
\university{\href{http://www.titech.ac.jp}{\textcolor{black}{Tokyo Institute of Technology}}} 
\department{Department of Mathematical and Computing Sciences} 
\group{\href{http://prg.is.titech.ac.jp/}{\textcolor{black}{Programming Research Group}}} 
\faculty{Graduate School of Information Science and Engineering} 

\AtBeginDocument{
\hypersetup{pdftitle=\ttitle} 
\hypersetup{pdfauthor=\authorname} 
\hypersetup{pdfkeywords=\keywordnames} 
}

\let\Chapter\chapter
\def\chapter{\addtocontents{lol}{\protect\addvspace{10pt}}\Chapter}

\begin{document}

\frontmatter 

\pagestyle{plain} 


\begin{titlepage}
\begin{center}
\includegraphics[scale=0.3]{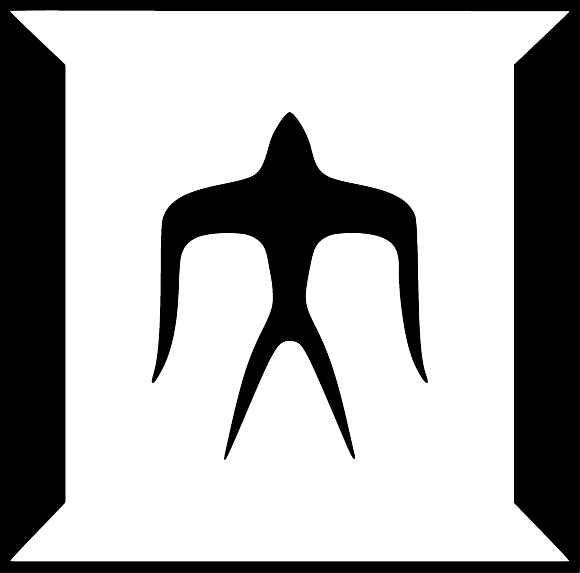} \\[0.1cm]
\vspace*{.01\textheight}
{\scshape\LARGE \univname\par}\vspace{1.5cm} 
\textsc{\Large Doctoral Thesis}\\[0.5cm] 

\HRule \\[0.4cm] 
{\huge \bfseries \ttitle\par}\vspace{0.4cm} 
\HRule \\[1.5cm] 

\begin{minipage}[t]{0.45\textwidth}
\begin{flushleft} \large
\emph{Author:}\\
\authorname\\[0.3cm] 
\emph{Student Number:}\\
\textsc{\studentnumname}
\end{flushleft}
\end{minipage}
\begin{minipage}[t]{0.45\textwidth}
\begin{flushright} \large
\emph{Supervisor:} \\
\supname 
\end{flushright}
\end{minipage}\\[3cm]

\vfill

\large \textit{A thesis submitted in fulfillment of the requirements\\ for the degree of \degreename}\\[0.3cm] 
\textit{in the}\\[0.4cm]
\groupname\\\deptname\\[2cm] 

\vfill

{\large \today}\\[4cm] 

\vfill
\end{center}
\end{titlepage}


\begin{declaration}
\addchaptertocentry{\authorshipname} 
\noindent I, \authorname, declare that this thesis titled, \enquote{\ttitle} and the work presented in it are my own. I confirm that:

\begin{itemize}
\item This work was done wholly or mainly while in candidature for a research degree at this University.
\item Where any part of this thesis has previously been submitted for a degree or any other qualification at this University or any other institution, this has been clearly stated.
\item Where I have consulted the published work of others, this is always clearly attributed.
\item Where I have quoted from the work of others, the source is always given. With the exception of such quotations, this thesis is entirely my own work.
\item I have acknowledged all main sources of help.
\item Where the thesis is based on work done by myself jointly with others, I have made clear exactly what was done by others and what I have contributed myself.\\
\end{itemize}

\noindent Signed:\\
\rule[0.5em]{25em}{0.5pt} 

\noindent Date:\\
\rule[0.5em]{25em}{0.5pt} 
\end{declaration}

\cleardoublepage






\setcounter{minitocdepth}{3}

\dominitoc

\begin{abstract}
\addchaptertocentry{\abstractname} 
Object-oriented programming (OOP) is often regarded as too inefficient for high-performance computing (HPC), even though many important HPC problems have an inherent object structure. To make HPC available to engineers and researchers in other domains, our goal is to bring efficient, object-oriented programming to massively parallel Single-Instruction Multiple-Data (SIMD) architectures, especially GPUs.

In this thesis, we develop techniques and prototypes for optimizing the memory access of object-oriented GPU code. Our first prototype \textsc{Ikra-Ruby} explores modular array-based GPU computing in Ruby, a high-level programming language. Our main prototype \textsc{Ikra-Cpp} explores object-oriented programming within parallel GPU code in CUDA/C++.

We propose a new object-oriented programming model called \emph{Single-Method Multiple-Objects} (SMMO) that can express many important HPC problems and that can be implemented efficiently on GPUs. Our main optimization is the well-known Structure of Arrays (SOA) data layout, which improves vectorized access and cache performance. The main contributions of this thesis are threefold: First, we develop an embedded C++/CUDA data layout for \textsc{Ikra-Cpp} that allows programmers to experience the performance benefit of SOA without breaking OOP abstractions. Second, we design \textsc{DynaSOAr}, a lock-free, dynamic GPU memory allocator for \textsc{Ikra-Cpp} that is based on hierarchical bitmaps and optimizes the usage of allocated memory with an SOA data layout. In contrast to other state-of-the-art GPU allocators, \textsc{DynaSOAr} trades raw (de)allocation performance for better memory access performance, resulting in an up to 3 times speedup of SMMO application code. Finally, we extend \textsc{DynaSOAr} with a memory defragmentation system called \textsc{CompactGpu}, which further increases the performance benefits of SOA and lowers the overall memory usage through less fragmented allocations.

We evaluated \textsc{Ikra-Cpp} with nine SMMO applications. Our experiments show that the SMMO programming model is powerful enough to express many important HPC problems from various domains. They also demonstrate that programmers can have the benefits of object-oriented programming and good runtime performance at the same time.

\end{abstract}



\begin{acknowledgements}
\addchaptertocentry{\acknowledgementname} 
First and foremost, I would like to thank my academic advisor Prof. Hidehiko Masuhara for accepting me as a doctoral student in the Programming Research Group, for giving me the opportunity to freely pursue my research interests, and for supervising, guiding and supporting my work throughout these last four years. I would also like to thank Prof. Robert Hirschfeld, my academic advisor during my Bachelor's and Master's studies for the continuous support and collaboration even after my Master's studies.

In 2016 and 2018, I had to pleasure to work with Peter Wauligmann and Yaozhu Sun, who joined the Programming Research Group as part of a university exchange program. Chapters~\ref{sec:thesis_expressing_parallel} and~\ref{sec:theis_opt_mem_access} contain material from two ARRAY workshop papers, which were co-authored by Peter Wauligmann and Yaozhu Sun. My co-authors have approved the inclusion of these papers in my thesis. I am grateful for their contributions to my research artifacts and wish them all the best with their own careers in industry and academia.

The financial support of MEXT (Japanese Ministry of Education, Culture, Science and Technology) and JSPS (Japan Society for the Promotion of Science) through their scholarship and fellowship programs allowed me focus entirely on my research. I am grateful for their support and their decision to fund my research proposal. I would also like to thank NVIDIA for donating a TITAN Xp GPU to our lab, which allowed me to run my experiments on the latest hardware architecture.
\end{acknowledgements}


\tableofcontents \mtcaddchapter

\listoffigures \mtcaddchapter

\listoftables \mtcaddchapter

\listofalgorithms 

\lstlistoflistings

\mainmatter 

\pagestyle{thesis} 



\newtheorem*{thstatement}{Thesis Statement}

\chapter{Introduction}

High-performance, general-purpose GPU computing has long been a tedious job, requiring programmers to write hand-optimized, low-level programs. Such programs are hard to develop, debug and maintain. Even though object-oriented programming is well established in other domains and appreciated for its good abstraction, expressiveness, modularity and developer productivity, it is regarded as too inefficient for high-performance computing (HPC). This is despite the fact that many important HPC applications exhibit an inherent object structure. Our goal is to bring fast object-oriented programming to SIMD architectures, especially GPUs.

\begin{thstatement}
Efficient object-oriented programming is feasible on GPUs.
\end{thstatement}

\begin{figure}
  \centering
  \includegraphics[width=0.8\textwidth]{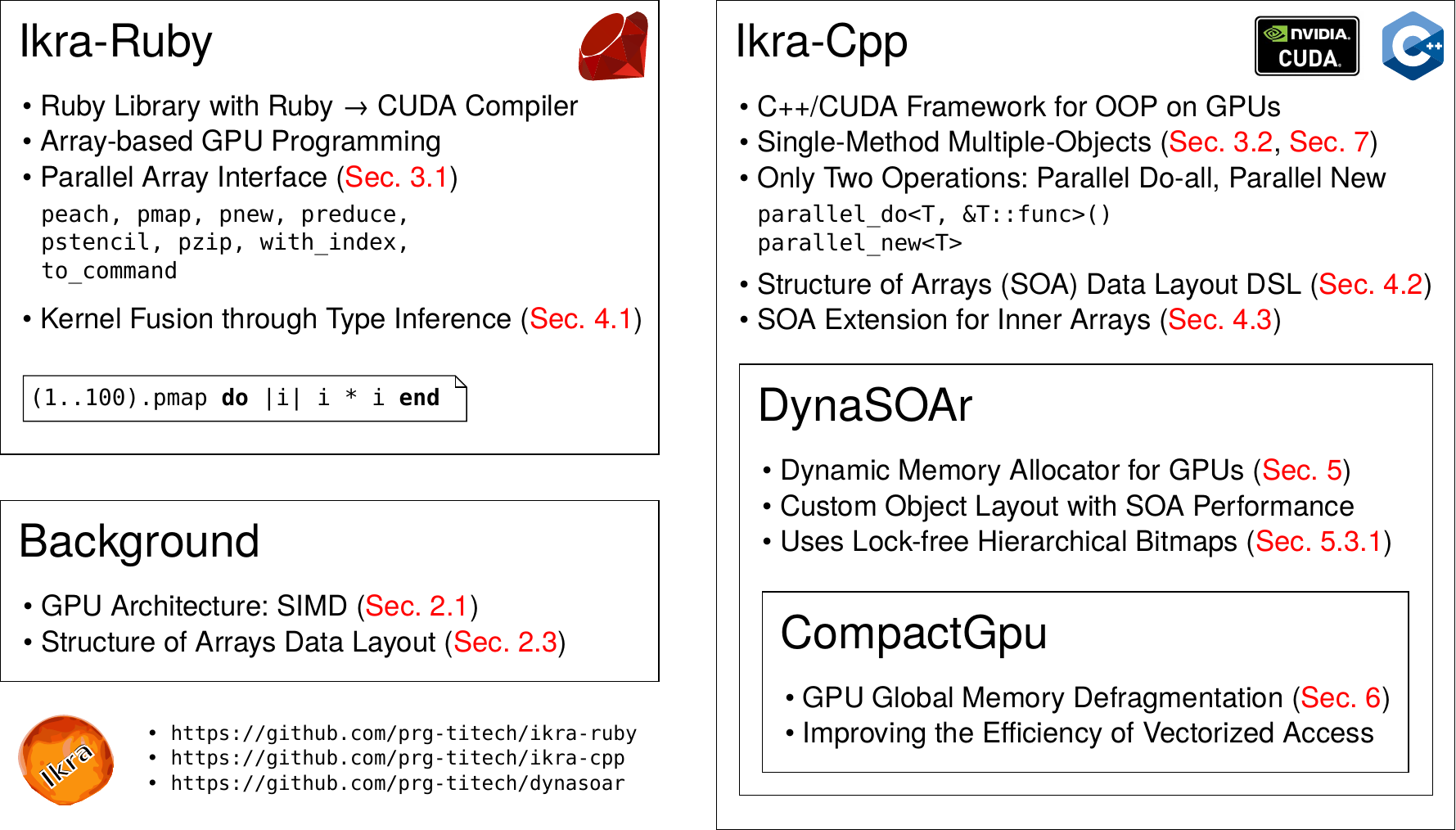}
  \caption{Overview of this thesis}
  \label{fig:thesis_overview}
\end{figure}

In the course of this thesis, we present the design and implementation of libraries/prototypes that allow programmers to utilize object-oriented programming on GPUs (Figure~\ref{fig:thesis_overview}). We show that an object-oriented programming style, with most of its benefits, is feasible on GPUs without sacrificing performance. We investigate (a) how GPU parallelism can be expressed with object orientation and (b) how object-oriented code that runs on GPUs can be optimized.

\paragraph{Expressing Parallelism: Parallel Array Interface}
Our first prototype, \textsc{Ikra-Ruby}, is a Ruby library for high-level, array-based GPU programming. Programmers express parallelism over an array with small, functional, customizable operations such as \emph{parallel map} or \emph{parallel stencil}. Programmers can utilize the full range of Ruby's object-oriented capabilities to compose a larger GPU program from multiple parallel operations in a modular way. By combining multiple parallel operations, programmers build a computation graph. \textsc{Ikra-Ruby} optimizes this computation graph by fusing multiple parallel operations into a small number of CUDA kernels. 

\paragraph{Expressing Parallelism: Single-Method Multiple-Objects}
Our second prototype, \textsc{Ikra-Cpp}, explores object-oriented programming within parallel GPU code in C++/CUDA. While \textsc{Ikra-Ruby} provides various parallel operations, we found that a simpler model with only one type of operation, \emph{parallel do-all}, is sufficient for a many applications. We discovered a broad programming model based on parallel do-all with many object-oriented applications that can be implemented efficiently on massively parallel SIMD accelerators. We call this model \emph{Single-Method Multiple-Objects} (SMMO), because parallelism is expressed by running a method on all objects of a type. SMMO fits well with the data-parallel execution pattern of GPUs and has many important real-world applications such as simulations for population dynamics, evacuations, wildfire spreading, finite element methods or particle systems. SMMO can also express breadth-first graph traversals and dynamic tree updates/constructions.

\paragraph{Optimizing Memory Access}
Object-oriented programming is slow on SIMD architectures mainly because of inefficient memory access. Getting data into and out of vector registers is often the biggest bottleneck and peak memory bandwidth utilization can be achieved only with efficient vector transactions. To optimize the memory access of \textsc{Ikra-Cpp} applications, we developed an embedded C++ DSL that stores objects of the same type in the well-known Structure of Arrays (SOA) data layout. SOA is a form of structure splitting that stores all values of a field together. SOA is a standard optimization and best practice for GPU programmers but without proper language support, it leads to less readable code and breaks language abstractions. \textsc{Ikra-Cpp} allows programmers to experience the performance benefit of SOA without sacrificing code readability or object-oriented abstractions in C++.

\paragraph{Dynamic Memory Allocation}
Dynamic memory management and the ability/flexibility of creating/deleting objects at any time is one of the corner stones of object-oriented programming. Unfortunately, existing GPU memory allocators are either notoriously slow in serving allocations or miss key optimizations for structured data, leading to poor data locality and low memory bandwidth utilization when accessing data. For this reason, many GPU programmers still avoid dynamic memory management entirely and try to statically allocate all memory.

As the main research contribution of this thesis, we extended \textsc{Ikra-Cpp} with \textsc{DynaSOAr}, a fully-parallel, lock-free, dynamic memory allocator. \textsc{DynaSOAr} improves the usage of allocated memory with an SOA data layout and achieves low memory fragmentation through efficient management of free and allocated memory blocks with lock-free, hierarchical bitmaps. Contrary to other state-of-the-art allocators, our design is heavily based on atomic operations, trading raw (de)allocation performance for better overall application performance. In our benchmarks, \textsc{DynaSOAr} achieves a speedup of SMMO application code of up to 3x over state-of-the-art allocators. Moreover, \textsc{DynaSOAr} manages heap memory more efficiently than other allocators, allowing programmers to run up to 2x larger problem sizes with the same amount of memory.

\paragraph{Memory Defragmentation}
In a system with dynamic memory allocation, the efficiency of an SOA data layout depends heavily on low memory fragmentation. If data is fragmented, more vector accesses are needed for the same number of bytes, reducing the benefit of SOA. To further improve memory fragmentation and vectorized access, we extended \textsc{DynaSOAr} with \textsc{CompactGpu}, an incremental, fully-parallel, in-place memory defragmentation system. \textsc{CompactGpu} defragments the heap by merging partly occupied memory blocks. We developed several implementation techniques for memory defragmentation that are efficient on SIMD/GPU architectures, such as finding defragmentation block candidates and fast pointer rewriting based on bitmaps.

\paragraph{Summary}
In this thesis, we show that, contrary to common belief, object-oriented programming is feasible on GPU without sacrificing performance, if the programming system provides well-designed abstractions and programming models that can be implemented efficiently on GPUs. SMMO is the main programming model throughout this thesis and we show through various examples that it is sufficiently expressive and that it can run efficiently on GPUs.

While this thesis focuses mostly on memory access performance, future work could focus more on control flow divergence or advanced features of object-oriented programming such as virtual function calls or advanced modularity constructs such as multiple inheritance. We also plan to integrate \textsc{Ikra-Cpp} into \textsc{Ikra-Ruby} in the future, so that programmers can develop SMMO applications in a high-level programming language.

\paragraph{Outline}
This thesis is organized as follows. Chapter~\ref{sec:thesis_background} gives an overview of the execution model and memory hierarchy of GPUs, as well as a brief summary of the object-oriented programming model and its challenges on GPUs. Chapter~\ref{sec:thesis_expressing_parallel} investigates how object orientation can be used to express GPU parallelism. In particular, we describe the APIs of \textsc{Ikra-Ruby} and \textsc{Ikra-Cpp}. Chapter~\ref{sec:theis_opt_mem_access} shows how to optimize the memory access of \textsc{Ikra-Ruby} and \textsc{Ikra-Cpp} applications with kernel fusion and with a Structure of Arrays (SOA) data layout. Chapter~\ref{sec:chapter_dynasoar} presents the design and implementation of the \textsc{DynaSOAr} dynamic memory allocator. Chapter~\ref{sec:chap_gpu_mem_defrag} describes how memory that was dynamically allocated with \textsc{DynaSOAr} can be defragmented with \textsc{CompactGpu} to improve the efficiency of the SOA layout and to lower the overall memory consumption. Chapter~\ref{chap:smmo_examples} illustrates the design and implementation of a various SMMO applications from different domains, which underlines the importance of the SMMO programming model. Finally, Chapter~\ref{sec:thesis_concl} concludes this thesis.

\chapter{Background}
\label{sec:thesis_background}
This chapter gives an overview of the architecture and execution model of recent NVIDIA GPUs (Section~\ref{sec:bg_gpu_exec_model}), as well as  a brief background on object-oriented programming (Section~\ref{sec:bg_oop}) and the Structure of Arrays (SOA) data layout (Section~\ref{sec:o_aos_vs_soa}).

\minitoc

\section{GPU Execution Model}
\label{sec:bg_gpu_exec_model}
Graphics processing units (GPUs) have become more and more popular in the last decade. While in the early years, GPUs were mainly used for producing visual output on a computer monitor and to accelerate graphical computations of computer games, they are now well established in a variety of other areas such as high-performance computing or machine learning.

The three main GPU manufacturers at present are AMD, Intel and NVIDIA~\cite{gpu_market_share}. We are focusing on NVIDIA architectures and the NVIDIA CUDA programming model in this work, as NVIDIA GPUs are most widely used in high-performance computing according to the \emph{TOP500} list of supercomputers~\cite{nvidia_supercomputer_list}. Since most GPU architectures follow similar design principles, we expect that our findings also apply to other architectures.

NVIDIA GPUs can be programmed with CUDA and OpenCL. Both are C++ dialects. While OpenCL works with a variety of GPU architectures, CUDA is specific to NVIDIA GPUs and provides access to more fine grained control flow, memory and synchronization primitives that may not be available on other architectures~\cite{cuda_vs_opencl_api}. On NVIDIA architectures, CUDA code has also been shown to be more performant than comparable OpenCL code~\cite{DBLP:journals/corr/abs-1005-2581, 6047190, 6419068}.

\subsection{Parallel Execution}
\label{sec:backg_parall_exec}
\begin{figure}
  \includegraphics[width=\textwidth]{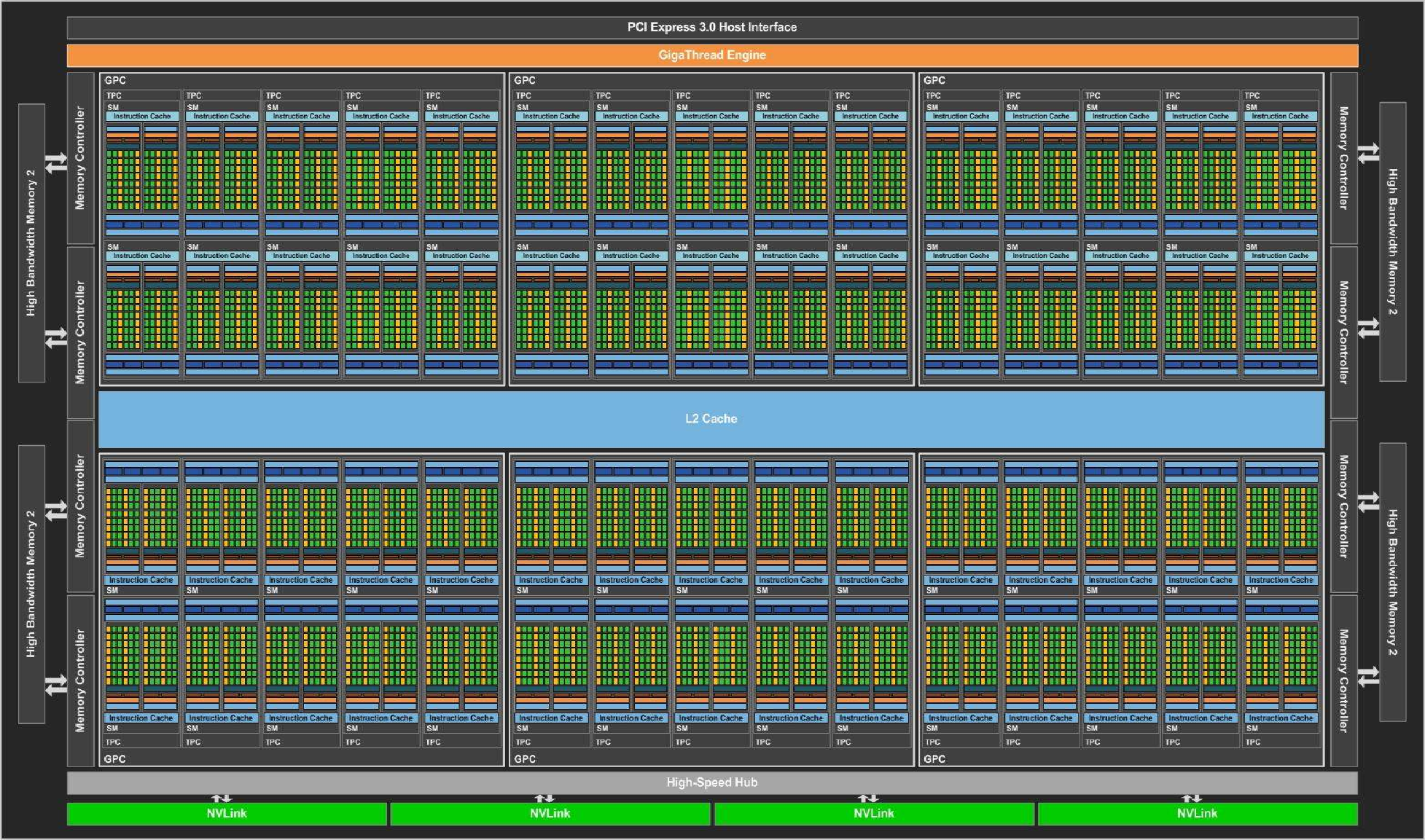}
  \caption[Architecture of NVIDIA GP100]{Architecture of NVIDIA GP100 (Source: NVIDIA Tesla P100 Whitepaper~\cite{pascal_whitepaper})}
  \label{fig:architecture_gp_100}
\end{figure}

\begin{figure}
  \centering
  \includegraphics[width=0.7\textwidth]{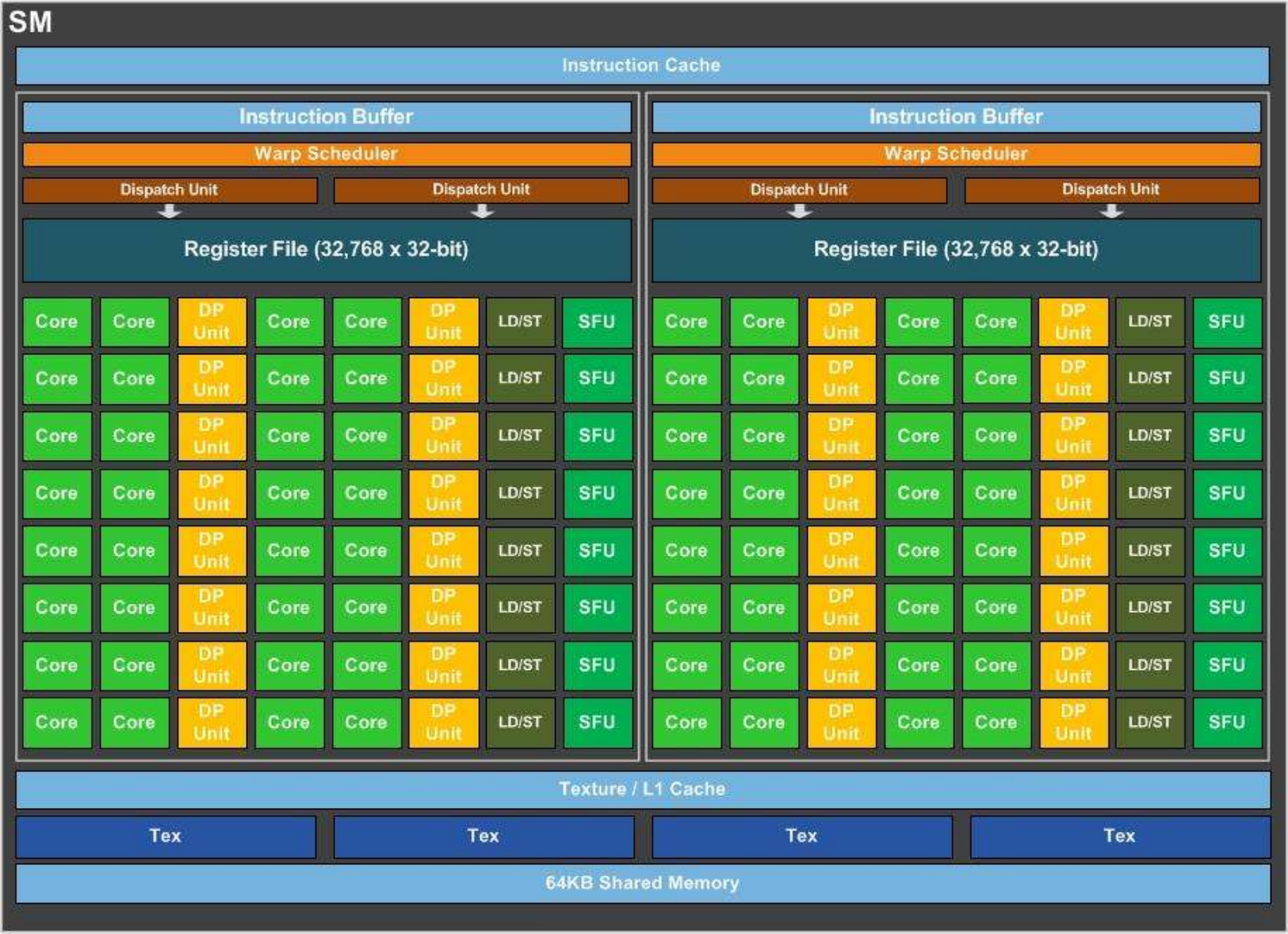}
  \caption[Architecture of a single NVIDIA GP100 SM]{Architecture of a single NVIDIA GP100 SM (Source: NVIDIA Tesla P100 Whitepaper~\cite{pascal_whitepaper})}
  \label{fig:architecture_gp_100_sm}
\end{figure}

Figures~\ref{fig:architecture_gp_100} and~\ref{fig:architecture_gp_100_sm} show the hardware architecture (\emph{SM block diagram}) of an NVIDIA GP100 GPU (Pascal architecture)\footnote{We used a GP102-450-A1 GPU (TITAN~Xp; Pascal architecture) for most of our benchmarks. NVIDIA has not published a whitepaper for this GPU, so there is no published SM block diagram for this GPU. However, its architecture is similar to a GP100 GPU.}. This GPU consists of 60 \emph{streaming multiprocessors} (SMs). Each SM has 64 \emph{CUDA cores} (green boxes), amounting to a total of 3840~CUDA cores. Every group of 32 CUDA cores (\emph{warp}) has a \emph{warp scheduler} that schedules processor instructions for all 32 cores. Each GP100 SM can also be seen as a dual-core processor, with each core operating on vector registers that hold 32 scalars (128~bytes).

GPUs are commonly referred to as \emph{massively parallel SIMD architectures}. However, they actually have three different ways of achieving parallelism.

\begin{itemize}
  \item \textbf{SIMD:} All 32 CUDA cores of a warp execute the same processor instruction as determined by their warp scheduler. Therefore, GPUs are \emph{Single-Instruction Multiple-Data} (SIMD) architectures. They achieve parallelism by executing the same instruction on a vector register (\emph{multiple data}).
  \item \textbf{MIMD:} Each warp has its own warp scheduler, so different warps can execute different instructions. Therefore, GPUs also have \emph{Multiple-Instruction Multiple-Data} (MIMD) parallelism. (Most CPU systems are MIMD architectures.)
  \item \textbf{Instruction-Level Parallelism (ILP):} Warp schedulers can issue two instructions per processor cycle (\emph{dual-issue}), e.g., an arithmetic instruction and a memory access instruction. Details vary from architecture to architecture. On older NVIDIA archiectures, ILP was necessary to achieve peak performance~\cite{volkov2010better, Volkov:EECS-2016-143}, but on Pascal \emph{single-issue} can fully utilize all CUDA cores~\cite{nvidia_instr_sched}.
\end{itemize}

SIMD and MIMD parallelism can be exploited with thread-level parallelism (TLP) in CUDA\footnote{This is a key difference from C++, where SIMD parallelism is not expressed via threads but via automatic vectorization and SIMD intrinsics (manual vectorization).}. There are no dedicated abstractions for instruction-level parallelism (ILP) in CUDA. The scheduling of instructions is up to the compiler and the hardware. However, programmers can sometimes achieve better ILP by manually unrolling loops~\cite{volkov2010better}.

\subsection{Memory Hierarchy}
Recent NVIDIA GPUs have five main types of memory. The size of each memory type can differ among GPU architectures and models. In this work, we are focusing mainly on \emph{device memory} (red access path in Figure~\ref{fig:nvvp_diagram}).

In the following list, we briefly describe each memory type and report memory sizes for an NVIDIA TITAN~Xp GPU, which we used for most of our benchmarks. Memory types are sorted from slowest to fastest.

\begin{figure}
  \centering
  \includegraphics[width=0.7\textwidth]{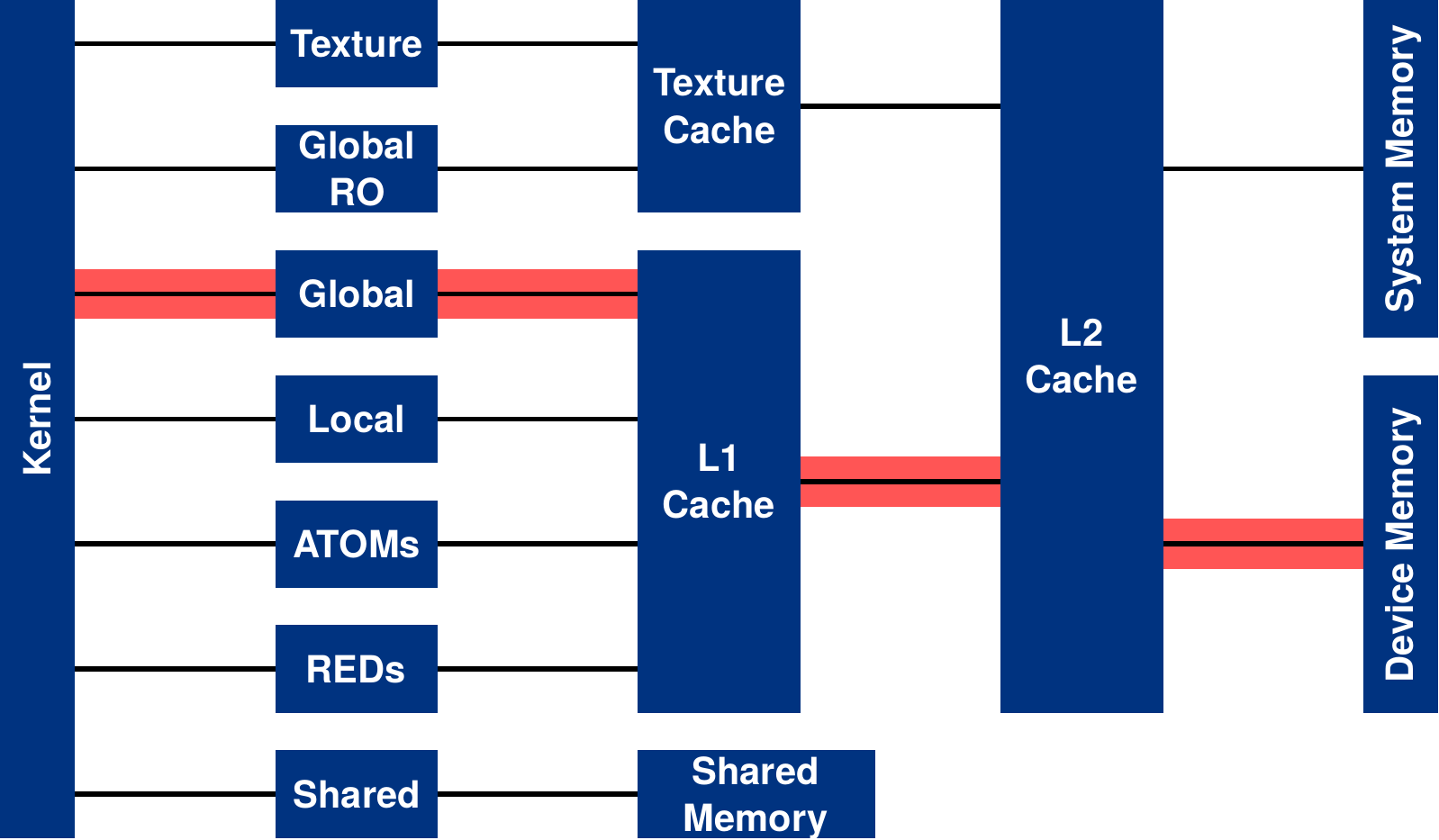}
  \caption[GPU memory access paths]{Memory access paths (Source: NVIDIA Visual Profiler)}
  \label{fig:nvvp_diagram}
\end{figure}

\begin{itemize}
  \item \textbf{Device Memory:} This is the largest type of memory. It resides in the GPU's DRAM. Device memory accesses are always routed through the L1/L2 caches~\cite{nvidia_routed_mem_access} (Figure~\ref{fig:nvvp_diagram}). The L1 cache acts as a \emph{coalescing buffer}: Data is accumulated in the L1 cache before it is delivered to a warp~\cite{nvidia_unified_l1_6x}. A memory request of a warp is potentially broken down into multiple requests, one per cache line (i.e., a memory instruction may be executed multiple times)~\cite{nvidia_global_mem_6x}. Our TITAN~Xp GPU has 12~GB of device memory. The latency of device memory access can be as high as 1,000~cycles in case of an L1/L2 cache miss~\cite{DBLP:journals/corr/abs-1804-06826}. In CUDA, device memory is referred to as \emph{global memory}.
  \item \textbf{Level 2 Cache:} All device memory accesses are cached in the L2 cache, which is 3~MB large on our TITAN~Xp GPU. The L2 cache is shared by all SMs and has a cache line size of 128~bytes. Each cache line consists of four 32~byte segments, which can be read/written independently. The latency of an L2 hit is around 220~cycles~\cite{DBLP:journals/corr/abs-1804-06826}.
  \item \textbf{Level 1 Cache:} By default, all device memory accesses are cached in the L1 cache. This behavior can be changed with compiler flags~\cite{nvidia_unified_l1_6x}. Our TITAN~Xp GPU has 48~KB (or 24~KB\footnote{There is no officially published number and other work reports contradicting numbers.}) of L1 cache per SM. The L1 cache line size is 128~bytes. Data is loaded from the L1 cache in 32~byte transactions, but written in 128~byte transactions~\cite{DBLP:journals/corr/abs-1804-06826}. L1 caches of different SMs are incoherent. The latency of an L1 hit is around 80~cycles~\cite{DBLP:journals/corr/abs-1804-06826}.
  \item \textbf{Shared Memory:} Until the NVIDIA Kepler architecture, L1 and shared memory used the same on-chip memory. In recent architectures, the shared memory is separate. Shared memory is the fastest kind of GPU memory that can be explicitly programmed/accessed in CUDA code. Our TITAN~Xp GPU has 96~KB of shared memory per SM.
  \item \textbf{Registers:} Processor registers are the fastest kind of memory. On Pascal, each SM has 65,536 32-bit registers (256~KB). Since each SM has 64~CUDA cores, an SM can execute exactly two warps at a time. However, more than two warps may be able to \emph{reside} in its register file, depending on the number of registers that each warp uses. To \emph{hide} the latency of (mostly memory) instructions, the SM can then issue instructions from another warp that is not blocked by a dependent instruction (\emph{latency hiding}~\cite{Volkov:EECS-2016-143}).
\end{itemize}

The CUDA Toolkit Documentations mentions two additional memory types. These types of memory are not relevant for our work and will not be mentioned outside of this section.
\begin{itemize}
  \item \textbf{Local Memory:} If an SM does not have enough resources, it uses (thread-)local memory. For example, registers may be spilled to local memory. On Pascal, the number of registers is limited to 255 per thread, but with a large number of threads per block, an SM may run out of registers earlier. Local memory actually resides in device memory and is not a separate address space.
  \item \textbf{Constant Memory:} Memory that does not change throughout a kernel can be stored in constant memory, which is a small part of the device memory but cached by a separate \emph{constant cache}. The size of the constant cache is not published, but believed to be 2~KB per SM on Pascal~\cite{DBLP:journals/corr/abs-1804-06826, nvidia_constant_cache}. Constant memory accesses cannot be coalesced. However, they can be broadcasted to other threads in a warp, should those threads access the same constant memory address. Our TITAN~Xp GPU has 64~KB of constant memory.
\end{itemize}

Details of memory access, the memory hierarchy and memory coalescing (Section~\ref{sec:background_mem_coal}) are constantly changing with new GPU architectures. With every new architecture, NVIDIA publishes a \emph{Tuning Guide} for achieving good performance on that architecture. However, many architectural/hardware details remain undisclosed, so that often the only way of understanding \emph{why} certain programming/memory access patterns achieve good performance is through exhaustive microbenchmarking, profiling and reverse engineering~\cite{5452013, DBLP:journals/corr/abs-1804-06826, 10.1007/978-3-662-44917-2_13, Zhang:2017:UGM:3018743.3018755}. Many architecture/hardware characteristics mentioned in this chapter were originally obtained in this way.

\subsection{CUDA Programming Model}
\label{sec:cuda_prog_model_backgr}
CUDA is an extension of C++. Programmers express thread-level parallelism with CUDA \emph{kernels}. A kernel is a C++ function, annotated with the \texttt{\_\_global\_\_} keyword. When programmers \emph{launch} a CUDA kernel, they have to specify the number of threads that should execute the kernel (C++ function). Inside the kernel, programmers can access special variables to determine a thread's ID, so that different threads run the same code but with different data.

In many cases, a thread ID is used as an index into an array of \emph{jobs}. If the number of threads equals the number of jobs, threads can be directly (one-to-one) mapped to jobs: A thread $t_i$ processes job $i$. However, CUDA limits the number of threads per kernel and, furthermore, using too many threads can cause runtime inefficiencies. If the number of jobs exceeds the number of threads, then each thread may have to process more than one job, typically with a \emph{grid-stride loop}~\cite{cuda_grid_stride}.

Besides the \texttt{\_\_global\_\_} keyword, functions and top-level variable declarations can be annotated with the \texttt{\_\_device\_\_} and/or \texttt{\_\_host\_\_} keyword. The former keyword instructs the compiler to generate GPU code for a function or to statically allocate the variable on \emph{device} (GPU) in global memory. The latter keyword (default if no keyword specified), instructs the compiler to generate CPU code for a function or to statically allocate the variable on the \emph{host} (CPU). In essence, \texttt{\_\_global\_\_} functions are like \texttt{\_\_device\_\_} functions but only the former ones can be invoked as kernels from CPU code. It is not possible to call \texttt{\_\_device\_\_} functions from \texttt{\_\_host\_\_} functions or vice versa.

\paragraph{Thread Hierarchy}
CUDA threads are organized in a hierarchy. Every thread belongs to one CUDA \emph{block}. The size of a block (\emph{block dimension}) can be between 1 and 1024 threads\footnote{We are focusing on 1D blocks and grids.}. The number of blocks (\emph{grid dimension}) can be between 1 and $2^{31} - 1$, so a kernel can have up to billions of CUDA threads. Programmers have to specify block/grid dimensions as kernel launch parameters. Inside a CUDA kernel, programmers can access four special variables.

\begin{itemize}
  \item \texttt{threadIdx.x}: The ID of the current thread within the block.
  \item \texttt{blockDim.x}: The size/dimension of each block.
  \item \texttt{blockIdx.x}: The ID of the block of the current thread.
  \item \texttt{gridDim.x}: The number of blocks (dimension of the \emph{grid}).
\end{itemize}

\noindent The ID of a thread and the total number of threads is then calculated as follows.

\begin{align*}
\mathit{tid} = \mathit{blockIdx.x} \cdot \mathit{blockDim.x} + \mathit{threadIdx.x} \tag{\emph{thread ID}}
\end{align*}
\begin{align*}
\mathit{gridDim.x} \cdot \mathit{blockDim.x} \tag{\emph{number of threads}}
\end{align*}

Besides blocks, there is a second thread hierarchy level that cannot be directly controlled programmers. Every consecutive group of 32 threads of a block (thread IDs $[32 \cdot i; 32 \cdot (i+1)$) is called a \emph{warp}.

\begin{align*}
\mathit{warp\_id} = \left\lfloor \frac{\mathit{threadIdx.x}}{32}\right\rfloor \tag{\emph{warp ID of a thread within a block}}
\end{align*}

\begin{align*}
\mathit{lane\_id} = \mathit{threadIdx.x} \,\,\, \% \,\,\, 32 \tag{\emph{lane ID of a thread}}
\end{align*}

CUDA warps map to physical warps, so threads of a warp execute the same instructions at the same time (SIMD parallelism). The threads within a warp are mapped to \emph{lanes} with IDs between 0 and 31. CUDA provides warp-level primitives for exchanging (shuffling) values between different lanes of a warp~\cite{cuda_website_using_warp_level}. Such primitives are not easy to use and require a certain amount of CUDA programming experience, but they are often necessary to achieve the best performance~\cite{DeGonzalo:2019:AGW:3314872.3314884}.

\paragraph{Block/Warp Scheduling}
Every CUDA block runs on a single SM and cannot be split to run on multiple SMs. If SM resources allow (as outlined below), multiple blocks may reside on a single SM. If the number of CUDA cores of an SM is less than the number of resident threads (i.e., block dimension times number of resident blocks per SM), then not all threads can run simultaneously and the SM has to context-switch between warps. This is usually the case so that memory latency can be hidden~\cite{Volkov:EECS-2016-143}. Among all resident warps (i.e., warps of all resident blocks), the warp scheduler selects a warp that is ready to execute (e.g., not blocked by a memory operation), until eventually the entire block finished executing. The \emph{block scheduler} then selects a new CUDA block for execution (if there are more blocks). Context-switching between warps is extremely fast on GPUs because the state of each warp is kept in the register file, which is much larger on GPUs compared to CPUs, so that registers do not have to be swapped to device memory.

The maximum number of resident blocks per SM depends on the computate capability version, the resource requirements of each block and the available resources of an SM, as described in the CUDA Toolkit Documentation. On our TITAN Xp with compute capability version 6.1, those limitations are as follows.

\begin{itemize}
  \item The total number of registers of all resident blocks must fit into the register file. I.e., all resident blocks together cannot use more than 256~KB of registers.
  \item The total amount of requested shared memory of all resident blocks must not exceed the SM's shared memory, which is 96~KB.
  \item No more than 32~blocks can reside on one SM.
  \item No more than 64~warps can reside on one SM. This limits the maximum number of resident threads per SM to 2048.
\end{itemize}

The NVIDIA Visual Profiler can be used to analyze such constraints. In the past, NVIDIA also provided an \emph{Occupancy Calculator} Excel sheet.

\paragraph{Grid-Stride Loops}
CUDA programs with fewer threads than jobs typically process jobs with a grid-stride loop~\cite{cuda_grid_stride, Khorasani:2015:EWE:2830772.2830796}. In a grid-stride loop, a thread $t_\mathit{tid}$ processes job $\mathit{tid}$ and increments of $n$, where $n$ is the total number of threads.

\begin{lstfloat}
\begin{lstlisting}[language=c++, caption={[Example: Grid-stride loop]Example of a grid-stride loop}, morekeywords={__device__, __global__}, label={lst:grid_stride_example}]
#define N 1000000
__device__ float data[N];

__global__ void kernel_increment_data() {
  for (unsigned int i = threadIdx.x + blockIdx.x * blockDim.x;
       i < N; i += blockDim.x * gridDim.x) { data[i] += 10.0f; }
}

int main() { kernel_increment_data<<</*gridDim.x=*/ 128, /*blockDim.x=*/ 128>>>(); }
\end{lstlisting}
\end{lstfloat}

Listing~\ref{lst:grid_stride_example} shows an example of a grid-stride loop. Grid-stride loops often have better performance than a different assignment of jobs to threads because of better memory coalescing (Section~\ref{sec:background_mem_coal}). This is because the threads of a warp process jobs that are spatially local in memory and can thus be serviced with more efficient vectorized memory loads/stores. This listing also illustrates CUDA's kernel launch notation, which allows programmers to specify the number of blocks/threads (Line~9).

\paragraph{Warp Divergence}
CUDA warps are mapped to physical warps of an SM. Since physical warps only have one warp scheduler, all threads of a warp execute the same instructions. If the control flow of the threads of a warp diverges (e.g., because they enter different \emph{if} branches), then the different control flow paths are executed sequentially until the control flow converges again. This phenomenon is called \emph{warp divergence} (also called \emph{thread divergence} or \emph{branch divergence}).

There are no guarantees about the order in which divergent branches are executed or when the control flow reconverges again~\cite{cuda_website_using_warp_level}. Warp divergence can make GPU programs significantly harder to reason about and is the main reason why locking is discouraged on GPUs.

As an example, consider the two implementations of a critical section in Listing~\ref{sec:impl_critic_sect}. Both implementations seem identical, but the first one deadlocks on our TITAN~Xp GPU and the second one works~\cite{Li:2015:FSD:2751205.2751232}\footnote{See also: \url{https://stackoverflow.com/questions/31194291/cuda-mutex-why-deadlock}.}. In fact, whether either one of them deadlocks or not may vary among GPU architectures and compilers.

\begin{lstfloat}
\begin{lstlisting}[language=c++, caption={Implementation of a critical section}, label={sec:impl_critic_sect}, morekeywords={__global__, __device__, __host__}]
__device__ int lock = 0;
__device__ int result = 0;  // Modified in critical section

__global__ void deadlock() {
  while (true) {
    if (atomicCAS(&lock, 0, 1) == 0) {  // lock
      result += 0x7777;
      atomicExch(&lock, 0);  // unlock
      break;
    }
  }
}

__global__ void no_deadlock() {
  bool continue_loop = true;
  while (continue_loop) {
    if (atomicCAS(&lock, 0, 1) == 0) {  // lock
      result += 0x7777;
      atomicExch(&lock, 0);  // unlock
      continue_loop = false;
    }
  }
}

int main() {
  deadlock<<<1, 32>>>();  // Run with 32 threads, 1 warp
}
\end{lstlisting}
\end{lstfloat} 

\begin{figure}
	\centering
  \includegraphics[width=\textwidth]{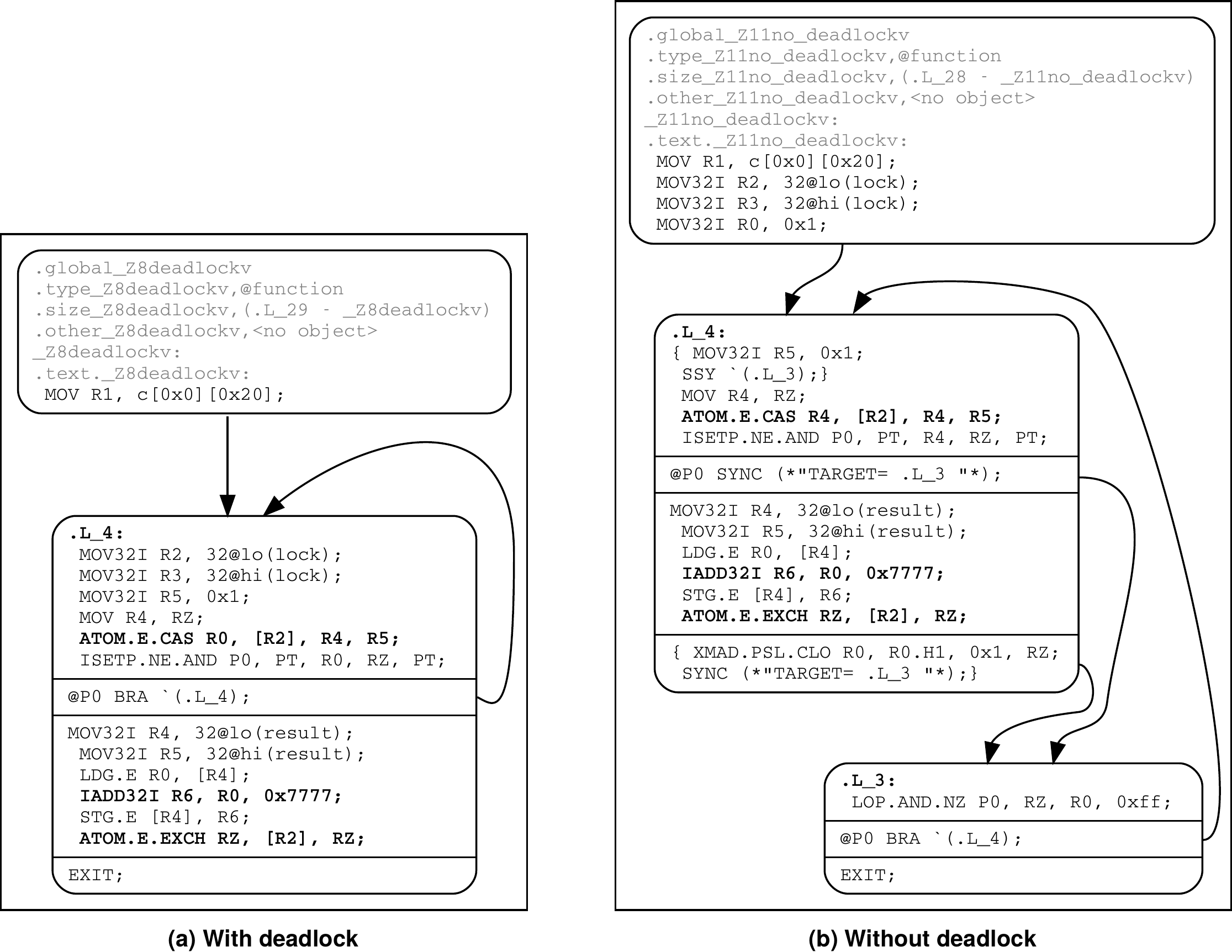}
 	\caption[Basic block diagram for critical sections in CUDA]{Basic block diagram for critical sections}
 	\label{fig:basic_block_crit}
\end{figure}

To understand why the first implementation deadlocks, we have to take a look at the basic block structure of the compiled PTX (Figure~\ref{fig:basic_block_crit}\textsc{a}). Each thread must aquire the lock in \texttt{L\_4} to enter the critical section and then release the lock in the third segment of \texttt{L\_4}, so that the next thread can enter the critical section. PTX instructions that start with an \emph{at} character are \emph{guarded} and conditionally executed (based on the result of the preceding instructions). The problem in (\textsc{a}) is that one thread acquires the lock, but then the other 31 unsuccessful threads jump back to the beginning of the loop. This is a control flow divergence and both control flows are executed sequentially. Unfortunately, our TITAN~Xp happens to always schedule the jump of the unsuccessful threads before executing the critical section with the successful thread, so the lock is never released.

Now consider (\textsc{b}). Before jumping back to the beginning of the loop, the compiler places a \texttt{SYNC} instruction, which converges threads after a conditional branch. If control flow is divergent, the GPU schedules the other not yet executed branches before continuing with the following instructions. Therefore, the successful thread now has a chance to execute the critical section and to release the lock.


Unfortunately, we do not know \emph{why} the compiler generates a \texttt{SYNC} instruction in (\textsc{b}) but not in (\textsc{a}). This is unspecified and may even change with different CUDA Toolkit versions~\cite{ElTantawy:2016:MSS:3195638.3195652}. Interestingly, if we invert the boolean flag \texttt{continue\_loop}\footnote{Initialize to \texttt{false}, check for \texttt{!continue\_loop} in the loop condition, set to \texttt{true} in the critical section.}, then our CUDA compiler generates the same deadlock-prone PTX as in (\textsc{a}). This example illustrates why locking within a warp is problematic on GPUs\footnote{If only one thread per warp acquires a lock, locking is unproblematic.}. We avoid locking in our work as much as possible and resort to lock-free algorithms. 

\paragraph{Consistency Model}
CUDA has a weak memory consistency model. In particular, global (device) memory accesses are not necessarily sequentially consistent. E.g., if a thread $t_1$ writes two variables $a$ and $b$, then another thread $t_2$ is not guaranteed to see the changes of $a$ and $b$ in the same order. Moreover, there are no guarantees that $t_2$ will see the changes of $a$ and $b$ at all. Such weak behaviors are in part due to incoherent L1 caches. For example, the old value of $a$ may still be in $t_2$'s L1 cache while $b$ is a cache miss. In such a case, $t_2$ would read the old value of $a$ but the new value of $b$. Unfortunately, weak GPU behavior is not well documented by GPU manufacturers, so that researchers have to resort to reverse engineering or litmus testing to fully understand GPU concurrency~\cite{Alglave:2015:GCW:2694344.2694391, Sorensen:2016:EER:2908080.2908114}.

Weak semantics on GPUs and how to avoid them has been discussed in previous work~\cite{Singh:2015:EES:2830772.2830778}. There are four main ways of addressing weak behavior in CUDA.
\begin{itemize}
  \item A \textbf{thread fence} (CUDA \texttt{\_\_threadfence()}) ensures that all global memory changes of the current thread before the thread fence become visible to other threads before global memory changes after the thread fence become visible. After executing the thread fence, the L1 cache of the executing thread is guaranteed to be consistent with the L2 cache.
  \item \textbf{Atomic operations} bypass the L1 cache and go straight to the L2 cache. Atomics can ensure sequential consistency and visibility of changes if data is both read/written with atomics. Atomic instructions have a higher latency than comparable non-atomic instructions, but they became considerably faster with recent GPU architectures~\cite{Gaihre:2019:XER:3307681.3326606, DBLP:journals/corr/abs-1804-06826}.
  \item The CUDA/C++ \textbf{\texttt{volatile} qualifier} indicates that a memory address may be concurrently read/written by another thread. It disables certain optimizations (such as keeping a value in a register) and always causes a memory read/write upon data access, bypassing the L1 cache.
  \item The \textbf{L1 cache} can also be entirely \textbf{turned off} with a compiler flag: \texttt{-Xptxas -dlcm=cg}. Note that this is not equivalent to making every memory access a \texttt{volatile} access, since the \texttt{volatile} qualifier also deactivates other optimizations (such as caching data in registers).
\end{itemize}

In this thesis, we develop various lock-free algorithms. Their implementation depends heavily on atomic operations, which are sequentially consistent. Atomic operations and retry loops are common patterns in lock-free algorithms~\cite{trevor_brown_phd}.

\subsection{Memory Coalescing}
\label{sec:background_mem_coal}
Many GPU applications are memory bound and for such applications, accessing device memory is the biggest bottleneck. \emph{Memory coalescing} is a fundamental optimization of the GPU memory controller that combines multiple memory requests into a smaller number of physical memory transactions, if certain conditions are met.

\begin{itemize}
  \item This optimization is specific to \textbf{device memory} (i.e., CUDA global memory).
  \item Only memory accesses (reads/writes) from the \textbf{same warp} can be combined.
  \item Memory addresses must be properly aligned and fall into an L1/L2 \textbf{cache line}.
\end{itemize}

The rules of memory coalescing differ among different NVIDIA compute capability versions. With compute capabilities between 3.x and 5.x, the NVIDIA CUDA C Programming Guide states that hits in the L1 cache can be coalesced into 128-byte transactions~\cite{nvidia_global_mem_6x}. I.e., if the threads of a warp access memory locations on the same L1 cache line (L1 cache lines are 128~bytes long and 128-byte aligned), then these accesses are serviced by one physical 128-byte transaction. Starting from compute capability 6.x (Pascal), the NVIDIA Pascal Tuning Guide states that the L1 cache services loads (but not stores) at a granularity of 32~bytes~\cite{nvidia_unified_l1_6x}, so memory loads can only be coalesced into smaller 32-byte transactions\footnote{This is presumably to avoid \emph{overfetching}, which required disabling the L1 cache for global memory access before Pascal.}.

In the worst case, if each thread of a warp loads a 32-bit word and those memory accesses cannot be coalesced, then each thread's access generates a 32-byte transaction, which increases the amount of memory transfer by 8~\cite{nvidia_dev_mem_access_cuda_c}, compared to one perfectly coalesced 128-byte transaction (CC 3.x-5.x) or four perfectly coalesced 32-byte transactions (CC 6.x).

Memory accesses of a warp that hit in the L2 cache and belong to the same L2 cache line, are coalesced into 32-byte transactions~\cite{nvidia_global_mem_6x}. Whether these memory addresses must correspond to the same 32-byte L2 cache line segment is not specified, but this is likely the case.

All device memory is accessed through the L1/L2 caches. If a memory load does not hit in the L1/L2 caches, then the data must first be loaded from the GPU's DRAM into the L2 cache (and maybe L1 cache). The DRAM is accessed with 32-byte transactions. In the worst case, if the uncached 32-bit word accesses of each thread in a warp cannot be coalesced, then each 32-bit word access triggers a 32-byte DRAM transaction, which increases the amount of memory transfer by~8.

Memory coalescing is one of the most fundamental optimizations. The CUDA C Best Practices Guide puts a \emph{high priority note} on coalesced access~\cite{nvidia_memoryco} and programmers should try to achieve good coalescing before trying any other optimizations.

The details of memory access and memory coalescing are changing with every new GPU architecture. While, until now, the number of generated memory transactions served as a good mental model to explain why memory coalescing increases performance, this no longer seems to be the case with the most recent architectures\footnote{See also: \href{https://stackoverflow.com/questions/56142674/memory-coalescing-and-nvprof-results-on-nvidia-pascal}{\texttt{https://stackoverflow.com/questions/56142674/memory-coalescing-and-nvprof- results-on-nvidia-pascal}}.}. In fact, NVIDIA removed certain metrics and performance counters from their profiling tools, so that those values are no longer exposed to programmers. While we can no longer clearly explain \emph{how} the hardware optimizes certain memory access patterns, the basic CUDA programming rules for achieving good memory coalescing and good memory bandwidth utilization fortunately remain the same.


\paragraph{Memory Coalescing vs. Vectorized Access}
Memory coalescing on GPUs is similar to vectorized memory access on CPUs but differs in one crucial aspect: Memory coalescing is a run-time optimization of the memory controller, whereas vectorization is a compile-time optimization\footnote{See also: \href{https://stackoverflow.com/questions/56966466/memory-coalescing-vs-vectorized-memory-access}{\texttt{https://stackoverflow.com/questions/56966466/memory-coalescing-vs-vectorized -memory-access}}}.

On CPU systems, there are two main techniques for utilizing SIMD parallelism: Either the programmer manually vectorizes the code with SIMD intrinsics, which is tedious. Or the compiler auto-vectorizes code as part of an optimization pass. Both techniques require the programmer/compiler to fully understand the memory access pattern. Otherwise, the code cannot be vectorized. The resulting vectorized assembly code contains vector load/stores such as \texttt{movaps}, which, given a memory address, loads four packed (consecutive) floats (128~bits) into an SSE vector register.

On GPU systems, the memory controller coalesces memory accesses by analyzing the requested memory addresses at runtime. There is nothing to do for the compiler. The resulting PTX assembly code contains scalar loads/stores such as \texttt{ld.global}, which, given one memory address per thread, performs a memory load for each thread in the warp. The memory controller analyzes the memory addresses and determines the number of times that this instruction has to be repeatedly executed such that all memory load requests in the warp can be fulfilled with 128-byte transactions.

\subsection{Memory Coalescing Experiment}
\label{sec:memory_coalescing_experiment_background}
To analyze the exact benefit of memory coalescing on a TITAN~Xp GPU, we implemented a simple benchmark that increments every value of a large 32-bit floating point numbers array. The size of the array is 8~GB. We measured the runtime performance with different assignments of array slots (jobs) to CUDA threads.

\paragraph{Experiment Setup}
\begin{figure}
  \centering
  \includegraphics[width=0.65\textwidth]{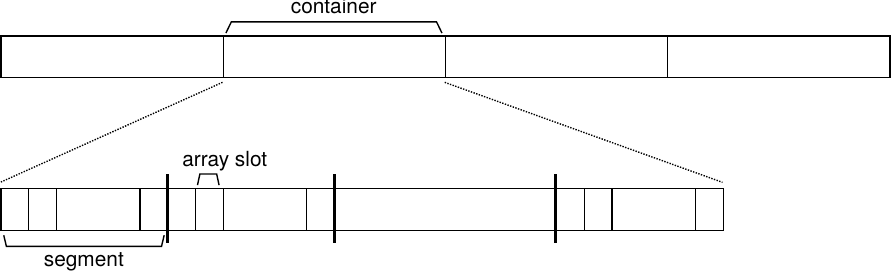}
  \caption[Array layout of memory coalescing experiment in CUDA]{Array layout of memory coalescing experiment}
  \label{fig:coalescing_experiment_sec3}
\end{figure}

The array of size (number of elements) $N$ is divided into $\mathit{num}_C \in \{1, 2, 4, 8, 16, 32\}$ containers of equal size, where $\mathit{num}_C$ is a configurable parameter. Each container is divided into segments of equal size (number of elements) $\mathit{size}_S$, where $\mathit{size}_S$ depends on the parameter $\mathit{num}_C$ (Figure~\ref{fig:coalescing_experiment_sec3}).

\begin{align*}
\mathit{size}_S = \frac{32}{\mathit{num}_C} \tag{\emph{segment size}}
\end{align*}

For the moment, let us assume that the number of array slots $N$ is equal to the number of GPU threads. We assign CUDA threads to array slots in such a way that every warp is processing one segment per container\footnote{Recall that the warp size is 32. Moreover, $\mathit{size}_S \cdot \mathit{num}_C = 32$.}. For example, for $\mathit{num}_C=4$, the threads of a warp are assigned to one segment (size $8$) per container. These segments are in totally different memory areas, so at least four memory transactions are required for a warp to load all assigned array values. By increasing $\mathit{num}_C$, we can increase the degree of \emph{scattering} and thus the number of required memory transactions. In the best case ($\mathit{num}_C=1$), each warp is accessing only consecutive array slots, which should result in perfect memory coalescing and best performance.

\paragraph{Access Pattern Details}
Let $\mathit{pos}$ be the array slot index that a given thread $t_{\mathit{tid}}$ is processing. This index is the sum of three parts: The offset into the array at which the respective container begins, the offset into the container at which the respective segment begins and the offset of the respective array slot within the segment.

\begin{align*}
\mathit{pos}(\mathit{tid}) = \mathit{array\_offset}(\mathit{tid}) + \mathit{container\_offset}(\mathit{tid}) + \mathit{segment\_offset}(\mathit{tid}) \tag{\emph{array slot}}
\end{align*}

The size of a container $\mathit{size}_C$ depends on the number of containers. All containers have the same number of array elements.

\begin{align*}
\mathit{size}_C = \frac{N}{\mathit{num}_C} \tag{\emph{container size}}
\end{align*}

Each warp consists of 32 threads. Within a warp, threads are numbered from 0 to 31 (\emph{lanes}). Containers are divided equally among lanes. $\mathit{size}_S$ is the number of lanes per container. For example, container 0 is processed by lanes 0 through $\mathit{size}_S - 1$ of all warps. The offset of the container for thread $t_\mathit{tid}$ is then computed as follows.
\begin{align*}
\mathit{array\_offset}(\mathit{tid}) & = \left\lfloor \frac{\mathit{lane\_id}(\mathit{tid})}{\mathit{size}_S} \right\rfloor \cdot \mathit{size}_C \tag{\emph{array offset}} \\
 & = \left\lfloor \frac{\mathit{tid} \,\, \% \,\, 32}{32 \,\, / \,\, \mathit{num}_C} \right\rfloor \cdot \frac{N}{\mathit{num}_C}
\end{align*}

Within a container, each segment is processed by the $\mathit{size}_S$ threads of each warp. Segment $i$ is processed by warp $i$. The offset of the segment for thread $t_\mathit{tid}$ is computed as follows. 
\begin{align*}
\mathit{container\_offset}(\mathit{tid}) & = \mathit{warp\_id}(\mathit{tid}) \cdot \mathit{size}_S \tag{\emph{container offset}} \\
 & = \left\lfloor \frac{\mathit{tid}}{32} \right\rfloor \cdot \frac{32}{\mathit{num}_C}
\end{align*}

Finally, among the $\mathit{size}_S$ threads per segment, we assign one thread to each array slot. The offset of the array slot within the segment for $t_\mathit{tid}$ is computed as follows.
\begin{align*}
\mathit{segment\_offset}(\mathit{tid}) & = \mathit{tid} \,\, \% \,\, \mathit{size}_s \tag{\emph{segment offset}} \\
 & = \mathit{tid} \,\, \% \,\, \frac{32}{\mathit{num}_c}
\end{align*}

\begin{lstfloat}
\begin{lstlisting}[language=c++, caption={Implementation of memory coalescing experiment}, morekeywords={__device__, __global__}, label={lst:full_source_code_bench_sec3}]
__global__ void kernel(float* volatile data, unsigned int N, unsigned int C) {
  for (unsigned int tid = threadIdx.x + blockIdx.x * blockDim.x;
       tid < N; tid += blockDim.x * gridDim.x) {
    unsigned int index = ((tid % 32) / (32 / C)) * (N / C)   // array offset
                       + (tid / 32) * (32 / C)               // container offset
                       + tid % (32 / C);                     // segment offset
    data[index] += 8.0f;  // 1 global memory read, 1 global memory write.
  }
}

int main() {
  float* data;
  cudaMalloc(&data, sizeof(float) * 2ULL * 1024 * 1024 * 1024);

  kernel<<<512, 512>>>(data, /*N=*/ 2ULL * 1024 * 1024 * 1024, /*C=*/ 4);
  cudaDeviceSynchronize();
}
\end{lstlisting}
\end{lstfloat}

Listing~\ref{lst:full_source_code_bench_sec3} shows the source code of the benchmark. Threads process array elements with a grid-stride loop, because spawning one thread per array slot would be too much overhead. This does not affect memory coalescing.

\paragraph{Results}
Figure~\ref{fig:sec3_bench_results} shows the results of the benchmark. We ran the benchmark with six different parameters $\mathit{num}_C \in \{1, 2, 4, 8, 16, 32\}$ and with a disabled L1 cache. Higher values of $\mathit{num}_C$ increase the degree of scattering, which reduces memory coalescing and increases the running time (Subfigure~\textsc{a}). The worst-case running time of $\mathit{num}_C = 32$ is 29.8 times slower than the best-case running time of $\mathit{num}_C = 1$.

Subfigure~\textsc{b} shows the number of bytes read/written from the GPU's DRAM. For $\mathit{num}_C \in \{1,2,4\}$, this value is almost the same because the DRAM is always accessed in 32-byte transactions. Only for higher values of $\mathit{num}_C$ does the number of bytes increase. In the worst case ($\mathit{num}_C = 32$), 127.4~GB are read/written from the DRAM. In this case, every 4-byte read/write access was serviced by a 32-byte transaction.

\begin{figure}
\subfloat[Running time]{\includegraphics[width=0.48\textwidth]{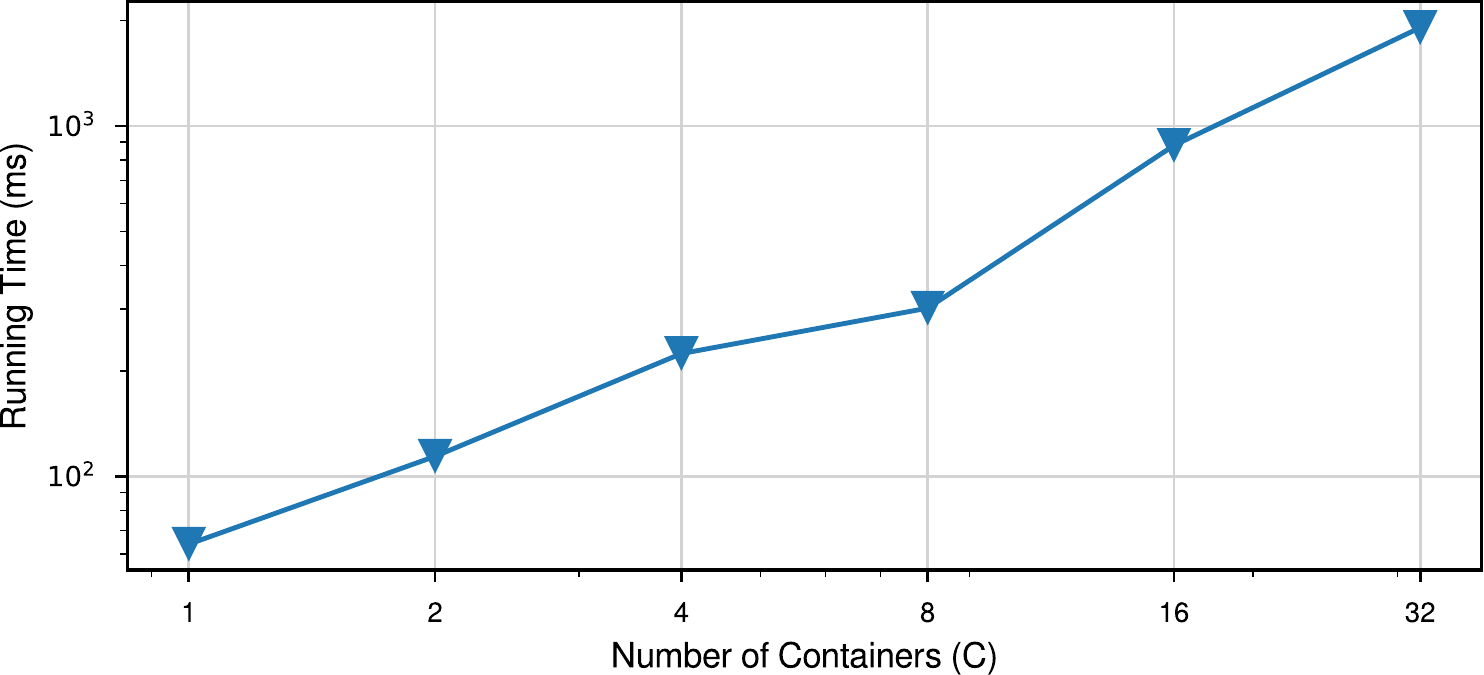}} \hfill
\subfloat[DRAM bytes read/written]{\includegraphics[width=0.48\textwidth]{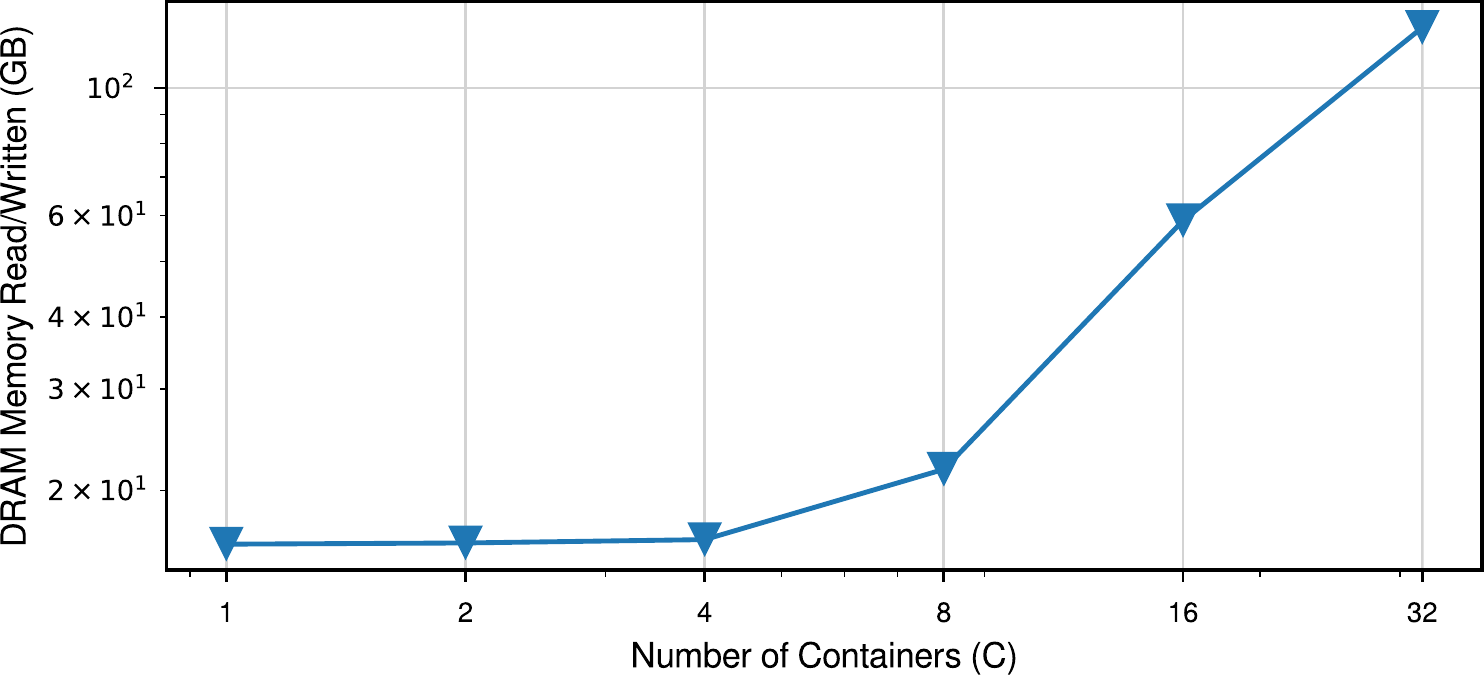}}
\caption[Memory coalescing experiment results]{Memory coalescing experiment results (lower is better)}
\label{fig:sec3_bench_results}
\end{figure}

\begin{align*}
2 \cdot \frac{8 \cdot 1024 \cdot 1024 \cdot 1024 \mbox{ byte}}{4 \mbox{ byte}} \cdot 32 \mbox{ byte} = 128 \mbox{ GB} \tag{\emph{exp. \#bytes accessed in DRAM}}
\end{align*}

The measured number of DRAM bytes is a bit less than 128~GB because a few transactions hit in the in the L2 cache and profiler metrics are not 100\% accurate.

The results of this experiment highlight the importance of coalesced memory access and show that significant running time speedups can be achieved by optimizing memory access. Our work focuses largely on achieving good memory coalescing with object-oriented programming on GPUs.

\section{Object-oriented Programming}
\label{sec:bg_oop}
Object-oriented programming (OOP) is a widely used programming paradigm. As of May 2019, 16 of the 20 most widely used programming languages (according to the TIOBE index\footnote{\url{https:
//www.tiobe.com/tiobe-index/}}) support object-oriented programming.

\paragraph{Objects}
In well designed programs, objects represent abstract or real-world entities. This can help programmers in building a good mental model of the data structures and code interactions in a program. An object in OOP consists of three fundamental parts.
\begin{itemize}
  \item \textbf{State:} The state of an object defines its properties. In pure OOP, the state of an object is \emph{private} and cannot be directly accessed by other objects.
  \item \textbf{Identity:} Even if two different objects have the same state, they have different identity. Identity can be seen as a special property that is different for each object.
  \item \textbf{Behavior:} In pure OOP, objects communicate by exchanging messages. Upon receipt, an object may process the message by executing code, which may in turn read/modify the object's state. Mainstream languages follow the notion of \emph{method calls} instead of \emph{message sends}.
\end{itemize}

Object-oriented programming is widely used in academia and industry. Some of its main advantages are good abstraction, encapsulation, modularity, good code understandability and high programmer productivity~\cite{AL:P20150611005-201503-201506110032-201506110032-39-46}.

\subsection{Class-based Object-oriented Programming}
The most prominent variant of object-oriented programming is \emph{class-based object-oriented programming}. Objects are instances of classes, which define the state (fields) and behavior (methods) of objects. CUDA and OpenCL are extensions of C++, a class-based, object-oriented programming language.

Class-based, object-oriented programming languages can increase code reuse and abstraction through subclassing/inheritance. In typed languages, a subclass relationship usually creates a subtype relationship. This means that, with respect to type safety, an object of a superclass can be substituted with an object of a subclass (\emph{Liskov Substitution}~\cite{Liskov:1994:BNS:197320.197383}). If a subclass \emph{overrides} a method of a superclass, the programming system should dispatch to the correct method based on an object's runtime type as opposed to its static type (\emph{virtual function call}).

\newsavebox{\sourcecpplayout}
\begin{lrbox}{\sourcecpplayout}
\begin{lstlisting}[numbers=none,linewidth=0.37\columnwidth,language=C++,morekeywords={override}]
class A {
  int f1; double f2; char f3;
  void foo();
  virtual void bar();
};

class B : public A {
  int f4;
  virtual void bar() override;
  virtual void qux();
};
\end{lstlisting}
\end{lrbox}

\begin{figure}
\subfloat{\begin{minipage}{.37\columnwidth}\vspace{-2.25cm}\usebox{\sourcecpplayout}\end{minipage}}\hfill
\subfloat{\includegraphics[width=0.6\textwidth]{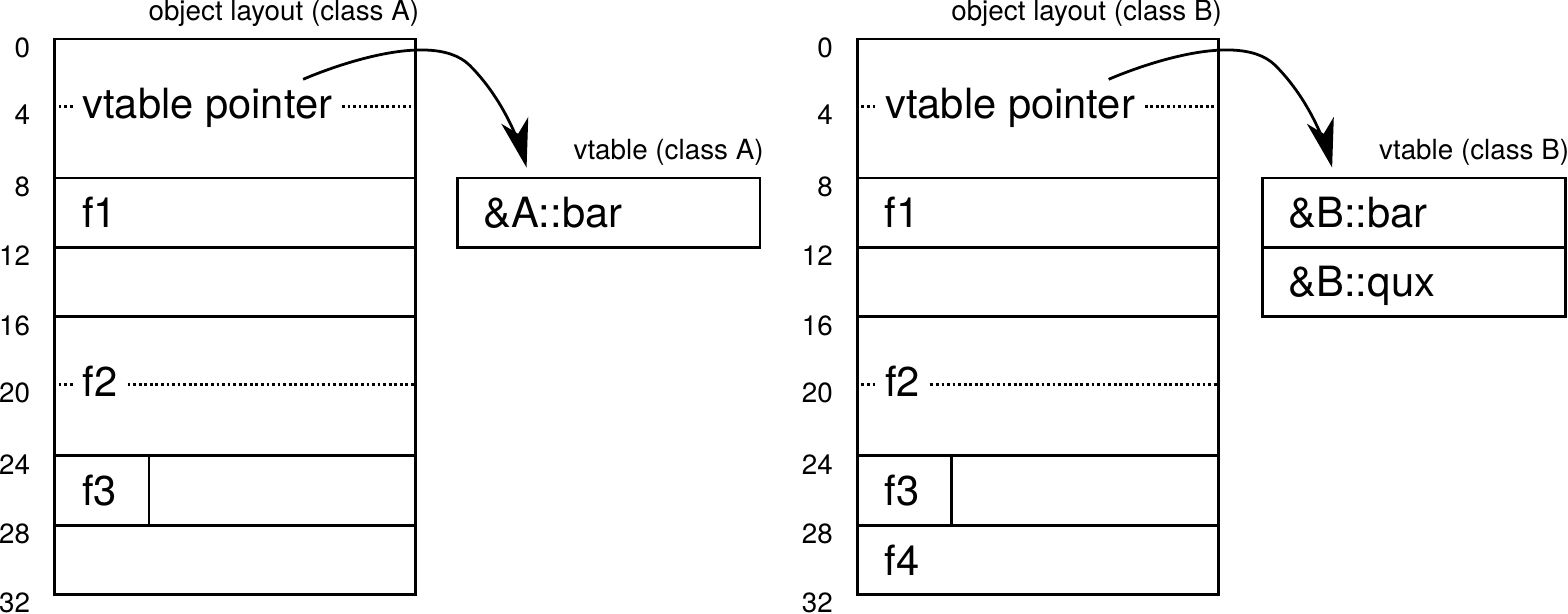}}\hfill
\caption{Example: C++ object layout on NVIDIA GPUs}
\label{fig:obj_layout_cpp_abi}
\end{figure}

\paragraph{C++ Object Layout}
The layout of objects in memory depends on the programming language/compiler and the platform's application binary interface (ABI)~\cite{abi_spark_ref1}. In C++, field values of an object are arranged in memory in the same order in which they are declared in the struct/class. If a struct/class inherits from another struct/class, then inherited fields come first\footnote{We are focusing on single inheritance. There are more complex rules for multiple inheritance.}. Furthermore, if the struct/class has at least one virtual member function, then the object starts with a pointer to the virtual method table (\emph{vtable}), which contains pointers to all virtual function implementations. The ABI defines the size and alignment of primitive types. For example, on recent NVIDIA GPUs, \texttt{char} is 1~byte, \texttt{int} is 4~bytes and \texttt{double} is 8~bytes. In addition, values of these types must be aligned to their respective byte size~\cite{cuda_abi}.

Figure~\ref{fig:obj_layout_cpp_abi} shows an example with two classes, where class \texttt{B} inherits from class \texttt{A}. Due to alignment, objects of both classes are 32~bytes in size, even though class \texttt{B} has an extra field. In particular, class \texttt{A} has a size of 32~bytes instead of 25~bytes. This is to ensure that if multiple \texttt{A} objects are allocated in an array, field values of each object are properly aligned.

\subsection{Problems of Object-oriented Programming on GPUs}
Even though object-oriented programming has a wide range of applications in high-performance computing~\cite{bandini2009, Kale:1993:CPC:165854.165874, allan2010survey, 10.1007/978-3-540-25934-3_2, CARY199720} it is often avoided due to bad performance~\cite{master_th_patel}. In C++, there are four main performance problems with object-oriented programming on GPUs.

\begin{itemize}
  \item \textbf{Data Layout:} While the object-oriented programming paradigm does not dictate any particular layout of objects in memory, most systems/compilers store each object in one contiguous block of memory. Such a layout may not be ideal for execution on GPUs. To achieve good memory access performance, GPU programmers have to tune data layouts to maximize memory coalescing and L1/L2 cache performance~\cite{5363313}, but current GPU languages such as CUDA and OpenCL do not allow programmers to change the layout of objects.
  \item \textbf{Dynamic Memory Allocation:} The ability/flexibility of creating and deleting objects at any time is one of the corner stones of object-oriented programming. CUDA provides a dynamic memory allocator (\texttt{malloc}/\texttt{free} interface) in GPU code, but this allocator is notoriously slow and unreliable~\cite{6339604}. Due to the massive number of threads and expensive inter-thread communication/synchronization, it is difficult to design efficient, dynamic memory allocators for GPUs.
  \item \textbf{Virtual Function Calls:} GPU compilers aggressively inline functions, because jumps and function calls are generally more expensive on GPUs~\cite{Wu:2016:GOG:2854038.2854041, 8366940}. In C++, virtual functions are typically compiled into a jump to an address in the virtual method table (\emph{vtable}). However, such jumps cannot be inlined by the compiler and are by a factor of 10x slower than regular function calls~\cite{cuda_virtual_methods_vtk}. Moreover, virtual function calls can lead to warp divergence.
  \item \textbf{64-Bit Pointers:} Memory pointers on recent GPU architectures are 64~bits long. Therefore, if an object stores a pointer to another object, an 8-byte field is required. Without object-oriented abstractions (i.e., object pointers), a 4-byte integer ID may be sufficient, so object-oriented programming can indirectly increase the memory usage. Similar overheads have been observed in Java applications when switching from a 32-bit address space to a 64-bit address space. Related work proposed \emph{pointer compression} to refer to objects with 32-bit values~\cite{10.1007/978-3-540-73589-2_5}. Similar techniques could be adopted to optimize GPU programs.
\end{itemize}

Related work describes techniques for optimizing the last two problems. In the course of this work, we are addressing the first two problems, which are largely unsolved and the main source of slowdowns: Data layout and dynamic memory allocation.

\paragraph{Optimizing Virtual Function Calls}
Virtual function calls are expensive because they are usually not inlined and compiled into jump/call instructions. Such jumps are particular expensive on GPUs, which were originally designed for highly regular control flow and neither have a branch target predictor nor execute code speculatively. We briefly review two techniques for inlining virtual function calls in C++:

\begin{itemize}
  \item \textbf{Switch-Case Statements:} Instead of generating a vtable pointer lookup and a jump, generate an exhaustive switch-case statement that dispatches to the correct method implementation based on the receiver's type (example in Listing~\ref{lst:handwritten_virt_meth_call}). This is possible because C++ is a statically-typed language and GPU programs are usually not separately compiled (i.e., one compilation unit). Therefore, the compiler knows all types that a receiver can have at runtime. Related work describes an instrumentation-based variant of this technique that works even in the absence of full compile-time type information~\cite{10.1007/BFb0053060}.
  \item \textbf{Expression Templates~\cite{Veldhuizen:1996:ET:260627.260749}:} This advanced C++ template metaprogramming technique encodes the structure of a computation in its type. It is heavily utilized in linear algebra libraries~\cite{doi:10.1137/110830125}.
\end{itemize}

As our focus is on efficient memory access, we do not attempt to optimize virtual function calls in this thesis. Instead, we hand-write the respective switch-case statements whenever we encounter a virtual function call in our examples.

\section{Array of Structure (AOS) vs. Structure of Arrays (SOA)}
\label{sec:o_aos_vs_soa}
Structure of Arrays (SOA) and Array of Structures (AOS) describe memory layouts for a fixed-size set of objects~\cite{intel_aos_soa}. In AOS, the standard layout of most systems, objects are stored as contiguous blocks of memory. In SOA, all values of a field are stored together (Figure~\ref{fig:nbody_aos_soa_ecoop}).


\begin{figure}
  \includegraphics[width=\textwidth]{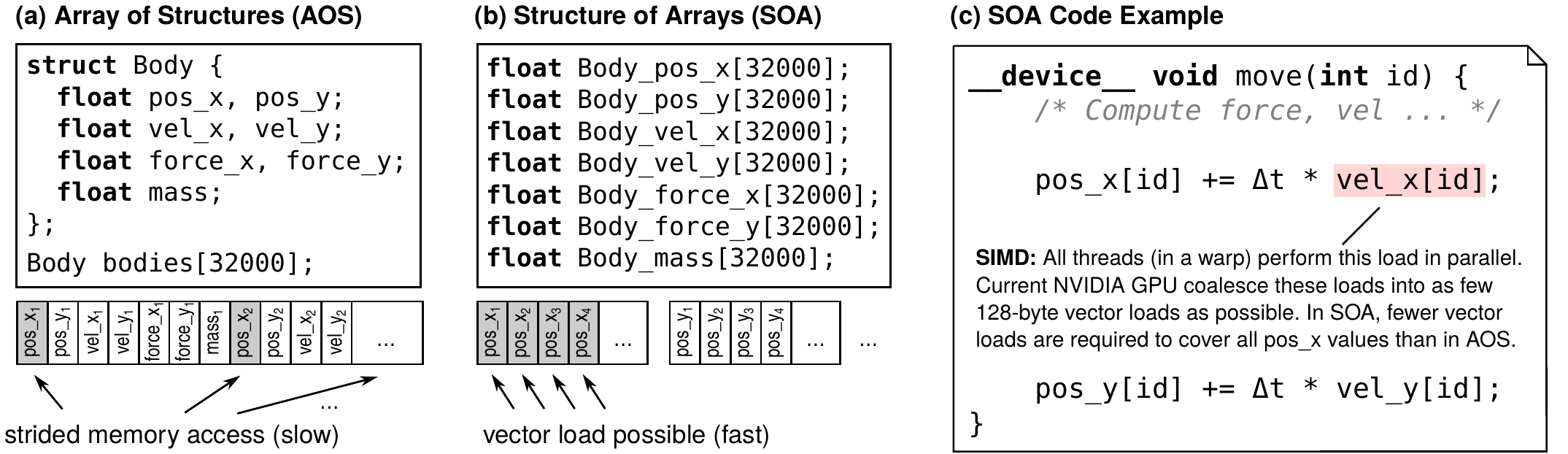}
  \caption{N-body simulation in AOS and SOA data layout}
  \label{fig:nbody_aos_soa_ecoop}
\end{figure}

\paragraph{Running Example}
As an example, we consider a simplified n-body simulation (full example in Section~\ref{sec:smmo_nbody_sec71}). Every body in the simulation is an object with fields for 2D position, 2D velocity, mass, etc. Listing~\ref{lbl:aos_layout} shows an excerpt of the source code in AOS layout. All bodies are stored in an array of \texttt{Body} base type, where \texttt{Body} is a C++ class/struct; thus the name \emph{Array of Structures}. Objects are constructed with C++'s placement-new syntax, which runs the constructor on a given memory address.

Alternatively, instead of allocating objects in an array, objects could be dynamically allocated on the heap. While objects were still stored as contiguous blocks of memory, there would be no guarantee that the dynamic memory allocator places them next to each other in an array-like form. Therefore, such a layout is no longer AOS and has likely performance characteristics different from AOS.

\begin{lstfloat}
\begin{lstlisting}[linewidth=\textwidth, language=c++, numbers=left, caption={N-body simulation in AOS layout}, label={lbl:aos_layout}, morekeywords={assert}]
class Body {
 public:
  float pos_x = 0.0;  float pos_y = 0.0;
  float vel_x = 1.0;  float vel_y = 1.0;
  float force_x;  float force_y;
  float mass;

  // Class constructor
  Body(float mass, float x, float y) : mass_(mass), pos_x(x), pos_y(y) {}

  void move(float dt) {
    pos_x = pos_x + vel_x * dt;
    pos_y = pos_y + vel_y * dt;
  }
};

Body bodies[50];
// Counter for number of objects (instances).
int Body_inst = 0;

void create_and_move() {
  Body* b = bodies + bodies_inst++;
  // Call constructor with C++ placement new.
  new(b) Body(150.0, 1.0, 2.0);

  b->move(0.5);
  assert(b->pos_x == 1.5);
}
\end{lstlisting}
\begin{lstlisting}[linewidth=\textwidth, numbers=left, language=c++, caption={N-body simulation in hand-written SOA layout},morekeywords={assert}, label={lbl:soa_layput}]
float Body_pos_x[50];  float Body_pos_y[50];
float Body_vel_x[50];  float Body_vel_y[50];
float Body_force_x[50];  float Body_force_y[50];
float Body_mass[50];

// Counter for number of objects (instances).
int Body_inst = 0;

int new_Body(float mass, float x, float y) {
  int id = Body_inst++;
  Body_vel_x[id] = 1.0;
  Body_vel_y[id] = 1.0;
  Body_mass[id] = mass;
  Body_pos_x[id] = x;
  Body_pos_y[id] = y;
  return id;
}

void Body_move(int id, float dt) {
  Body_pos_x[id] += Body_vel_x[id] * dt;
  Body_pos_y[id] += Body_vel_y[id] * dt;
}

void create_and_move() {
  int b = new_Body(150.0, 1.0, 2.0);
  Body_move(b, 0.5);
  assert(Body_pos_x[b] == 1.5);
}
\end{lstlisting}
\end{lstfloat}

Listing~\ref{lbl:soa_layput} shows the same program in SOA layout. The array of base type \texttt{Body} was replaced by seven arrays, one for every field; thus the name \emph{Structure of Arrays}. We call these arrays \emph{SOA arrays}.

\subsection{Abstractions for Object-oriented Programming}
At first sight, the SOA code is harder to read/understand than the AOS code. This is mainly because most programmers are familiar with object-oriented programming and its syntax/notation. However, even objectively, the AOS code does a better job at expressing the object-oriented design:

\begin{itemize}
  \item \textbf{Abstraction and Encapsulation:} Values that belong together are defined together. In AOS, \texttt{pos\_x}, \texttt{pos\_y}, etc. are defined and encapsulated inside the scope of class \texttt{Body}. Furthermore, \texttt{Body}'s fields could be made private, such that they can only be accessed from within the class.
  \item \textbf{Type System:} In AOS, objects can be referred to with class pointers. In SOA, they are referred to with integer IDs. In AOS, the type system can catch programming mistakes early on, while, in SOA, many programming mistakes can remain unnoticed until runtime.
  \item \textbf{C++ OOP Abstractions:} AOS allows C++ programmers to use C++ abstractions for object-oriented programming, such as constructor syntax, the \texttt{new} keyword, inheritance, member visibility, member function (method) calls and virtual functions. This is not possible with SOA code.
\end{itemize}

There are programming languages that allow programmers to specify custom memory layouts without breaking abstractions. \emph{Shapes}~\cite{Franco:2017:YAG:3133850.3133861} is one example for such a language. However, we are focusing on C++/CUDA in this work because GPUs are predominantly programmed in a C++ dialect. Unfortunately, standard C++/CUDA does not allow programmers to specify custom memory layouts. This is possible only with custom C/C++ dialects such as \emph{ispc}~\cite{6339601} or custom preprocessors such as \emph{ROSE}~\cite{doi:10.1142/S0129626400000214}. Some compilers automatically perform data layout optimizations~\cite{Zhong:2004:ARS:996841.996872, Chilimbi:1999:CSD:301618.301635, 10.1007/978-3-642-54420-0_19}, but they often fail at complex programs.

\subsection{Performance Characteristics of Structure of Arrays}
SOA is a well-studied best practice on SIMD architectures. Such architectures achieve parallelism by executing the same processor instruction on a \emph{vector} register. Getting data into and out of vector registers is often the biggest bottleneck and peak memory bandwidth utilization can be achieved only with memory coalescing. SOA is one of the most basic optimizations that experienced GPU/CPU-SIMD programmers apply to optimize global memory access~\cite{nvidia_soa_layout_devs}. Previous work has reported speedups of SOA over AOS by multiple factors~(e.g.,~\cite{HOMANN2018325}; see also Section~\ref{sec:memory_coalescing_experiment_background}).

Due to the data-parallel execution model, all threads of a SIMD work group (\emph{warp} in CUDA) execute the same processor instruction in parallel on a vector of scalars. Therefore, if one thread is reading/writing a field such as \texttt{Body::pos\_x}, then all other threads in the SIMD work group are likely reading/writing the same field (but usually of a different object). If neighboring threads process neighboring objects, then these reads/writes can be combined into a small number of physical memory transactions when using an SOA data layout (Figure~\ref{fig:nbody_aos_soa_ecoop}). On GPUs, this is implemented in hardware: The memory controller coalesces memory reads/writes at run-time. On CPUs, the instruction set provides vectorized (packed) versions of many commonly used instructions (including vector loads/stores) and the compiler generates those instructions instead of their scalar versions (see Section~\ref{sec:background_mem_coal}). In AOS, the memory access follows a much more inefficient strided pattern, which does not allow for vector loads/stores because the memory addresses are not contiguous\footnote{Recent CPU SIMD extensions such as AVX2 provide gather/scatter vector loads/stores~\cite{intel_ia_arch, Hofmann:2014:CPD:2568058.2568068}, but those instructions are not as fast as regular vector loads/stores.}.

Besides better memory coalescing, SOA can also improve cache utilization if not all fields are used in a computation. Such fields do not occupy cache lines in SOA. For example, the implementation of \texttt{Body::move()} does not access the fields \texttt{force\_x}, \texttt{force\_y} and \texttt{mass}. However, since the L1 cache line size is 128~bytes on NVIDIA GPUs, in AOS, these fields will occupy the same cache lines as the actually accessed \texttt{pos\_x}, \texttt{pos\_y}, \texttt{vel\_x}, \texttt{vel\_y} fields, leaving less space in the cache for those fields.

\begin{figure}
\centering
\subfloat[Host running time]{\includegraphics[width=0.49\columnwidth]{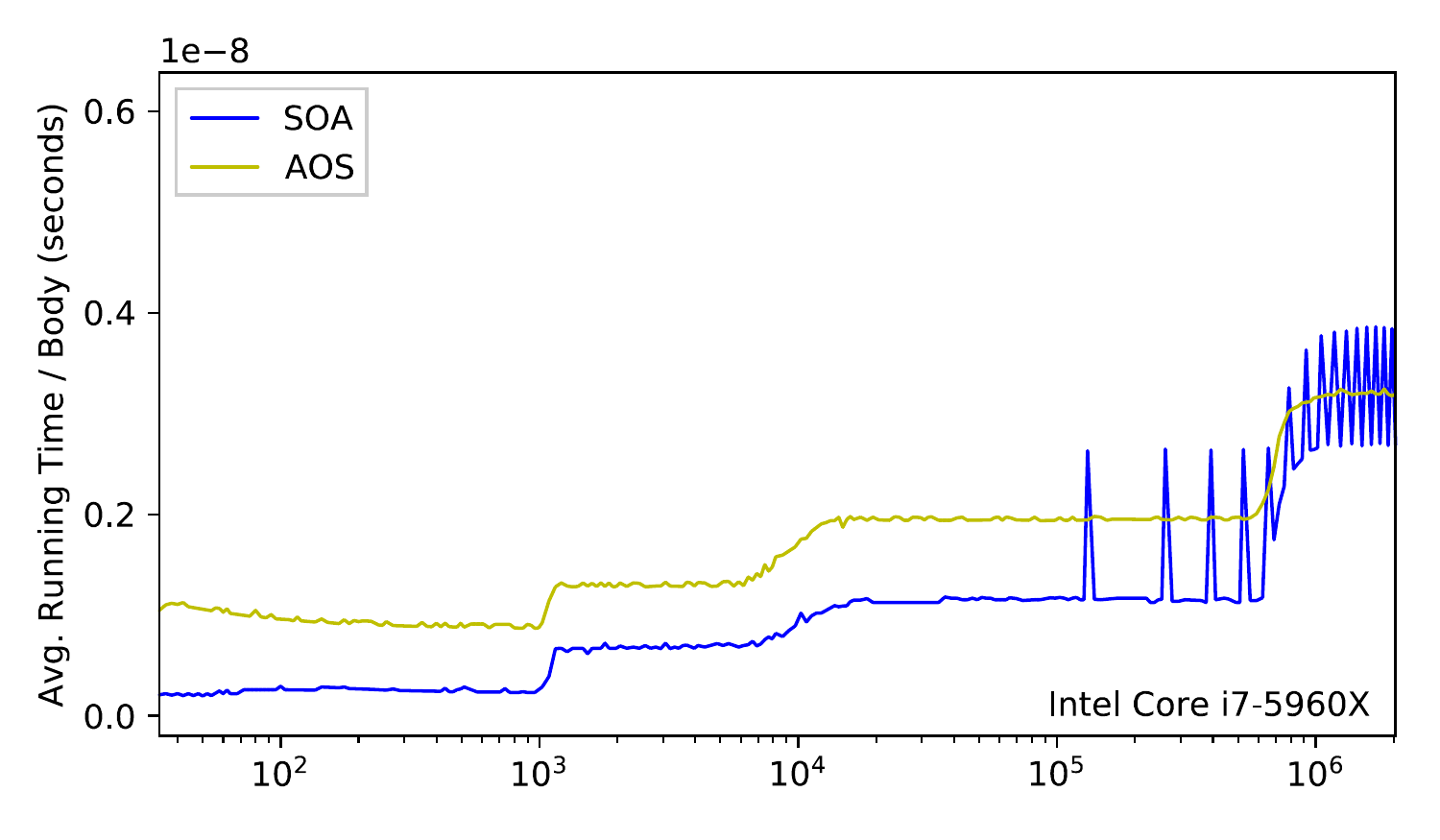}}\hfill
\subfloat[Device running time]{\includegraphics[width=0.49\columnwidth]{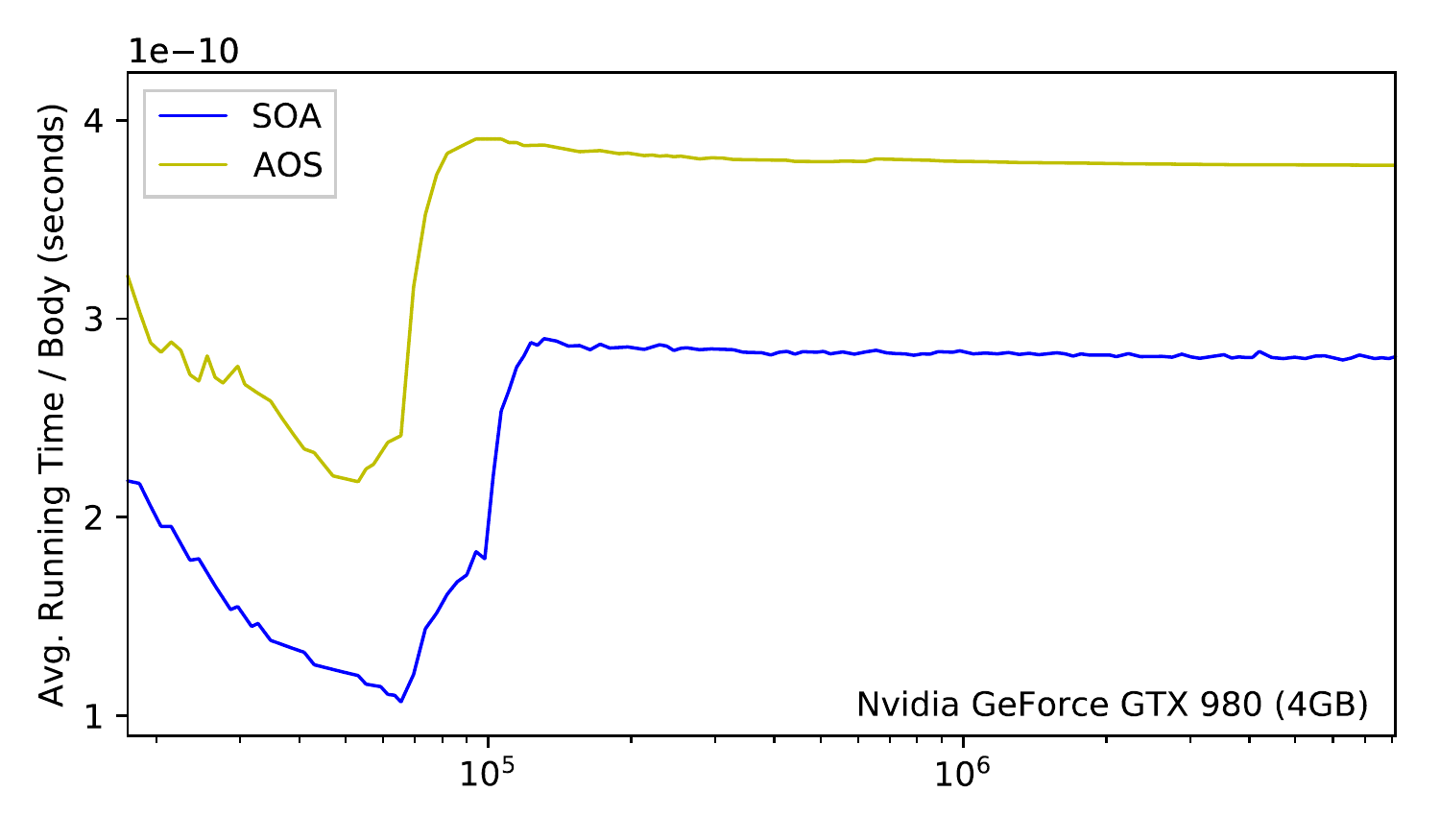}}
\caption{Running time of n-body simulation in AOS and SOA data layout}
\label{fig:running_time_soa_aos_body}
\end{figure}

Figure~\ref{fig:running_time_soa_aos_body} shows the running time of \texttt{Body::move} with AOS and SOA data layout for a varying number of \texttt{Body} objects (x-axis). We report the average running time per \texttt{Body} object (y-axis). In host (CPU, Subfigure~\textsc{a}) code, SOA is much faster at the beginning. This is because the compiler generates SSE vector instructions for SOA code. Furthermore, the cache is utilized more efficiently in SOA. This benefit starts fading away with higher body numbers. The spikes in the SOA graph are due to cache associativity issues. This could be resolved with a hybrid layout (Section~\ref{sec:choosing_dl_hybr}). In device (GPU, Subfigure~\textsc{b}) code, SOA is always faster than AOS, mostly due to better memory coalescing. For low body numbers, kernel invocation dominates the running time, thus the poor performance at the beginning.

\subsection{Object vs. SOA Array Alignment}
\label{sec:obj_vs_soa_array_alignment}
An object set stored in SOA layout can sometimes require less memory than the same object set in AOS\footnote{For reasonably large object sets, SOA never requires more space than AOS.}. In AOS, every field and every object must be properly aligned, whereas in SOA, only the SOA arrays themselves must be aligned.

\begin{lstfloat}
\begin{lstlisting}[language=c++, numbers=none, columns=fixed, caption={Example: Memory consumption of an object set in AOS/SOA}, label={lst:alignment_aos_soa}]
class DummyClass {
  double field_0;          // sizeof(double)       = 8
  char   field_1;          // sizeof(char)         = 1
};                         // sizeof(DummyClass)   = 16
DummyClass objects[50];    // sizeof(objects)      = 800

struct SoaContainer {
  double soa_field_0[50];  // sizeof(double[50])   = 400
  char   soa_field_1[50];  // sizeof(char[50])     = 50
};                         // sizeof(SoaContainer) = 456
\end{lstlisting}
\end{lstfloat}

As an example, consider the C++ source code in Listing~\ref{lst:alignment_aos_soa}. 50~objects in AOS layout require 800~bytes because every \texttt{double} field value (and \texttt{DummyClass} object) must be aligned to 8~bytes. However, 50~objects in SOA layout require only 456~bytes.

As a side note, the size of a C++ struct/class can sometimes be reduced by rearranging its fields (\emph{structure packing}). In the example of Listing~\ref{lst:alignment_aos_soa}, structure packing cannot reduce the size of \texttt{DummyClass}. 

\subsection{Choosing a Data Layout}
\label{sec:choosing_dl_hybr}
Choosing the best data layout for an application is challenging and depends on hardware characteristics and on the data access patterns of the application. Previous work has shown that, on different platforms, different layouts can sometimes achieve the best performance. For example, on platforms with low cache associativity, SOA is known to cause more cache evictions than AOS, which can increase the number of cache misses and slow down an application~\cite{10.1007/978-3-642-54420-0_19, 7853809}.

Previous work has shown that a mixture of AOS and SOA such as AoSoA can sometimes achieve the best performance~\cite{Franco:2017:YAG:3133850.3133861, 10.1007/978-3-662-48096-0_21, Weber:2017:MAL:3132652.3106341, Lee:2017:VPP:3105762.3105768, 7967111}. How to find good data layouts has been studied before~\cite{10.1007/978-3-662-48096-0_21, Abel:1999:ATS, haine2017} and is out of the scope of our work. When talking about custom data layouts, we are focusing on SOA in the remainder of this theis. However, our work could be extended to other layouts in the future.

\chapter{Expressing Parallelism in Object-oriented Programs}
\label{sec:thesis_expressing_parallel}
This chapter investigates how object orientation can be utilized in GPU programs. We are particularly interested in how object orientation can be used to express parallelism.

\minitoc

\paragraph{Overview and Outline}
We developed two techniques/prototypes for expressing GPU parallelism in object-oriented programs.

\begin{itemize}
\item \textsc{Ikra-Ruby}\footnote{\url{https://prg-titech.github.io/ikra-ruby/index.html}} (Section~\ref{sec:express_parallel_arr_int}): A data-parallel array class for Ruby with \texttt{each}/\texttt{map}/ \texttt{reduce}/\texttt{stencil}/... operations that run on the GPU. Object orientation allows programmers to compose GPU programs of small parallel sections in a modular way. \textsc{Ikra-Ruby} is suitable mainly for mathematical computations.
\item \textsc{Ikra-Cpp}\footnote{\url{https://github.com/prg-titech/ikra-cpp}} (Section~\ref{sec:express_smmo}): To allow for efficient object-oriented programming within GPU code, we developed a CUDA framework for a subclass of object-oriented applications that we call \emph{Single-Method Multiple-Objects} (SMMO). In SMMO, parallelism is expressed by running a method on all existing objects of a type. Many interesting applications in high-performance computing can be expressed with SMMO (Section~\ref{chap:smmo_examples}).
\end{itemize}

SMMO applications can also be implemented in \textsc{Ikra-Ruby} by running a parallel \texttt{each} operation on an array of objects. Besides being implemented in a different programming language, \textsc{Ikra-Cpp} is, in essence, a more restricted version of \textsc{Ikra-Ruby} with only one kind of operation: A parallel method call. This allows us to develop better optimizations for object-oriented code that runs on the GPU. 

Our vision is that \textsc{Ikra-Cpp} will eventually become a part of \textsc{Ikra-Ruby}, so that programmers can develop SMMO applications in a high-level language. However, this is out of the scope of this thesis and up to future work.



\paragraph{Publications}
This chapter is in part based on the following papers.
\begin{itemize}
  \item Matthias Springer, Hidehiko Masuhara. \textbf{``Object Support in an Array-based GPGPU Extension for Ruby.''} In: \emph{Proceedings of the 3rd ACM SIGPLAN International Workshop on Libraries, Languages, and Compilers for Array Programming.} ARRAY 2016. ACM, 2016, pp.~25--31. \texttt{\doi{10.1145/2935323.2935327}}
  \item Matthias Springer, Peter Wauligmann, Hidehiko Masuhara. \textbf{``Modular Array-Based GPU Computing in a Dynamically-Typed Language.''} In: \emph{Proceedings of the 4th ACM SIGPLAN International Workshop on Libraries, Languages, and Compilers for Array Programming.} ARRAY 2017. ACM, 2017, pp.~48--55. \texttt{\doi{10.1145/3091966.3091974}}
  \item Matthias Springer. \textbf{``DynaSOAr: Accelerating Single-Method Multiple-Objects Applications on GPUs.''} Extended Abstract. In: \emph{ACM Student Research Competition, Grand Finals, Graduate Category.} Originally submitted to SPLASH 2018. \href{https://arxiv.org/pdf/1809.07444.pdf}{\texttt{arXiv:1809.07444}} (reviewed and published on ACM SRC website, no formal proceedings)
\end{itemize}

\section{\textsc{Ikra-Ruby}: A Parallel Array Interface}
\label{sec:express_parallel_arr_int}
Most mainstream, object-oriented programming languages, such as Java, Python or Ruby, provide a rich collection API for arrays, linked lists, hash maps, etc. These containers often have functional operations that execute a scalar computation on each (sometimes multiple) element(s) of the collection. The most prominent operations are \texttt{map} and \texttt{reduce}, popularized by the Map-Reduce programming model~\cite{Dean:2008:MSD:1327452.1327492}.

In this section, we investigate how collections can be used to express GPU parallelism. Arrays are particularly interesting because they allow for efficient random element access. Since the (functional) scalar computations are independent of each other, they can run in parallel on the GPU. 

\paragraph{Ruby Array Interface}
As an example, consider the array interface (class \texttt{Array}) of the Ruby programming language. This interface provides many operations, but the following four ones are among the most used ones.

\begin{itemize}
  \item \texttt{Array.new(n, \&block)}: Creates a new array of size $n$. Each array slot is initialized with the return value of the block.
  \item \texttt{Array.map(\&block)}: Applies a block (lambda function) to every element of an array and returns a new array of results. The original array remains unchanged.
  \item \texttt{Array.each(\&block)}: Applies a block (lambda function) to every element of an array. In contrast to all other operations, this operation does not return a new array. The original array remains unchanged. However, the block execution usually has a side effect.
  \item \texttt{Array.reduce(\&block)}: Combines all elements of an array by applying a binary operation over and over. Returns a scalar value. The original array remains unchanged.
\end{itemize}

A Ruby \emph{block} is an anonymous/lambda function. Programmers can customize the above operations with a block that specifies the computation of each scalar value. Blocks are ordinary Ruby objects. However, when defining/calling functions, Ruby distinguishes between ordinary parameters and block parameters. Ruby provides a special syntax for passing a block to a function (Listing~\ref{lst:example_ruby_array_map}).

\begin{lstfloat}
\begin{lstlisting}[language=Ruby, caption={Example: Ruby \texttt{Array::map}}, label={lst:example_ruby_array_map}, numbers=none]
increment = Proc.new do |x| x + 1 end  # Define a computation
[1, 2, 3].map(&increment)  # Must pass blocks with ampersand ("block argument")
# => [2, 3, 4]

# More compact: Inline block definition
[1, 2, 3].map do |x| x + 1 end
# => [2, 3, 4]
\end{lstlisting}
\end{lstfloat}

\paragraph{Parallel Array Operations}
To explore array-based GPU computing in a high-level, object-oriented programming language, we developed \textsc{Ikra-Ruby}, a language extension for data-parallel and scientific computations on NVIDIA GPUs in Ruby. \textsc{Ikra-Ruby} uses arrays as an abstraction for expressing parallelism: We extended Ruby's \texttt{Array} class with parallel versions of the above operations and with a few extra operations that we describe later. The notation and API for these operations is similar to Ruby's sequential counterpart operations but method names are prefixed with \texttt{p} for \emph{parallel}. In the simplest case, one CUDA thread is assigned to each array element, such that the array can be processed in parallel.


\paragraph{Programming Style}
When using \textsc{Ikra-Ruby}, we encourage a \emph{dynamic} programming style that is governed by the following two concepts.
\begin{description}
    \item [Integration of Dynamic Language Features] Code in (blocks passed to) parallel array operations is limited to a restricted set of types and operations (dynamic typing and object-oriented programming is allowed). This allows us to generate efficient GPU code for parallel operations. All Ruby features (incl. metaprogramming) may still be used in other parts of a Ruby program. Therefore, programmers can still use external libraries (e.g., I/O or GUI libraries).
    \item [Modularity~\cite{Meyer:1988:OSC:534929}] While optimized low-level programs typically consist of a small number of kernels performing a variety of tasks, \textsc{Ikra-Ruby} allows programmers to compose a GPU computation from multiple reusable, smaller parallel array operations. \textsc{Ikra-Ruby} generates efficient CUDA kernels from these operations.
\end{description}

\paragraph{Compilation Process}
\textsc{Ikra-Ruby} is essentially a source-to-source compiler that compiles Ruby code to CUDA code. Ruby is a dynamically-typed programming language, while CUDA is a statically-typed language. Due to Ruby's dynamic language features, whole-program static (ahead of time) analysis is difficult. Therefore, \textsc{Ikra-Ruby} generates CUDA programs at runtime (just in time) when all type information is known.

\begin{figure}
    \centering
    \includegraphics[width=0.7\textwidth]{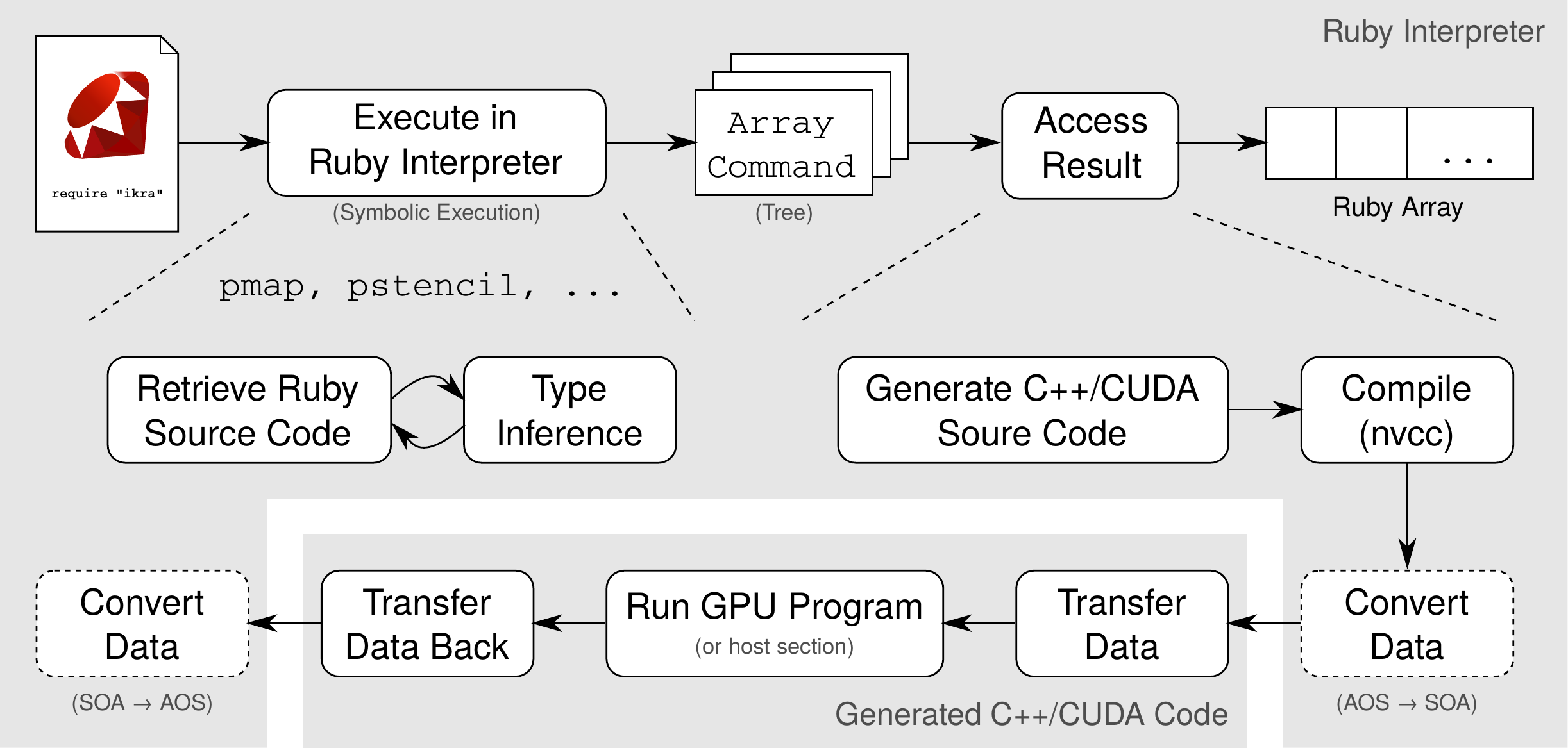}
    \caption[High-level overview of the \textsc{Ikra-Ruby} compilation process]{High-level overview of the compilation process}
    \label{fig:overview_arc}
\end{figure}

Figure~\ref{fig:overview_arc} gives a high-level overview of \textsc{Ikra-Ruby}'s compilation process. \textsc{Ikra-Ruby} executes parallel operations \emph{symbolically}~\cite{King:1976:SEP:360248.360252} in the Ruby interpreter. The result is an \emph{array command} object. Such an object contains all information required for CUDA code generation and execution. An array command can be used like a normal Ruby array. However, only when its contents are accessed for the first time, \textsc{Ikra-Ruby} generates CUDA/C++ source code, compiles it using the CUDA compiler and runs the generated GPU program. The generated program copies data to the GPU, executes the parallel operation(s) on the GPU, copies the result back to the host memory and returns the result\footnote{If object-oriented programming is used inside a kernel, data is first converted to a memory coalescing-friendly \emph{Structure of Arrays layout} (Section~\ref{sec:o_aos_vs_soa}).}. If a parallel operation modifies instance variables of objects as a side effect, then these changes are also copied back to Ruby.

Array commands can have other array commands as inputs, forming a \emph{computation graph} of array commands. Such a graph can be optimized as a whole before emitting CUDA code. Recently, many machine learning framework follow such a design~\cite{Abadi:2016:TSL:3026877.3026899, chainer_learningsys2015}. The main optimization of \textsc{Ikra-Ruby} is \emph{kernel fusion}: Compiling the computation graph into a small number of CUDA kernels. We describe this optimization in detail in Section~\ref{sec:kernel_fusion_in_ikra_ruby_sec4}.

\paragraph{Symbolic Execution}
During symbolic execution, \textsc{Ikra-Ruby} retrieves the source code of the block that is passed to a parallel operation and generates an abstract syntax tree\footnote{We utilize the Ruby \emph{parser} library. See also: \url{https://rubygems.org/gems/parser}.}. The result of symbolic execution is an array command. Within a parallel operation, \textsc{Ikra-Ruby} currently supports expressions of primitive types, user-defined classes and polymorphic types. Advanced Ruby features such as metaprogramming, features that cannot be easily compiled to C++/CUDA (e.g., system calls, file I/O, GUI or FFI), and nested parallel operations are not supported. Parallel sections typically perform mathematical computations, which can be implemented with most basic Ruby functionality.



\paragraph{Array Commands}
A command object in the Command Design Pattern~\cite{Gamma:1995:DPE:186897} is an object that contains all information that is necessary to perform an action at a later point of time. An array command in \textsc{Ikra-Ruby} is an object that contains all information required for code generation and running a parallel operation. For example, a stencil array command contains the AST of the computation (block), a reference to the input array command, an array of neighborhood indices, and an out of bounds value (explained in Section~\ref{sec:parallel_operations}). 

An array command can be seen as a special \emph{Ikra array}. It has all methods that an ordinary Ruby array has, but its contents are computed once it is accessed for the first time. The result of the computation is cached in Ruby, so that multiple accesses do not trigger recomputation. If the input of an array command is changed, the cached result is \emph{not} invalidated (Listing~\ref{lst:caching_result_parallel_arr}). This is a deliberate decision and similar to how standard Ruby behaves with normal (non-parallel) array operations. However, in contrast to standard Ruby, results are computed upon access, so modifications of the input of an array command after symbolic execution can affect its result. 

\begin{lstfloat}
\begin{lstlisting}[language=Ruby, caption={[Example: Caching the result of \textsc{Ikra-Ruby} operations]Example: Caching the result of parallel array operations}, label={lst:caching_result_parallel_arr}, numbers=none]
arr = [10, 20, 30, 40, 50]
squared = arr.pmap do |x| x * x end  # No computation yet
arr[0] = 5  # Modifying "arr" before computation affects result of "squared"
squared[0]  # First access triggers compilation and execution on GPU
# => 25

squared[1]  # Result of "squared" is already in the cache
# => 400

arr[3] = 1000  # No recomputation of "squared"
squared[3]
# => 1600  (not 1000000)
\end{lstlisting}
\end{lstfloat}

Figure~\ref{fig:integration} gives an overview of \textsc{Ikra-Ruby}'s integration in Ruby and the design of array commands.
\begin{figure}
    \centering
    \includegraphics[width=0.7\textwidth]{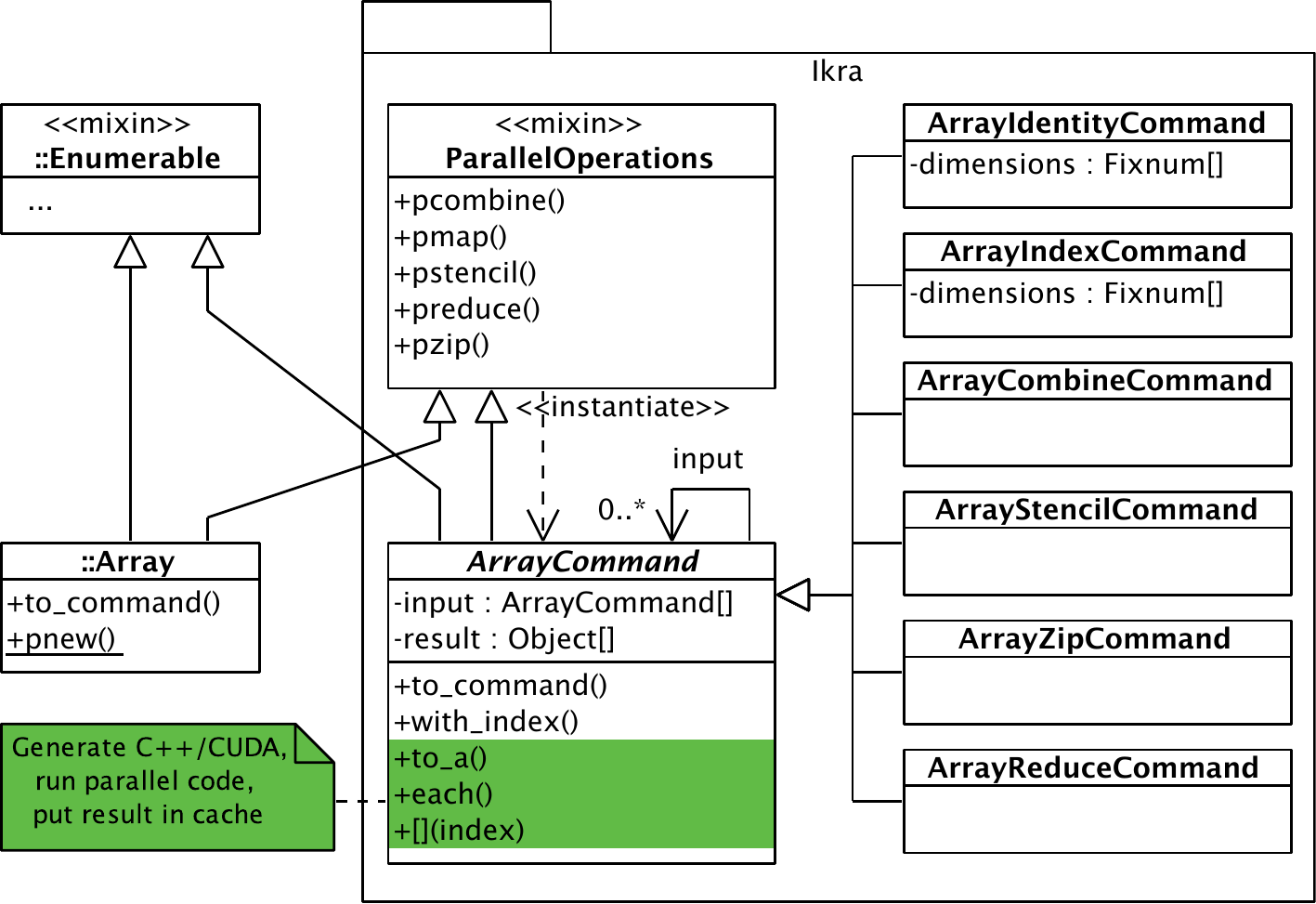}
    \caption{Integration of \textsc{Ikra-Ruby} in Ruby}
    \label{fig:integration}
\end{figure}
There are a variety of subclasses of \texttt{ArrayCommand} corresponding to the parallel operations that are supported in \textsc{Ikra-Ruby} (Section~\ref{sec:parallel_operations}). The standard Ruby \texttt{Array} class and \textsc{Ikra-Ruby}'s \texttt{ArrayCommand} class include two mixins~\cite{Bracha:1990:MI:97945.97982}: \texttt{Enumerable} and \texttt{ParallelOperations}. The first mixin provides standard collection API functionality and requires an implementation of the \texttt{each} method in the class that it is mixed into. The second mixin provides \textsc{Ikra-Ruby}'s parallel operations which are executed on the GPU.

\subsection{Parallel Operations}
\label{sec:parallel_operations}
This section gives an overview of the parallel operations that are provided by \textsc{Ikra-Ruby}. All operations can handle multidimensional \textsc{Ikra-Ruby} arrays, making code more readable if data is inherently multidimensional (e.g., images). For presentation reasons, we use only one dimension for most operations in this section.

If an operation performs a computation, then the size of the receiver (\emph{base array}) determines the number of CUDA threads that are used. By default, \textsc{Ikra-Ruby} uses one CUDA thread per array element, but this can be changed.

\paragraph{Array Identity}
This operation wraps a Ruby array $A$ in an \textsc{Ikra-Ruby} array (command), denoted by $\mathit{id}(A)$. Since an \textsc{Ikra-Ruby} computation graph consists of only array commands, this operation is necessary to make an external Ruby array $A$ (which is not computed on the GPU) available in \textsc{Ikra-Ruby}. 

Array identity is applied implicitly where required. For example, when a \emph{parallel map} operation is applied to a Ruby array, \textsc{Ikra-Ruby} first wraps the Ruby array in an array identity command.

Array identity can also be used to \emph{reshape} a Ruby array. E.g., this is useful if programmers want to convert a one dimensional Ruby array into a multidimensional \textsc{Ikra-Ruby} array. Reshaping does not change the actual data layout and is ``for free''. Array identity is exposed to Ruby programmers as \texttt{to\_command}, taking an optional parameter for dimensions.

\smallskip
\begin{lstlisting}[numbers=none]
<@$A$@>.to_command()
<@$A$@>.to_command(dimensions: [15, 20])
\end{lstlisting}
\smallskip

\textsc{Ikra-Ruby} arrays can be converted back to Ruby arrays with \texttt{to\_a}. This executes the computation graph on the GPU, unless the graph was already executed and the result is already cached. 

\paragraph{Combine}
This operation is used to map over one or more arrays $A_i$ of same size $m$ and dimensions. It takes as input $n$ arrays and a block (anonymous function) $f$ taking $n$ scalar values. It applies $f$ to every element of the input and retains the original shape of the input (all dimensions).

\begin{align*}
\mathit{combine}(A_1, \ldots, A_n, f) &= \begin{bmatrix}
       f(A_1[0], \ldots, A_n[0]) \\
       \vdots \\
       f(A_1[m-1], \ldots, A_n[m-1])
     \end{bmatrix}
\end{align*}

By default, \textsc{Ikra-Ruby} allocates $m$ CUDA threads, i.e., every thread processes one tuple. This operation is exposed to Ruby programmers as \texttt{pcombine}:

\smallskip
\begin{lstlisting}[numbers=none]
<@$A_1$@>.pcombine(<@$A_2$@>, ..., <@$A_n$@>, &f)
\end{lstlisting}

\paragraph{Map}
This operation is a special case of \textit{combine} with only one input array. It corresponds to an ordinary map operation but is executed in parallel.

\begin{equation*}
\mathit{map}(A_1, f) = \mathit{combine}(A_1, f) = [f(A_1[0]), \ldots, f(A_1[m-1])]
\end{equation*}

\noindent This operation is exposed to Ruby programmers as \texttt{pmap}:

\smallskip
\begin{lstlisting}[numbers=none]
<@$A_1$@>.pmap(&f)
\end{lstlisting}

\paragraph{For-Each}
This operation is similar to \emph{map}: It runs a given Ruby block for every element in an array $A_1$. However, in contrast to all other operations, \emph{for-each} is not a \emph{functional} operation and does not have a return value. The Ruby block may as a side effect, just as in all other operations, modify the elements of $A_1$ (if they are objects), as well as other objects that are in scope, such as lexical variables. These changes are written back to Ruby after running the operation. This operation is exposed to Ruby programmers as \texttt{peach}:

\smallskip
\begin{lstlisting}[numbers=none]
<@$A_1$@>.peach(&f)
\end{lstlisting}
\smallskip

Since \emph{for-each} does not have a return value, it is not symbolically executed but immediately executed. Furthermore, certain optimizations that are described later, such as kernel fusion, are not applied to this operation.

\paragraph{Index}
This operation generates an array of size $m$ of consecutive indices starting from $0$ and ending with $m-1$.

\begin{equation*}
\mathit{index}(m) = [0, 1, 2, \ldots, m-1]
\end{equation*}

In a multidimensional case, \textit{index} takes $d$ arguments $m_i$ ($d$ is the number of dimensions), where $m_i$ is the size of the $i$-th dimension. Every value in the resulting array is then an array of size $d$ containing the indices for every dimension. 

\begin{align*}
\mathit{index}(m_1, m_2) &= \begin{bmatrix}
       [0, 0], \ldots, [0, m_1 - 1], \ldots, \\
       [m_2 - 1, 0], \ldots, [m_2 - 1, m_1 - 1]
     \end{bmatrix}
\end{align*}

This operation is not directly exposed to programmers. Similar to standard Ruby, programmers must invoke the method \texttt{with\_index} after a parallel operation (the parameter $m$ is provided implicitly) or use \texttt{Array.pnew} (see below).

\smallskip
\begin{lstlisting}[numbers=none, language=Ruby]
<@$A_1$@>.pmap.with_index(&f)

# API Example:
[10, 20, 30, 40, 50].map.with_index do |x, y| x + y end
# => [10, 21, 32, 43, 54]
\end{lstlisting}

\paragraph{New}
This operation is a combination of \textit{index} and \textit{map}. It creates a new array of size $m$ and initializes it using the block (anonymous function) $f$. It is a parallel version of Ruby's \texttt{Array.new}.

\begin{equation*}
\mathit{new}(m, f) = \mathit{map}(\mathit{index}(m), f) = [f(0), \ldots, f(m-1)]
\end{equation*}

Similiar to \textit{index}, this operation takes multiple arguments in a multidimensional case. This operation is exposed to Ruby programmers as \texttt{pnew}:

\smallskip
\begin{lstlisting}[numbers=none]
Array.pnew(m, &f)
\end{lstlisting}

\paragraph{Stencil (Convolution)}
This operation takes as arguments an input array $A$, an array of relative indices $I$ (neighborhood) of size $k$, a block $f$, and an out-of-bounds value $o$. It creates an array of same size and dimensionality where position $i$ is initialized using $f$, passing the values in the neighborhood of $A[i]$ as arguments to $f$. Simple examples of stencil computations are image filtering kernels.

The following formula is used to calculate the value in the resulting array at position $i$. If all indices are within bounds (case 1), i.e., $0 \leq i + I[j] < m$ for all $0 \leq j < k$ (where $m$ is the size of $A$), the value of the stencil computation is used. Otherwise (case 2), the fallback value $o$ is used.

\begin{equation*}
\mathit{v}(A, I, f, o, i) = 
\begin{cases}
    f([A[i + I[0]], \ldots, A[i + I[k-1]]]), & \text{(case 1)} \\
    o, & \text{(case 2)}
\end{cases}
\end{equation*}

\noindent Using this helper function $v$, a stencil computation is defined as follows.

\begin{equation*}
\begin{split}
\mathit{stencil}(A, I, f, o) = [v(A, I, f, o, 0), \ldots, v(A, I, f, o, m-1)]
\end{split}
\end{equation*}

\noindent This operation is exposed to Ruby programmers as \texttt{pstencil}. The value of the parameter $I$ is inferred from the code of the block.

\smallskip
\begin{lstlisting}[numbers=none,columns=fixed, language=Ruby]
<@$A$@>.pstencil(o, &f)

# API Example (Gaussian blur 3x3 image kernel):
img = load_pixels("pic.png").to_command(dimensions: [640, 480])
filtered = img.pstencil(0) do |v|
  ( 1 * v[-1][-1]  +  2 * v[0][-1]  +  1 * v[1][-1] +
    2 * v[-1][ 0]  +  4 * v[0][ 0]  +  2 * v[1][ 0] +
    1 * v[-1][ 1]  +  2 * v[0][ 1]  +  1 * v[1][ 1] ) / 16.0
end
\end{lstlisting}

\paragraph{Zip}
This operation does not perform a computation but groups values of two or more arrays of same size and dimensionality.

\begin{align*}
\mathit{zip}(A_1, \ldots, A_n) &= \begin{bmatrix}
       [[A_1[0], \ldots, A_n[0]] \\
       \vdots \\
       [A_1[m-1], \ldots, A_n[m-1]]]
     \end{bmatrix}
\end{align*}

The result of this operation is an array of arrays. This operation is exposed to Ruby programmers as \texttt{pzip}:

\smallskip
\begin{lstlisting}[numbers=none]
<@$A_1$@>.pzip(<@$A_2$@>, ..., <@$A_n$@>)
\end{lstlisting}

\paragraph{Reduce}
This operation takes as arguments an input array $A$ and a block $f$, whose function must be associative. Every block application reduces two elements into a single one. The block is applied until only one element is left (dimensions are ignored\footnote{Machine learning libraries provide more powerful versions that reduce values along one or multiple specified dimensions, e.g. \texttt{reduce\_sum} in TensorFlow. This is not currently supported in \textsc{Ikra-Ruby}.}). This operation is similar to Ruby's \texttt{Array.reduce}, but the return value is an array with one element instead of a scalar value.

\begin{equation*}
\begin{split}
\mathit{reduce}(A, f) = [f(\ldots f(f(A[0], A[1]), f(A[2], A[3]), \ldots) \ldots)]
\end{split}
\end{equation*}

There are no guarantees about the order in which elements are reduced, because reduction is done in parallel. This operation is exposed to Ruby programmers as \texttt{preduce}:

\smallskip
\begin{lstlisting}[numbers=none, language=Ruby]
# Passing a block (function)
<@$A$@>.preduce(&f)
# Symbols are also possible as shortcuts
# E.g.: <@$A$@>.preduce do |a, b| a + b; end
<@$A$@>.preduce(:+)
\end{lstlisting}

\subsection{Mapping Ruby Types to C++ Types}
Ruby is a dynamically typed programming language. During type inference, \textsc{Ikra-Ruby} infers which types each expression in a parallel operation can have, based on the runtime types of the input to the parallel operation. If an expression is monomorphic (has a single type), a Ruby type is directly mapped to the corresponding type in the C++/CUDA source code.

\begin{table}[!htb]
\caption[\textsc{Ikra-Ruby}: Mapping Ruby types to C++ types]{Mapping Ruby types to C++ types}
\label{fig:mapping_types}
\small
\begin{tabularx}{\columnwidth}{l|X}
\hline\hline
\textbf{Ruby Type} & \textbf{C++ / CUDA Type} \\
\hline
\texttt{Fixnum} & \texttt{int} \\
\texttt{Float} & \texttt{float} \\
\texttt{TrueClass} & \texttt{bool} \\
\texttt{FalseClass} & \texttt{bool} \\
\texttt{NilClass} & \texttt{int} \\
\texttt{Array} & \texttt{array\_t}, generated struct type\\
\texttt{ArrayCommand} & \texttt{array\_command\_t *} \footnotesize \textit{\textcolor{gray}{(only in host sections, Section~\ref{sec:hs_symbolic})}} \\
(other) & \texttt{object\_id\_t} (\texttt{int}) \\
(polymorphic) & \texttt{union\_t} \\
\hline\hline
\end{tabularx}
\end{table}

Table~\ref{fig:mapping_types} shows the mapping of Ruby data types to CUDA/C++ data types. Numeric values are currently represented by \texttt{int} or \texttt{float}. \texttt{nil} is represented by \texttt{int} value 0. Arrays are either represented by \texttt{array\_t} (a pointer-size pair) or a generated struct type for zip types. Other objects are represented by \texttt{int} object IDs generated by \textsc{Ikra-Ruby}'s object tracer (Section~\ref{sec:impl_tracer}).

\newsavebox{\unionta}
\begin{lrbox}{\unionta}
\begin{lstlisting}[numbers=none,linewidth=0.55\columnwidth,language=C++]
union union_v_t {
  int int_;
  float float_;
  bool bool_;
  void *pointer;
  array_t array;
  array_command_t array_command; // later...
};
\end{lstlisting}
\end{lrbox}
\newsavebox{\uniontb}
\begin{lrbox}{\uniontb}
\begin{lstlisting}[numbers=none,linewidth=0.3\columnwidth,language=C++]
struct union_t {
  int class_id;
  union_v_t value;
};
\end{lstlisting}
\end{lrbox}

\begin{figure}[!htb]
    \centering
    \subfloat{\begin{minipage}{.6\columnwidth}\usebox{\unionta}\end{minipage}}
    \subfloat{\begin{minipage}{.3\columnwidth}\usebox{\uniontb}\end{minipage}}
    \caption[Union type struct definition for polymorphic types]{Union type struct definition}
    \label{fig:union_type_struct}
\end{figure}

For polymorphic expressions (e.g., if \texttt{Fixnum}s and \texttt{nil} are assigned to a variable), union type~\cite{Abadi:1991:DTS:103135.103138} structs (Figure~\ref{fig:union_type_struct}) are used. Values are stored in \texttt{union\_v\_t} which can hold values (or pointers to values) for all C++ types of Table~\ref{fig:mapping_types}. The class ID field contains a number that identifies the runtime type\footnote{This is different from standard C++, where a vtable pointer is stored at the allocation site of an object. Such an implementation would impose a large overhead in \textsc{Ikra-Ruby} because we would then have to heap-allocate primitive types and refer to them with pointers.}. If a method is called on a polymorphic expression, \textsc{Ikra-Ruby} generates a switch-case statement with all types that the expression can have at runtime. Such an implementation is much more efficient in CUDA than a virtual function call because all call targets can be inlined by the compiler. If a monomorphic value is assigned to a polymorphic lvalue, \textsc{Ikra-Ruby} wraps the value in a union type struct. Arrays of union type structs are used to represent polymorphic arrays.

In contrast to other compilers that just-in-time generate GPU code from dynamic language code~\cite{Fumero:2017:JGC:3050748.3050761}, \textsc{Ikra-Ruby} performs a \emph{conservative} type inference pass, such that no unexpected types can appear at GPU program runtime. Therefore, \textsc{Ikra-Ruby} does not need to insert runtime type checking guards and code is never invalidated due to a failing runtime type check. This design decision is based on our assumption that most GPU code is monomorphic or exhibits a very small set of runtime types. In fact, we could not find any real GPU programs for which our conservative type inference system would infer an excessively large set of types of which only a few actually appear at GPU program runtime.

\subsection{Object Tracer}
\label{sec:impl_tracer}
If object-oriented programming is used within a parallel operation, we have to transfer all objects to the GPU that the GPU program may be accessing during its runtime. Object allocation and deallocation is not supported within parallel operations; only existing objects can be modified. Before CUDA code is generated, the \emph{object tracer} identifies all objects that may potentially be accessed in a parallel operation. Since the exact runtime behavior of a program cannot be predicted, this analysis is conservative and may identify objects that are not actually accessed at runtime. The object tracer performs three main tasks.

\begin{itemize}
  \item Decide \textbf{which Ruby objects may be accessed} in the parallel operation and must, therefore, be transferred to the GPU.
  \item Assign a \textbf{unique ID} to each such object.
  \item For each Ruby class that is used within a parallel operation, determine \textbf{which instance variables are read/written}. Only those instance variables are transitively traced and copied to the GPU (and copied back to Ruby after the kernel finished executing).
\end{itemize}


The object tracer starts tracing with a set of root objects: All elements of the base array and lexical variables that are accessed. Then, it traverses the object graph by following all instance variables that are read or written inside the parallel section. \textsc{Ikra-Ruby}'s object tracer is somewhat similar to what a \emph{system tracer}~\cite{smalltalk_tracer, Krasner:1983:SBH:226} does in Smalltalk-80 systems.

Object tracing can result in additional type inference passes. The reason for that is Ruby's dynamic typing. During object tracing, \textsc{Ikra-Ruby} might find that the set of possible types of an instance variable is larger than previously assumed. If a method is called on such an instance variable, \textsc{Ikra-Ruby} has to translate also this new method to CUDA code. Furthermore, this method may, for the first time, access another instance variable of objects of a class that we have already begun tracing. The instance variable values of those already traced objects must now also be traced.

\textsc{Ikra-Ruby}'s object tracer is currently not optimized for performance. The current implementation allows us to explore object-oriented programming within parallel operations, but object tracing and data transfers to the GPU can take a considerable amount of time. Assuming that most computation-intensive code runs on the GPU, future versions of \textsc{Ikra-Ruby} should optimize this by allocating most objects on the GPU and transferring them to the Ruby interpreter only upon access. 

\subsection{Example: Image Manipulation Library}
\label{sec:example_img}
To illustrate \textsc{Ikra-Ruby}'s API, we design a simple image manipulation library. This library provides methods for loading images from the file system and a few filters (\emph{image kernels}) and effects.

\begin{lstfloat}
\begin{lstlisting}[numbers=none, language=Ruby, caption={Example: Definition of image manipulation filters}, label={lst:def_of_filters}]
module ImageLibrary::Filters
  def self.load_png(filename)
    # Pixels are just integers.
    image = read_png(filename)
    return image.pixels.to_command(dimensions: [image.height, image.width])
  end

  def self.blend(other, ratio)
    return CombineFilter.new(other) do |p1, p2|
        # pixel_add, pixel_scale: Helper functions for dealing with RGB values
        pixel_add(pixel_scale(p1, 1.0 - ratio), pixel_scale(p2, ratio))
    end
  end

  def self.pixel_add(p1, p2)  # Add values for each channel, cap at 255.
    return min(((p1 & 0xff0000) >> 16 + (p2 & 0xff0000) >> 16), 255) << 16  # red
         + min(((p1 & 0x00ff00) >>  8 + (p2 & 0x00ff00) >> 8), 255) << 8    # green
         + min((p1 & 0x0000ff) + (p2 & 0x0000ff), 255)                      # blue
  end

  def self.pixel_scale(p1, r)  # Ratio r must be between 0.0 and 1.0
    return (r * ((p1 & 0xff0000) >> 16)) << 16  # red
         + (r * ((p1 & 0x00ff00) >> 8)) << 8    # green
         +  r * (p1 & 0x0000ff)                 # blue
  end

  def self.blur
    return StencilFilter.new(0) do |v|
            r = v[-1][-1][0] + ... + v[1][1][0]
            g = v[-1][-1][1] + ... + v[1][1][1] 
            b = v[-1][-1][2] + ... + v[1][1][2]
            [r / 9.0, g / 9.0, b / 9.0]
    end
  end

  # Source code of other filters omitted.
end
\end{lstlisting}

\begin{lstlisting}[numbers=none, caption={Example: Image manipulation library usage}, label={lst:img_manip_usage}, language=Ruby, morekeywords={require}]
require "image_library"

tt = ImgLib.load_png("tokyo_tower.png")
for i in 0...3
  tt = tt.apply_filter(ImgLib::Filters.blur)
end

sun = ImgLib.load_png("sunset.png")
combined = tt.apply_filter(ImgLib::Filters.blend(sun, 0.3))

forest = ImgLib.load_png("forest.png")
forest = forest.apply_filter(ImgLib::Filters.invert)
combined = combined.apply_filter(
    ImgLib::Filters.overlay(forest, ImgLib::Masks.circle(tt.height/4)))

ImgLib::Output.render(combined)  # Draw pixels
\end{lstlisting}
\end{lstfloat}

As an example, Figure~\ref{fig:img_manip_ex} shows which filters we are using and in which order we apply them: First, we load a picture of the Tokyo Tower. Then, we apply a blur filter multiple times. Next, we load a picture of a sunset and merge (blend) both pictures. Finally, we load a picture of a forest, invert it, and overlay it with the previously merged picture. Listing~\ref{lst:img_manip_usage} shows the source code for this example. Notice how the code is modular with respect to composability, reusability, understandability: Image filters are provided by the library and can be arbitrarily combined.

\begin{figure}
    \centering
    \includegraphics[width=0.7\textwidth]{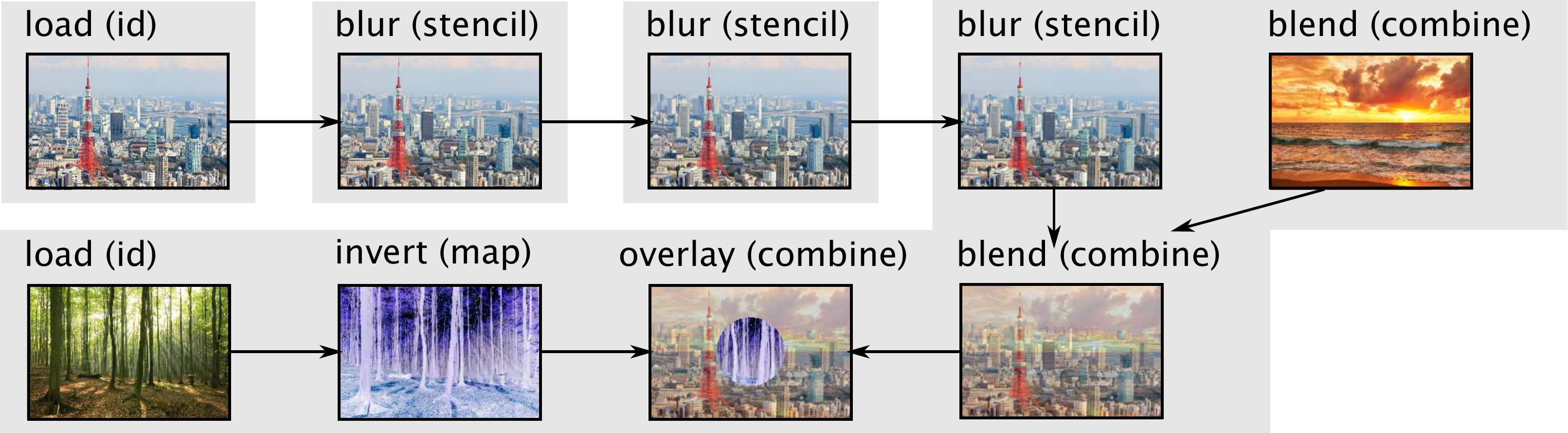}
    \caption[Usage of image manipulation library example]{Example: Gray boxes indicate CUDA kernels (\emph{kernel fusion}, Section~\ref{sec:kernel_fusion_in_ikra_ruby_sec4}).}
    \label{fig:img_manip_ex}
\end{figure}

\begin{figure}
    \centering
    \includegraphics[width=0.7\textwidth]{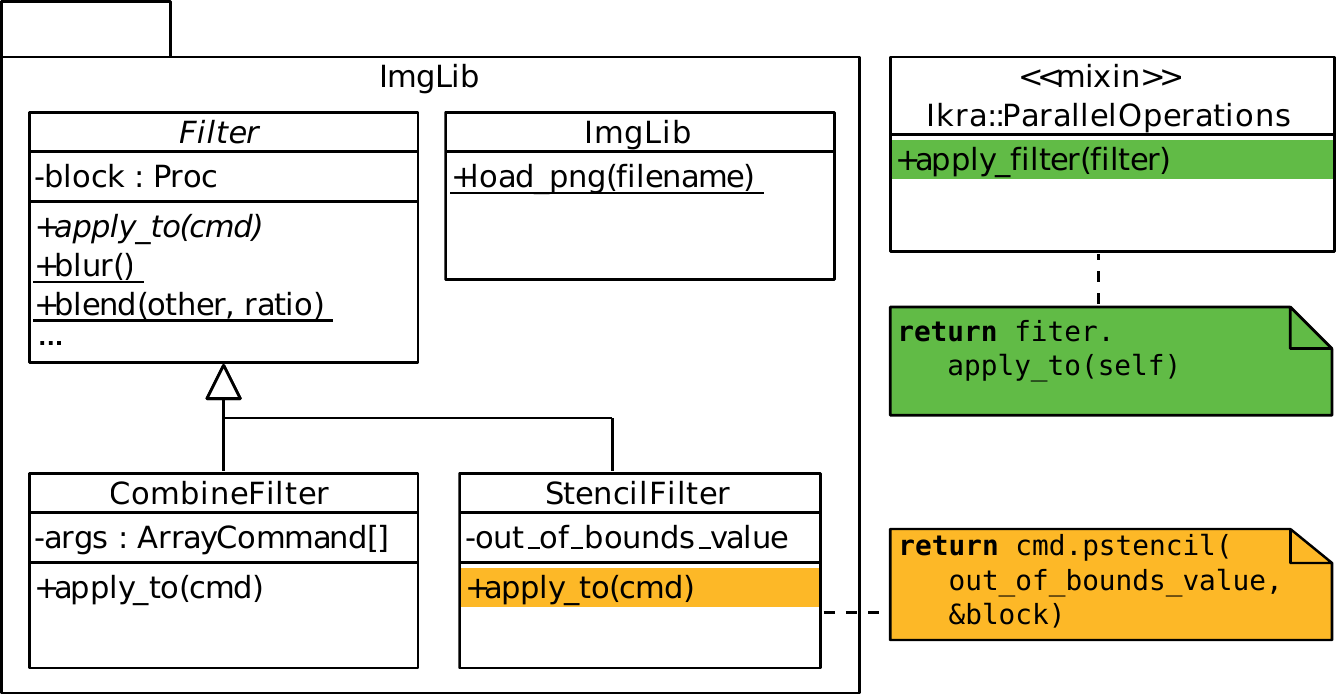}
    \caption{Architecture of image manipulation library example}
    \label{fig:img_lib_arch}
\end{figure}

The filters are implemented in the library using parallel map or stencil operations (Figure~\ref{fig:img_lib_arch}, Listing~\ref{lst:def_of_filters}). Images are represented as 2D array command objects. The library defines an extension method \texttt{apply\_filter} for applying a filter to an image with double dispatch~\cite{Ingalls:1986:STH:28697.28732}. Multiple filters can be applied in a sequence by chaining \texttt{apply\_filter} method calls. Only when the final result is accessed, does \textsc{Ikra-Ruby} generate an optimized CUDA program, copy the images to device memory, run the CUDA kernels and copy back the result.

\subsection{Summary}
In this section, we presented the design and implementation of \textsc{Ikra-Ruby}, a Ruby library for data-parallel computations. Programmers express parallelism by running parallel operations on an array. \textsc{Ikra-Ruby} allows programmers to write modular code with respect to reusability and composability of parallel operations. Section~\ref{sec:kernel_fusion_in_ikra_ruby_sec4} describes how \textsc{Ikra-Ruby} optimizes such GPU programs with kernel fusion.


\section[\textsc{Ikra-Cpp}: A C++/CUDA Library for SMMO Applications]{\textsc{Ikra-Cpp}: A C++/CUDA Library for Single-Method Multiple-Objects Applications}
\label{sec:express_smmo}
Based on our experiments, we found that \textsc{Ikra-Ruby} is suitable for mathematical and purely functional applications, but less suitable for applications that utilize object-oriented programming inside of parallel operations. Such applications do not take advantage of the full range of \textsc{Ikra-Ruby}'s operations. Most of them use only \emph{for-each} (\texttt{Array::peach}) and modify object fields as a side effect, as opposed to expressing state transitions in a functional way through a series of \emph{map} operations.

In fact, in object-oriented programming languages, many programmers express computation as objects imperatively changing their internal state upon receipt of a message (method call). This is in contrast to functional programming, which favors immutability of state. If the state of an object is modifed in a purely functional system, the result is a brand new object~\cite{Kjolstad:2011:TCI:1985793.1985803}. Such a functional programming style certainly has its benefits~\cite{10.1007/978-3-540-24851-4_12} and can make programs easier to reason about~\cite{1214329}. However, such \emph{functional objects}~\cite{Wegner:1990:CPO:382192.383004}/\emph{value objects}~\cite{Riehle:2006:VO:1415472.1415507} lack object identity, which many programmers see as a fundamental concept in object-oriented programming.


\subsection{Single-Method Multiple-Objects}
We identified a broad object-oriented programming model that can be implemented efficiently on SIMD architectures such as GPUs and has many real-world applications in the area of high-performance computing. We call this model \emph{Single-Method Multiple-Objects} (SMMO). We can think of SMMO as OOP-speech for SIMD (\emph{Single-Instruction Multiple-Data}). The most fundamental operation of SMMO is parallel \emph{do-all} (\emph{for-each}): Running one method in parallel on all existing objects of a type (\emph{object set}). Parallel do-all operations typically perform some kind of computation by imperatively modifying the state of objects.

SMMO fits well with the data-parallel SIMD execution model of GPUs and can be implemented very efficiently. Since we only process objects (jobs) of the same type in a parallel do-all operation, we expect those computations to be mostly uniform with little warp divergence.

\paragraph{Definition and Runtime Semantics}
The SMMO programming model provides one main abstraction/operation for expressing parallelism: \emph{Parallel do-all}. This operation runs a given method \texttt{T::func} for all objects of a given type $S$ (and subtypes) that exist at the beginning of the parallel do-all operation. $T$ must be a supertype of $S$ or be equal to $S$, i.e., $S <: T$. If \texttt{T::func} is virtual, then the operation runs the most specific (overridden) method for each object, similar to a virtual method call.

A parallel do-all operation typically runs in multiple threads. The number of threads and the assignment of objects to threads is implementation-specific, but on GPUs, certain assignments may result in better performance than others (Section~\ref{sec:par_do_all_sec3}). Since objects are typically\footnote{A parallel do-all implementation that runs sequentially would also satisfy our definition of SMMO.} processed in parallel, there is no defined order in which objects are processed. A parallel do-all operation returns when all threads finished executing, i.e., when all objects were processed.

New objects of any type may be created within a parallel do-all operation. However, these objects are not being processed by the same parallel do-all operation. Furthermore, existing objects may be deleted within a parallel do-all operation with two limitations. Let $S$ be the type of objects that are processed in a parallel do-all operation and \texttt{T::func} be the method that is executed for all objects of type $S$.

\begin{itemize}
  \item No object may be deleted more than once. This is forbidden even in ordinary object-oriented programming.
  \item An object \texttt{obj} of type $S$ may be deleted only at the end of \texttt{T::func} (last statement) and only by the method execution that is bound to \texttt{obj}. In other words, a method \texttt{T::func} can delete the object that it is bound to (\texttt{delete this}) but no other objects of type $S$. This is to avoid that an object is deleted while it is being processed by another thread.
\end{itemize}

\begin{figure}
  \centering
  \includegraphics[width=0.8\textwidth]{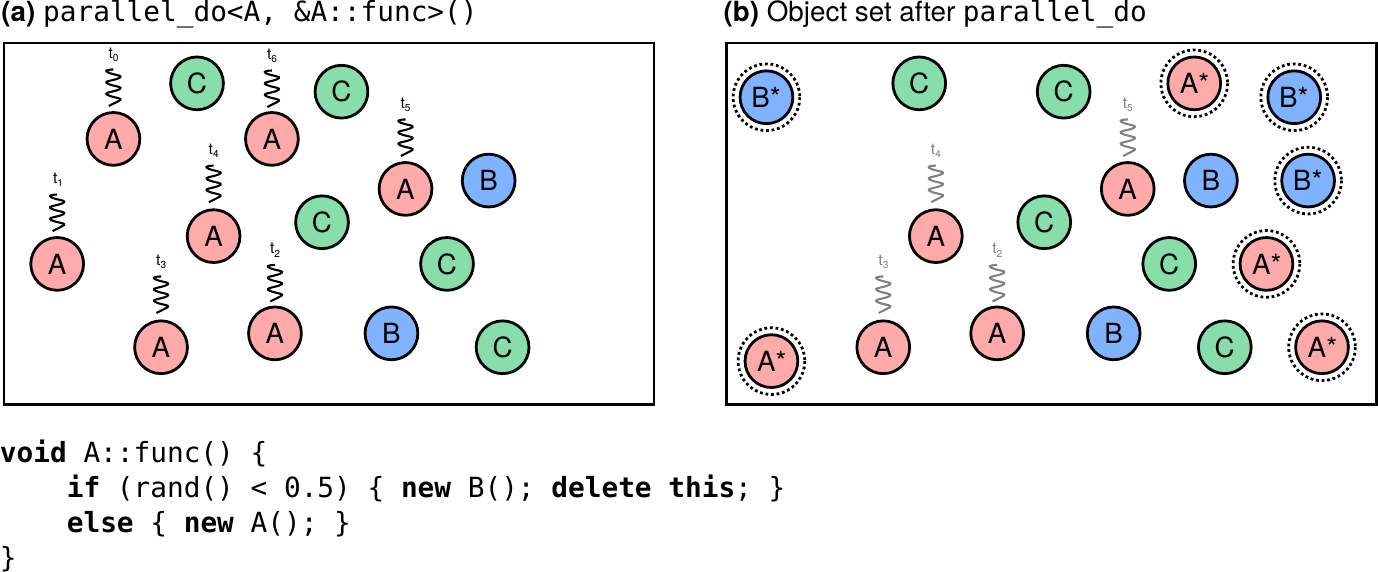}
  \caption[Example: Object set before and after a parallel do-all operation]{Example: Object set before (during) and after a parallel do-all operation}
  \label{fig:smmo_vizz}
\end{figure}

Figure~\ref{fig:smmo_vizz} illustrates a parallel do-all operation of a method \texttt{A::func} with an example. In this example, one thread is spawned for each object of type $A$ (Subfigure~\textsc{a}). The method \texttt{A::func} instantiates new objects of various types (annotated with a star), including objects of type $A$. However, these new objects are not processed by the same parallel do-all operation.

\paragraph{Example Applications}
SMMO is applicable to a broad class of problems\footnote{We implemented a few SMMO applications from different domains: \url{https://github.com/prg-titech/dynasoar/wiki/Benchmark-Applications}} with many real-world applications, such as simulations for population dynamics, (e.g., Sugarscape~\cite{RePEc:mtp:titles:0262550253}), evacuations~\cite{doi:10.1002/cpe.3808}, wildfire spreading~\cite{doi:10.1080/21580103.2016.1262793}, finite element methods or particle systems, to name just a few. SMMO can also express breadth-first search graph traversals and dynamic tree updates/constructions, e.g., in Barnes-Hut~\cite{BURTSCHER201175}. We show the design and implementation of a few SMMO applications in Chapter~\ref{chap:smmo_examples}, highlighting their SMMO structure and their parallel do-all operations.





\subsection{Programming Interface and Notation}
\label{sec:ikra_cpp_api_sect3}
We developed \textsc{Ikra-Cpp}, a C++/CUDA framework that implements the SMMO programming model. \textsc{Ikra-Cpp} facilitates the development of SMMO applications by providing abstractions for parallel do-all operations. In the course of this thesis, we will extend \textsc{Ikra-Cpp} with a data layout DSL (Section~\ref{sec:data_layout_dsls_ikracpp}), a dynamic memory allocator (Chapter~\ref{sec:chapter_dynasoar}) and a memory defragmentation system (Chapter~\ref{sec:chap_gpu_mem_defrag}). We plan to provide a Ruby frontend of \textsc{Ikra-Cpp} in the future, so that programmers can write SMMO applications in a high-level language.

\paragraph{Class Definition}
\textsc{Ikra-Cpp} applications must follow a certain notation. Listing~\ref{lst:sec_3_api_e_ikra_cpp} shows \textsc{Ikra-Cpp}'s notation with an n-body simulation (details in Section~\ref{sec:smmo_nbody_sec71}). User-defined classes/structs must inherit from a special template class \texttt{IkraBase}\footnote{User-defined classes/structs cannot inherit from other user-defined classes/structs. We will remove this limitation later in Chapter~\ref{sec:chapter_dynasoar}.}. This template has two arguments: The class/struct itself (\emph{curiously recurring template pattern}~\cite{Coplien:1995:CRT:229227.229229}) and the number of objects $n$ of the class/struct that can exist at runtime. The class/struct must be initialized with a helper macro \texttt{IKRA\_INITIALIZE\_CLASS}. Finally, the programer must decide whether objects should reside on the host (CPU) or on the device (GPU): \texttt{IKRA\_HOST\_STORAGE} or \texttt{IKRA\_DEVICE\_STORAGE}. These macros statically allocate a \emph{storage array} of size $n$, either on the host or on the device.

\begin{lstfloat}
\begin{lstlisting}[language=c++, caption={[\textsc{Ikra-Cpp}: N-body simulation with AOS layout]API example of \textsc{Ikra-Cpp}}, label={lst:sec_3_api_e_ikra_cpp}, morekeywords={__host__, __device__}, numbers=none]
class Body : public IkraBase<Body, 50> {
 public: <@\emph{IKRA\_INITIALIZE\_CLASS}@>
  float pos_x = 0.0f;
  float pos_y = 0.0f;
  float vel_x = 0.0f;
  float vel_y = 0.0f;
  float force_x = 0.0f;
  float force_y = 0.0f;
  float mass = 0.0f;

  __device__ __host__ Body(float m, float x, float y)
      : mass(m), pos_x(x), pos_y(y) {}

  __device__  __host__ Body(int index) { /* ... */ }  // parallel_new constructor

  __device__ void apply_force(Body* other) {
    if (other != this) {
      // To avoid race conditions: Update other instead of this.
      float dx = pos_x - other->pos_x;
      float dy = pos_y - other->pos_y;
      float dist = sqrt(dx*dx + dy*dy);
      float F = kGravityConstant * mass * other->mass / (dist * dist);
      other->force_x += F*dx / dist;
      other->force_y += F*dy / dist;
    }
  }

  __device__ void compute_force() {
    force_x = force_y = 0.0f;
    device_do<Body>(&Body::apply_force, this);
  }

  __device__ void update() {
    pos_x = pos_x + vel_x * kDt;
    pos_y = pos_y + vel_y * kDt;
  }
};

// Preallocate memory for 50 Body objects on the GPU.
<@\emph{IKRA\_DEVICE\_STORAGE}@>(Body);

int main() {
  // Create a few random Body objects on the host. Alternative: parallel_new.
  for (int i = 0; i < kNumBodies; ++i) {
    new Body(/*m=*/ rand_float(100.0f, 1000.0f),
             /*x=*/ rand_float(-1.0f, 1.0f), /*y=*/ rand_float(-1.0f, 1.0f));
  }

  for (int i = 0; i < kNumIterations; ++i) {
    parallel_do<Body, &Body::update>();
    parallel_do<Body, &Body::move>();
  }
}
\end{lstlisting}

\begin{lstlisting}[language=c++, caption={[\textsc{Ikra-Cpp}: Dynamic object allocation]Object creation in \textsc{Ikra-Cpp}}, label={lst:ikra_cpp_object_creation}, numbers=none]
parallel_new<Body>(49);  // Create 49 objects from host code
Body* b = new Body(100.0f, 0.0f, 0.0f);  // Create a new object from device/host code
delete b;  // Not supported yet. Later in this thesis...
Body* b2 = new Body(50.0f, 0.0f, 0.0f);  // Out of memory
\end{lstlisting}
\end{lstfloat}

\paragraph{Data Layout and Object Creation}
All objects of a type are stored in their respective statically-allocated storage array. Such a layout is called an \emph{Array of Structures} (AOS; Section~\ref{sec:o_aos_vs_soa}). \textsc{Ikra-Cpp} classes are not designed for heap allocation. This is in part because the default CUDA dynamic memory allocator is so slow that runtime heap allocation could severely reduce the performance of a GPU program, unless a custom memory allocator is used.

There are two ways of creating new objects in \textsc{Ikra-Cpp} (Listing~\ref{lst:ikra_cpp_object_creation}): With the C++ \texttt{new} keyword or with a \emph{parallel new} operation (next paragraph). In either case, new objects are stored inside the AOS storage buffer and not on the heap. Existing objects cannot be deleted, so we can create at most 50 objects in the example code before running out of memory. We will extend \textsc{Ikra-Cpp} with an efficient and fully fledged dynamic memory allocator in Chapter~\ref{sec:chapter_dynasoar}.

\paragraph{Expressing Parallelism}
\textsc{Ikra-Cpp} provides a simple API for expressing parallelism. This API allows programmers to run a member function for all objects of a type, which is the foundation of SMMO. Our API abstracts from CUDA implementation details, resulting in more compact and readable code. 

Programmers do not have to define kernels or write kernel invocation statements. A boolean flag \texttt{CUDA} controls whether an operation is executed on the device (GPU) or on the host (CPU). This flag is optional and omitted in all following examples. By default, parallel operations run where the data is located, as indicated by \texttt{IKRA\_DEVICE\_STORAGE} or \texttt{IKRA\_HOST\_STORAGE}.

\begin{itemize}
  \item \texttt{parallel\_do<CUDA, S, \&T::func>(args...)}: Runs a member function \texttt{T::func} for all objects of type $S$ that exist at invocation time, where $S <: T$. If \texttt{CUDA} is \texttt{true}, this operation spawns a GPU kernel and runs the the member function inside the kernel in parallel. Otherwise, the code executes on the host\footnote{Currently sequentially, but we could use OS threads in the future.}.
  \item \texttt{parallel\_do\_and\_reduce<CUDA, S, \&T::func, \&T::reducer>(args...)}: \\ Same as \texttt{parallel\_do}, but the return values of \texttt{T::func} are reduced into a single value by applying a binary, static function \texttt{T::reducer} over and over. Parallel reductions can be implemented efficiently in CUDA with shared memory~\cite{cuda_reduction_web}. This operation is useful for termination detection of iterative algorithms where the termination criteria depends on a property of multiple objects.
  \item \texttt{parallel\_new<CUDA, T>(n, args...)}: Instantiates $n$ objects of type $T$. This operation calls the constructor of $T$ in parallel with an object index (between $[0; n)$), followed by \texttt{args...}. $T$ must have a suitable constructor.
  \item \texttt{device\_do<S, \&T::func>(args...)}: Runs a member function \texttt{T::func} for all objects of type $S$ in the current CPU/GPU thread, where $S <: T$. This is a sequential \emph{for-each} loop. It is typically used inside of a parallel do-all for processing all pairs of objects (e.g., in n-body simulations). Whether objects created within the enclosing parallel do-all operation are enumerated is unspecified. 
\end{itemize}

Note that a parallel operation does not necessarily have to run where the data is located, as indicated by \texttt{IKRA\_*\_STORAGE}, as long as the function/constructor has the correct \texttt{\_\_device\_\_} and/or \texttt{\_\_host\_\_} qualifiers. \textsc{Ikra-Cpp} will take care of the necessary data transfers. The following four configurations are supported.

\begin{itemize}
  \item \textbf{Code on GPU, Data on GPU:} If the program is suitable for GPU execution, then this configuration achieves the best performance.
  \item \textbf{Code on CPU, Data on CPU:} This configuration is useful if no GPU is available and for debugging purposes.
  \item \textbf{Code on CPU, Data on GPU:} During our experiments, we found that it was often very convenient to run setup code (e.g., loading and parsing data from an external source and creating the necessary objects) on the host.
  \item \textbf{Code on GPU, Data on CPU:} This configuration may be useful for highly compute-bound applications that access almost no memory.
\end{itemize}

To reduce the amount of data transfers and to achieve good runtime performance, performance-critical code should in general run where the data is located.





\subsection{Implementation Details}
\textsc{Ikra-Cpp} is implemented entirely in C++. It does not need a separate compiler or preprocessor/code generator. \textsc{Ikra-Cpp} consists mainly of two preprocessor macros (\texttt{IKRA\_INITIALIZE\_CLASS} and \texttt{IKRA\_*\_STORAGE}) and API functions for running a constructor or running a member functions on all objects of a type.

\paragraph{Storage and Allocation}
The purpose of the two preprocessor macros is to declare an array that contains all objects of the class. In particular, \texttt{IKRA\_DEVICE\_STORAGE} in Listing~\ref{lst:sec_3_api_e_ikra_cpp} generates a variable that contain the array of structures and an object counter variable.

\smallskip
\begin{lstlisting}[language=c++, numbers=none, morekeywords={__device__}]
__device__ Body objects_Body[50];
__device__ int counter_Body = 0;
\end{lstlisting}
\smallskip

The purpose of \texttt{IKRA\_INITIALIZE\_CLASS} is mainly to overload the C++ \texttt{new} operator, such that newly instantiated objects are always stored in the previously generated array. Note that the counter variable must be incremented with an atomic operation because multiple threads may be simultaneously instantiating objects.

\smallskip
\begin{lstlisting}[language=c++, numbers=none, morekeywords={__device__, assert}]
__device__ void* Body::operator new() {
  assert(counter_Body < 50);
  return &objects_Body[atomicAdd(&counter_Body)];
}
\end{lstlisting}
\smallskip

\noindent In C++, object instantiation with the \texttt{new} keyword involves three steps.

\begin{enumerate}
  \item Allocate heap memory for the new object. This request is delegated to the system-wide heap allocator (\texttt{malloc}) or to an overloaded \texttt{operator new}.
  \item Zero-initialize the allocated memory. CUDA seems to omit this step.
  \item Run the constructor. This includes field initializers.
\end{enumerate}

Note that C++ operators can be overloaded for specific types. For example, the above source code listing overloads the \texttt{new} operator only for class \texttt{Body}, so other classes/structs are still heap-allocated with the system-wide memory allocator.

\paragraph{API Functions}
In the case of GPU execution, \texttt{parallel\_do} and \texttt{parallel\_new} variants launch a CUDA kernel with a configurable number of threads. By default, \textsc{Ikra-Cpp} launches 256 blocks with 256 threads per block. Inside the template-generated CUDA kernel, \textsc{Ikra-Cpp} processes objects with a grid-stride loop.

\texttt{parallel\_do} and \texttt{parallel\_new} are always called from host code. The functions \texttt{T::func} and \texttt{T::reducer} are device functions and passed to \texttt{parallel\_do} as function pointers. Taking the address of a device function in host code is not allowed. This is because the host code and device code are essentially two different programs. They are compiled separately. However, the address of a device function can be taken in host code if it is used to instantiate a template that resides in device code. When calling such a template function pointer, the (CUDA) compiler can fully inline the function~\cite{cpp_func_templtaes} (instead of generating a jump), because the template instantiation is specific to that function.

\paragraph{Memory Transfers}
If objects are accessed on a device that differs from the allocation device (e.g., accessed on CPU, allocated on GPU), \textsc{Ikra-Cpp} automatically performs the required memory transfers. Programmers have to set a flag to enable support for this. Instead of statically allocating a storage array, \textsc{Ikra-Cpp} then allocates the array at runtime in CUDA unified memory~\cite{nvidia_unified_cuda_unif}. This CUDA feature automatically performs the required memory transfers and was introduced with CUDA 6.

\subsection{Conclusion}
In this section, we presented a second, more restricted way of expressing GPU parallelism: By running a method in parallel on all objects of a type. This simple programming model (\emph{Single-Method Multiple-Objects}) is expressive enough for many important applications in high-performance computing (Chapter~\ref{chap:smmo_examples}) and the full range of \textsc{Ikra-Ruby} operations is often not necessary.

We designed and implemented a small C++/CUDA framework \textsc{Ikra-Cpp} for SMMO applications. At this point, \textsc{Ikra-Cpp} is not much more than a wrapper around kernel invocation statements. The purpose of this section is to make the reader familiar with the SMMO programming model. We will extend \textsc{Ikra-Cpp} with a data layout DSL (Section~\ref{sec:data_layout_dsls_ikracpp}) and a full dynamic memory allocator (Chapter~\ref{sec:chapter_dynasoar}) later in this thesis.

\section{Related Work}
\label{sec:express_related}
%
This section gives an overview of how GPU parallelism is expressed in other systems. Most systems fall into one of two categories: Parallel array interfaces or \emph{for} loop parallelization.

\subsection{Parallel Array/Tensor Interface}
There are many libraries for various programming languages that provide parallel array functions that execute on the GPU, similar to \textsc{Ikra-Ruby}. This section is not an exhaustive description of all of them, but gives an overview of a few major ones.

Ishizaki et al. developed a JIT compiler than translates lambda expressions of Java~8 bytecode to NVVM~IR, an LLVM~IR-based intermediate representation that can be compiled to NVIDIA PTX code~\cite{Ishizaki:2015:COJ:2923305.2923809}. Their compiler targets array-based computations that are expressed with the Java Streams API. Object-oriented programming is supported within lambda expressions and virtual method calls are devirtualized: Either directly, if the receiver type is found to be monomorphic at JIT compilation time, or based on runtime profiling with a guard, otherwise. If the guard fails, the lambda expression is executed on the host. Their compiler follows the same assumption as \textsc{Ikra-Ruby}, namely that GPU code is mostly monomorpic, so guards are expected to fail rarely.

Fumero et. al. also developed an array-based GPU library for Java~8~\cite{Fumero:2014:CAF:2627373.2627381}. Instead of targeting the Java Streams API, their JIT compiler comes with a custom Java array interface that provides functions that are similar to \textsc{Ikra-Ruby}. Their JIT compiler does not support object-oriented programming in GPU code, but parallel array operations can be chained, similar to \textsc{Ikra-Ruby}.

Fumero et al. designed also a GPU extension for the R programming language, a dynamically-typed language~\cite{Fumero:2017:JGC:3050748.3050761}. Their implementation is built on top of the Truffle AST interpreter framework~\cite{Wurthinger:2013:OVR:2509578.2509581} and the Graal JIT compiler. They use partial evaluation to generate optimized OpenCL code for hot code sections. In contrast to \textsc{Ikra-Ruby}, the resulting OpenCL code can handle only monomorphic types, whereas \textsc{Ikra-Ruby} generates a single CUDA program with union types that can handle all types which could theoretically show up during runtime. Consequently, their generated OpenCL code is more efficient, but requires recompilation if the runtime types are changing.

Lime is an object-oriented programming language for task-based and array-based parallelism~\cite{Dubach:2012:CHL:2254064.2254066}. The Lime compiler decides automatically which code should be executed on the device and which code should be executed on the host. In the former case, the compiler generates OpenCL code and in the latter case, the compiler generates a mixture of Java bytecodes.

Delite is a Scala framework for building and executing implicitly parallel DSLs~\cite{Chafi:2011:DAH:1941553.1941561}. Parallel operations form a computation graph of \emph{Delite ops}, which are similar to array commands in \textsc{Ikra-Ruby}. Delite schedules the execution of the computation graph. Independent parts of the computation graph can be executed in parallel. Delite also provides abstractions for defining data-parallel computations.

Firepile is a Scala library for array-based GPU programming~\cite{Nystrom:2011:FRC:2047862.2047883}. Firepile provides a parallel array class whose functional operations are executed in parallel on the GPU. Similar to \textsc{Ikra-Ruby}, the Firepile compiler is usually invoked right before a kernel run (just-in-time). Object-oriented programming is supported in parallel GPU code. Similar to \textsc{Ikra-Ruby}, Firepile uses union types to represent polymorphic types and devirtualizes virtual method calls with switch-case statements.

TensorFlow is a machine learning framework developed by Google~\cite{Abadi:2016:TSL:3026877.3026899}. Programmers build a computation graph of linear algebra nodes with a Python or C++ frontend. This graph is then scheduled to execute on one or multiple devices such as CPUs, GPUs or TPUs (tensor processing units). In contrast to \textsc{Ikra-Cpp}, TensorFlow operations cannot be customized with code. There are a variety of machine learning frameworks that follow a similar design.

\subsection{For-Loop Parallelization}
There are a variety of systems that accelerate programs with GPUs by finding and offloading loops to the GPU. MegaGuards is a Python framework that transparently compiles and offloads compute-intensive loop to the GPU~\cite{qunaibit_et_al:LIPIcs:2018:9221}. It detects parallelizable loops with polyhedral optimization techniques. MegaGuards avoids runtime type checks (guards) inside GPU code with a type stability analysis. If positive, all type checks within GPU code are replaced by a single, larger guard before the loop.

There are also parallelization frameworks that let programmers manually annotate \emph{for} loops for acceleration on GPUs. For example, OpenACC allows programmers to annotate C/C++ loops with pragmas that trigger GPU code generation~\cite{openacc_cuda}.

JaMP is a Java extension with OpenMP-style directives for parallel \emph{for} loops~\cite{6550592}. It supports object-oriented programming within parallel \emph{for} loops and generates a C struct type for every Java class. Instead of referring to referenced objects with pointers, objects are fully inlined into the C struct. JaMP distinguishes between shared objects, which can be accessed by any thread and are replicated on all devices, and managed arrays, which are partitioned among all devices.

Concord is a heterogeneous C++ programming framework for running irregular application code on integrated GPUs~\cite{Barik:2014:EMI:2581122.2544165}. It provides two abstractions for expressing parallelism: A parallel \emph{for} loop and a parallel reduction loop. Concord is implemented with Clang and LLVM and generates OpenCL code for parallel loop bodies. Most C++ features, including virtual function calls and multiple inheritance, are supported in parallel GPU code. Virtual function calls are implemented with vtable pointers, as usual in C++. Concord also provides a software-based shared virtual memory, so that pointers can be shared between GPU code and CPU code. Concord's main selling point is improved energy efficiency of application code over a CPU-only implementation.

%
%

\chapter{Optimizing Memory Access}
\label{sec:theis_opt_mem_access}
Memory access is the one of the biggest bottlenecks in many GPU applications. In this chapter, we describe two memory access optimizations for \textsc{Ikra-Ruby} and \textsc{Ikra-Cpp}: \emph{Kernel fusion} and a data layout DSL for the \emph{Structure of Arrays} (SOA) layout. These optimizations can improve memory bandwidth utilization and cache performance, and thus increase the overall application performance.

\minitoc

\paragraph{Outline}
This chapter is organized as follows. Section~\ref{sec:kernel_fusion_in_ikra_ruby_sec4} describes how we extended \textsc{Ikra-Ruby} with kernel fusion. Section~\ref{sec:data_layout_dsls_ikracpp} describes the design and implementation of an embedded SOA data layout DSL for \textsc{Ikra-Cpp}. This DSL makes SOA easier to use for CUDA/C++ programmers. Section~\ref{sec:inner_arrays_in_soa_papaer} presents an extension of SOA to inner array types and its implementation in \textsc{Ikra-Cpp}. Finally, Section~\ref{sec:summary_chap4} concludes this chapter.

\paragraph{Publications}
This chapter is in part based on the following papers.
\begin{itemize}
  \item Matthias Springer, Peter Wauligmann, Hidehiko Masuhara. \textbf{``Modular Array-Based GPU Computing in a Dynamically-Typed Language.''} In: \emph{Proceedings of the 4th ACM SIGPLAN International Workshop on Libraries, Languages, and Compilers for Array Programming.} ARRAY 2017. ACM, 2017, pp.~48--45. \texttt{\doi{10.1145/3091966.3091974}}
  \item Matthias Springer, Hidehiko Masuhara. \textbf{``Ikra-Cpp: A C++/CUDA DSL for Object-Oriented Programming with Structure-of-Arrays Layout.''} In: \emph{Proceedings of the 4th Workshop on Programming Models for SIMD/Vector Processing.} WPMVP 2018. ACM, 2018, pp.~1--9. \texttt{\doi{10.1145/3178433.3178439}}
  \item Matthias Springer, Yaozhu Sun, Hidehiko Masuhara. \textbf{``Inner Array Inlining for Structure of Arrays Layout.''} In: \emph{Proceedings of the 5th ACM SIGPLAN International Workshop on Libraries, Languages, and Compilers for Array Programming.} ARRAY 2018. ACM, 2018, pp.~50--58. \texttt{\doi{10.1145/3219753.3219760}}
\end{itemize}

\section{Kernel Fusion in \textsc{Ikra-Ruby}}
\label{sec:kernel_fusion_in_ikra_ruby_sec4}
Unoptimized \textsc{Ikra-Ruby} programs suffer from two types of slowdowns. \textsc{Ikra-Ruby} optimizes GPU programs using two techniques that are well-known in statically-typed languages but not in dynamically-typed languages such as Ruby.

\begin{description}
  \item[Global Memory Access] \textsc{Ikra-Ruby} encourages a programming style where a GPU application is composed of many small parallel operations. The image manipulation library of Section~\ref{sec:example_img} illustrates this programming style. A naive compilation strategy compiles each parallel operation into a separate CUDA kernel. This is problematic because the intermediate results of each GPU kernel must be loaded from/stored in global memory. To eliminate this slowdown, \textsc{Ikra-Ruby} \emph{fuses} multiple parallel operations into one large CUDA kernels, such that intermediate results can remain in GPU registers. Kernel fusion is a common and well-studied GPU optimization~\cite{Wahib:2014:SKF:2683593.2683615, Filipovic2015, Sato2009}.
  \item[Ruby Interpreter Loops] Many GPU applications are iterative applications of parallel operations. For example, the \emph{Himeno} benchmark~\cite{5470394} is an iterative application of a stencil operation. If the loops of iterative applications are executed in the Ruby interpreter, the overall program performance can suffer for two reasons. First, the Ruby interpreter is much slower than C++ code. Second, every time a parallel operation executed in Ruby, it is symbolically executed, which incurs some overheads even if the same parallel operation was seen before. To eliminate such slowdowns, we introduce the concept of a \emph{host section}. A host section is a block of Ruby code that contains a more complex program (typically with a loop) with multiple parallel operations. In such a case, the entire block is translated to C++ code, avoiding switching from the Ruby interpreter to external C++ programs multiple times.
\end{description}

Microbenchmarks show that both techniques together achieve performance that is comparable to hand-written CUDA code. As the main contribution of this section, we show how kernel fusion can be implemented as part of the type inference process of host section code.


\subsection{Kernel Fusion}
All array commands except for \emph{index} and \emph{array identity} have at least one input array command (\texttt{input} in Figure~\ref{fig:integration}). E.g., the inputs of a \textit{combine} operation are the array commands that are being mapped over. During code generation, \textsc{Ikra-Ruby} traverses the tree of dependent (input) commands (\emph{computation graph}). Depending on the access pattern of the dependent commands, \textsc{Ikra-Ruby} may fuse multiple array commands into a single CUDA kernel. This can increase the performance of the generated CUDA code, because intermediate results can remain in GPU registers and do not have to be written back to global memory.

\begin{lstfloat}
\begin{lstlisting}[language=Ruby, caption={[Example program for \textsc{Ikra-Ruby} kernel fusion]Example program for kernel fusion}, label={lst:ruby_kf_example_code}]
r0 = (0...1000).to_a
r1 = r0.pmap do |x| 2 * x end
r2 = r1.pmap do |x| x + 1 end
r2.to_a
\end{lstlisting}

\begin{lstlisting}[language=c++, morekeywords={__global__}, caption={Generated CUDA code without kernel fusion}, label={lst:without_kernel_fusion_code}]
__global__ void kernel_1(float* input, float* output) {
  for (unsigned int i = threadIdx.x + blockIdx.x * blockDim.x;
       i < 1000; i += blockDim.x * gridDim.x) { output[i] = 2 * input[i]; }
}

__global__ void kernel_2(float* input, float* output) {
  for (unsigned int i = threadIdx.x + blockIdx.x * blockDim.x;
       i < 1000; i += blockDim.x * gridDim.x) { output[i] = input[i] + 1; }
}

void run_gpu_program(float* input, float* output) {
  float *d_data1, *d_data2;
  cudaMalloc(&d_data1, sizeof(float) * 1000);
  cudaMalloc(&d_data2, sizeof(float) * 1000);
  cudaMemcpy(d_data1, input, sizeof(float) * 1000, cudaMemcpyHostToDevice);

  // Read input from d_data1, store temp. result in d_data2.
  kernel_1<<<256, 256>>>(d_data1, d_data2);
  cudaDeviceSynchronize();
  // Read temp. result from d_data2, store final result in d_data1.
  kernel_2<<<256, 256>>>(d_data2, d_data1);
  cudaDeviceSynchronize();

  cudaMemcpy(output, d_data1, sizeof(float) * 1000, cudaMemcpyHostToDevice);
}
\end{lstlisting}

\begin{lstlisting}[language=c++, morekeywords={__global__}, caption={Generated CUDA code with kernel fusion}, label={lst:code_with_kernel_fus}]
__global__ void kernel_fused(float* input, float* output) {
  for (unsigned int i = threadIdx.x + blockIdx.x * blockDim.x;
       i < 1000; i += blockDim.x * gridDim.x) {
    float tmp_result_1 = 2 * input[i];
    output[i] = tmp_result_1 + 1;
  }
}

void run_gpu_program(float* input, float* output) {
  float *d_data1, *d_data2;
  cudaMalloc(&d_data1, sizeof(float) * 1000);
  cudaMalloc(&d_data2, sizeof(float) * 1000);
  cudaMemcpy(d_data1, input, sizeof(float) * 1000, cudaMemcpyHostToDevice);

  // Read input from d_data1, store final result in d_data2.
  kernel_fused<<<256, 256>>>(d_data1, d_data2);
  cudaDeviceSynchronize();

  cudaMemcpy(output, d_data2, sizeof(float) * 1000, cudaMemcpyHostToDevice);
}
\end{lstlisting}
\end{lstfloat}

Consider the \textsc{Ikra-Cpp} code in Listing~\ref{lst:ruby_kf_example_code} as an example. This listing contains two chained parallel map operations. Without kernel fusion (Listing~\ref{lst:without_kernel_fusion_code}), the first kernel stores the temporary result of \texttt{r1} in global memory and the second kernel loads the temporary result of \texttt{r1} from global memory. These two loads/stores can be eliminated with kernel fusion (Listing~\ref{lst:code_with_kernel_fus}).

\begin{table}[!htp]
\caption[\textsc{Ikra-Ruby}: Input access patterns of array commands]{Input access patterns of array commands}
\label{fig:access_pattern_ac}
\begin{tabularx}{\columnwidth}{l|X}
\hline\hline
\textbf{Command} & \textbf{Input Access Pattern} \\
\hline
\textit{combine} & same location (for all inputs) \\
\textit{stencil} & multiple (fixed pattern) \\
\textit{reduce} & multiple \\
\textit{zip} & same location (for all inputs) \\
(\texttt{with\_index}) & no input, but always accessed as ``same location'' \\
\hline\hline
\end{tabularx}
\end{table}

Table~\ref{fig:access_pattern_ac} lists access patterns for dependent computations for all array commands. ``Same location'' means that for the computation of the element at position $i$ only the element at the same position in the dependent command(s) is required. For example, to calculate element 12 in \emph{map}, only element 12 from the input is required. ``Multiple'' means that an array command needs multiple elements from the dependent command(s). For example, a stencil computation requires an entire neighborhood of values from the input.

\textsc{Ikra-Ruby} can currently merge dependent computations only if the access pattern is ``same location''. In that case, one thread can first compute the dependent operation and then directly proceed with the following computation without any synchronization (Listing~\ref{lst:code_with_kernel_fus}).

In this section, we focus on the process of identifying \emph{which} parallel operations should be merged into a CUDA kernel. We omit implementation details of the compilation process, e.g., \emph{how exactly} the code of one parallel operation is inlined into another parallel operation and how C++/CUDA code is generated from an AST.

\paragraph{Example}
Figure~\ref{fig:ex_kernel_fusion} shows an example. White boxes indicate parallel operations and gray boxes indicate fused CUDA kernels. The leftmost gray box corresponds to Lines~2--3. Those operations are fused together because the input access pattern for \textit{combine}/\textit{map} is ``same location''\footnote{\textit{id} is added implicitly when programmers use Ruby arrays.}. The first input of the stencil computation cannot be fused because its access pattern is ``multiple''. The \textit{index} input can be fused because input generated by \texttt{with\_index} is always accessed as ``same location''. 

\newsavebox{\fusionexlst}
\begin{lrbox}{\fusionexlst}
\begin{lstlisting}[linewidth=\textwidth,xleftmargin=0.325cm]
<@\textcolor{blue}{A1}@> = [1, 2, 3]; <@\textcolor{blue}{A2}@> = [10, 20, 30]               <@\hfill \textcolor{gray}{\# Ruby arrays}@>
a = <@\textcolor{blue}{A1}@>.pmap.with_index do |e, idx| ... end      <@\hfill \textcolor{gray}{\# combine}@>
b = a.pcombine(<@\textcolor{blue}{A2}@>) do |e1, e2| ... end          <@\hfill \textcolor{gray}{\# combine}@>
c = b.pstencil([-1, 0, 1], 0).with_index do |values, idx| ... end <@\hfill \textcolor{gray}{\# stencil}@>
d = c.preduce do |r1, r2| ... end               <@\hfill \textcolor{gray}{\# reduce}@>
\end{lstlisting}
\end{lrbox}

\begin{figure}
    \usebox{\fusionexlst}\\
    \begin{center}
    \includegraphics[width=0.7\textwidth]{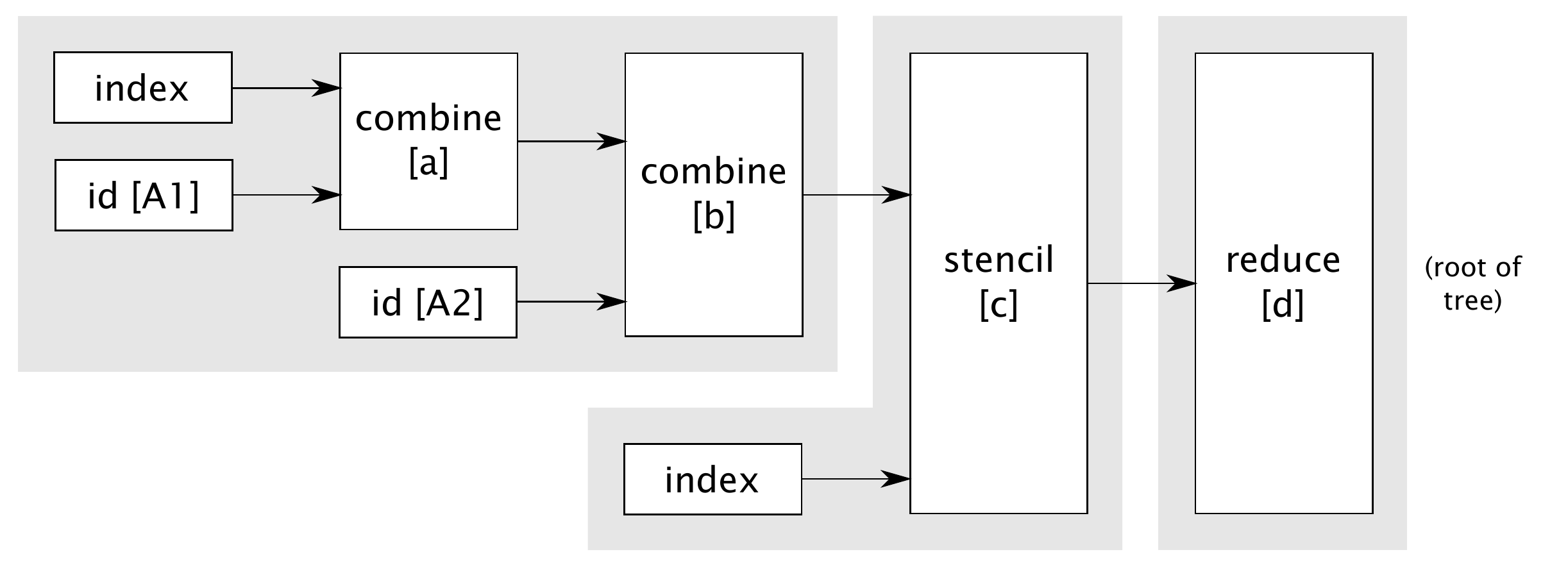}
    \end{center}
    \caption[Example: Kernel fusion of \textsc{Ikra-Ruby} operations]{Example: Kernel fusion result. Gray boxes are kernels.}
    \label{fig:ex_kernel_fusion}
\end{figure}

\subsection{Host Sections}

Many scientific computations (e.g., numerical partial differential equations) exhibit an iterative structure where an array or matrix is updated for a fixed number of times or until convergence. The example code in Listing~\ref{fig:iterative_host_ex_new} does not compute anything meaningful but illustrates how to write such computations in \textsc{Ikra-Ruby}. The loop is enclosed in a host section, a code block that is compiled to C++ and executed on the host, as opposed to \emph{parallel operations} which are executed on the device. The value of the last statement of a host section is the return value of the host section. Inside host sections, only \emph{simple} Ruby code may be written: Everything that is allowed inside a parallel operation plus parallel operations themselves. More advanced Ruby features (such as metaprogramming) are forbidden.

\begin{lstfloat}
\begin{lstlisting}[caption={[Example: Iterative \textsc{Ikra-Ruby} computation in host section]Example: Iterative computation in host section}, label={fig:iterative_host_ex_new}]
input = [10, 20, 30, 40, 50, 60]
result = <@\textcolor{blue}{Ikra}@>.host_section do
  arr = input.to_command(dimensions: [2, 3])
  for i in 0...10
    if arr.preduce(:+)[0] % 2 == 0
      arr = arr.pmap do |i| i + 1; end <@\hfill \textcolor{gray}{\textit{\# $\mathit{map}_A$}}@>
    else
      arr = arr.pmap do |i| i + 2; end <@\hfill \textcolor{gray}{\textit{\# $\mathit{map}_B$}}@>
    end
    arr = arr.pmap do |i| i + 3; end <@\hfill \textcolor{gray}{\textit{\# $\mathit{map}_C$}}@>
  end
  arr
end
\end{lstlisting}
\end{lstfloat}

One C++/CUDA program is generated for the code in Listing~\ref{fig:iterative_host_ex_new}. That program contains a C++ function for the host section and multiple CUDA kernels. As part of code generation, control flow statements inside host sections are executed symbolically as opposed to control flow statements outside of host sections, which are executed by the Ruby interpreter.

Host sections are translated with a conservative kind of ahead-of-time compilation and kernel fusion technique: In the above example, it is not clear until runtime if the parallel operation in Line~10 will be executed together with the one in Line~6 ($\mathit{map}_A$ + $\mathit{map}_C$) or the one in Line~8 ($\mathit{map}_B$ + $\mathit{map}_C$), or both in different iterations. Therefore, \textsc{Ikra-Ruby} generates both fused kernel variants and launches the appropriate one at runtime.

\subsection{Symbolic Execution in Host Sections}
\label{sec:hs_symbolic}
Host sections are pieces of Ruby code that are entirely translated to C++ code. They may contain one or more parallel operations but no advanced language features like metaprogramming. The compilation process of host sections is identical to the one of parallel operations, but there are additional steps to handle parallel operations within them. Since no code generation can be done once a host section (C++ code) is executing, fused kernels must be generated up front. 

\textsc{Ikra-Ruby} statically analyzes all code paths through a host section with parallel operations and generates a number of fused kernels, even some that might never be used at runtime. At runtime, \textsc{Ikra-Ruby} keeps track of which parallel operations were executed symbolically and eventually launches a fused kernel, which may contain multiple parallel operations, when the result is accessed.

\paragraph{Kernel Fusion via Type Inference}
Within host sections, there are additional type inference rules to handle parallel operations. \textsc{Ikra-Ruby} performs kernel fusion through type inference: The type of a parallel operation is the array command object that it evaluates (symbolically) to in the Ruby interpreter.

\begin{lstfloat}
\begin{lstlisting}[caption={Example: \textsc{Ikra-Ruby} type inference}, label={lst:ikra_ruby_type_inf_ex}]
a = 10  # type(a) = Int

b = Array.pnew(a) do |i| 2.5 + i; end
# Eval Ruby: Array.pnew(CodeRef.new(:a)) do ... end
# => ArrayCombineCommand instance
# type(b) = ArrayCombineCommand<@$_3$@>[Float, (ArrayIndexCommand[Float, <@$\emptyset$@>])]

c = b.pmap do |i| 0.5 + i; end
# Eval Ruby: type(b).pmap do ... end
# => (another) ArrayCombineCommand instance
# type(c) = ArrayCombineCommand<@$_8$@>[Float, (
# <@\,\,\,\,\,\,\,\,@> ArrayCombineCommand<@$_3$@>[Float, (ArrayIndexCommand[Float, <@$\emptyset$@>])])]

d = Array.pnew(a) do |i| 2.5 + i; end
# Eval Ruby: Array.pnew(CodeRef.new(:a)) do ... end
# => (yet another) ArrayCombineCommand instance
# type(d) = ArrayCombineCommand<@$_{14}$@>[Float, (ArrayIndexCommand[Float, <@$\emptyset$@>])]
\end{lstlisting}
\end{lstfloat}

Listing~\ref{lst:ikra_ruby_type_inf_ex} illustrates the type inference rules with an example. The type of variable \texttt{a} is \texttt{Int}. The type of variable \texttt{b} is an \texttt{ArrayCombineCommand} object (Figure~\ref{fig:integration}). This is because a parallel \emph{new} operation is implemented as a \emph{map}/\emph{combine} operation wrapped around an \emph{index} operation (see Section~\ref{sec:parallel_operations}). The type of variable \texttt{c} is another \texttt{ArrayCombineCommand} object.

All array commands with a code block have a subscript that identifies the block. In the above example, code blocks are referred to with line numbers. Therefore, even though the types of \texttt{b} and \texttt{d} look similar, they are actually different types. For presentation reasons, we only account for array command base types (\texttt{Float} in all types in the above example), code blocks and computation structure (i.e., the nesting of array commands) in this section and omit other properties of array command types such as array dimensions or out-of-bounds values of stencil operations.

Notice that \texttt{ArrayCommand} objects are not only used for code generation but also serve as type representations/objects in our type inference engine. After type inference, \textsc{Ikra-Ruby} generates a CUDA kernel for each array command type that is \emph{executed}. An expression of array command type is executed when its result is accessed, via either Ruby method \texttt{to\_a} or array subscript notation (\texttt{[]}).


\paragraph{Compilation Overview}
A host section is translated to C++/CUDA as shown below. This process is similar to symbolic execution of parallel operations in the Ruby interpreter (see Section~\ref{sec:express_parallel_arr_int}). However, since host sections are entirely translated to C++ code, parallel operations are now symbolically executed in C++ instead of in the Ruby interpreter.

\begin{enumerate}
    \item Retrieve the Ruby source code of the host section.
    \item Generate an AST (abstract syntax tree) of the source code.
    \item Convert the AST to SSA (static single assignment) form.
    \item Insert a \texttt{to\_a} method call on last expression.
    \item Eliminate expressions with circular types by inserting \texttt{to\_a} method calls.
    \item Perform type inference. (This step also performs kernel fusion.)
    \item Generate C++ source code for the host section and CUDA source code for all array commands on which \texttt{to\_a} or \texttt{[]} is called.
\end{enumerate}

The SSA form simplifies type inference: If a variable is written a second time with a value of different type, a new C++ variable is allocated and a union type can sometimes be avoided. The return value (last expression) of a host section must be an array command or an array. In the former case, to ensure that an array command is executed, \textsc{Ikra-Ruby} wraps the last expression in a \texttt{to\_a} method call. That method does not have any effect for arrays but triggers array command execution. The next two steps, eliminating circular types and type inference, are described in detail in the following paragraphs.

\subsection{Type Inference}
Ruby is a dynamically typed language. Since C++ and CUDA are statically typed languages, we represent polymorphic expressions with union types (also known as \emph{sum types}~\cite{Pierce:2002:TPL:509043}).

For simplicity, we assume that all expressions are monomorphic in this section. Therefore, our types and typing rules do not account for union types. Furthermore, we only describe the typing rules that involve array commands. In addition, we assume that arrays are explicitly wrapped (\emph{boxed}) in array commands before a parallel operation is invoked on them.


\paragraph{Types}
Figure~\ref{fig:ikra_ruby_list_of_types} gives a simplified overview of the types in our system that we need for fusion of array commands. An expression can be of primitive type (\texttt{TPrim}), of array type with a certain primitive base type (\texttt{TArray}) or of array command type (\texttt{TArrayCommand}). There are multiple subtypes of \texttt{ArrayCommand}, depending on the type of the operation (Figure~\ref{fig:integration}).

Array command types have two parameters. The first parameter is the base type of the array command, similar to the base type of an array. The second parameter is a list of all dependent (input) array command types (\texttt{AcInput}). Similar to C++ expression templates~\cite{Veldhuizen:1996:ET:260627.260749}, an array command type in \textsc{Ikra-Ruby} encodes the structure of a computation in its type. The code generator can use this information to apply additional code optimizations, kernel fusion in the case of \textsc{Ikra-Ruby}.

\paragraph{Typing Rules}
To perform kernel fusion, our type inference engine captures method calls whose receiver types are a subclass of \texttt{ArrayCommand}. The type of such a method call is the result of the evaluation of that method call in the Ruby interpreter, where the receiver and all arguments are replaced with their respective types (\textsc{T-ParOp}). The rule \textsc{T-ParOp} can be broken down into one rule per operation type: \textsc{T-ParMap}, \textsc{T-ParCombine} and rules for other operation types that we omit here.

As an example, consider \textsc{T-ParCombine}. This typing rule captures \texttt{pcombine} method calls on receivers with \texttt{ArrayCommand} subtype. If this method call has $n$ arguments (excluding the code block), then the code block \texttt{b} must be a function taking $n+1$ arguments: The type of the first argument is the base type of the receiver array command. The types of the following $n$ arguments are the base types of the argument array command types. The resulting type is an \texttt{ArrayCombineCommand} with a list of the receiver type and all argument types as the second type parameter, encoding the structure of the computation.

Before programmers can run a parallel operation on an array, they must wrap it in an array command, as described above (\textsc{T-Box}). Similarly, an array command can be converted into an array, which triggers execution on the GPU (\textsc{T-UnBox}).

\begin{figure}
\begin{align*}
  \texttt{TPrim} ::= & \,\, \texttt{Bool} \left.\right| \texttt{Float} \left.\right| \texttt{Int}\\
 \\
  \texttt{TArray} ::= & \,\, \texttt{Array}[\texttt{TPrim}] \\
 \\
  \texttt{TArrayCommand} ::= & \,\, \texttt{ArrayCommand}_\mathit{loc}[\texttt{TPrim}, \texttt{AcInput}] \\
                             & \,\, \left.\right| \texttt{ArrayIdentityCommand}[\texttt{TPrim}, \emptyset] \\
                             & \,\, \left.\right| \texttt{ArrayIndexCommand}[\texttt{Int}, \emptyset] \\
                             & \,\, \left.\right| \texttt{ArrayCombineCommand}_\mathit{loc}[\texttt{TPrim}, \texttt{AcInput}] \\
                             & \,\, \left.\right| \texttt{ArrayStencilCommand}_\mathit{loc}[\texttt{TPrim}, (\texttt{TArrayCommand})] \\
                             & \,\, \left.\right| \texttt{ArrayReduceCommand}_\mathit{loc}[\texttt{TPrim}, (\texttt{TArrayCommand})] \\
  \\
  \texttt{AcInput} ::= & \,\, \emptyset \left.\right| \texttt{TArrayCommand} \circ \texttt{AcInput} \\
  \\
  \texttt{T} ::= & \,\, \texttt{TPrim} \left.\right| \texttt{TArray} \left.\right| \texttt{TArrayCommand}
\end{align*}
\vspace{-0.5cm}
\begin{center}
  \textit{(types omitted for classes, zip types, nil and arrays with non-primitive base types)}
\end{center}
\caption[\textsc{Ikra-Ruby}: List of types]{List of types}
\label{fig:ikra_ruby_list_of_types}

\vspace{1cm}

\begin{equation}
\inferrule
  {\Gamma \vdash e_i : T_i \and T_i <: \texttt{ArrayCommand}_*[*, *]}
  {\Gamma \vdash \texttt{$e_0$.m($e_1$, \ldots, $e_n$, \&b)} : \mathit{eval}(\texttt{$T_0$.m($T_1$, \ldots, $T_n$, \&b)})} \tag{\textsc{T-ParOp}}
\end{equation}

\begin{equation}
\inferrule
  {\Gamma \vdash e_0 : T_0 \and T_0 <: \texttt{ArrayCommand}_*[B_0, *] \and \Gamma \vdash \texttt{b} : B_0 \rightarrow R}
  {\Gamma \vdash \texttt{$e_0$.pmap(\&b)} : \texttt{ArrayCombineCommand}_{\mathit{loc}(\texttt{b})}[R, (T_0)]} \tag{\textsc{T-PMap}}
\end{equation}

\begin{equation}
\inferrule
  {\Gamma \vdash e_i : T_i \and T_i <: \texttt{ArrayCommand}_*[B_i, *] \and \Gamma \vdash \texttt{b} : B_0 \rightarrow B_1 \rightarrow \ldots \rightarrow B_n \rightarrow R}
  {\Gamma \vdash \texttt{$e_0$.pcombine($e_1$, \ldots, $e_n$, \&b)} : \texttt{ArrayCombineCommand}_{\mathit{loc}(\texttt{b})}[R, (T_0, T_1, \ldots, T_n)]} \tag{\textsc{T-PCombine}}
\end{equation}

\begin{center}
  \textit{(rules omitted for \texttt{pnew}, \texttt{preduce}, \texttt{pstencil} and \texttt{pzip})}
\end{center}

\begin{equation}
\inferrule
  {\Gamma \vdash e_0 : \texttt{Array}[B]}
  {\Gamma \vdash \texttt{$e_0$.to\_command()} : \texttt{ArrayIdentityCommand}[B, \emptyset]} \tag{\textsc{T-Box}}
\end{equation}

\begin{equation}
\inferrule
  {\Gamma \vdash e_0 : T_0 \and T_0 <: \texttt{ArrayCommand}_*[B, *]}
  {\Gamma \vdash \texttt{$e_0$.to\_a()} : \texttt{Array}[B]} \tag{\textsc{T-UnBox}}
\end{equation}
\caption[\textsc{Ikra-Ruby}: Typing rules]{Typing rules}
\label{fig:ikra_ruby_typing_rules}
\end{figure}

\paragraph{Breaking Circular Types}
The type of an array command encodes the structure of its computation, so an array command type effectively encodes a data flow path through the program. If there are multiple possible paths (one of which is chosen at runtime), our type inference system uses a union type that contains all paths. This is problematic with loops or recursion. Since the number of loop iterations is not generally known ahead of time, the union type would grow infinitely, because every additional loop iteration adds another path to the union type.

\begin{lstfloat}
\begin{lstlisting}[caption={[\textsc{Ikra-Ruby}: Circular union type]Circular union type}, label={lst:circ_union_type}]
arr = [1, 2, 3].to_command

for i in 1...100 do
  arr = arr.pmap do |x|
    x + 1
  end
end

result = arr.to_a
\end{lstlisting}
\end{lstfloat}

Consider Listing~\ref{lst:circ_union_type} as an example. After Line~1, the type of variable \texttt{arr} is \texttt{ArrayIdentityCommand[Int, $\emptyset$]}. After the first loop iteration, the type of \texttt{arr} is \texttt{ArrayCombineCommand$_4$[Int, ArrayIdentityCommand[Int, $\emptyset$]]}. Every additional loop iteration wraps the type in another \texttt{ArrayCombineCommand}$_4$.

Since the number of loop iterations is generally unknown until runtime, the type of \texttt{arr} after the loop is an infinite union type with one type for each possible number of loop iterations.

\begin{align*}
  \emph{type}(\texttt{arr}) = \{ & \,\, \mathit{id}, \\
                                 & \,\, \mathit{map}_4[\mathit{id}], \\
                                 & \,\, \mathit{map}_4[\mathit{map}_4(\mathit{id}]], \\
                                 & \,\, \ldots, \\
                                 & \,\, \mathit{map}_4[\mathit{map}_4[ \ldots \mathit{map}_4[\mathit{id}] \ldots ]] \,\, \}
\end{align*}

We abbreviate \texttt{ArrayCombineCommand} with \emph{map} and \texttt{ArrayIdentityCommand} with \emph{id} in the above formula for presentation reasons. Furthermore, we omit base types.

The infinitely large type of \texttt{arr} is a problem. Line~9 of the source code accesses the result of  \texttt{arr} and will trigger CUDA kernel execution. Since \texttt{arr} is a union type, this will translate to a switch-case statement with one case per type. Our type inference engine avoids such circular types by changing the source code of the program. The type $T$ of an expression is \emph{circular} if $T$ is included as a dependent array command of $T$. Not necessarily as a directly dependent array command, but maybe somewhere nested deep inside $T$. For example, given the type $T=\mathit{map}_4[\mathit{id}]$, another iteration of the loop yields the type $T' = \mathit{map}_4[\mathit{map}_4[\mathit{id}]] = \mathit{map}_4[T]$, so \texttt{arr} has a circular type.

Whenever a circular type is detected, we insert two chained method calls \texttt{to\_a()} and \texttt{to\_command()}, which execute the CUDA kernel and reset the type (Listing~\ref{lst:circ_union_type_elim}).

\begin{lstfloat}
\begin{lstlisting}[caption={[\textsc{Ikra-Ruby}: Eliminated circular union type]Eliminated circular union type}, label={lst:circ_union_type_elim}]
arr = [1, 2, 3].to_command

for i in 1...100 do
  arr = arr.pmap do |x|
    x + 1
  end <@\fbox{\texttt{.to\_a.to\_command}}@>
end

result = arr.to_a
\end{lstlisting}
\end{lstfloat}

The elimination of circular types is baked into our type inference engine and we do not describe it in more detail in this section. There may be multiple ways of eliminating a circular type. Where exactly we insert the two chained method calls is an implementation details and not important. What is important is that the type of \texttt{arr} after the loop of Listing~\ref{lst:circ_union_type_elim} is now no longer circular.

\begin{align*}
  \emph{type}(\texttt{arr}) = \{ & \,\, \mathit{id}, \mathit{map}_4[\mathit{id}] \,\, \}
\end{align*}

\paragraph{Type Inference by Example}
To illustrate type inference in a larger example, consider the host section source code of Listing~\ref{lst:hs_ssa_form_ex}, which is identical to Listing~\ref{fig:iterative_host_ex_new}, but in SSA form and with a \texttt{to\_a} method call at the end.

\begin{lstfloat}
\begin{lstlisting}[caption={[Example: Iterative \textsc{Ikra-Ruby} computation in SSA form]Example: Iterative \textsc{Ikra-Ruby} computation in host section in SSA form}, label={lst:hs_ssa_form_ex}]
result = Ikra.host_section do
    arr<@$_1$@> = input.to_command(dimensions: [2, 3])
    for i in 0...10
        arr<@$_2$@> = <@$\phi$@>(arr<@$_1$@>, arr<@$_6$@>)
        if arr<@$_2$@>.preduce(:+)[0] % 2 == 0
            arr<@$_3$@> = arr<@$_2$@>.pmap do |i| i+1; end <@\hfill \textcolor{gray}{\textit{\# $\mathit{map}_A$}}@>
        else
            arr<@$_4$@> = arr<@$_2$@>.pmap do |i| i+2; end <@\hfill \textcolor{gray}{\textit{\# $\mathit{map}_B$}}@>
        end
        arr<@$_5$@> = <@$\phi$@>(arr<@$_3$@>, arr<@$_4$@>)
        arr<@$_6$@> = arr<@$_5$@>.pmap do |i| i+3; end <@\hfill \textcolor{gray}{\textit{\# $\mathit{map}_C$}}@>
    end
    arr<@$_7$@> = <@$\phi$@>(arr<@$_1$@>, arr<@$_6$@>)
    arr<@$_7$@>.to_a
end
\end{lstlisting}
\end{lstfloat}

In the following paragraphs, we take a look at the inferred types of all \texttt{arr}$_i$ variables. \texttt{arr}$_1$ is an identity command for the Ruby array \texttt{input} and \texttt{arr}$_2$ cannot be fully inferred yet because \texttt{arr}$_6$ is still unknown.

\begin{equation*}
\begin{split}
\mathit{type}(\texttt{arr}_1) & = \mathit{id} \\
\mathit{type}(\texttt{arr}_2) & = \{ \,\, \mathit{type}(\texttt{arr}_1), \mathit{type}(\texttt{arr}_6) \,\, \} = \{ \,\, \mathit{id}, \mathit{type}(\texttt{arr}_6) \,\, \}
\end{split}
\end{equation*}

Next, we infer the types for the first two \textit{map} operations. Different subscripts of \textit{map} operations indicate that the operations are different array commands.

\begin{equation*}
\begin{split}
\mathit{type}(\texttt{arr}_3) = & \, \mathit{map}_A[\mathit{type}(\texttt{arr}_2)] = \{ \,\, \mathit{map}_A[\mathit{id}], \mathit{map}_A[\mathit{type}(\texttt{arr}_6)] \,\, \} \\
\mathit{type}(\texttt{arr}_4) = & \, \mathit{map}_B[\mathit{type}(\texttt{arr}_2)] = \{ \,\, \mathit{map}_B[\mathit{id}], \mathit{map}_B[\mathit{type}(\texttt{arr}_6)] \,\, \} \\
\mathit{type}(\texttt{arr}_5) = & \, \{ \,\, \mathit{type}(\texttt{arr}_3), \mathit{type}(\texttt{arr}_4) \,\, \} \\
 = & \, \{ \,\, \mathit{map}_A[\mathit{id}], \mathit{map}_A[\mathit{type}(\texttt{arr}_6)], \mathit{map}_B[\mathit{id}], \mathit{map}_B[\mathit{type}(\texttt{arr}_6)] \,\, \} \\
\end{split}
\end{equation*}

\noindent Next, we infer the type of the last \textit{map} operation.

\begin{equation*}
\begin{split}
\mathit{type}(\texttt{arr}_6) = & \, \mathit{map}_C[\mathit{type}(\texttt{arr}_5)] \\
 = & \, \{ \,\, \mathit{map}_C[\mathit{map}_A[\mathit{id}]], \mathit{map}_C[\mathit{map}_A[\mathit{type}(\texttt{arr}_6)]], \\
  & \,\,\,\,\,\,\, \mathit{map}_C[\mathit{map}_B[\mathit{id}]], \mathit{map}_C[\mathit{map}_B[\mathit{type}(\texttt{arr}_6)]] \,\, \}
\end{split}
\end{equation*}

As can be seen from the definitions above, the type of \texttt{arr}$_6$ is circular. If we try to fully expand its definition, it will have an infinite number of elements. 

\textsc{Ikra-Ruby} breaks this circular type by inserting a \texttt{to\_a.to\_command} method call, which will launch the kernel, return its result as an array and then box it in another array command. Consequently, this method call will stop the kernel fusion process. \textsc{Ikra-Ruby} currently inserts the method call in Line~11, but it could also be inserted in Lines~4 or 10.

\smallskip
\begin{lstlisting}[firstnumber=11]
arr<@$_6$@> = arr<@$_5$@>.to_a.to_command.pmap do |i| i + 3; end <@\hfill \textcolor{gray}{\textit{\# $\mathit{map}_C$}}@>
\end{lstlisting}
\smallskip

We can now complete the type inference process and fill in the \texttt{arr}$_2$ placeholders in the other definitions.

\begin{equation*}
\begin{split}
\mathit{type}(\texttt{arr}_6) = & \, \mathit{map}_C[\mathit{id}] \\
\mathit{type}(\texttt{arr}_2) = & \, \{ \,\, \mathit{id}, \mathit{map}_C[\mathit{id}] \,\, \} \\
\mathit{type}(\texttt{arr}_5) = & \, \{ \,\, \mathit{type}(\texttt{arr}_3), \mathit{type}(\texttt{arr}_4) \,\, \} \\
 = & \, \{ \,\, \mathit{map}_A[\mathit{id}], \mathit{map}_A[\mathit{map}_C[\mathit{id}]], \mathit{map}_B[\mathit{id}], \mathit{map}_B[\mathit{map}_C[\mathit{id}]] \,\, \} \\
\mathit{type}(\texttt{arr}_7) = & \, \{ \,\, \mathit{type}(\texttt{arr}_1), \mathit{type}(\texttt{arr}_6) \,\, \} = \{ \,\, \mathit{id}, \mathit{map}_C[\mathit{id}] \,\, \}
\end{split}
\end{equation*}

For code generation, only \texttt{arr}$_2$, \texttt{arr}$_5$ and \texttt{arr}$_7$ are of interest, because their result is accessed. \textsc{Ikra-Ruby} generates kernels and invocations for them: The static type of a variable (or union type class ID field if polymorphic) determines the kernel to be launched. In total, \textsc{Ikra-Ruby} generates the following kernels in this example. Some of them might never be launched at runtime.

\begin{multicols}{2}
\begin{description}
    \item [\normalfont (5.1)] $\mathit{map}_A[\mathit{id}]$
    \item [\normalfont (5.2)] $\mathit{map}_A[\mathit{map}_C[\mathit{id}]]$
    \item [\normalfont (5.3)] $\mathit{map}_B[\mathit{id}]$
    \item [\normalfont (5.4)] $\mathit{map}_B[\mathit{map}_C[\mathit{id}]]$
    \item [\normalfont (7.1)] $\mathit{id}$
    \item [\normalfont (7.2)] $\mathit{map}_C[\mathit{id}]$
    \item [\normalfont (r.1)] $\mathit{reduce}[\mathit{id}]$
    \item [\normalfont (r.2)] $\mathit{reduce}[\mathit{map}_C[\mathit{id}]]$
\end{description}
\end{multicols}

\subsection{Code Generation}
After type inference, \textsc{Ikra-Ruby} generates C++ source code for the host section. Expressions of array command type, have a generated \texttt{array\_command\_t} type in the C++ code. This struct contains a pointer to the cached result of the array command. If an expression's polymorphic type can be one of multiple array commands, then \textsc{Ikra-Ruby} uses a union type in the generated C++ code and the class ID indicates the exact array command at runtime.

In the above example, all variables except for \texttt{arr}$_1$ are polymorphic and will have type \texttt{union\_t} in the generated C++ code. The class ID field is used to determine which kernel should be launched. For example, either kernel 7.1 or kernel 7.2 should be launched in Line~14 depending on whether there was at least one loop iteration. This information is implicitly encoded in the class ID field and \textsc{Ikra-Ruby} generates a switch-case statement, similar to polymorphic method calls. In the generated code, Lines~2, 4, 6, 8, 10 and 13 do not launch a kernel but merely return a new \texttt{array\_command\_t} object, possibly wrapped inside a union type struct containing the class ID for the command.

Whenever an array command access is detected, \textsc{Ikra-Ruby} generates a CUDA kernel for the array command and a kernel invocation snippet which checks if a cached result is available and otherwise transfers data (if necessary) and launches the kernel. Since array commands may have dependent commands, generated kernels may consist of multiple fused parallel operations.

\subsection{Benchmarks}
\begin{sidewaysfigure}
    \includegraphics[width=\textwidth]{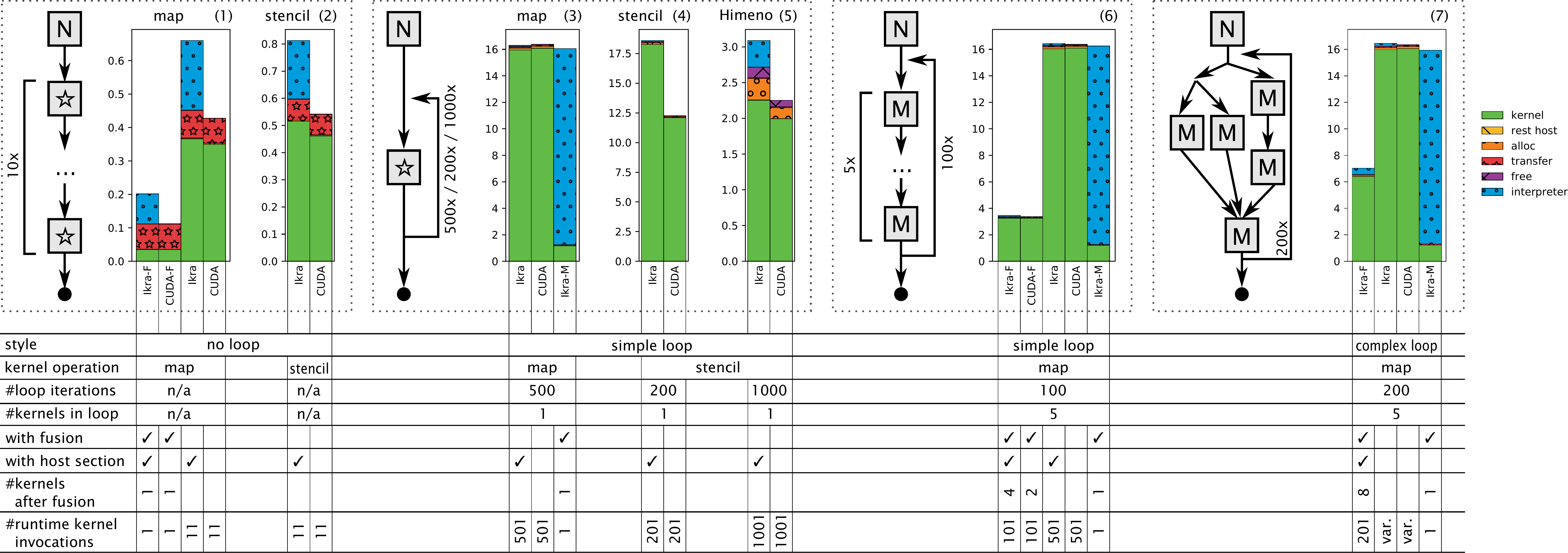}
    \caption[\textsc{Ikra-Ruby} microbenchmark results]{Microbenchmark runtime in seconds. \emph{Ikra-F} is \textsc{Ikra-Ruby} with all code in a single host section and kernel fusion. \emph{Ikra} is without kernel fusion. \emph{Ikra-M} is a lower bound where all code is in a single kernel (even among iterations). \emph{CUDA-F} and \emph{CUDA} are hand-written baseline implementations with/without (manual) kernel fusion. Compilation time (not shown here) is around 2 seconds for \textsc{Ikra-Ruby}-generated code.}
    \label{fig:ikra_ruby_benchmarks}
\end{sidewaysfigure}
Figure~\ref{fig:ikra_ruby_benchmarks} shows the runtime performance of a number of microbenchmarks\footnote{Source code: \url{https://github.com/prg-titech/ikra-ruby}, branch \texttt{array17}} (small parallel operations, only \textit{for} loops) of the current \textsc{Ikra-Ruby} implementation in various configurations. Benchmarks were run on a computer with an Intel Core i7-6820HQ CPU (2.70 GHz), 32 GB RAM, an NVIDIA GeForce 940MX GPU, Ubuntu 16.04.1 (kernel version \texttt{4.4.0-43-generic}), Ruby 2.3.1 and the CUDA Toolkit V8.0.44. Program~5 is a 3D stencil computation on a matrix of size $129 \times 65 \times 65$. All other programs operate on matrices of size $6 \times 10^7$ (228~MB) with a one-dimensional CUDA block size of 1024. 

The benchmarks show the performance speedup due to kernel fusion and how generated \textsc{Ikra-Ruby} code performs in comparison to hand-written CUDA code. We compare 7 different programs with various compilation strategies. For every program, we show the structure of parallel operations (control/data flow graph). The square boxes indicate parallel operations and the arrows indicate data flow. The letter \emph{M} indicates a \textit{map} operation, the letter \emph{N} a \textit{new} operation, and the star a can be either \textit{map} or \textit{stencil}.

Programs~1, 6 and 7 show the benefit of kernel fusion. In Program~1, a single kernel is launched for all 11 parallel operations, giving a 10x speedup compared to the version without kernel fusion. In Program~6, 101 kernels are launched for 501 parallel operations. All 5 kernels within the loop are fused together, giving a 5x speedup compared to the version without kernel fusion. If all 501 \textit{map}/\textit{new} operations are fused together by executing the loop in the Ruby interpreter (no host section; \emph{Ikra-M}), another 2x speedup is possible. However, in the general case, the number of loop iterations is unknown. Moreover, the resulting kernel code becomes large; increasing the number of iterations too much even results in an \texttt{nvcc} compilation error. As can be seen from the \emph{interpreter} time, \textsc{Ikra-Ruby} also spends a very long time in the Ruby interpreter if a tree of 501 array commands is analyzed and fused together. This shows that there is still potential for optimization of the part of \textsc{Ikra-Ruby} that runs in the Ruby interpreter.

Program~7 shows the benefit of kernel fusion in a program with more complex control flow, giving a speedup of 2x. This program generates source code with union types to keep track of which kernel to launch at the end of an iteration. This leads to additional runtime overheads. However, this overhead (\emph{rest host}) is much smaller than the kernel runtime and could not even be measured with confidence in our experiments. The reason that Program~7 does not achieve the same speedup from kernel fusion as Program~6 is that the number of parallel operations inside the loop is smaller and the number of loop iterations is larger (200 vs. 100).

Programs~3--5 consist of a loop with a single \textit{map} or \textit{stencil} operation. In such cases, \textsc{Ikra-Ruby} cannot perform kernel fusion inside the loop, which is why we only report the performance for \textit{Ikra}. \emph{Himeno} is a benchmark with a memory-bound stencil computation. It has a high \emph{alloc} time because \textsc{Ikra-Ruby} (and the CUDA baseline) do currently not reuse memory but allocate a new piece of memory every time an array is updated within the loop. 

All benchmarks have a low \emph{transfer} time, i.e., time spent for transferring data between the device and the host. This is because data is transferred to the host only when the result of a parallel operation is accessed. The loops in all benchmarks are \emph{for} loops with a fixed number of iterations. Only after the last iteration, the result is accessed and transferred back to the host. In a more realistic case, where the control flow (e.g., number of iterations) depends on the data, more data transfer will occur.

When comparing \textsc{Ikra-Ruby}'s performance with hand-written CUDA code, we can see that the kernel running times are almost identical for \textit{map} operations. The generated code of \textit{stencil} operations is not yet fully optimized (or fused). In addition, \textsc{Ikra-Ruby} spends some time in the Ruby interpreter for performing type inference and generating Ruby code. Overall, it spends very little time in the remaining host section code (allocating/comparing union types/array commands, loop overhead).

\paragraph{Number of Generated Kernels}
Due to the kernel fusion process described in Section~\ref{sec:hs_symbolic}, \textsc{Ikra-Ruby} generates one kernel per data flow path (excluding loops). This can lead to a combinatorial explosion of the number of generated kernels. Based on an analysis of the kernel structure of a large number of parallel programs by Shen et al.~\cite{shen2015study}, we believe that the number of kernels remains manageable in real applications. Their work showed that the kernel structure of all analyzed programs is similar to the ones in our benchmarks\footnote{\textsc{Ikra-Ruby}'s programming style could increase the number of parallel sections.} and never more complex than the structure in Program~7.

\subsection{Future Work}
\label{sec:future_work}
Besides improving the overall efficiency of our implementation, we plan extend \textsc{Ikra-Ruby} with kernel fusion of stencil operations and with better memory management in future. Furthermore, there are cases when kernel fusion can lead to a slowdown of programs, e.g., if the first kernel contains highly divergent control flow. We will analyze such programs in the future.

\paragraph{Kernel Fusion of Stencil Operations}
\textsc{Ikra-Ruby} can currently only fuse operations whose input pattern is ``same location''. We plan to extend kernel fusion to certain stencil computations that exhibit a \emph{simple} neighborhood. Stencil computations can be fused with shared memory or with redundant computations. Consider, for example, the following two stencil operations.

\begin{equation*}
\begin{split}
A_1 &= \mathit{stencil}(A_0, [-1, 0], f, 10) \\
A_2 &= \mathit{stencil}(A_1, [-1, 0], g, 10) \\
\end{split}
\end{equation*}
The resulting arrays after each iteration are defined as follows.
\begin{align*}
A_1 &= \begin{bmatrix}
       10 \\
       f(A_0[0], A_0[1]) \\
       f(A_0[1], A_0[2]) \\
       \ldots
     \end{bmatrix}
\end{align*}
\begin{align*}
A_2 &= \begin{bmatrix}
       10 \\
       g(10, f(A_0[0], A_0[1])) \\
       g(f(A_0[0], A_0[1]), f(A_0[1], A_0[2])) \\
       \ldots
     \end{bmatrix}
\end{align*}
The definition of $A_2$ represents the computation of a fused stencil operation. Many arguments of function $g$ are computed twice (redundantly). To reduce this overhead, \textsc{Ikra-Ruby} could split $A_0$ into multiple subarrays, assign each subarray to a CUDA block and store intermediate results in shared memory. Inter-block synchronization or redundant computation is then only necessary at the subarray borders (\emph{ghost region}~\cite{Kjolstad:2010:GCP:1953611.1953615}). An optimized version of this technique is known as \emph{temporal blocking}~\cite{temp_blocking_naoya}. 

\paragraph{Reusing Memory}
Many programs that update a vector or matrix iteratively only need access to the data of the preceding iteration. To save memory and for performance reasons (e.g., caching), CUDA programmers allocate only one (for \textit{combine} operations) or two (for \textit{stencil} operations) arrays in hand-written code and keep writing to these arrays. However, all parallel operations in \textsc{Ikra-Ruby} (except \texttt{peach}) are \emph{functional} and return a new array instead of modifying the existing array in-place. Furthermore, \textsc{Ikra-Ruby} does not have a garbage collector, so the memory is released only after the execution of the host section or if the programmer frees memory explicitly in the Ruby code. 

To decide whether it is safe to reuse a previously allocated array, \textsc{Ikra-Ruby} must perform an escape analysis. This is to ensure that data from a previous iteration is not accessed at a later point of time. Such advanced memory management issues are subject to future work.

\subsection{Related Work}
Kernel fusion is an optimization that is supported in many other GPGPU frameworks and languages, but their focus is on different aspects. For example, Harlan~\cite{Holk:2014:RMM:2660193.2660244} supports nested kernels, Futhark~\cite{Henriksen:2016:DGP:2935323.2935326} has support for nested parallelism and a powerful fusion engine for map/reduce combinations, and Kernel Weaver focuses on database queries~\cite{Wu:2012:KWA:2457472.2457490}. Furthermore, all of these tools focus on statically-typed programming languages, making translation within a kernel easier compared to \textsc{Ikra-Ruby}, because no union types are required and making kernel fusion itself easier, because it is known ahead of execution time which kernels are executed together.

\section{A Data Layout DSL for \textsc{Ikra-Cpp}}
\label{sec:data_layout_dsls_ikracpp}
Maintaining an SOA layout in C++/CUDA manually is troublesome. SOA code is less readable and less expressive than AOS code and native C++ language constructs for object-oriented programming cannot be used in SOA code.

To give programmers the performance benefit of SOA and the expressiveness of AOS-style object-oriented programming at the same time, we extended \textsc{Ikra-Cpp} with a lightweight, embedded DSL\footnote{\textsc{Ikra-Cpp}'s data layout DSL is a \emph{performance-oriented DSL}, aiming to ``make the compiler more productive (producing better code)''~\cite{DBLP:journals/corr/abs-1109-0778}. We call it a \emph{DSL} because it defines a new syntax/notation for application code.} (Listing~\ref{lbl:ikra_cpp_example_code}). This DSL is implemented entirely in C++ with template metaprogramming, operator overloading, helper classes and preprocessor macros. It can be used on CPUs and GPUs and should work with every modern C++14 compiler and the NVIDIA CUDA Toolkit 9.0 or higher (in GPU mode).

\begin{lstfloat}
\begin{lstlisting}[linewidth=\textwidth, language=c++, numbers=left, caption={[\textsc{Ikra-Cpp}: N-body simulation with SOA layout but AOS notation]N-body simulation in \textsc{Ikra-Cpp}: SOA layout but AOS notation}, label={lbl:ikra_cpp_example_code}, morekeywords={float_, assert}]
class Body : public IkraSoaBase<Body, 50> {
 public: <@\emph{IKRA\_INITIALIZE\_CLASS}@>
  float_ pos_x = 0.0;
  float_ pos_y = 0.0;
  float_ vel_x = 1.0;
  float_ vel_y = 1.0;
  float_ force_x;
  float_ force_y;
  float_ mass;

  Body(float m, float x, float y) : mass(m), pos_x(x), pos_y(y) {}

  void move(float dt) {
    pos_x = pos_x + vel_x * dt;
    pos_y = pos_y + vel_y * dt;
  }
};

// This macro defines SOA arrays and an object (instances) counter.
<@\emph{IKRA\_HOST\_STORAGE}@>(Body);

void create_and_move() {
  Body* b = new Body(150, 1.0, 2.0);
  b->move(0.5);
  assert(b->pos_x == 1.5);
}
\end{lstlisting}
\end{lstfloat}




\paragraph{Contribution}
The main contribution of \textsc{Ikra-Cpp} is twofold. First, to the best of our knowledge, \textsc{Ikra-Cpp} is the first C++ data/memory layout DSL that supports a wide range of OOP abstractions, most notably member function calls and constructors. Second, there exist other solutions that allow for a limited set of OOP abstractions (e.g., referencing objects with class pointers instead of IDs) with a custom data/memory layouts in C++ (Section~\ref{sec:related_work_ikra_cpp_paper}). However, these solutions rely on external tools such as custom preprocessors or compiler/language extensions. In contrast, \textsc{Ikra-Cpp} is implemented entirely in C++ and requires only a modern C++ compiler.



\subsection{Language Overview}
\label{sec:arch}
In this section, we describe the basic functionality of \textsc{Ikra-Cpp}'s data layout DSL, focusing on host (CPU) code.

\paragraph{Notation and API}
A class whose objects are stored as SOA is called a \emph{SOA class} and its instances are called \emph{SOA objects}. Recall that classes in \textsc{Ikra-Cpp} must inherit from \texttt{IkraBase}. This class template is for AOS layouts. If we would like to store objects in SOA layout, classes must inherit from \texttt{IkraSoaBase}. This class template provides useful helper methods and type aliases. Same as with \texttt{IkraBase}, the maximum number of instances of an SOA class is a compile-time constant and template parameter of \texttt{IkraSoaBase} (Listing~\ref{lbl:ikra_cpp_example_code}, Line~1).

The main programming restriction of \texttt{IkraSoaBase} compared to \texttt{IkraBase} is that SOA objects can only be created with the \texttt{new} keyword and must be referred to with pointers or references (Listing~\ref{lst:ikra_soa_layput_restrict}). Stack/static allocation is not allowed, because the fields of an object are not stored as a consecutive chunk of data as in a traditional AOS layout. They can only reside within a structure of arrays, i.e., within the storage buffer generated by \texttt{IKRA\_*\_STORAGE}.

\begin{lstfloat}
\begin{lstlisting}[language=c++, caption={[\textsc{Ikra-Cpp}: Allocation restrictions of \texttt{IkraSoaBase}]Allocation restrictions of \texttt{IkraSoaBase}}, label={lst:ikra_soa_layput_restrict}, morekeywords={assert}]
class A : public IkraBase<A, 100> { /* ... */ };
class B : public IkraSoaBase<B, 100> { /* ... */ };

<@\emph{IKRA\_HOST\_STORAGE}@>(A);
<@\emph{IKRA\_HOST\_STORAGE}@>(B);

int main() {
  A obj1;               // OK, but stored outside of the array of structures.
  B obj2;               // Error
  A* ptr1 = new A();    // OK. Allocated as AOS within storage buffer.
  B* ptr2 = new B();    // OK. Allocated as SOA within storage buffer.
  B& obj3 = *obj2;      // OK
}

void function1(B obj);   // Error
void function2(B* ptr);  // OK
void function3(B& ptr);  // OK
\end{lstlisting}
\end{lstfloat}

There is one main change in notation compared to our original \textsc{Ikra-Cpp} implementation from Section~\ref{sec:express_smmo}. Regardless of AOS or SOA layout, programmers must now use special \emph{proxy field types} for field declarations. These proxy types are available for all primitive C++ types and end with an underscore, e.g., \texttt{float\_} instead of \texttt{float} (Listing~\ref{lbl:ikra_cpp_example_code}, Lines~3--9). While the old notation (plain \texttt{float}) is technically still supported for AOS layouts, we recommend using proxy types because the layout of a class can then be switched from AOS to SOA or vice versa with mininal effort by changing the superclass from \texttt{IkraBase} to \texttt{IkraSoaBase} or vice versa.


\paragraph{Supported C++ OOP Features}
In SOA mode, \textsc{Ikra-Cpp} supports many but not all OOP features/abstractions of C++. This paragraph gives an overview of the main supported features and the main restrictions.

\begin{itemize}
  \item Same as in AOS mode, classes are defined with \textbf{standard C++ notation} (\texttt{class} keyword). Classes cannot be templatized and class inheritance is not possible. Programmers must use the two proprocessor macros \texttt{IKRA\_INITIALIZE\_CLASS} and \texttt{IKRA\_*\_STORAGE}.
  \item Objects must be referred to with \textbf{class pointers} or object references.
  \item Member fields must be declared with \textbf{proxy types}. \textsc{Ikra-Cpp} provides proxy types for all \textbf{primitive types}. Furthermore, there is a notation for defining proxy types for other base types.
  \item Member functions must be \textbf{non-virtual} but can be \textbf{templatized}.
  \item \textbf{Class constructors}, including field initializers, are supported.
  \item Class constructors and member functions can be \textbf{overloaded}.
  \item Instance creation is supported \textbf{only with the \texttt{new} keyword}. Same as in AOS mode, existing objects cannot be deleted (allocation only, no deallocation).
  \item Given an object pointer, members can be accessed with the C++ \textbf{member of pointer} (arrow) operator. Given an object reference, members can be accessed with the C++ \textbf{member of object} (dot) operator. This is the main feature that gives \textsc{Ikra-Cpp} code an object-oriented \emph{look and feel}, even in SOA mode.
\end{itemize}

We will lift some of these restrictions with \textsc{DynaSOAr}, which is an extension of \textsc{Ikra-Cpp} with a dynamic memory allocator (Section~\ref{sec:chapter_dynasoar}).



\subsection{Implementation Details}
\label{sec:implementation_details}
\textsc{Ikra-Cpp} is implemented entirely in C++. It does not need a separate compiler or preprocessor/code generator. It is based on four ideas:

\begin{itemize}
  \item All data is stored in a large, statically-allocated \textbf{storage buffer} (\texttt{char} array, generated by \texttt{IKRA\_*\_STORAGE}).
  \item Upon allocation, all objects are assigned unique \textbf{integer IDs} (starting from 1).
  \item Objects are referenced with \textbf{fake pointers} that encode their object ID. This is a form of \emph{type punning}, i.e., ``a programming technique that subverts or circumvents the type system of a programming language in order to achieve an effect that would be difficult or impossible to achieve within the bounds of the formal language''~\cite{wiki:typeunning}. Type punning can lead to fragile code and a future, more mature implementation of \textsc{Ikra-Cpp} could rely on a more powerful preprocessor such as ROSE~\cite{doi:10.1142/S0129626400000214}.
  \item Several \textbf{operators} of proxy field types are \textbf{overridden} to decode IDs from fake pointers and to calculate physical memory addresses within the storage buffer.
\end{itemize}

In this section, we explain those concepts in more detail, using a more verbose notation. The code in Listing~\ref{lst:ikra_expanded_code} is identical to the one in Listing~\ref{lbl:ikra_cpp_example_code} (without constructor), but with simplified, expanded preprocessor macros. In particular, note that tokens like \texttt{float\_} are preprocessor macros that expand to different proxy field type instantiations like \texttt{float\_\_<...>}\footnote{Internally, this is implemented with the \texttt{\_\_COUNTER\_\_} preprocessor macro, which is supported by most compilers. See \url{https://gcc.gnu.org/onlinedocs/cpp/Common-Predefined-Macros.html}.}, from now on simply called \emph{proxy types}.

\begin{lstfloat}
\begin{lstlisting}[linewidth=\textwidth, language=c++, label={lst:ikra_expanded_code}, numbers=left, caption={[\textsc{Ikra-Cpp}: Macro-expanded \texttt{Body} class]Macro-expanded \texttt{Body} class from a Listing~\ref{lbl:ikra_cpp_example_code} (simplified)}]
class Body : public IkraSoaBase<Body> {
  const static int kMaxInst = 50;
  const static int kObjSize = 7 * 4;  <@\hfill \textcolor{gray}{\textsf{7 floats = 28 bytes}}\;\;@>

  static char storage[kMaxInst * kObjSize];
  static int size = 0; 

  float__<1, 0> pos_x = 0.0; <@\hfill \textcolor{gray}{\textsf{field index = 1, offset = 0}}\;\;@>
  float__<2, 4> pos_y = 0.0; <@\hfill \textcolor{gray}{\textsf{field index = 2, offset = 4}}\;\;@>
  float__<3, 8> vel_x = 1.0; <@\hfill \textcolor{gray}{\textsf{field index = 3, offset = 8}}\;\;@>
  float__<4, 12> vel_y = 1.0; <@\hfill \textcolor{gray}{\textsf{field index = 4, offset = 12}}\;\;@>
  float__<5, 16> force_x; <@\hfill \textcolor{gray}{\textsf{field index = 5, offset = 16}}\;\;@>
  float__<6, 20> force_y; <@\hfill \textcolor{gray}{\textsf{field index = 6, offset = 20}}\;\;@>
  float__<7, 24> mass; <@\hfill \textcolor{gray}{\textsf{field index = 6, offset = 24}}\;\;@>

  static Body* get(int id); <@\hfill \textcolor{gray}{\textsf{Calculate Address}}\;\;@>

  void* operator new() { return get(++size); }

  void move(float dt) {
    pos_x = pos_x + vel_x * dt;
    pos_y = pos_y + vel_y * dt;
  }
};
\end{lstlisting}
\end{lstfloat} 

\paragraph{SOA Objects}
SOA objects can only be referred to with pointers or references. In particular, they cannot be stack-allocated. This is because the address of an SOA object is not a physical memory location but encodes an object ID (\emph{fake pointer}). Based on an object ID and the memory locations of the SOA arrays within the storage buffer, the actual, physical memory location of each field value of that object can be computed. This can be seen as a form of \emph{address translation}. This translation is performed transparently and implemented in C++, as part of the data layout DSL.

All objects of an SOA class \texttt{C} have unique IDs between $[1; \mathit{maxInst}(\texttt{C})]$, where \emph{maxInst}(\texttt{C}) is the user-specified maximum number of objects of \texttt{C}. ID 0 is reserved for null pointers. E.g., running \texttt{new C()} for the first time returns a \texttt{C*} pointer encoding ID \texttt{1}. This pointer does \emph{not} point to an actual memory location. Accessing the address of this fake pointer would most likely result in a memory access violation (Listing~\ref{lst:deref_fake_ptr}).

\begin{lstfloat}
\begin{lstlisting}[language=c++, caption={[\textsc{Ikra-Cpp}: Dereferencing a fake pointer]Dereferencing a fake pointer}, label={lst:deref_fake_ptr}]
C* obj = new C();
char* ptr = reinterpret_cast<char*>(obj);
printf("%c\n", *ptr);  // This usually works, but with IkraSoaBase it segfaults!
\end{lstlisting}
\end{lstfloat}

\paragraph{Proxy Types}
SOA fields are declared with \emph{proxy types} (\texttt{Field\_} template instantiations, Listing~\ref{lst:proxy_types}). These types behave like normal C++ types (base types) in most cases, but access data at a different physical location inside the storage buffer. Proxy type values always appear as \emph{lvalues}, i.e., as values that have a memory location. This is because their implementation calculates the actual, physical data location based on their own lvalue address. Proxy types support the following operations, implemented via operator overloading.

\begin{itemize}
  \item \emph{Reading a Value:} A proxy type lvalue can be converted to its base type value without an explicit typecast (\emph{implicit conversion operator}\footnote{The \texttt{auto} keyword is not supported. E.g., a proxy type lvalue cannot be assigned to a variable declared as \texttt{auto} without an explicit type cast, unless the variable is of reference type.}, Line~17).
  \item \emph{Writing a Value:} A base type value and a proxy type lvalue can be assigned to a proxy type lvalue by overloading the  \emph{assignment operator} (Lines~20, 45).
  \item \emph{Method Call:} For proxy type lvalues with a pointer base type, a method call is forwarded to the object at the physical data location by overloading the \emph{member of pointer ``arrow'' operator}\footnote{Note for experienced C++ programmers: This is similar to how \texttt{std::unique\_ptr} is implemented.} (Line~23).
  \item \emph{Address-of:} It is possible to take the address of a proxy type lvalue (\emph{address-of operator}, Line~26).
  \item \emph{Dereference:} It is possible to dereference a proxy type lvalue (\emph{pointer dereference operator}, Line~29) if its base type is a pointer type.
  \item \emph{Array acccess:} It is possible to access a proxy type lvalue of array base type with array syntax (\emph{array subscript operator}, Line~34).
  \item \emph{Initialization:} Proxy type lvalues can be initialized with base type values or proxy type lvalues; e.g., with a field initializer of a class constructor (constructor, Lines~14, 40).
\end{itemize}

\begin{lstfloat}
\begin{lstlisting}[linewidth=\textwidth, numbers=left, label={lst:proxy_types}, language=c++, caption={[\textsc{Ikra-Cpp}: Implementation of proxy types]Implementation of proxy types}]
template<class Self> <@\hfill \textcolor{gray}{\textsf{Template Parameter ``Self'': CRTP~\cite{Coplien:1995:CRT:229227.229229}}}\;\;@>
class IkraSoaBase {
  template<int Index, int Offset>
  using float__ = Field_<float, Index, Offset, Self>;

  <@\textcolor{gray}{\textbf{template}<\textbf{int} Index, \textbf{int} Offset>}@>
  <@\textcolor{gray}{\textbf{using} int\_\_ = Field\_<\textbf{int}, Index, Offset, Self>;}@>
};

template<typename T, int Index, int Offset, class Owner>
class Field_ {
  // Constructor
  Field_() {}
  Field_(const T& value) { *data_ptr() = value; }

  // Implicit conversion (implicit type cast)
  operator T&() const { return *data_ptr(); }

  // Assignment
  void operator=(const T& value) { *data_ptr() = value; }

  // Method call
  T* operator->() const { return data_ptr(); }

  // Address-of
  T* operator&() { return data_ptr(); }

  // Pointer dereference
  typename std::remove_pointer<T>::type& operator*() {
    return **data_ptr();
  }

  // Array access (subscript operator) to std::array
  T::value_type& operator[](size_t pos) {
    return (*data_ptr())[pos];
  }

  // Support assignment of other proxy types.
  template<int Index2, int Offset2, class Owner2> 
  Field_(const Field_<T, Index2, Offset2, Owner2>& value) {
    *data_ptr() = *value.data_ptr();
  }

  template<int Index2, int Offset2, class Owner2> 
  void operator=(const Field_<T, Index2, Offset2, Owner2>& value) {
    *data_ptr() = *value.data_ptr();
  }

  T* data_ptr() const; <@\hfill \textcolor{gray}{\textsf{Calculate Address (details later)}}\;\;@>
};
\end{lstlisting}
\end{lstfloat}

SOA field types are defined in \texttt{IkraSoaBase} as template instantiations of \texttt{Field\_} (Listing~\ref{lst:proxy_types}). This class provides the necessary operator implementations and calculates the address inside the storage buffer at which the field value of a certain object can be found (Line~49).

\paragraph{Custom Proxy Types}
\textsc{Ikra-Cpp} provides proxy type macros such as \texttt{float\_} for all primitive C++ types. Proxy types of other types can be defined with the helper macro \texttt{field\_} (Listing~\ref{lst:proxy_field_type_notation}). This macro takes as argument the desired base type.

\begin{lstfloat}
\begin{lstlisting}[language=c++, numbers=none, morekeywords={float_, int_, field_}, caption={[\textsc{Ikra-Cpp}: Field proxy type notation]Field proxy type notation}, label={lst:proxy_field_type_notation}]
class DummyClass : IkraSoaBase<DummyClass, 50> {
 public: <@\emph{IKRA\_INITIALIZE\_CLASS}@>
  float_ field_0;  // Base type: float
  field_(void*) field_1;  // Base type: void*
  field_(std::array<float, 10>) field_2;  // Base type: std:array<float, 10>
  int_ field_3;  // Base type: int
  /* ... */
}; <@\emph{IKRA\_HOST\_STORAGE}@>(DummyClass);
\end{lstlisting}
\end{lstfloat}

\paragraph{Address Computation -- Simplified}
Proxy types perform a form of address translation: They translate a fake pointer into a physical memory address (Figure~\ref{fig:addr_fake_ptr_tr}). Proxy objects always appear as lvalues, i.e., values that have an address. This address does not point to allocated memory but is based on a fake pointer.

\begin{figure}
  \centering
  \includegraphics[width=0.8\textwidth]{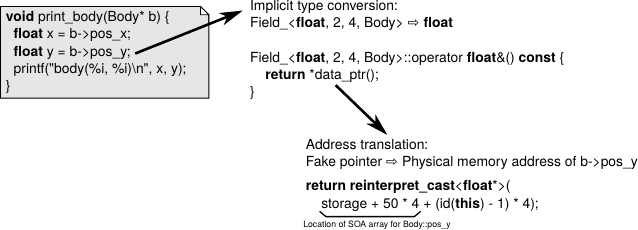}
  \caption[\textsc{Ikra-Cpp}: Address translation of a fake pointer]{Address translation of a fake pointer}
  \label{fig:addr_fake_ptr_tr}
\end{figure}

In the most basic case, given the address of a field proxy object $\mathit{this}$ (i.e., an object of type \texttt{Field\_<...>}), the address of a field value \texttt{C::f} can be computed with the formula below. $\mathit{id}$ is a function that decodes the object ID from a proxy object address. Note that we always translate addresses from the perspective of a field proxy object and not from the perspective of an SOA object, because we overloaded the operators of \texttt{Field\_} and not the operators of the SOA class.

\begin{align*}
\mathit{addr}(\mathit{this}, \texttt{C::f}) = & \; {\mathit{storage}} \\
 & + {\mathit{maxInst}(\texttt{C}) \cdot \mathit{offset}(\texttt{C::f})} \\
 & + ({\mathit{id}(\mathit{this})} - 1) \cdot {\mathit{sizeof}(\texttt{C::f})}
\end{align*}

The first two lines in the equation compute the beginning of the \emph{SOA array} storing all values of \texttt{C::f}. The third line computes the offset into that array. How exactly object IDs are encoded in fake pointers is determined by the \emph{addressing mode} and described in the next few paragraphs. \textsc{Ikra-Cpp} supports three different addressing modes, one of which must be chosen at compile time: \emph{Zero Addressing} and two variants of \emph{Valid Addressing}. The former one is more space-efficient but relies on non-standard C++ constructs, so it might not work with some compilers\footnote{We verified that it works with g++ 5.4.0, clang 3.8.0 and CUDA 9.0.}.

\subsection{Addressing Modes}
\label{sec:addr_mode}
This section describes of three addressing modes. Zero addressing and storage-relative zero addressing are implemented in \textsc{Ikra-Cpp}. In accordance with the C++ zero overhead principle~\cite{Stroustrup2012}, zero addressing is the default mode.

\paragraph{Zero Addressing}
In zero addressing mode (Figure~\ref{fig:boxes_zero_addr}, Listing~\ref{lst:addr_comp_zero_addr}), an object of an SOA class \texttt{C} with ID $i$ is referenced with a \texttt{C*} fake pointer pointing to address $i$ (Listing~\ref{lst:fake_ptrs_zero_addr}).

Values are grouped by field within the storage buffer (SOA layout). No field values are stored for object 0 (null pointer). Given a \texttt{C*} fake pointer $\mathit{obj}$, the physical memory location within the storage buffer of the field value \texttt{C::f}, i.e., \texttt{\&obj->f}, is calculated as follows. Compile-time constants are in blue.

\begin{lstfloat}
\begin{lstlisting}[linewidth=\textwidth, language=c++, numbers=left, caption={[\textsc{Ikra-Cpp}: Address computation in zero addressing]Address computation in zero addressing}, label={lst:addr_comp_zero_addr}, morekeywords={uintptr_t}]
Body* Body::get(int id) {
  return reinterpret_cast<Body*>(id)
}

template<typename T, int Index, int Offset, class Owner>
T* Field_<T, Index, Offset, Owner>::data_ptr() {
  Owner* obj = reinterpret_cast<Owner*>(this);
  return reinterpret_cast<T*>(Owner::storage
      + Offset * Owner::kMaxInst - sizeof(T)
      + sizeof(T) * reinterpret_cast<uintptr_t>(obj));
}
\end{lstlisting}

\begin{lstlisting}[language=c++, caption={[\textsc{Ikra-Cpp}: Fake pointers in zero addressing mode]Fake pointers in zero addressing mode}, label={lst:fake_ptrs_zero_addr}]
C* a = new C();  // First object. ID 1
C* b = new C();  // Second object. ID 2
printf("%p, %p\n", a, b);  // Prints: 0x00000001, 0x00000002
\end{lstlisting}
\end{lstfloat}

\begin{figure}
  \centering
  \includegraphics[width=0.5\textwidth]{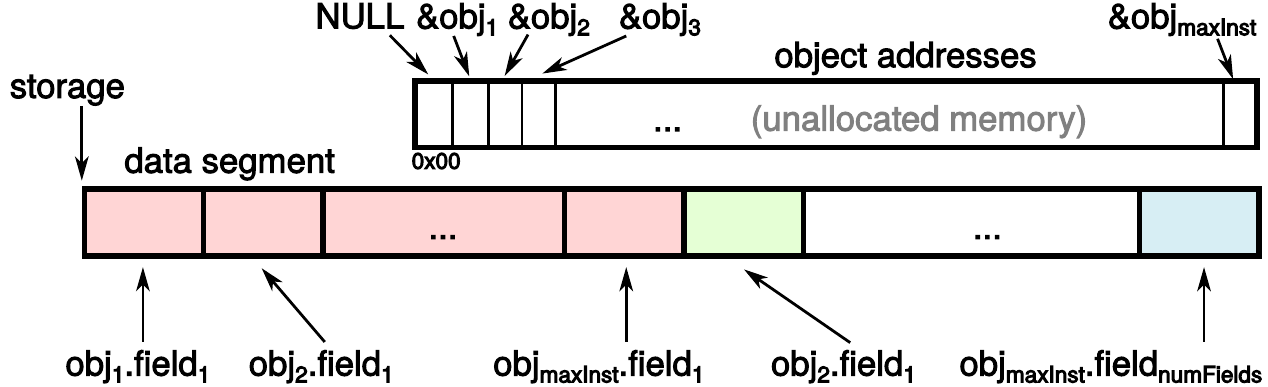}
  \caption[\textsc{Ikra-Cpp}: Storage buffer layout in zero addressing]{Storage buffer layout in zero addressing} 
  \label{fig:boxes_zero_addr}
\end{figure}

\begin{align*}
\mathit{addr}_\mathit{zero}(\mathit{obj}, \texttt{C::f}) = & \; \textcolor{blue}{\mathit{storage}} \\
\mbox{\raisebox{-.25\height}{\crule[blue]{1em}{1em}} \textsf{\textcolor{gray}{constant}} \;\;\;\;\;\;\;\;\;\;\;\;\;\;\;\;\;\;\;\;\;} & + \textcolor{blue}{\mathit{maxInst}(\texttt{C}) \cdot \mathit{offset}(\texttt{C::f})} \\
\mbox{\raisebox{-.25\height}{\crule[red]{1em}{1em}} \textsf{\underline{\textcolor{gray}{variable}}} \;\;\;\;\;\;\;\;\;\;\;\;\;\;\;\;\;\;\;\;\;\;} & - \textcolor{blue}{\mathit{sizeof}(\texttt{C::f})} \\
 & + \textcolor{red}{\underline{\mathit{obj}}} \cdot \textcolor{blue}{\mathit{sizeof}(\texttt{C::f})}
\end{align*}

Intuitively, the address is computed as follows: Start with the address of the storage buffer. Add the offset at which the SOA array for \texttt{C::f} begins. Finally, add the offset at which the \emph{(obj - 1)}\textsuperscript{th} value begins. We have to subtract 1 to account for the fact that object IDs start at 1 but C++ arrays start indexing with 0.

Since the storage buffer is statically allocated, the first three lines of the address calculation are compile-time constants and the fourth line is a strided memory access. After constant folding, this is identical to a hand-written SOA layout with statically allocated SOA field arrays. In a hand-written SOA layout, the address of the field value \texttt{C::f} of an object with ID $i$, i.e., \texttt{\&C\_f[i]}, is computed as follows.

\begin{align*}
\mathit{addr}_\mathit{manual}(\mathit{i}, \texttt{C::f}) = & \; \textcolor{blue}{\texttt{\&C\_f[0]}} \\
 & +\textcolor{red}{ \underline{\mathit{i}}} \cdot \textcolor{blue}{\mathit{sizeof}(\texttt{C::f})}
\end{align*}

Modern compilers are good at peephole optimizations. We compared the compiled assembly code of a single field access for both a hand-written SOA layout and for the equivalent \textsc{Ikra-Cpp} code in C++ and CUDA. The resulting assembly code was identical (Section~\ref{sec:ikra_cpp_code_generation_experiment}).

One crucial assumption of zero addressing is that the C++ object size of SOA classes and \texttt{Field\_} instantiations is zero bytes (e.g., \texttt{sizeof(Body) = 0}). In that case, the fake address of an SOA object is equal to the addresses of all its proxy field objects\footnote{E.g., for a \texttt{Body* b}: $\texttt{b} = \texttt{\&b->pos\_x} = \texttt{\&b->pos\_y} = ...$}. Listing~\ref{lst:addr_comp_zero_addr} shows the implementation of the function \texttt{data\_ptr}, which calculates the physical memory address of an SOA field value within the storage buffer. If \texttt{Field\_} instantiations had a byte size larger than zero, Line~7, which calculates the fake pointer of the SOA object, would have to be changed.

Zero addressing has one main advantage: It allows us to use the C++ constructor syntax without wasting memory. If \texttt{Field\_} instantiations and SOA objects have a C++ object size of zero bytes\footnote{This break pointer arithmetics with SOA object pointers. However, we provide a custom iterator type that programmers can use instead.}, we can use the \texttt{new} keyword for instance creation. In host code (not GPU code), all memory is zero-initialized before running a constructor. Zero-initializing a memory segment of zero bytes is a \emph{no operation}, even if the segment starts at a bogus memory address (\emph{fake pointer}). On the contrary, zero-initializing at least one byte at a bogus memory address will likely result in a memory access violation and crash the program.

According to the C++ standard, the size of a class or struct should be greater than zero (even if it has no members)~\cite{ISO:2012:III}, but many compilers can be instructed to use a size of zero. If this is not supported by a compiler, either valid addressing or a different mechanism for instance creation must be used.

\paragraph{Valid Addressing Mode}
Since zero addressing does not conform to the C++ standard, \textsc{Ikra-Cpp} provides an alternative addressing mode. In \emph{valid addressing}, the C++ object size of every proxy field object is one byte (e.g., \texttt{sizeof(double\_) = 1}). Consequently the C++ object size of every SOA object is $\mathit{numFields}$ bytes. Note that in either addressing mode, the C++ \texttt{sizeof} keyword reports a number that is different from the actual memory consumption an SOA object within the storage buffer.

In order to support the \texttt{new} keyword, the address of an SOA object must then point to valid (allocated) memory; thus the name \emph{valid} addressing. Otherwise, zero-initialization would cause a memory access violation. The challenge of valid addressing is to add an as small as possible amount of padding (\emph{wasted} memory) such that no data is overwritten by zero initialization.

Programmers should use zero addressing if supported by their compiler, since it does not waste any padding memory. We expect the same runtime performance as in zero addressing, because address computation is in both cases reduced to a strided memory access after constant folding.

\paragraph{Storage-relative Zero Addressing}
This addressing mode is one of two variants of valid addressing. In this addressing mode, an object of SOA class \texttt{C} with ID $i$ is referred to with a fake \texttt{C*} pointer pointing to the $i$\textsuperscript{th} byte of the storage buffer (Figure~\ref{fig:rel_zero_addddr}). E.g., if the storage buffer is allocated at address \texttt{0x4000}, then the address of $\texttt{obj}_3$ is $\texttt{0x4003}$. 

The \emph{data segment}, where field values are stored, is identical to the one in zero addressing. However, it starts at byte offset \emph{padding}, i.e., $\mathit{padding}$ many bytes are wasted in this addressing mode.

\begin{align*}
\mathit{padding} = \emph{maxInst}(\texttt{C}) + 1 + \mathit{numFields}(\texttt{C})
\end{align*}

Since all fake object pointers point to one of the first \emph{maxInst}(\texttt{C}) bytes of the storage buffer (\emph{object ptr. address} part in Figure~\ref{fig:rel_zero_addddr}), not a single byte of the data segment is overwritten by zero-initialization, which zeros out \emph{numFields}(\texttt{C}) many bytes. Note that, similar to zero addressing, we have to add 1 in the above formula because object IDs start with 1 (0 is reserved for null pointers).

In general, the physical memory location of a field value \texttt{C::f} of an object with fake pointer $\mathit{obj}$ is calculated as follows. This formula is identical to the one in zero addressing, except for the offset into the data segment and the object ID computation/decoding part.

\begin{figure}
  \includegraphics[width=0.5\textwidth]{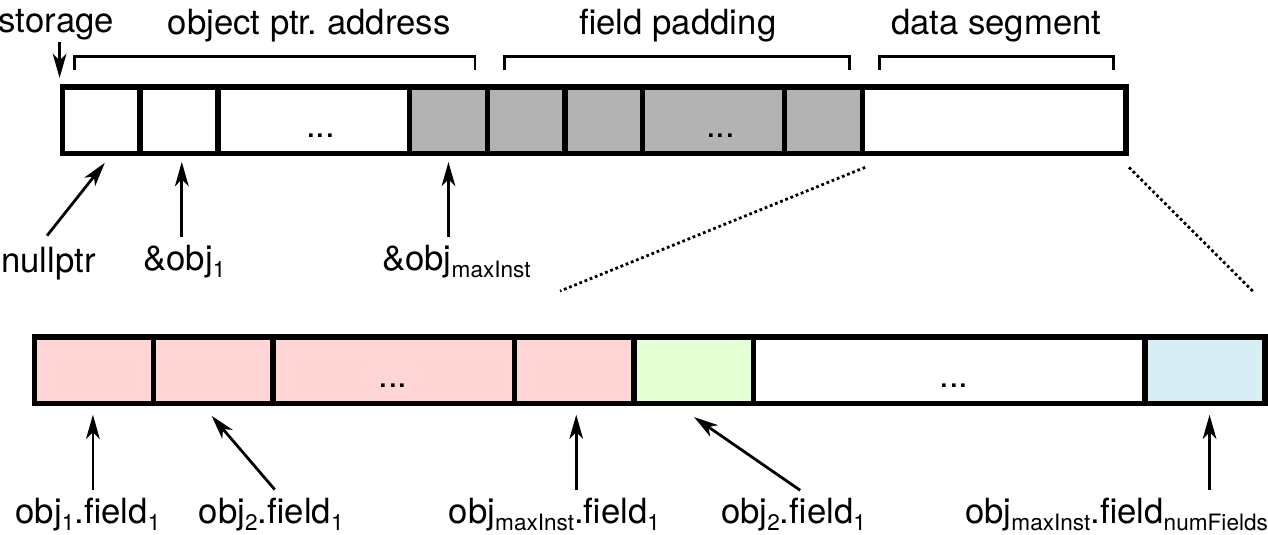}
  \centering
  \caption[\textsc{Ikra-Cpp}: Storage buffer layout in storage-relative zero addressing]{Storage buffer layout in storage-relative zero addressing}

  \label{fig:rel_zero_addddr}
\end{figure}

\begin{align*}
\mathit{addr}_\mathit{valid}(\mathit{obj}, \texttt{C::f}) =\;  & \mathit{storage} \\
 \mbox{\textsf{\textcolor{gray}{data segment offset}} \;\;\;\;\;\;} & \mbox{\raisebox{-.5\height}{\shadowbox{$+ \mathit{maxInst}(\texttt{C}) + 1 + \mathit{numFields}(\texttt{C})$}}} \\
 & + \mathit{maxInst}(\texttt{C}) \cdot \mathit{offset}(\texttt{C::f}) \\
 & - \mathit{sizeof}(\texttt{C::f}) \\
 \mbox{\textsf{\textcolor{gray}{ID computation}\;\;\;\;\;\;\;\;\;\;\;\;}} & + \mbox{\raisebox{-.5\height}{\shadowbox{$(\mathit{obj} - \mathit{storage})$}} $ \cdot \mathit{sizeof}(\texttt{C::f})$}
\end{align*}

Since address computation is done inside field proxy types (i.e., \texttt{Field\_} instantiantions, \emph{not} \texttt{IkraSoaBase}), we have to express the above formula in terms of the address (\textit{this} pointer) of the field proxy object instead of the fake object address \textit{obj}. The fake object address \textit{obj} of an SOA object can be computed based on the address of a field proxy object of \texttt{C::f} as follows, where $\mathit{index}(\texttt{C::f})$ is the field index of \texttt{C::f} (start counting from 1).

\begin{align*}
\mathit{obj} = \mathit{this} - \mathit{index}(\texttt{C::f}) + 1
\end{align*}

Since every \texttt{Field\_} instantiation is 1 byte in size, we can retrieve the base address (fake pointer) \emph{obj} of an object from the address \emph{this} of the $i$\textsuperscript{th} field proxy object of \emph{obj} by subtracting $i-1$ from it. For example, let \texttt{0x4003} be the address of the third field proxy object (\texttt{vel\_x}) of a \texttt{Body} object. Then, $\mathit{obj} = \texttt{0x4003} - 3 + 1 = \texttt{0x4001}$. This fake object pointer can be plugged into $\mathit{addr}_\mathit{valid}(\emph{obj}, \texttt{Body::vel\_x})$. Putting both definitions together, the physical memory location of a field value \texttt{C::f} with respect to its proxy object's address $\mathit{this}$ is then calculated as follows.


\begin{align*}
\mathit{addr}_\mathit{valid}(\mathit{this}, \texttt{C::f}) \; & =  \textcolor{blue}{\mathit{storage}} \\
 + \; & \textcolor{blue}{\mathit{maxInst}(\texttt{C}) + 1 + \mathit{numFields}(C)} \\
 + \; & \textcolor{blue}{\mathit{maxInst}(\texttt{C}) \cdot \mathit{offset}(\texttt{C::f})} \\
 - \; & \textcolor{blue}{\mathit{sizeof}(\texttt{C::f}) \cdot (\mathit{index}(\texttt{C::f}) + \mathit{storage})} \\
 + \; & \textcolor{red}{\underline{\mathit{this}}} \cdot \textcolor{blue}{\mathit{sizeof}(\texttt{C::f})}
\end{align*}

The formula above was rearranged to keep the number of terms small. After constant folding, the address of a field value can be calculated with the same instructions as in zero addressing mode. 

\paragraph{First Field Addressing}
This addressing mode is the second variant of valid addressing and not currently implemented in \textsc{Ikra-Cpp}. Its purpose is to reduce the amount of waste due to the padding area. It could potentially also be useful for virtual function support in the future.

An object of SOA class \texttt{C} with ID $i$ is referred to with a fake \texttt{C*} pointer pointing to the physical memory location of the value of the first field of object $i$ (Figure~\ref{fig:mode_first_field}). If the SOA class has at least one virtual function, then the first field is the vtable pointer\footnote{Most C++ compilers store the vtable pointer in the first 8 bytes of an object.}. If the number of fields of the SOA class is larger than the size of the first field, then the memory of the first field must be padded with $\mathit{sizeof}(\texttt{C::}\mathit{first}) - \mathit{numFields}(\texttt{C})$ bytes to avoid overwriting values of the first field of other objects (with larger object IDs) due to zero initialization. Since object deallocation is not yet supported in \textsc{Ikra-Cpp}, this would currently not be problem because object slots with greater IDs are guaranteed to be empty. However, we will extend \textsc{Ikra-Cpp} with dynamic object deallocation in the next chapter. Given a fake \texttt{C*} pointer $\mathit{obj}$, the memory location of a field \texttt{C::f} is calculated as follows.

\begin{figure}
  \centering
  \includegraphics[width=0.5\columnwidth]{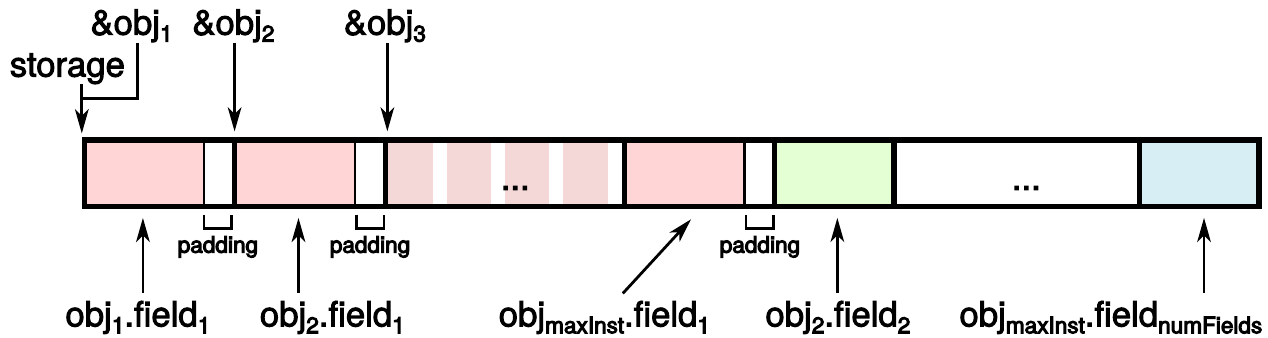}
  \caption[\textsc{Ikra-Cpp}: Storage buffer layout in first field addressing]{Storage buffer layout in first field addressing}
  \label{fig:mode_first_field} 
\end{figure}

\begin{align*}
\mathit{addr}_\mathit{first}(\mathit{obj}, \texttt{C::f})  & = \mathit{storage} \\
 + \; & \mathit{maxInst}(\texttt{C}) \cdot \mathit{offset}^*(\texttt{C::f}) \\
 + \; & \left(\frac{\mathit{obj} - \mathit{storage}}{\mathit{sizeof}^*(\texttt{C::}\mathit{first})} - 1\right) \cdot \mathit{sizeof}^*(\texttt{C::f})
\end{align*}

$\mathit{sizeof}(\texttt{C::f})$ and $\mathit{offset}(\texttt{C::f})$ denote the C++ size of a field type and the offset of a field within the class (Listing~\ref{lbl:ikra_cpp_example_code}). In addition, $\mathit{sizeof}^*(\texttt{C::f})$ and $\mathit{offset}^*(\texttt{C::f})$ also take into account padding that may be added to the first field. The physical memory location of a field value \texttt{C::f} with respect to its proxy object's address $\mathit{this}$ is calculated as follows. We get this formula by plugging in the definition of \emph{obj} from storage-relative zero addressing.

\begin{align*}
\mathit{addr}_\mathit{first}(\mathit{this}, \texttt{C::f})  & = \textcolor{blue}{\mathit{storage} - \mathit{sizeof}(\texttt{C::f})} \\
 + \; & \textcolor{blue}{\mathit{maxInst}(\texttt{C}) \cdot \mathit{offset}^*(\texttt{C::f})} \\
 - \; & \textcolor{blue}{(\mathit{index}(\texttt{C::f}) + \mathit{storage} - 1) \cdot R} \\
 + \; & \textcolor{red}{\underline{\mathit{this}}} \cdot \textcolor{blue}{R} \\
\mbox{where } R  & = \textcolor{blue}{\frac{\mathit{sizeof}^*(\texttt{C::f})}{\mathit{sizeof}^*(\texttt{C::}\mathit{first})}}
\end{align*}

Even though the definition of $\mathit{addr}_\mathit{first}$ contains a fraction, its value is always an integer. However, its calculation is not straightforward. Similar to the previous addressing modes, we rearranged the terms in the above formula such that it can be computed with one addition and one multiplication after constant folding (strided memory access). While the formula always gives us integer values, single parts such as $\mathit{this} \cdot R$ may be fractions. Floating point operations as part of the address computation are highly inefficient and must be avoided.

Therefore, there are two options for implementing first field addressing in \textsc{Ikra-Cpp}. Either we compute the formula differently (with more arithmetic operations) or we enforce $R$ to be an integer, i.e., the size of every field must be a multiple of the size of the first field.

This addressing mode is superior to storage-relative zero addressing only if the field padding size is zero or one byte. Otherwise, there would be no space savings, because field padding is incurred for every object, i.e., $\mathit{maxInst}$ many times. 

\subsection{Code Generation Experiment}
\label{sec:ikra_cpp_code_generation_experiment}
\textsc{Ikra-Cpp}'s data layout DSL is \emph{embedded} into C++/CUDA and does not rely on an external preprocessor or code generator. It is relies heavily on template metaprogramming, operator overloading and type punning. Therefore, the code that the C++ compiler is seeing is quite complex. In this section, we analyze how well C++ compilers can optimize such code.

\begin{lstfloat}
\begin{lstlisting}[language=c++, morekeywords={uintptr_t, int_}, caption={Example: Writing a field of an object in AOS/SOA/\textsc{Ikra-Cpp}}, label={lst:writing_field_of_obj}, numbers=none]
// Case 1: AOS style
class DummyClassAos {
 public:
  int field0;
  int field1;
};

// Case 2: Handwritten SOA
int soa_field0[100];
int soa_field1[100];

// Case 3: Ikra-Cpp (SOA)
class DummyClassIkraCpp : public IkraSoaBase<DummyClassIkraCpp, 100> {
 public: <@\emph{IKRA\_INITIALIZE\_CLASS}@>
  int_ field0;
  int_ field1;
}; <@\emph{IKRA\_HOST\_STORAGE}@>(DummyClassIkraCpp)

void write_field0(DummyClassIkraCpp* obj)        { obj->field0    = 0x7777; }
void write_field0_handwritten_soa(uintptr_t id)  { soa_field0[id] = 0x7777; }
void write_field0_aos(DymmyClassAos* obj)        { obj->field0    = 0x7777; }
\end{lstlisting}

\begin{lstlisting}[columns=fixed, caption={[Generated assembly code for field write]Generated assembly code for functions in Listing~\ref{lst:writing_field_of_obj}}, label={lst:generated_assembly_code_funcs}, numbers=none]
00000000004005e0 <_Z12write_field0P9DummyClassIkraCpp>:
  <@\textcolor{red}{4005e0:}@> c7 04 bd b0 a3 60 00  <@\textcolor{blue}{movl}@>   $0x7777,0x60a3b0(,%rdi,4)
  <@\textcolor{red}{4005e7:}@> 77 77 00 00 
  <@\textcolor{red}{4005eb:}@> c3                    <@\textcolor{blue}{retq}@>   
  <@\textcolor{red}{4005ec:}@> 0f 1f 40 00           <@\textcolor{blue}{nopl}@>   0x0(%rax)

0000000000400620 <_Z21write_field0_handwritten_soam>:
  <@\textcolor{red}{400620:}@> c7 04 bd 60 10 60 00  <@\textcolor{blue}{movl}@>   $0x7777,0x601060(,%rdi,4)
  <@\textcolor{red}{400627:}@> 77 77 00 00 
  <@\textcolor{red}{40062b:}@> c3                    <@\textcolor{blue}{retq}@>   
  <@\textcolor{red}{40062c:}@> 0f 1f 40 00           <@\textcolor{blue}{nopl}@>   0x0(%rax)

0000000000400640 <_Z25write_field0_aosP16DummyClassAos>:
  <@\textcolor{red}{400640:}@> c7 07 77 77 00 00     <@\textcolor{blue}{movl}@>   $0x7777,(%rdi)
  <@\textcolor{red}{400646:}@> c3                    <@\textcolor{blue}{retq}@>   
  <@\textcolor{red}{400647:}@> 66 0f 1f 84 00 00 00  <@\textcolor{blue}{nopw}@>   0x0(%rax,%rax,1)
  <@\textcolor{red}{40064e:}@> 00 00 
\end{lstlisting}
\end{lstfloat}

Listing~\ref{lst:writing_field_of_obj} shows the setup of our experiment. We would like to analyze the assembly code that a C++ compiler generates for writing an integer value to a field of a given object. We consider three cases.
\begin{enumerate}
  \item An ordinary C++ class, where the object is referred to with an object pointer (AOS style).
  \item A hand-written SOA layout, where the object is referred to with an integer index into SOA arrays.
  \item An \textsc{Ikra-Cpp} class in SOA layout, where the object is referred to with a fake pointer in zero addressing mode.
\end{enumerate}

Listing~\ref{lst:generated_assembly_code_funcs} shows the generated assembly code of gcc~5.4.0 with \texttt{-O3} optimization. Lower optimization levels result in considerably less efficient code because the overloaded operators and the \texttt{data\_ptr} function (Listing~\ref{lst:proxy_types}) are not inlined.

We can see that the generated assembly code of the \textsc{Ikra-Cpp} version is nearly identical to the hand-written SOA version. The assembly code of both functions differs only in the address of the (conceptual) SOA array. This shows that, at least in this small experiment, a modern compiler is able to optimize \textsc{Ikra-Cpp} code through peephole optimizations such as constant folding.


\paragraph{Automatic Vectorization}
Unfortunately, this is not always the case for other compiler optimizations. One example is automatic loop vectorization. We analyzed the generated assembly code of a benchmark that runs \texttt{Body::move} (Listing~\ref{lbl:ikra_cpp_example_code}) on the host (CPU) in a loop with many iterations. In the hand-written SOA code, gcc and clang vectorize the method call \texttt{Body::move} of multiple loop iterations with SSE (\emph{Streaming SIMD Extensions}) processor instructions. However, only gcc performs the equivalent loop vectorization with \textsc{Ikra-Cpp} code. Clang is able to apply optimizations like loop unrolling but considers the memory reads/writes\footnote{The fundamental problem is pointer casting. In the simplest case, an expression like \texttt{array[reinterpret\_cast<uintptr\_t>(id)]}, where \texttt{id} is a pointer encoding an integer array offset is already considered unsafe. This could potentially be solved with the C/C++ \texttt{restrict} keyword.} as potentially \emph{dependent} memory operations and thus unsafe for vectorization.

There are three approaches to solve this problem. First, we can try rewriting the address computation part of \textsc{Ikra-Cpp}, in an attempt to give the compiler additional hints that trigger optimizations. Due to our type punning-based implementation, this approach is fragile and could break at any time. Second, code can be vectorized manually, either with C++ SSE intrinsics or with a vectorization framework like \emph{Sierra}~\cite{LeiBa:2012:ECL:2145816.2145825, LeiBa:2014:SSE:2568058.2568062}. Considering that real applications, which exhibit code that is more complex than our example here, cannot be automatically vectorized (yet) with today's compilers, even if written in SOA style, this approach seems feasible to us. Third, \textsc{Ikra-Cpp} could be implemented as a compiler extension or with a custom preprocessor/code generator, which is the cleanest and most stable solution. 

Note that automatic vectorization is a minor issue for GPU code. CUDA follows the SIMT (Single-Instruction Multiple-Threads) model. Since SIMD parallelism is exposed to programmers as threads, programmers effectively vectorize their code manually. Similarly, programs written for the Intel SPMD Prgram Compiler (ispc)~\cite{6339601} code are implicitly vectorized.

\subsection{Preliminary Performance Evaluation}
\label{sec:benchmarks}
We evaluated \textsc{Ikra-Cpp} on a computer with an Intel Core i7-5960X CPU (4x 3.00 GHz), 32 GB RAM and an NVIDIA GeForce GTX 980 GPU, a 64-bit Ubuntu 16.04.1, gcc~5.4.0 and the NVIDIA CUDA Toolkit 9.0.176 in zero addressing mode.

We benchmarked an iterative application of \texttt{Body::move} (Listing~\ref{lbl:ikra_cpp_example_code}) in a parallel do-all operation. This benchmark is quite simple, but it clearly isolates the overheads of \textsc{Ikra-Cpp}, specifically address computation. Since the generated assembly code for \textsc{Ikra-Cpp} is nearly identical to a hand-written SOA layout, we expect minimal overheads. The number of iterations was chosen such that every program ran for at least 5 seconds. We calculated the average running time per iteration and report the minimum time out of 12 program runs.

\begin{figure}
\centering
\includegraphics[width=0.49\columnwidth]{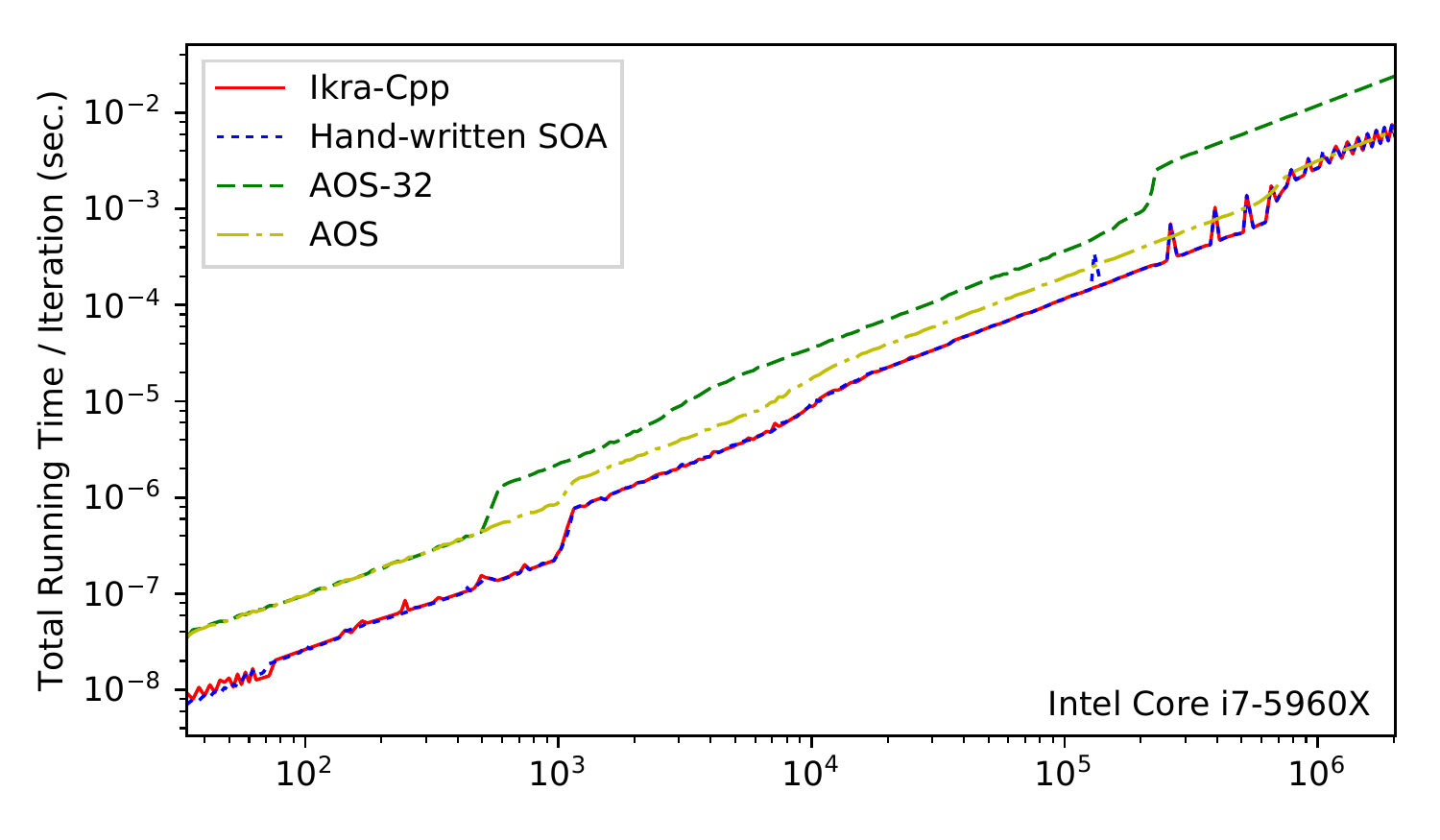}\hfill
\includegraphics[width=0.49\columnwidth]{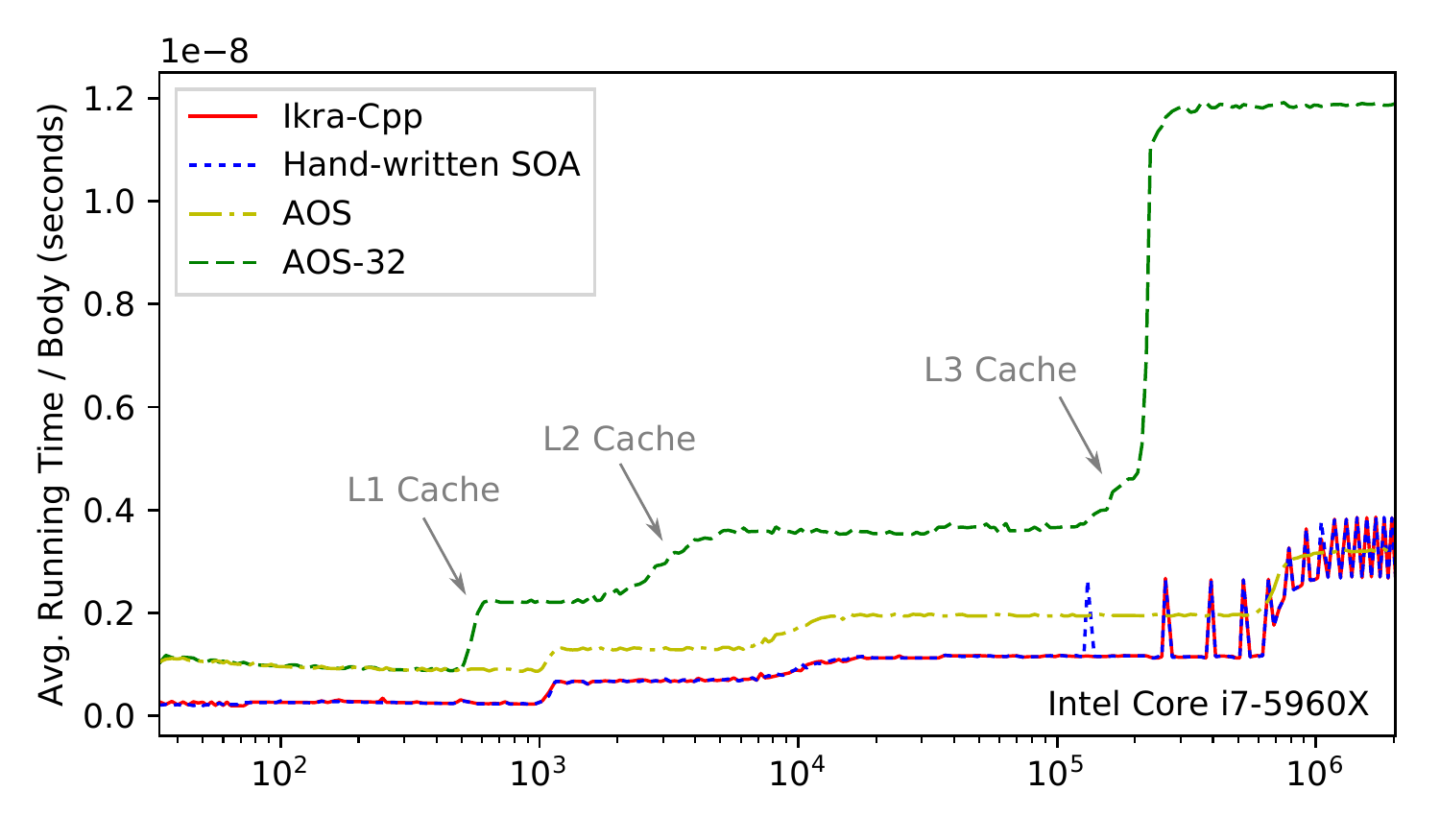}
\caption[N-body in \textsc{Ikra-Cpp}: Host mode running time]{Host mode running time}
\label{fig:workstation_bench_host}
\end{figure}

\begin{figure}
\centering
\includegraphics[width=0.49\columnwidth]{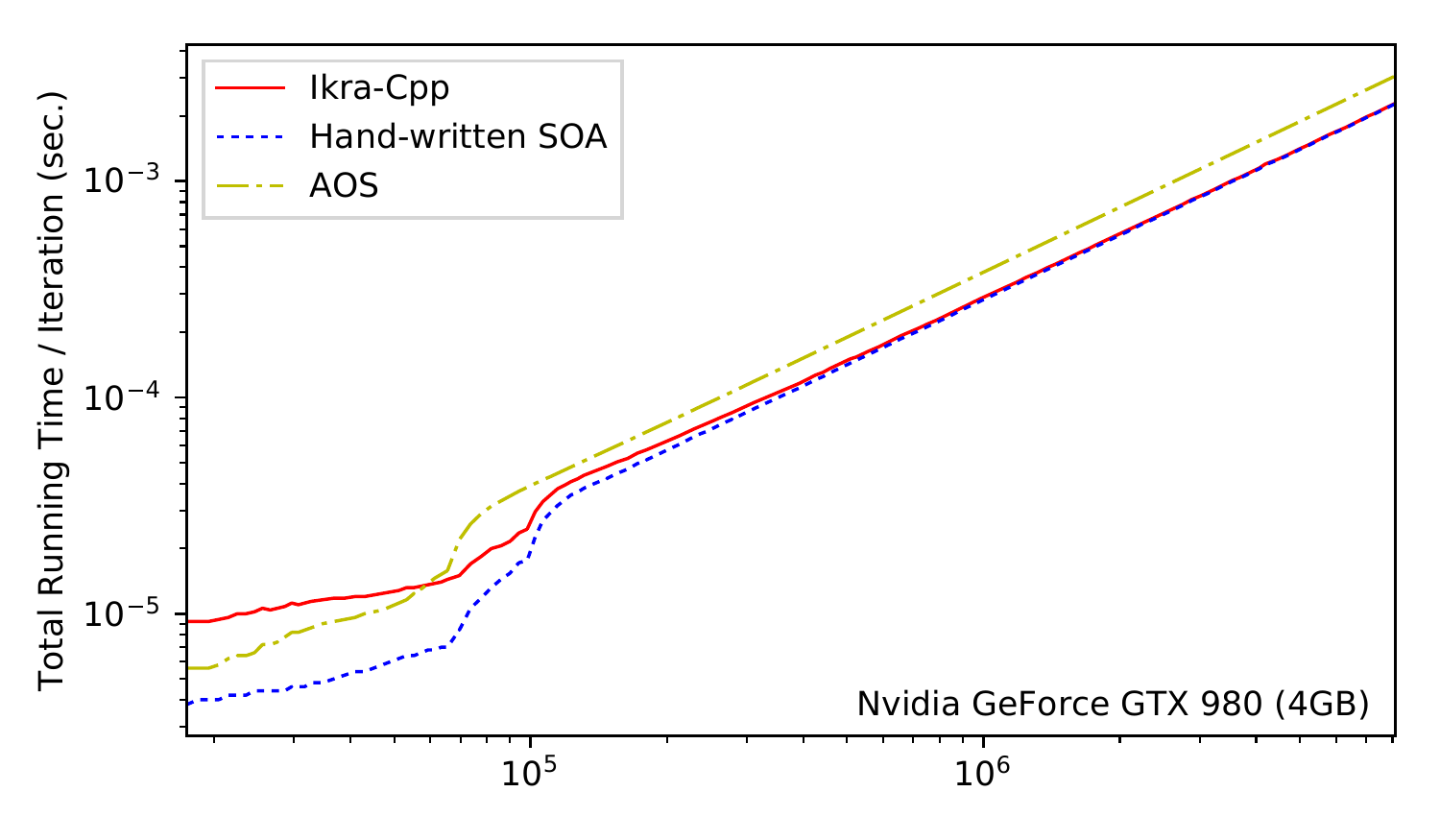}\hfill
\includegraphics[width=0.49\columnwidth]{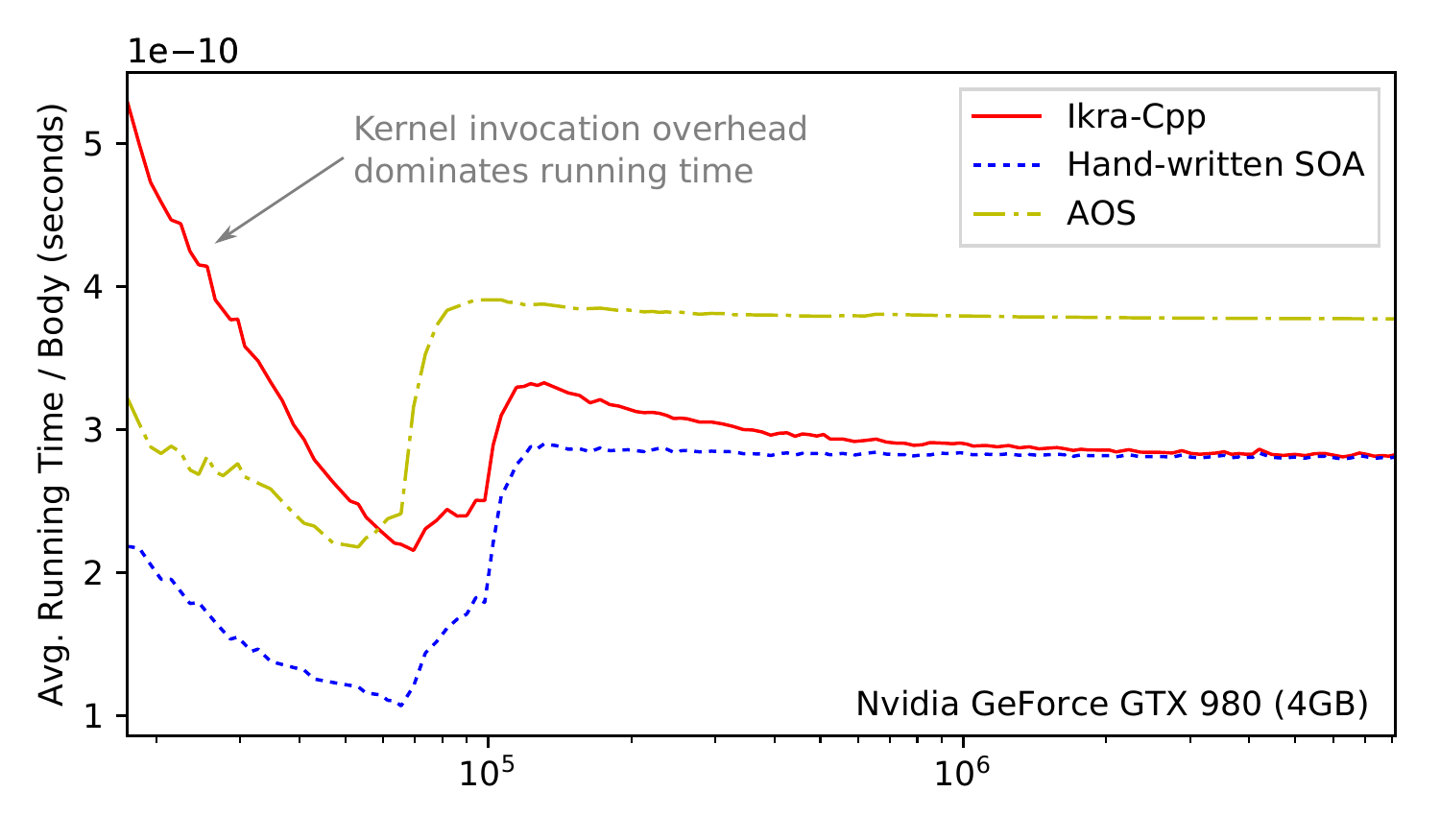}
\caption[N-body in \textsc{Ikra-Cpp}: Device mode running time]{Device mode running time}
\label{fig:workstation_bench_device}
\end{figure}

Figures~\ref{fig:workstation_bench_host} and \ref{fig:workstation_bench_device} show the running time on CPU and GPU. The x-axis denotes the number of objects and the y-axis denotes the running time in seconds. The left subfigure shows the average running time of one entire iteration and the right subfigure shows the average running time for a single \texttt{Body} object.

In host mode, \textsc{Ikra-Cpp}'s performance is almost identical to the hand-written SOA code. AOS-32 is a variant of AOS where 16 supplemental \texttt{double} fields were added to the \texttt{Body} class, similarly to the SoAx benchmark section~\cite{HOMANN2018325}. We can think of such fields as additional properties of a body (e.g., mass or radius) that are not utilized in this particular computation. The AOS-32 line of the right subfigure clearly shows the effect of the L1, L2, L3 caches (32~KB, 256~KB, 20~MB) in host mode.

The performance difference between \textsc{Ikra-Cpp} and the hand-written SOA layout in device mode is due to a higher kernel invocation overhead of \textsc{Ikra-Cpp}. More than 10,000 iterations (kernel invocations) are performed for small problem sizes. With a larger number of bodies, we get closer to hand-written SOA code, because each kernels performs more work.

\subsection{Related Work}
\label{sec:related_work_ikra_cpp_paper}
The AOS-SOA tradeoff is a well-known problem and has been studied in previous work in the context of C structs. To the best of our knowledge, there is no system that provides an AOS-like programming style for object-oriented programming with SOA performance characteristics.

\emph{SoAx}~\cite{HOMANN2018325} is a C++ library for AOS-style C/C++ programming with an implicit SOA layout. It is based on preprocessor macros and template metaprogramming. SoAx does not support OOP concepts like classes or methods. SOA struct types are defined with \texttt{std::tuple} instantiations and a helper macro that defines every SOA array separately. Object field values can only be accessed through a getter method of an SOA container object, which takes an object ID as argument, and not through SOA pointers. While such code is less expressive, it has two benefits: First, such code is easier to optimize for compilers than \textsc{Ikra-Cpp} code because it does not require decoding an object ID from a pointer. Second, it allows programmers to create multiple containers (structures of arrays), each of which has its own object ID range.

\emph{Array of Structures eXtended (ASX)}~\cite{STRZODKA2012429} is a library similar to SoAx. Objects in ASX can be allocated in ASX containers as well as on the stack (as single objects). ASX containers support both SOA and AOS data layouts, one of which must be chosen as a template parameter. ASX allows fields to be accessed with C++ member syntax (dot/arrow operators), but there are certain restrictions with respect to the size of field types.

\emph{Columnar Objects}~\cite{Mattis:2015:COI:2814228.2814230} is a Python extension (implemented in PyPy) that stores objects as SOA. Columnar Objects can deliver significant speedups of object-oriented analytical applications. In contrast to \textsc{Ikra-Cpp} and other libraries, subclassing is possible. However, newly introduced SOA arrays of subclasses store \emph{null} values for objects of superclasses, which can waste memory. Similar to \textsc{Ikra-Cpp}, objects are accessed through proxy objects which delegate field accesses to the actual, physical memory locations.

The \emph{Intel SPMD Program Compiler (ispc)}~\cite{6339601, Brodman:2014:WSS:2568058.2568065} is an experimental C compiler with language features for better SIMD support. Among other features, it can store an array of C structs in a hybrid SOA layout (also called \emph{Array of Structures of Arrays} (AoSoA)~\cite{Strzodka:2012:DLO:2337604.2337713, Weber:2014:ACA:2855568.2855580} or \emph{Tiled AOS}~\cite{10.1007/978-3-662-48096-0_21}). If a struct type is annotated with the \texttt{soa<$N$>} keyword and used to declare an array (where $N$ should be the SIMD width), then the array is stored as hybrid SOA with an SOA length of $N$. Array elements can be accessed with the usual C syntax. Furthermore, it is possible to take the address of an SOA object and fields can be accessed using an SOA object pointer. From that perspective, ispc's functionality is very similar to \textsc{Ikra-Cpp}. It would interesting to see how easily ispc can be extended to support OOP concepts like methods.

\emph{Shapes} is a high-level programming language that allows programmers to specify custom data layouts for better memory cache performance~\cite{Franco:2017:YAG:3133850.3133861}. Objects are stored in \emph{pools} and every pool can have a different object layout, e.g., AOS, SOA or a mixed layout. The developers of Shapes are recently working on compiling Shapes applications for SIMD architectures~\cite{Tasos:2018:ESS:3242947.3242951}.

\textsc{Ikra-Cpp} could be implemented as a compiler extension. To the best of our knowledge, no such extension exists for a widely used language. We believe that this is due to the high engineering effort of writing a new compiler or such an invasive compiler extension~\cite{Mernik:2005:DDL:1118890.1118892}.

\subsection{Summary}
\label{sec:conclusion}
We presented a first implementation of \textsc{Ikra-Cpp}'s data layout DSL for object-oriented programming with SOA performance characteristics. \textsc{Ikra-Cpp} allows programmers to write object-oriented code in AOS notation, while data is stored as SOA for better performance. SOA object members are always accessed through fake pointers or object references. How exactly an object ID is encoded in a pointer is determined by the addressing mode. Our main insights are that (a) object ID decoding and field address computations can be done efficiently after constant folding and that (b) an AOS-style notation can be achieved transparently in C++ with operator overloading, template metaprogramming, and preprocessor macros. Preliminary benchmarks show that simple examples written with \textsc{Ikra-Cpp} and compiled with gcc are on par with hand-written SOA code.

\textsc{Ikra-Cpp} is the basis of \textsc{DynaSOAr} (Chapter~\ref{sec:chapter_dynasoar}). The data layout DSL of \textsc{DynaSOAr} is heavily based on \textsc{Ikra-Cpp}'s DSL.

\section{Inner Arrays in a Structure of Arrays}
\label{sec:inner_arrays_in_soa_papaer}
SOA works well with simple data structures, but cannot be easily applied to structs that contain inner arrays (i.e., fields of array type), potentially of non-constant size. Such structures appear frequently in graph-based applications and in object-oriented designs with associations of high multiplicity.

Such arrays are typically allocated separately on the heap, requiring an additional pointer indirection. In this section, we analyze different data layout techniques for inner arrays. We extended \textsc{Ikra-Cpp} with additional proxy field types that implement those layouts.

Applications that iterate over arrays one-by-one or in a fashion that is \emph{uniform} among all objects are particularly interesting, because their memory access can be optimized. While a standard SOA layout does not affect the layout of inner arrays, a different, more SIMD-friendly layout can group elements by array index and increase memory coalescing on GPUs for such applications.

\paragraph{Examples}
To experiment with various inner array layouts, we implemented two important real-world SMMO applications in \textsc{Ikra-Cpp}: Breadth-first search (BFS) and an agent-based, object-oriented traffic flow simulation. In graphs, vertices often have a varying number of neighbors and adjacency lists (arrays) are the preferred representation for BFS on GPUs~\cite{10.1007/978-3-540-77220-0_21}. The traffic flow simulation exhibits graph-based features for representing street networks and utilizes array-based data structures within the simulation logic.


\begin{figure}
\begin{lstlisting}[language=c++, morekeywords={__host__, __device__, int_}, caption={Data structure of frontier-based BFS in \textsc{Ikra-Cpp}},label={lst:bfs_in_ikra_cpp_impl}]
class Vertex : public IkraBase<Vertex, 100> {  // or: IkraSoaBase<Vertex, 100>
 public: <@\emph{IKRA\_INITIALIZE\_CLASS}@>
  int_ distance = std::numeric_limits<int>::max();
  <@\fbox{???}@> /* (some array type) */ neighbors;

  // Constructor can run on host and device.
  __device__ ___host__ Vertex(int num_neighbors) : neighbors(num_neighbors) {}

  __device__ int num_neighbors() { return neighbors.size(); }

  __device__ void visit(int frontier);  // Implementation later...
};

<@\emph{IKRA\_DEVICE\_STORAGE}@>(Vertex)
\end{lstlisting}
\vspace{20pt}

\begin{minipage}{0.48\textwidth}
\includegraphics[width=\columnwidth]{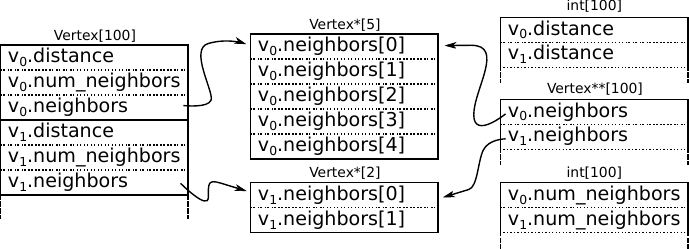}
\caption[Inner array layout: No inlining]{No inlining: AOS (left side) and SOA (right side)} 
\label{fig:ex_strat_no}
\end{minipage}\hfill \begin{minipage}{0.48\textwidth}
\includegraphics[width=\columnwidth]{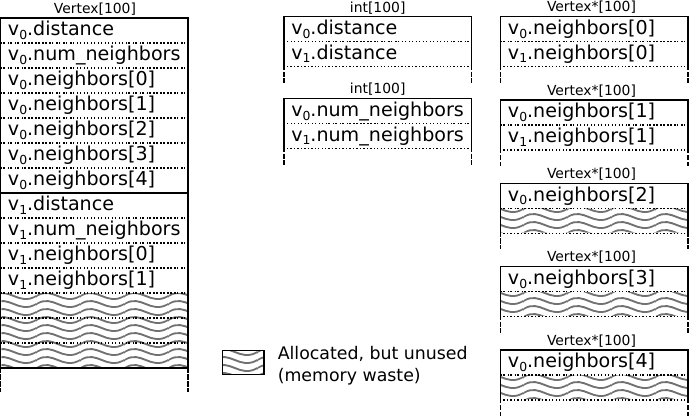}
\caption[Inner array layout: Full inlining]{Full inlining: AOS (left side) and SOA (middle, right side)} 
\label{fig:ex_strat_full}
\end{minipage}
\vspace{20pt}

\begin{minipage}{0.48\textwidth}
\includegraphics[width=\columnwidth]{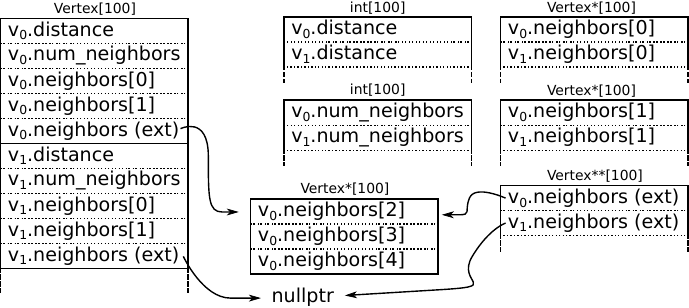}
\caption[Inner array layout: Partial inlining]{Partial inlining: AOS (left side) and SOA (middle, right side)} 
\label{fig:ex_strat_part}
\end{minipage} \hfill \begin{minipage}{0.48\textwidth}
\includegraphics[width=\columnwidth]{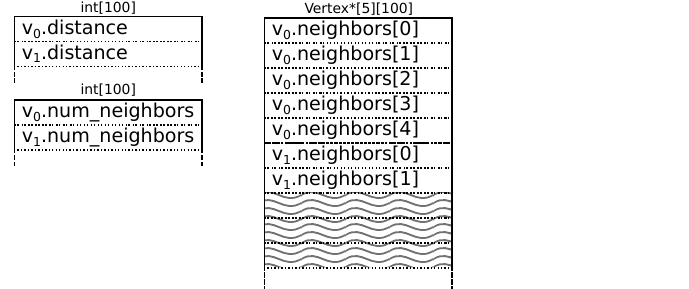}
\caption[Inner array layout: Array as object]{SOA, array as object} 
\label{fig:ex_strat_obj}
\end{minipage}
\end{figure}

\subsection{Data Layout Strategies for Inner Arrays}
\label{sec:data_layout_str}


This section gives an overview of seven inner array layout strategies. We focus on C++-style arrays with fixed size after allocation. Figures~\ref{fig:ex_strat_no}--\ref{fig:ex_strat_obj} illustrate these strategies visually, using the \texttt{Vertex} class for a breath-first algorithm as an example (Listing~\ref{lst:bfs_in_ikra_cpp_impl}). This class has two fields: A distance field storing the computed distance of the vertex from the source vertex and an array of \texttt{Vertex} pointers (adjacency list). Note that the adjacency lists of different vertices may have different lengths. Therefore, we also have to store the size of each list/array (\texttt{num\_neighbors}).

In the following paragraphs, we describe multiple layout strategies of the adjacency list array.  We implemented these strategies in \textsc{Ikra-Cpp} with additional proxy field types. 

We consider three categories of layout strategies: \emph{AOS}, \emph{SOA} and \emph{SOA with Array as Object} (i.e., SOA without handling inner arrays specially). In every category, inner arrays can either be \emph{fully inlined}, \emph{partially inlined} or not inlined at all (\emph{without inlining}). Whether objects of a class are stored in AOS or SOA depends on the chosen superclass (\texttt{IkraBase} or \texttt{IkraSoaBase}).

\paragraph{AOS without Inlining (Figure~\ref{fig:ex_strat_no} (left), Listing~\ref{lst:ikra_cpp_aos_wo_inlining})}
This is the default data layout that many programmers choose intuitively. Objects are stored as AOS (\texttt{IkraBase}), i.e., all field values of an object are stored together. Inner arrays are heap-allocated. This layout is useful if inner arrays have different sizes, because then no memory is wasted by array inlining. In standard C++, inner arrays can be allocated manually (using \texttt{malloc}/\texttt{new}) or with a helper class like \texttt{std::vector<T>}.

\textsc{Ikra-Cpp} provides a proxy field type \texttt{inlined\_array\_(T, 0)} which internally stores a memory pointer to a heap-allocated array and the size of the array. Fields of this type must be initialized in the constructor with the desired array size (Listing~\ref{lst:bfs_in_ikra_cpp_impl}, Line~7). \textsc{Ikra-Cpp} will then allocate the array with the system/CUDA-wide dynamic memory allocator (\texttt{malloc}). Note that the CUDA allocator is slow, so this implementation is not suitable if the constructor is part of performance-critical code.

\begin{figure}
\makeatletter\def\@captype{lstlisting}\makeatother
\caption[\textsc{Ikra-Cpp}: Notation: AOS without inlining]{Notation: AOS without inlining}
\label{lst:ikra_cpp_aos_wo_inlining}
\vspace{-0.25cm}
\begin{minipage}{.32\textwidth}
\begin{lstlisting}[language=c++, numbers=none, morekeywords={inlined_array_, field_}]
// Preferred notation:
inlined_array_(Vertex*, 0)
    neighbors;
\end{lstlisting}
\end{minipage} \hfill
\begin{minipage}{.32\textwidth}
\begin{lstlisting}[language=c++, numbers=none, morekeywords={int_, field_}]
// Alternative 1:
int_ num_neighbors;
field_(Vertex**) neighbors;
\end{lstlisting}
\end{minipage}\hfill
\begin{minipage}{.28\textwidth}
\begin{lstlisting}[language=c++, numbers=none, morekeywords={int_, field_}]
// Alternative 2:
field_(std::vector<
    Vertex*>) neighbors;
\end{lstlisting}
\end{minipage}
\end{figure}

Depending on the class structure, compilers may have to add padding to ensure that all fields are properly aligned. E.g., the location that stores the 64-bit memory pointer to the inner array must be aligned to a multiple of 64-bit. This is a general disadvantage of AOS (Section~\ref{sec:obj_vs_soa_array_alignment}). Note that this layout is identical to ``AOS with Partial Inlining'' with an inlining size of zero, but the conditional branch of that layout strategy is optimized away by the compiler.

\paragraph{AOS with Full Inlining (Figure~\ref{fig:ex_strat_full} (left), Listing~\ref{lst:ikra_cpp_aos_full_inl})}
This layout strategy still stores objects as AOS (\texttt{IkraBase}), but inlines inner arrays fully into objects, such that they can be accessed more efficiently without a pointer indirection. Very small inner arrays may be able to share cache lines with other fields, and thus benefit cache utilization.

The downside of this layout is that it potentially wastes memory; all inner arrays must have the same size, i.e., the largest size among all inner arrays, to be able to hold all elements. The amount of wasted memory depends on the variance among inner array sizes.

In standard C++, fully-inlined inner arrays of size $N$ can be declared with \texttt{std::array<T, N>}. \textsc{Ikra-Cpp} provides a proxy field type \texttt{fully\_inlined\_array\_}.

\begin{figure}
\makeatletter\def\@captype{lstlisting}\makeatother
\caption[\textsc{Ikra-Cpp}: Notation: AOS with full inlining]{Notation: AOS with full inlining}
\label{lst:ikra_cpp_aos_full_inl}
\vspace{-0.25cm}
\begin{minipage}{0.45\textwidth}
\begin{lstlisting}[language=c++, numbers=none, morekeywords={fully_inlined_array_}]
// Preferred notation:
fully_inlined_array_(Vertex*, N)
    neighbors;
\end{lstlisting}
\end{minipage} \hfill
\begin{minipage}{0.5\textwidth}
\begin{lstlisting}[language=c++, numbers=none, morekeywords={field_, int_}]
// Alternative:
int_ num_neighbors;
field_(std::array<Vertex*, N>) neighbors;
\end{lstlisting}
\end{minipage}
\end{figure}

\paragraph{AOS with Partial Inlining (Figure~\ref{fig:ex_strat_part} (left), Listing~\ref{lst:ikra_cpp_aos_part_inl})}
This layout is a mixture of the previous two strategies. Up to $N$ inner array elements are inlined into objects, where $N$ is a compile-time constant. Elements with an index $\geq N$ are stored externally on the heap. The benefit of this approach is efficient access to the first $N$ elements. Access to elements on the external storage is as expensive as with ``Without Inlining'', i.e., it requires a pointer indirection. This strategy requires an additional conditional branch to determine whether an element is stored in the inline storage or on the external storage. This imposes little overhead on GPUs, which do generally not execute instructions speculatively.

In ordinary C++, inner arrays can be partially inlined with a helper class such as \texttt{absl::InlinedVector<T, N>}\footnote{\texttt{InlinedVector} is part of the Abseil library. See \url{https://github.com/abseil/abseil-cpp}}. \textsc{Ikra-Cpp} supports this layout with the previously mentioned proxy field type \texttt{inlined\_array\_(T, N)}. The second argument $N$ is the number of inlined array slots. The total array size must be specified during field initialization.

\begin{figure}
\makeatletter\def\@captype{lstlisting}\makeatother
\caption[\textsc{Ikra-Cpp}: Notation: AOS with partial inlining]{Notation: AOS with partial inlining}
\label{lst:ikra_cpp_aos_part_inl}
\vspace{-0.25cm}
\begin{minipage}{1.0\textwidth}
\begin{lstlisting}[language=c++, numbers=none, morekeywords={int_, field_, inlined_array_}]
inlined_array_(Vertex*, N) neighbors;  // Preferred notation
field_(absl::InlinedVector<Vertex*, N>) neighbors;  // Alternative

// Data layout is equivalent to, however, much harder to use:
int_ num_neighbors;
field_(std::array<Vertex*, N>) neighbors;
field_(Vertex**) neighbors_other;
\end{lstlisting}
\end{minipage}
\end{figure}

\paragraph{SOA without Inlining (Figure~\ref{fig:ex_strat_no} (right), Listing~\ref{lst:ikra_cpp_soa_wo_inlining})}
This layout is identical to ``AOS without Inlining'', but stores objects as SOA (\texttt{IkraSoaBase}). It has the usual benefits of SOA-style allocation: First, if not all fields are used all the time, it can improve cache utilization because those fields will not occupy cache lines. Second, it allows for efficient loads/stores from/into vector registers if objects with consecutive IDs are simultaneously accessed (memory coalescing). Third, less memory is wasted for object padding compared to AOS, because only the SOA arrays themselves must be aligned to certain byte sizes, but not each object.

With respect to inner arrays, the same advantages and disadvantages as in the first strategy apply. \textsc{Ikra-Cpp} provides a proxy field type \texttt{inlined\_array\_(T, 0)} which implements this layout.

\begin{figure}
\makeatletter\def\@captype{lstlisting}\makeatother
\caption[\textsc{Ikra-Cpp} Notation: SOA without inlining]{Notation: SOA without inlining}
\label{lst:ikra_cpp_soa_wo_inlining}
\vspace{-0.25cm}
\begin{minipage}{.32\textwidth}
\begin{lstlisting}[language=c++, numbers=none, morekeywords={inlined_array_, field_}]
// Preferred notation:
inlined_array_(Vertex*, 0)
    neighbors;
\end{lstlisting}
\end{minipage} \hfill
\begin{minipage}{.32\textwidth}
\begin{lstlisting}[language=c++, numbers=none, morekeywords={int_, field_}]
// Alternative 1:
int_ num_neighbors;
field_(Vertex**) neighbors;
\end{lstlisting}
\end{minipage}\hfill
\begin{minipage}{.28\textwidth}
\begin{lstlisting}[language=c++, numbers=none, morekeywords={int_, field_}]
// Alternative 2:
field_(std::vector<
    Vertex*>) neighbors;
\end{lstlisting}
\end{minipage}
\end{figure}

\paragraph{SOA with Full Inlining (Figure~\ref{fig:ex_strat_full} (right), Listing~\ref{lst:ikra_cpp_soa_full_inl})}
This layout is identical to ``AOS with Full Inlining'', but stores objects as SOA (\texttt{IkraSoaBase}). In particular, inner arrays are also stored in SOA, as if every array slot were a separate field of the class. There is a separate SOA array for each inner array index.

This layout provides opportunities for vectorized operations and memory coalescing not only for primitive fields but also when accessing inner array elements. This is possible if objects with consecutive IDs are processed in the same warp and inner array elements with the same indices are accessed.

Unfortunately, many parallel graph algorithms~\cite{Malewicz:2010:PSL:1807167.1807184} on GPUs do not benefit much from additional memory coalescing in this layout. Even though reading the pointers of an adjacency list can be coalesced, data reads/writes on neighboring vertices are still uncoalesced, because vertex IDs of neighbors are usually \emph{random} and not consecutive (Section~\ref{sec:inner_arrays_perf_eval}).

Similar to ``AOS with Full Inlining'', this layout provides more efficient array access without a pointer redirection but can waste memory. It is supported by \textsc{Ikra-Cpp} as \texttt{fully\_inlined\_array\_(T, N)}.

\begin{figure}
\makeatletter\def\@captype{lstlisting}\makeatother
\caption[\textsc{Ikra-Cpp}: Notation: SOA with full inlining]{Notation: SOA with full inlining}
\label{lst:ikra_cpp_soa_full_inl}
\vspace{-0.25cm}
\begin{minipage}{1.0\textwidth}
\begin{lstlisting}[language=c++, numbers=none, morekeywords={int_, field_, fully_inlined_array_}]
// Preferred notation:
fully_inlined_array_(Vertex*, N) neighbors;

// Data layout is equivalent to, however, much harder to use:
int_ num_neighbors;
field_(Vertex*) neighbors_1;
field_(Vertex*) neighbors_2;
/* ... */
field_(Vertex*) neighbors_N;
\end{lstlisting}
\end{minipage}
\end{figure}

\texttt{fully\_ininlined\_array\_} is used in both this layout and ``AOS with Full Inlining''. Whether this proxy type stores the inner array as AOS or as SOA depends on the layout of the class (\texttt{IkraBase} or \texttt{IkraSoaBase}).

\paragraph{SOA with Partial Inlining (Figure~\ref{fig:ex_strat_part} (right), Figure~\ref{lst:ikra_cpp_soa_part_inl})}
This strategy is identical to ``AOS with Partial Inlining'', but stores objects as SOA (\texttt{IkraSoaBase}). The first $N$ inner array elements are inlined and stored as SOA, as in the previous strategy. This layout is supported by \textsc{Ikra-Cpp} as \texttt{inlined\_array\_(T, N)}.

\begin{figure}
\makeatletter\def\@captype{lstlisting}\makeatother
\caption[\textsc{Ikra-Cpp}: Notation: SOA with partial inlining]{Notation: SOA with partial inlining}
\label{lst:ikra_cpp_soa_part_inl}
\vspace{-0.25cm}
\begin{minipage}{1.0\textwidth}
\begin{lstlisting}[language=c++, numbers=none, morekeywords={int_, field_, inlined_array_}]
// Preferred notation:
inlined_array_(Vertex*, N) neighbors;

// Data layout is equivalent to, however, much harder to use:
int_ num_neighbors;
field_(Vertex*) neighbors_1;
field_(Vertex*) neighbors_2;
/* ... */
field_(Vertex*) neighbors_N;
field_(Vertex**) neighbors_other;
\end{lstlisting}
\end{minipage}
\end{figure}

\begin{figure}
\makeatletter\def\@captype{lstlisting}\makeatother
\caption[\textsc{Ikra-Cpp}: Notation: SOA with array as object]{Notation: SOA with array as object}
\label{lst:ikra_cpp_soa_array_as_obj}
\vspace{-0.25cm}
\begin{minipage}{1.0\textwidth}
\begin{lstlisting}[language=c++, numbers=none, morekeywords={field_}]
field_(std::array<Vertex*, N>) neighbors;
\end{lstlisting}
\end{minipage}
\end{figure}

\texttt{ininlined\_array\_} is used in both this layout and ``AOS with Partial Inlining''. Whether this proxy type stores the inlined inner array slots as AOS or as SOA depends on the layout of the class (\texttt{IkraBase} or \texttt{IkraSoaBase}).

\paragraph{SOA with Array as Object (Figure~\ref{fig:ex_strat_obj}, Listing~\ref{lst:ikra_cpp_soa_array_as_obj})}
This layout treats inner arrays as normal C++ objects and does not perform any data layout transformations on them. There is \emph{one} SOA array for every field, including inner arrays. This layout is useful for GPU applications with nested parallelism. If threads with consecutive IDs simultaneously access consecutive inner array slots, then those accesses can be coalesced. E.g., this is the case when optimizing BFS with virtual warp-centric programming~\cite{Hong:2011:ACG:1941553.1941590}. This layout is supported by \textsc{Ikra-Cpp} as \texttt{field\_(std::array<T, N>)}.

From an inlining perspective, this layout inlines the inner array fully. It could easily be adapted to \emph{partial inlining} (\texttt{field\_(absl::InlinedVector<T, N>)}) or \emph{without inlining} (\texttt{field\_(std::vector<T>)}), but we do not analyze such layouts any further in this work.

\paragraph{Choosing a Layout Strategy}
Programmers have a variety of layout strategies to choose from. Which strategy is best depends on the hardware architecture, the data access patterns of the application and the characteristics of the dataset. Even for experienced programmers this process can to some degree be a trial and error.

With \textsc{Ikra-Cpp}, programmers still have to take all these factors into account, but switching between strategies is now much easier. As a rule of thumb, we suggest to start experimenting with a partial inlining size that ensures that 80\% of all inner array elements are inlined.

\subsection{Performance Evaluation}
\label{sec:inner_arrays_perf_eval}
We evaluated the previously described data layouts with three benchmarks: A synthetic benchmark, a frontier-based BFS implementation and a traffic flow simulation.

We ran all experiments on a machine with an Intel i7-5960X CPU (8x 3.00~GHz), 32~GB main memory, an NVIDIA GeForce GTX 980 GPU (4~GB memory), Ubuntu~16.04 and the \texttt{nvcc} compiler from the NVIDIA CUDA Toolkit version 9.1. This work focuses on GPU execution, but similar performance effects can be observed on CPUs with a compiler with good auto-vectorization. 

\paragraph{Synthetic Benchmark}
To isolate the performance effects of array inlining, we created a synthetic benchmark (Listing~\ref{lst:source_synth_inner_array_bench}) with a dummy class containing an \texttt{int} data field and an \texttt{int} array. The array has between 32 and 64 elements (chosen randomly). The benchmark adds the data field value to all array elements in a loop, i.e., all inner array elements are read and written. 

\begin{figure}
\centering
\includegraphics[width=0.65\columnwidth]{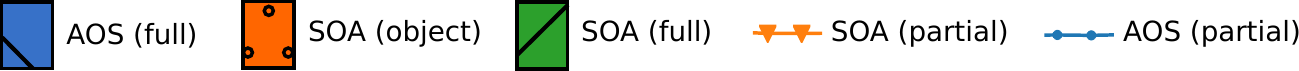}
\begin{minipage}{0.48\textwidth}
\centering
\includegraphics[width=0.7\columnwidth]{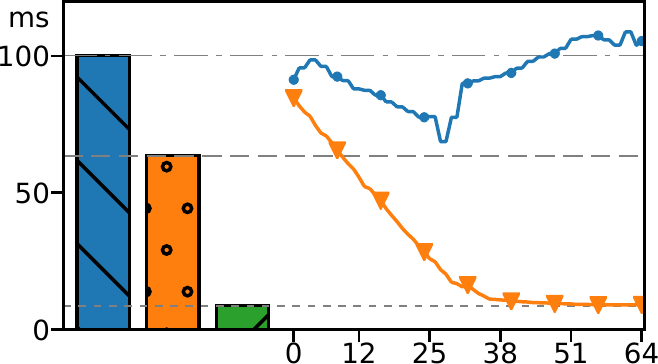}
\caption[Inner array inlining: Synthetic benchmark]{Synthetic benchmark (ms)} 
\label{fig:synthetic_bench_array_onl}
\end{minipage}\hfill \begin{minipage}{0.48\textwidth}
\centering
\includegraphics[width=0.7\columnwidth]{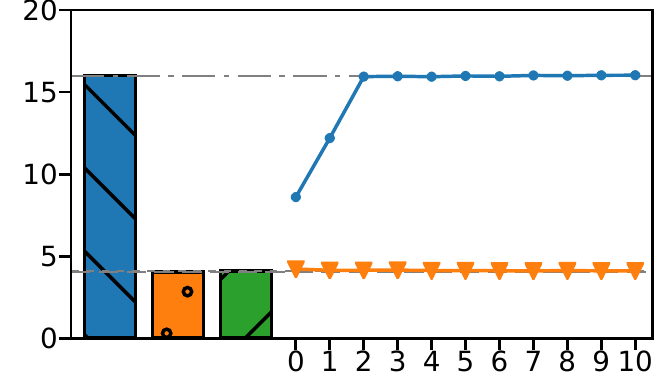}
\caption[Inner array inlining: Frontier-based BFS benchmark]{Frontier-based BFS (sec.)} 
\label{fig:bfs_bench_array_inl}
\end{minipage}
\vspace{0.5cm}

\begin{minipage}{1.0\textwidth}
\begin{lstlisting}[language=c++, numbers=none, morekeywords={__device__, int_}, caption={[Inner Array Inlining: Synthetic Benchmark]Source code of synthetic benchmark}, label={lst:source_synth_inner_array_bench}]
__device__ int rand_int(int min, int max) { /* return random int */ }

class DummyClass : public <@\fbox{???}@><DummyClass, 262144> {
 public: <@\emph{IKRA\_INITIALIZE\_CLASS}@>
  int_ data;
  <@\fbox{???}@>(int, <@\fbox{?}@>) arr;

  __device__ DummyClass() : data(rand_int(0, 100)), arr(rand_int(32, 65)) {}

  __device__ void benchmark() {
    // Can the accesses of arr[i] be coalesced?
    for (int i = 0; i < arr.size(); ++i) { arr[i] += data; }
  }
}; <@\emph{IKRA\_DEVICE\_STORAGE}@>(DummyClass);

void init_benchmark() { parallel_new<DummyClass>(262144); }
void run_benchmark() { parallel_do<DummyClass, &DummyClass::benchmark>(); }
\end{lstlisting}

\begin{lstlisting}[label={lst:vertex_class}, language=c++, numbers=left, caption={[Frontier-based BFS]Source code of frontier-based BFS}, morekeywords={__device__}]
__device__ still_running = false;

__device__ void Vertex::visit(int frontier) {
  if (distance == frontier) {
    for (int i = 0; i < num_neighbors(); ++i) {
      if (neighbors[i]->distance > frontier + 1) {
        still_running = true;
        neighbors[i]->distance = frontier + 1;
      }  // else: Vertex was already visited.
    }
  }
}

void run_bfs(Vertex* start_vertex) {
  // Ikra-Cpp objects can be accessed from host and device.
  start_vertex->distance = 0;

  // Read/write still_running with cudaMemcpyTo/FromSymbol.
  for (int iteration = 0; still_running; iteration++, still_running = false) {
    // Notation: Parameter types (int) must be specified before function pointer.
    parallel_do<Vertex, int, &Vertex::visit>(iteration);
  }
}
\end{lstlisting}
\end{minipage}
\end{figure}

Figure~\ref{fig:synthetic_bench_array_onl} shows the running time for 262,144 dummy objects. The ``Array as Object'' SOA version is more than 30\% faster than the AOS version. The fully inlined SOA version is an order of magnitude faster than the AOS version.

The ``SOA (partial)'' line shows the running time with various partial inner array inlining sizes $N$ (x-axis) in SOA layout. For $N=0$ (no inlining), the performance is worse than ``Array as Object''. In both cases, the entire inner array is stored in one block of memory. However, in the former case, the array is located on the heap and can only be accessed with a pointer indirection, causing a slowdown. Note that, while ``Array as Object'' is faster, it wastes a considerable amount of memory.

The performance of partial inlining (SOA) increases with partial inlining sizes, because inner array accesses can be coalesced. Partial array inlining sizes larger than 32 do not improve the SOA-mode performance anymore, because, due to warp divergence caused by differing inner array sizes, those additional inner array accesses are unlikely to be coalesced.

The performance of partial inlining (AOS) decreases with partial inlining sizes after 32 because more and more cache entries are blocked by non-existing array slots.

\paragraph{Breadth-first Search}

BFS is an important and fundamental algorithm in graph processing. A variety of implementation strategies have been proposed for GPUs, some based on advanced techniques such as hierarchical queues~\cite{Luo:2010:EGI:1837274.1837289} or virtual warp-centric programming~\cite{Hong:2011:ACG:1941553.1941590}. The frontier-based BFS algorithm~\cite{Merrill:2015:HSG:2737841.2717511} is among the simplest ones and provides a reasonable speedup compared to CPU execution.

Frontier-based BFS (Listing~\ref{lst:vertex_class}) computes the distance of every vertex from a designated start vertex. At first, the start vertex has distance zero and all other vertices have distance infinity. The algorithm now proceeds iteratively. In iteration $i$, all vertices with distance $i$ (i.e., the \emph{frontier}) are processed in parallel: For every vertex in the frontier, all of its neighbors are updated with distance $i + 1$, unless they already have a smaller/equal distance value. The algorithm terminates if no updates are performed anymore. BFS is an interesting example for \textsc{Ikra-Cpp} because the adjacency lists are arrays of different sizes.

Figure~\ref{fig:bfs_bench_array_inl} shows the running time of the frontier-based BFS algorithm with different layout strategies on the Pennsylvania road network (1,088,092 vertices, 3,083,796 directed edges, avg. degree 2.83)~\cite{DBLP:journals/corr/abs-0810-1355}. The graph clearly shows the benefit of SOA over AOS. The performance of AOS degrades with a growing inlining size, because more cache entries for non-existing \texttt{neighbors} array slots are wasted. BFS does not benefit from any additional memory coalescing when inlining inner arrays in SOA mode (compare ``SOA (object)'' and ``SOA (full)''), because the neighbors accessed in the inner \emph{for} loop (Listing~\ref{lst:vertex_class}, Line~5) have \emph{random} (as opposed to consecutive) IDs. While array reads \texttt{neighbors[i]} can be coalesced, accesses to fields of \texttt{neighbors[i]} cannot be coalesced. Different optimization techniques can be applied to speed up such algorithms, most notably virtual warp-centric programming~\cite{Hong:2011:ACG:1941553.1941590}, which is a form of nested parallelism. That optimization would see a speedup from an \emph{Array as Object} layout.

\paragraph{Traffic Flow Simulation}
In Section~\ref{sec:smmo_traf_flow_sia}, we are developing a complex traffic flow simulation that implements the Nagel-Schreckenberg model~\cite{nagel_schr}. This implementation has six array fields in five different classes. These arrays are accessed in most parallel do-all operations of \textsf{traffic} and we can improve the overall runtime performance by optimizing their access\footnote{Refer to Section~\ref{sec:smmo_traf_flow_sia} for application implementation details.}. We consider two inner arrays in this paragraph.

\begin{itemize}
  \item \texttt{Cell::outgoing}: Read sequentially in a parallel do-all operation of a method of class \texttt{Car}.
  \item \texttt{Car::path}: Cleared at the beginning of a Nagel-Schreckenberg iteration, then filled sequentially and read sequentially in parallel do-all operations of methods of class \texttt{Car}.
\end{itemize}



\begin{figure}

    \subfloat[\texttt{Cell::incoming}]{\includegraphics[width=0.31\columnwidth]{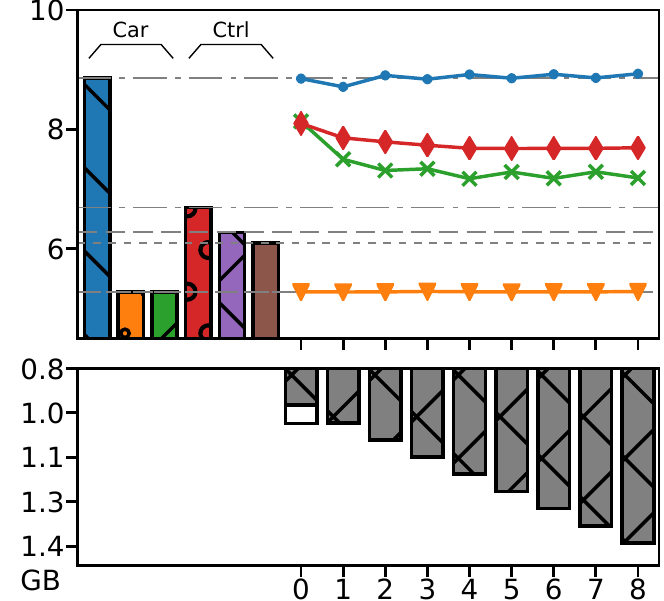}}\hfill
    \subfloat[\texttt{Car::path}]{\includegraphics[width=0.31\columnwidth]{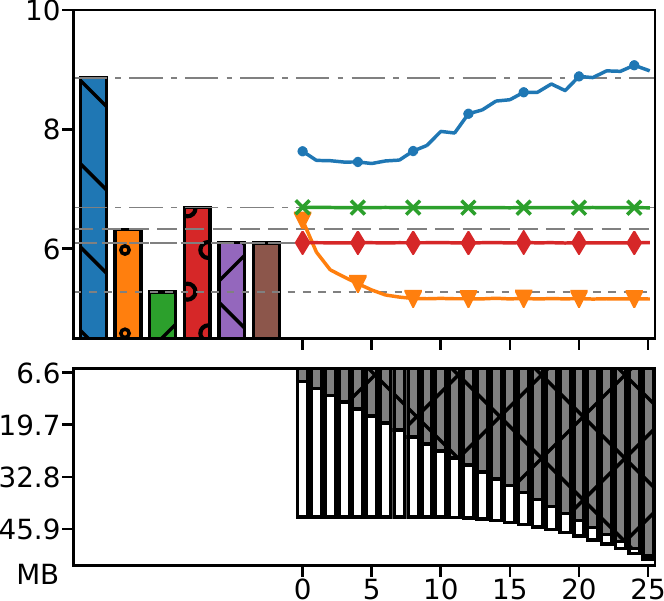}}\hfill
    \begin{minipage}{0.31\textwidth}
      \vspace{-2.75cm}
      \includegraphics[width=0.95\columnwidth]{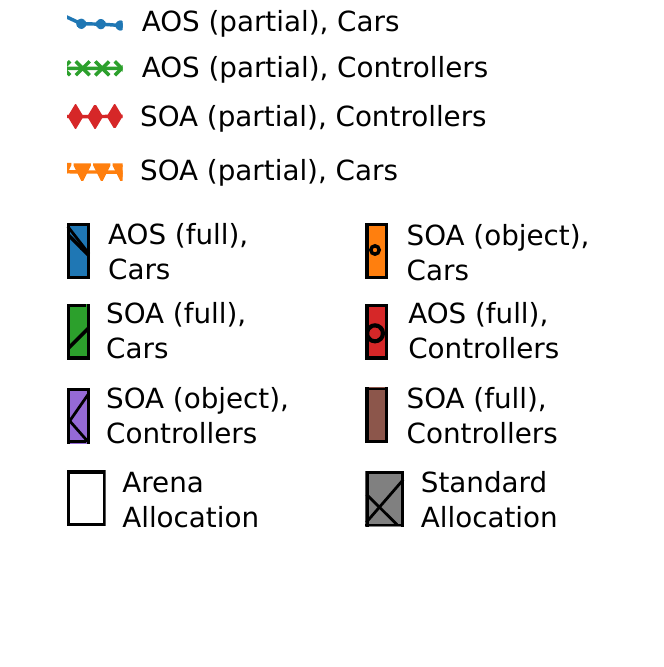}
    \end{minipage}
    \caption[Inner array inlining: \textsf{traffic} benchmark]{Running time and memory requirement of traffic simulation}
    \label{fig:bench_traffic}
\end{figure}

The benefit of an SOA layout over an AOS layout is clearly visible in this benchmark (Figure~\ref{fig:bench_traffic}). Regardless of which inner array inlining strategy is chosen, the running time spent on processing cars is always significantly lower in SOA (e.g., first bar vs. second bar). We can observe a similar behavior for traffic controllers.

The benefit of array inlining can be observed best in Subfigure~(\textsc{b}) when comparing the running time of a fully inlined \texttt{Car::path} (third bar) with the one that is stored ``as object'' (second bar). The fully inlined version is 15\% faster. This is because our implementation of the Nagel-Schreckenberg algorithm iterates over the path array in every thread. Because all cars are processed in order (i.e., thread $i$ processes car with ID $i$), memory accesses are coalesced. Furthermore, a slightly better speedup can be achieved with partial array inlining (orange line), starting from an inlining size of 6. Accesses to \texttt{Car::path} at indices higher than 6 are unlikely to be coalesced because very few threads actually access these array slots. We believe that the additional speedup is due to better cache utilization and prefetching: With partial inlining, a cache line can contain multiple elements of the same inner array from different objects.


The lower part of every subgraph shows the memory usage of inner arrays. The gray part represents inlined allocation (regular fields and inlined slots of inner arrays). The white part represents external storage (heap allocation). Too much inlining wastes memory because not all inner arrays are equally large. Recall that the number of inlined slots of an inner array is the same for every object. On the contrary, inner arrays of different objects may have different sizes on the external storage.  




\subsection{Conclusion and Related Work}
\label{sec:soa_layout_conclusion}
In this section, we presented an overview of various data layout strategies for inner arrays in a Structure of Arrays layout. Depending on the data access pattern, such arrays can be split and regrouped by array index to take advantage of memory coalescing when accessing inner array elements. Since writing and maintaining such low-level code is tedious, we extended \textsc{Ikra-Cpp} with additional proxy types that store inner arrays in the discussed layouts. 

Object inlining has been proposed for Java-like object-oriented languages for better cache performance and reducing overheads due to allocation and pointer indirections~\cite{Dolby:2000:AOI:349299.349344}. Later work applied the idea of data inlining to arrays, which can simplify address arithmetics of array accesses and eliminate load instructions in the assembly code~\cite{Wimmer:2008:AAI:1356058.1356061}. \textsc{Ikra-Cpp} applies the same idea to arrays in an SOA layout and we are seeing similar speedups. However, we inline arrays under simplified assumptions: Arrays are of fixed size and inlining is controlled manually by the programmer. Future work could attempt to automate this process.

\section{Summary}
\label{sec:summary_chap4}
We discussed two global memory access optimizations in this chapter: (a) Kernel fusion (\textsc{Ikra-Ruby}) and (b) the Structure of Arrays (SOA) data layout (\textsc{Ikra-Cpp}). These two optimizations improve memory access by (a) reducing the amount of memory transfer between global memory and GPU registers and (b) improving memory coalescing and cache performance.

In the remainder of this thesis, our focus will be on \textsc{Ikra-Cpp}. In the next two chapters, we will extend \textsc{Ikra-Cpp} with a dynamic memory allocator and a memory defragmentation system.

\chapter[Dynamic Mem. Allocation with SOA Performance Characteristics]{Dynamic Memory Allocation with SOA Performance Characteristics}
\label{sec:chapter_dynasoar}
\newcommand\soaalloc{\textsc{DynaSOAr}}
Dynamic memory management and the ability/flexibility of creating/deleting objects at any time is one of the corner stones of object-oriented programming. We believe that poor support for dynamic memory allocation is one of the reasons why many GPU programmers avoid object-oriented programming entirely.

In this chapter, we present \soaalloc{}, a CUDA framework for SMMO applications. \soaalloc{} is a parallel, lock-free, dynamic memory allocator, combined with an efficient parallel do-all operation and an embedded C++/CUDA DSL to enable object-oriented abstractions in an SOA layout. The DSL built on top of \textsc{Ikra-Cpp} and \textsc{DynaSOAr} can be seen as \textsc{Ikra-Cpp} extended with a dynamic memory allocator, although with a slightly different notation and syntax.

\minitoc

\paragraph{Outline}
This chapter is organized as follows. Section~\ref{sec:design_goals} describes \textsc{DynaSOAr}'s design goals and its programming interface, an extension of SMMO. Section~\ref{sec:dynasoar_arch_sec} presents \textsc{DynaSOAr}'s high-level design and architecture. Section~\ref{sec:optimizations} describes optimizations that improve the runtime performance of (de)allocate operations. Section~\ref{sec:concurrency_dynasoar} describes how \textsc{DynaSOAr} handles concurrency and argues for its correctness. Section~\ref{sec:related_work} compares \textsc{DynaSOAr} with other dynamic GPU memory allocators. Section~\ref{sec:benchmark} presents a performance evaluation with synthetic benchmarks and a variety of SMMO applications, which are described in more detail in Chapter~\ref{chap:smmo_examples}. Finally, Section~\ref{sec:conclusion} concludes this chapter.

\paragraph{Overview}
Dynamic memory management on GPUs is a hard problem. Due to the massive parallelism and data-parallel execution of GPUs, the number of simultaneous (de)allocations is significantly higher than on other parallel hardware architectures. Massively parallel SIMD allocations also follow patterns different from CPU/MIMD allocations: Most memory requests are small in size\footnote{If thousands of threads were to request large memory allocations, a GPU would run out of memory immediately.} and due to mostly regular control flow, many allocations have the same byte size. In recent years, fast, dynamic memory allocators have been developed for GPUs~\cite{6339604,5577907, Widmer:2013:FDM:2458523.2458535, Vinkler:2015:RED:3071494.3071506, DBLP:journals/corr/abs-1710-11246, Spliet:2014:KDM:2588768.2576781, osti_1398234, Gelado:2019:TGM:3293883.3295727} and demanded by application developers~\cite{Zhu:2015:PIM:2817095.2817115, master_th_cuda_allc, doi:10.1002/cpe.3808, IJNC126, Schafer:2013:RLD:2492045.2492052, Li:2014:ENS:2701002.2701020, Li:2015:CAS:2769458.2769470}, showing a growing interest in better programming models and abstractions that have long been available on other platforms. However, while these allocators often provide good (de)allocation performance, they miss key optimizations for structured data (such as SOA), leading to poor data locality and memory bandwidth utilization when accessing allocated memory.


In contrast to state-of-the-art allocators, \textsc{DynaSOAr} does not only control the data placement/layout through its memory allocator, but also data access through its do-all operation. In SMMO applications, \soaalloc{} achieves superior performance compared to state-of-the-art allocators due to three main optimizations.

\begin{itemize}
\item Objects are stored in a \textbf{Structure of Arrays (SOA)} data layout, a best practice for structured data in SIMD programs, making usage of allocated memory more efficient when used in conjunction with \soaalloc{}'s do-all operation.
\item Memory fragmentation caused by dynamic object allocation/deallocation is minimized with \textbf{hierarchical bitmaps}. This is important because fragmentation diminishes the benefit of the SOA layout (Section~\ref{sec:chap_gpu_mem_defrag}) and adversely affects cache performance~\cite{Grunwald:1993:ICL:155090.155107}.
\item Object allocation and deallocation performance is optimized with a number of \textbf{low-level techniques}. For example, \soaalloc{} combines allocation requests within SIMD thread groups (\emph{warps}) to reduce the number of memory accesses during allocations~\cite{5577907} and takes advantage of efficient bitwise operations/integer intrinsics.
\end{itemize}

\paragraph*{Contributions}
This chapter makes the following contributions.

\begin{itemize}
  \item The design and implementation of \soaalloc{}, a dynamic object allocator for CUDA; with fast (de)allocation and a parallel do-all operation. To the best of our knowledge, \textsc{DynaSOAr} is the first dynamic allocator that stores objects in an SOA data layout.
  \item An extension of the SOA data layout to dynamic object sets and subclassing. While in \textsc{Ikra-Cpp}, the maximum number of objects of each class had to be specified as a compile-time constant, objects can be freely allocated in \textsc{DynaSOAr} without such restrictions.
  \item A concurrent, lock-free, hierarchical bitmap data structure, based on atomic operations and retry loops.
  \item A comparison and evaluation of existing state-of-the-art GPU memory allocators on SMMO applications.
\end{itemize}

\paragraph{Publications}
This chapter is in part based on the following papers.
\begin{itemize}
  \item Matthias Springer, Hidehiko Masuhara. \textbf{``DynaSOAr: A Parallel Memory Allocator for Object-oriented Programming on GPUs with Efficient Memory Access.''} In: \emph{Proceedings of the 33rd European Conference on Object-oriented Programming}. ECOOP 2019. Leibniz-Zentrum f{\"u}r Informatik, Dagstuhl Publishing, 2019, LIPIcs, Vol. 134, pp.~17:1--17:37. \texttt{\doi{10.4230/LIPIcs.ECOOP.2019.17}}.
\end{itemize}

\section{Design Goals}
\label{sec:design_goals}
\textsc{DynaSOAr} is a CUDA framework for SMMO applications and consists of three parts.

\begin{description}
  \item[Memory Allocator] We developed a dynamic memory allocator that provides \texttt{new}/ \texttt{delete} operations in GPU code and stores objects in an SOA data layout. The main task of the allocator is to decide where to store each field value of each object on the heap.
  \item[Data Layout DSL] We developed an embedded C++ DSL to support OOP abstractions while storing objects in a custom layout. We could alternatively implement \textsc{DynaSOAr} in a language that allows programmers to specify custom data layouts (e.g., Shapes~\cite{Franco:2017:YAG:3133850.3133861} or ispc~\cite{6339601}), but such languages have limited GPU support.
  \item[Parallel Do-All] We developed an object enumeration strategy for SMMO applications that achieves efficient access of allocated memory on SIMD architectures. By controlling memory allocation and memory access, applications can achive better performance with \textsc{DynaSOAr} than with other state-of-the-art allocators, which are only concerned with memory allocation.
\end{description}

\textsc{DynaSOAr}'s DSL builds on top of \textsc{Ikra-Cpp}'s embedded C++ DSL for object-oriented programming with SOA layout (Section~\ref{sec:data_layout_dsls_ikracpp}). Its purpose is to make \textsc{DynaSOAr} easier to use for programmers. This chapter is mainly about the memory allocator and the parallel do-all operation.

\subsection{Programming Interface}
In contrast to general memory allocators, \soaalloc{} is an \emph{object allocator}. The types (classes/structs) that can be allocated must be specified at compile time. \soaalloc{} provides five basic operations. The operations for object enumeration follow the \textsc{Ikra-Cpp} API (Section~\ref{sec:ikra_cpp_api_sect3}), but are implemented differently and do not support host-side execution. All operations except for \texttt{parallel\_do} and \texttt{parallel\_new} are \emph{device} functions that can only be called from GPU code.

\begin{itemize}
  \item \texttt{\textbf{new}(d\_allocator) T(args...)}: Allocates a new object of type $T$ and returns a pointer to the object. The \emph{placement new} notation~\cite{cpp_placement} is a common C++ pattern for arena allocation and \texttt{d\_allocator} is the allocator/arena in which the object is allocated.
  \item \texttt{destroy(d\_allocator, ptr)}: Deletes an object with pointer \texttt{ptr}, assuming that the object was allocated with \texttt{d\_allocator}\footnote{There is no \emph{placement delete} syntax, so it is a common pattern to use a separate \texttt{destroy} function~\cite{placement_delete}.}.
  \item \texttt{HAllocatorHandle::parallel\_do<S, \&T::func>(args...)}: Launches a GPU kernel that runs a member function \texttt{T::func} for all objects of type $S$ and subtypes\footnote{To avoid branch divergence, we launch a separate kernel for every type.} existing at launch time (\emph{parallel do-all}), where $S <: T$. \texttt{T::func} may allocate new objects, but they are not enumerated by this parallel do-all. \texttt{T::func} may deallocate any object of different type $U \not= S$, but \texttt{\textbf{this}} is the only object of type $S$ it may deallocate (delete itself). This is to avoid race conditions.
  \item \texttt{HAllocatorHandle::parallel\_new<T>(n, args...)}: Launches a GPU kernel that instantiates $n$ objects of type $T$. This operation calls the constructor of $T$ in parallel with an object index (between 0 and $n$) as first argument, followed by \texttt{args...}. $T$ must have a suitable constructor.
  \item \texttt{DAllocatorHandle::device\_do<S, \&T::func>(args...)}: Runs a member function \texttt{T::func} for all objects of type $S$ in the current GPU thread, where $S <: T$. Can only be used inside of a \texttt{parallel\_do} or a manually launched GPU kernel. This is a sequential \emph{for-each} loop. It is typically used for processing all pairs of objects (e.g., in n-body simulations). 
\end{itemize}

Listing~\ref{lst:short_example} shows parts of an implementation of an n-body simulation with collisions to illustrate \textsc{DynaSOAr}'s API and DSL (full example in Section~\ref{sec:nody_with_coll}). 

\begin{lstfloat}
\begin{lstlisting}[language=c++,caption={\textsc{DynaSOAr} API Example: \textsf{n-body}}, label={lst:short_example}, morekeywords={__device__}]
#include "dynasoar.h"

class Body;  // Pre-declare all classes. This simple example has only one class.
using AllocatorT = SoaAllocator</*max_num_obj=*/ 16777216, /*T...=*/ Body>;
__device__ DAllocatorHandle<AllocatorT> d_allocator;

class Body : public AllocatorT::Base {  // Can subclass other user-defined class.
 public:
  // Pre-declare all field types. DynaSOAr uses these to compute the size of blocks.
  declare_field_types(Body, float /*pos_x_*/, float /*pos_y_*/,
                            /* ... */, bool /*was_merged_*/)
 private:
  // Declare fields with proxy types but use like normal C++ fields (as in Ikra-Cpp).
  Field<Body, 0> pos_x_;               // Position X
  Field<Body, 1> pos_y_;               // Position Y
  /* other fields omitted... */
  Field<Body, 9> was_merged_;          // Was this body merged into another one?

 public:
  __device__ Body(float pos_x, float pos_y, float vel_x, float vel_y, float mass)
    : pos_x_(pos_x), pos_y_(pos_y), vel_x_(vel_x), vel_y_(vel_y), mass_(mass) {}

  // This constructor is invoked by parallel_new.
  __device__ Body(int id) : Body(/*pos_x=*/ random_float(0, 1), /*...*/) {}

  __device__ void apply_force(Body* other) {
    if (other != this) {
      float dx = pos_x_ - other->pos_x_;  float dy = pos_y_ - other->pos_y_;
      float dist = sqrt(dx*dx + dy*dy);
      float F = kGravityConstant * mass_ * other->mass_ / (dist * dist);
      other->force_x_ += F * dx / dist;  other->force_y_ += F * dy / dist;
    }
  }

  __device__ void step_1_compute_force() {
    force_x_ = force_y_ = 0.0f;
    d_allocator->device_do<Body, &Body::apply_force>(this);
  }

  __device__ void step_2_move(float dt) {
    vel_x_ += force_x_ * dt / mass_;  vel_y_ += force_y_ * dt / mass_;
    pos_x_ += dt * vel_x_;            pos_y_ += dt * vel_y_;
  }

  __device__ void step_6_delete_merged() {
    if (was_merged_) { destroy(d_allocator, this); }
  }
};

int main() {
  // Create new allocator. This will allocate a large buffer ("heap") on the GPU.
  auto* h_allocator = new HAllocatorHandle<AllocatorT>();
  // Allocate 65536 new bodies, randomly initialized.
  h_allocator->parallel_new<Body>(65536);
  for (int i = 0; i < kIterations; ++i) {
    h_allocator->parallel_do<Body, &Body::step_1_compute_force>();
    h_allocator->parallel_do<Body, &Body::step_2_move>(/*dt=*/ 0.5);
    /* some steps omitted... */
    h_allocator->parallel_do<Body, &Body::step_6_delete_merged>();
  }
  delete h_allocator;  // Deallocate buffer and all allocations within.
}
\end{lstlisting}
\end{lstfloat}

\paragraph{Data Layout DSL}
Similar to \textsc{Ikra-Cpp}, programmers must define their classes with our data layout DSL. Through this DSL, \textsc{DynaSOAr} can programmatically \emph{reflect} on the classes/fields that are defined in an application, somewhat similar to the Java Reflection API or metaobject protocols~\cite{Chiba:1995:MPC:217838.217868}.

Before the actual application code begins, programmers must pre-declare all classes (Line~3) and define an allocator type (Line~4). Every allocator type is a template instantiation of \texttt{SoaAllocator}. The first template argument is the maximum number of objects of the smallest type, which indirectly defines the heap size, as described in detail later. The other template arguments list all types that are under control of the allocator (variadic template). This design imposes some restrictions with respect to separate compilation: While different compilation units can be developed/compiled separately and then linked together in C++, the compilation unit containing the \textsc{DynaSOAr} memory allocator requires definitions of all classes/structs that should be dynamically allocated with \textsc{DynaSOAr}. However, this is a minor issue at the moment, since GPU applications rarely use separate compilation.

All user-defined classes/structs must in some way inherit from a special class \texttt{AllocatorT::Base}; either as a direct subclass or by inheriting from a class that already inherits from that class. Fields types must be pre-declared (Line~10) and then declared with proxy types (Lines~14--17). 

\paragraph{Class Inheritance}
In contrast to \textsc{Ikra-Cpp}, a user-defined class can inherit from another user-defined class that is managed by the allocator. Multiple inheritance is not supported at the moment. A class can be marked as \emph{abstract} by defining a static field \texttt{static const bool kIsAbstract = true;} after the field pre-declarations.

Programmers can downcast object pointers with the \texttt{cast<T>()} member function defined by \texttt{AllocatorT::Base}. This function is similar to C++ \texttt{dynamic\_cast<T*>} \texttt{(ptr)}; it performs a dynamic type check and returns \texttt{nullptr} if the object is not of runtime type \texttt{T}.

The n-body simulation example in this chapter does not utilize class inheritance, but Chapter~\ref{chap:smmo_examples} contains a number of source code examples with class inheritance.

\subsection{Memory Access Performance}
The main insight of our work is that optimizing only for fast (de)allocations is not enough. Optimizing the access of allocated memory can result in much higher speedups, because device (\emph{global}) memory access is the biggest bottleneck of memory-bound GPU applications:

\begin{description}
\item[Latency] Global memory access instructions have a very high latency at around 400--800 clock cycles, compared to arithmetic instructions at around 6--24 cycles. Programmers can hide latency with \emph{high occupancy}~\cite{Volkov:EECS-2016-143} (i.e., running many threads).

\item[Memory Bandwidth] The global memory bandwidth is a limiting factor. Peak memory transfer rates can be achieved only with \emph{memory coalescing}: When the threads in a GPU application simultaneously access different memory addresses, the GPU coalesces accesses from the same SIMD thread group (\emph{warp} in CUDA, every 32 consecutive threads) into one physical transaction if the addresses are on the same 128-byte cache line~\cite{5473222}. However, if threads access data on multiple cache lines (e.g., non-contiguous, spread-out addresses), more transactions are needed\footnote{This is similar to vectorized loads, but coalescing is performed at runtime by the hardware.}, which reduces transfer rates significantly. The CUDA Best Practices Guide puts a \emph{high priority} note on coalesced memory accesses~\cite{nvidia_memoryco}. 

\item[Caches] Hits in the L1/L2 cache are served much faster (less latency, memory bandwidth pressure) than global memory loads. Field reordering and structure splitting are common techniques for increasing the number of hot fields in cache~\cite{Chilimbi:1999:CSD:301618.301635}.

\end{description}

\textsc{DynaSOAr} achieves good memory access performance with an SOA-style data layout: First, SOA increases memory coalescing because values of the same field, which are accessed simultaneously in SIMD, are stored together. Second, SOA is an extreme form of structure splitting and can improve cache utilization because fields that are not accessed do not occupy cache lines.

\subsection{High Density Memory Allocation}
An SOA data layout (Figure~\ref{fig:clustering_ex}a) achieves good memory performance but is not suitable for dynamic allocation: The size of SOA arrays is fixed and new allocations cannot be accommodated once all array slots are occupied.

\textsc{DynaSOAr}'s design is based on the insight that a \emph{clustered layout} with SOA-style structure splitting (Figure~\ref{fig:clustering_ex}b) has the same cache/vector performance characteristics as an SOA layout, if scalar values are stored in dense clusters of at least 128~bytes (vector and cache line size) and clusters are aligned to 128~bytes, regardless of where the clusters themselves are located in memory.

Figure~\ref{fig:clustering_ex} illustrates the number of required memory transactions for reading 24 floats simultaneously. For illustration purposes, we assume a warp size (vector length) of 4 instead of 32. Both layouts require the same number memory transactions because accesses can be equally well coalesced/vectorized in both layouts. This shows that a perfect SOA layout is not required for perfect memory coalescing. This insight is exploited by \textsc{DynaSOAr}'s allocation policy and gives \textsc{DynaSOAr} more freedom in the placement of allocations.


\begin{figure}[!t]
  \centering
  \includegraphics[width=\textwidth]{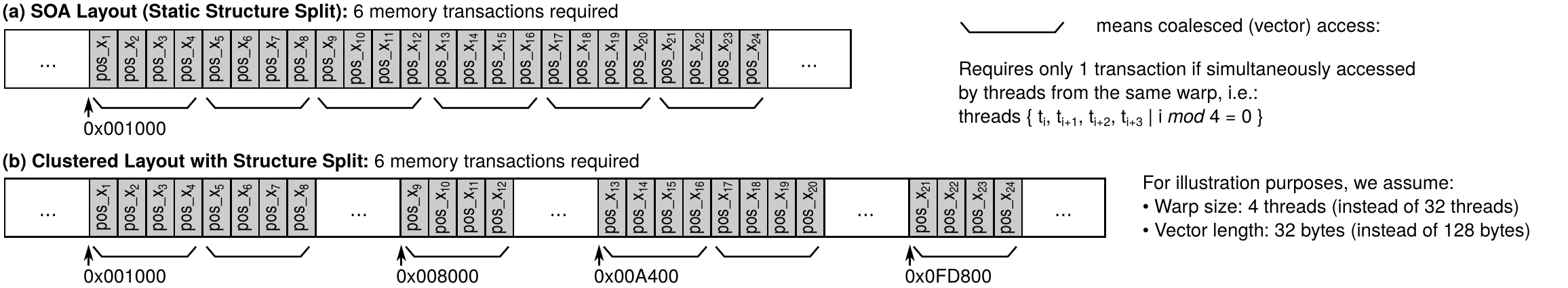}
  \caption{Data layouts: SOA layout and \textsc{DynaSOAr}'s SOA-style layout}
  \label{fig:clustering_ex}
\end{figure}

\subsection{Parallel Object Enumeration Strategy}
Current GPUs follow the Single-Instruction Multiple-Threads (SIMT) execution model. Intuitively, every SIMD lane corresponds to a thread and every group of consecutive 32 threads forms a \emph{warp} which executes an instruction on a vector register.

To benefit from memory coalescing, the threads of a warp must access addresses on the same 128-byte L1 cache line. In an SOA data layout, this is achieved when the threads of a warp read/write the same fields of objects with \emph{contiguous indices} at the same time. Intuitively, threads in a warp should process \emph{neighboring} objects.

In \textsc{DynaSOAr}, programmers invoke GPU kernels with parallel do-all operations. These operations must (a) spawn enough GPU threads to hide latency, but not too many to avoid inefficiencies, and (b) assign objects to threads in such a way that memory access is coalesced.



\subsection{Scalability}
Memory allocations require some sort of synchronization between threads to prevent \emph{collisions}, i.e., two threads allocating the same memory location. To avoid collisions, some allocators such as Cilk~\cite{Blumofe:1999:SMC:324133.324234} and Hoard~\cite{Berger:2000:HSM:378993.379232} utilize private heaps, but such designs can lead to high memory consumption (\emph{blowup})~\cite{Berger:2000:HSM:378993.379232} and are infeasible on massively parallel architectures with thousands of threads.

State-of-the-art GPU allocators such as ScatterAlloc~\cite{6339604} and Halloc~\cite{hallocweb} reduce collisions with hashing, which scatters allocations almost randomly in the heap. This would render an SOA layout useless and defeat one of \textsc{DynaSOAr}'s main optimizations.

With such design restrictions, \textsc{DynaSOAr} is bound to have less efficient allocations than other allocators. However, as we show throughout this chapter, \textsc{DynaSOAr} can more than make up for slow allocations with more efficient access of allocated memory.


\paragraph{False Sharing}
Previous CPU memory allocator designs emphasize mechanisms for reducing false sharing, which can degrade performance~\cite{Berger:2000:HSM:378993.379232}. A memory segment $A$ is falsely shared if it is brought into a processor cache because it happens to be located on the same cache line as another memory segment $B$ that is actually accessed. If caches are \emph{coherent}, the entire cache line is invalidated if another processor modifies $A$, even though $A$ is never read or written by the original processor.

False sharing is not an issue on GPUs, because L1 caches are not coherent. Programmers must use the \texttt{volatile} keyword or atomic operations to enfore a read or write to the shared L2 cache or global memory. 

\section{Architecture Overview}
\label{sec:dynasoar_arch_sec}
\soaalloc{} manages a single, large heap in global memory on device. The heap is divided into $M$ blocks and every block has the same number of bytes. $M$ and the size of each block are compile-time constants that are calculated based on the set of classes that are under control of the allocator (template arguments to \texttt{SoaAllocator}), as described below.

A block contains only objects of the same type (class/struct), stored in SOA data layout (Figure~\ref{fig:overview_heap_soaalloc}). Once a block is initialized (\emph{allocated}) for a certain type, only objects of that type can be stored in that block until the block is deallocated.

The maximum number of objects in a block (\emph{block capacity}) depends on its type, because structs/classes may have different sizes. To improve clustering, \soaalloc{} allocates new objects in already existing, non-full blocks (\emph{fast path}). We call such blocks \emph{active}, because they participate in allocations (Figure~\ref{fig:block_states_dynasoar}). If no active block could be found, a new block is allocated and becomes active (\emph{slow path}).

\label{sec:arch_overview}
\begin{figure}
  \centering
  \includegraphics[width=\columnwidth]{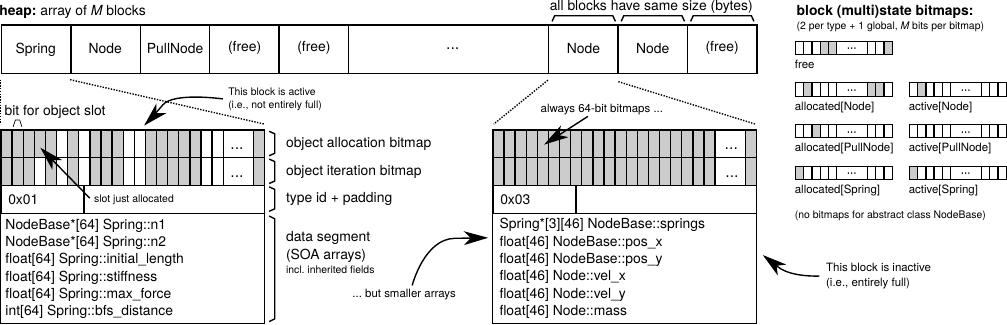}
  \caption[\textsc{DynaSOAr} heap layout of \textsf{structure}]{Example: Heap layout for an FEM simulation (Section~\ref{sec:example_structure_sec7}) of a crack in a composite material. The heap is divided into $M$ blocks of equal size. Every block has the same structure: an allocation bitmap, an iteration bitmap, and a type identifier, followed by a data segment storing objects in SOA layout.} %
  \label{fig:overview_heap_soaalloc}
\end{figure}

\begin{figure}
  \begin{minipage}[c]{0.4\textwidth}
    \includegraphics[width=\textwidth]{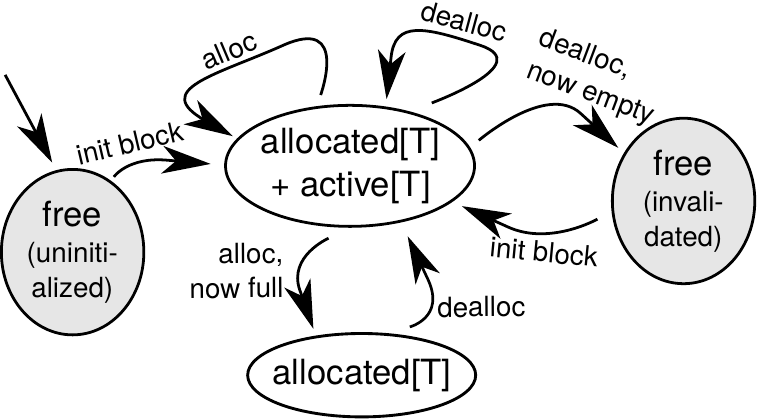}
  \end{minipage}\hfill
  \begin{minipage}[c]{0.56\textwidth}
    \caption[\textsc{DynaSOAr} block state transitions]{Block state transitions. At first, blocks are in an uninitialized state. As part of allocation, new active blocks may be initialized (\emph{allocated}). Active blocks become inactive when they are full. Inactive blocks become active again an object is deallocated. Active blocks are invalidated when their last object is deallocated. Invalidated blocks can be reinitialized (to any type) and are handled similar to uninitialized blocks.
    } \label{fig:block_states_dynasoar}
  \end{minipage}
\end{figure}

\subsection{Block Structure}
Every block has two 64-bit object bitmaps: An \emph{object allocation bitmap} and an \emph{object iteration bitmap}. The allocation bitmap tracks allocated slots in the block. The iteration bitmap is used for object enumeration and overwritten with the allocation bitmap before every parallel do-all operation. Its purpose is to ensure that objects that were created during a parallel do-all operation are not enumerated by the same parallel do-all operation; that would a race condition.

The \emph{type identifier} is a unique ID for the type $T$ of a block. The remainder of the block is occupied by padding and the \emph{data segment}, storing $1 \leq N_T \leq 64$ objects in SOA layout. The data segment begins with SOA arrays for inherited fields and ends with SOA arrays of newly introduced fields.

Slots are marked as (de)allocated with atomic AND/OR operations that change a single bit of the object allocation bitmap. Based on their return value\footnote{An atomic operation returns the value in memory before modification.}, we know ...

\begin{itemize}
  \item ... if an allocation was successful or another thread was faster allocating the same slot.
  \item ... if a particular allocation filled up a block (i.e., allocated the last slot).
  \item ... if a particular deallocation emptied a block (i.e., deallocated the last slot).
\end{itemize}

If a thread filled up a block or emptied a block, it is that thread's \emph{responsibility} to update the other internal data structures. This is a common pattern in lock-free designs~\cite{Michael:2004:SLD:996841.996848}. Note that every block has the same byte size and structure; e.g., the bitmaps are always at the same offset. This is an important property for the correctness of our (de)allocation algorithms.

\subsection{Block Capacity}
\label{sec:block_capacity_dynasoar_s5}
The capacity of a block (maximum number of objects) depends on the size (bytes) of the type of objects in the block. If \soaalloc{} manages objects of types $T_1$, $T_2$, ..., $T_n$ and $s=\argmin_{i \in 1...n} \mathit{size}(T_i)$ is the index of the smallest type, then the capacity $N_{T}$ of a block of type $T$ is determined as follows.

\begin{align*}
N_{T} = \left\lfloor{\frac{64 \cdot \mathit{size}(T_s)}{\mathit{size}(T)}}\right\rfloor \tag{\emph{block capacity}}
\end{align*}

According to this formula, a block of the smallest type $T_s$ has capacity 64. This is why the bitmaps within a block have 64 bits. Given a fixed heap size, the size of $T_s$ determines the block size in bytes and thus the number of blocks $M$.

\begin{align*}
M = \frac{\mbox{heap size (bytes)}}{\mbox{size of block of type $T_s$ (bytes)}} \tag{\emph{block capacity}}
\end{align*}

In our current implementation, programmers specify the size of the heap not in bytes but in terms of the number of objects of the smallest type $T_s$ (first template argument of \texttt{SoaAllocator}). This number must be a multiple of 64. This is merely an implementation detail and will change in future versions of \textsc{DynaSOAr}.

\paragraph{Very Small/Large Object Sizes}
As soon as a type $T$ is more than twice as big as $T_s$, the benefit of the SOA layout is starting to fade away for $T$, because the number of objects in such a block $N_T$ will be smaller than 32. However, assuming 32-bit scalar types, the maximum amount of memory coalescing can only be achieved with vector loads (cluster sizes) of 32 values (Section~\ref{sec:memory_coalescing_experiment_background}). Moreover, \soaalloc{} cannot handle cases in which a type is more than 64 times bigger than the smallest type. In reality, these limitations proved to be insignificant. None of our benchmarks experienced a slowdown due to unfavorable block sizes.

\subsection{C++ Data Layout DSL and Object Pointers}
\label{sec:dynasoar_cpp_datalayout_dsl_fp}
\begin{figure}
  \begin{minipage}[c]{0.5671195351024146\textwidth}
    \includegraphics[width=\textwidth]{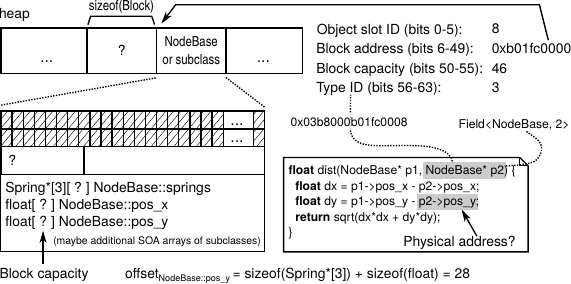}
  \end{minipage}\hfill
  \begin{minipage}[c]{0.4\textwidth}
    \caption[\textsc{DynaSOAr} fake object pointer example]{Object pointer example. The static type of \texttt{p2} is \texttt{NodeBase*}. The corresponding block has SOA arrays for \texttt{NodeBase} fields and for the additional fields of the runtime type of \texttt{p2}. The size of those arrays is not statically known and depends on the runtime type of \texttt{p2}.
    } \label{fig:global_ref}
  \end{minipage}
\end{figure}

Field access is simple in most object-oriented systems: Given an object pointer, which is a memory location, a field value is stored at a fixed offset from the memory location.

%


In \textsc{DynaSOAr}, an object pointer is not a memory location, but a combination of various components (\emph{fake pointer}), similar to \emph{global references} in \emph{Shapes}~\cite{Franco:2017:YAG:3133850.3133861}. Upon field access, the \textsc{DynaSOAr} DSL transparently converts object pointers to memory locations, without breaking C++'s OOP abstractions. We follow the implementation strategy of \textsc{Ikra-Cpp}, where fields are declared with proxy types \texttt{Field<B, N>}, which can be implicitly converted to \texttt{T\&} values~\cite{cpp_obj}, where \texttt{T} is the N-th predeclared field type of \texttt{B}. This conversion is defined by our DSL and computes the actual, physical memory location within a data segment.

Fake pointers in \textsc{DynaSOAr} are different from fake pointers in \textsc{Ikra-Cpp}. They do not just encode an object ID, but four different components. A \textsc{DynaSOAr} object pointer (Figure~\ref{fig:global_ref}) is based on the address of the block in which the object is located. All blocks are aligned to 64~bytes, so we can store the object slot ID in the 6 least significant bits\footnote{Recall that blocks never have more than 64 objects, so 6 bits are enough.}. Since recent GPU architectures have at most 24~GB of memory and no virtual memory, only the 35 least significant bits are used in memory addresses and the remaining 29 bits are always zero\footnote{We experimentally verified this on NVIDIA Maxwell and NVIDIA Pascal.}. We store additional information in these bits: The 8 most significant bits store the type identifier of the class type for fast instance-of checks. The next 6 bits store the capacity of the block. The remaining 15 bits are effectively not utilized at the moment. Note that, while C++ stores runtime types with a vtable pointer at the beginning of an object, we store runtime type information in unused pointer bits.

\paragraph{Field Access}
While in most object-oriented systems, runtime type information is only required for virtual function calls, \textsc{DynaSOAr} requires the block capacity (a property of the runtime type!) also for field accesses, because SOA array offsets within the data segment depend on it.

For example, \texttt{p2} in Figure~\ref{fig:global_ref} is statically known to be of type \texttt{NodeBase*}, but the block capacity (size of SOA arrays) depends on the runtime type, which can be any subclass of \texttt{NodeBase}. Those subclasses can have different block capacities. The size of SOA arrays and the object slot ID are required to compute the physical location of \texttt{p2->pos\_y} based on the block address, so we store both inside object pointers.

\begin{lstfloat}
\begin{lstlisting}[language=c++,caption={[\textsc{DynaSOAr}: Field address computation]Field address computation}, label={lst:dynasoar_impl_conv}, morekeywords={uint64_t}]
// Impl. conv. operator: E.g., convert Field<NodeBase, 2> to float& in Figure <@\ref{fig:global_ref}@>.
// <@\textit{BaseType}@>: N-th predeclared type in B (within <@\textit{declare\_field\_types}@>).
template<typename B, int N>
Field<B, N>::operator <@\textit{BaseType}@>&() {
  int offset = ...;  // Computed with templ. metaprog. <@$\textsf{offset}_{\textsf{B::fieldname}}$@> in Figure <@\ref{fig:global_ref}@>.
  auto obj_ptr = reinterpret_cast<uint64_t>(this) - N;
  // Bits 0-49 and clear 6 least significant bits.
  auto* block_address = reinterpret_cast<char*>(obj_ptr & 0x3FFFFFFFFFFC0);
  int obj_slot_id = obj_ptr & 0x3F;  // Bits 0-5
  int block_capacity = (obj_ptr & 0xFC000000000000) >> 50;  // Bits 50-55
  auto* soa_array = reinterpret_cast<<@\textit{BaseType}@>*>(
      block_address + field_offset * block_capacity);
  return soa_array[obj_slot_id];
}
\end{lstlisting}
\end{lstfloat}




This computation (Listing~\ref{lst:dynasoar_impl_conv}), along with bit-shifting and bit-AND operations for extracting all components from a fake object pointer, is performed on every field read/write. Compilers may eliminate redundant computations of multiple accesses with peephole optimizations, but field accesses still require more complex arithmetics compared to an AOS layout. This overhead may seem large, but arithmetic operations are much faster than memory access, even in case of an L2 cache hit. Overall, the performance benefit of SOA is much larger than the address computation overhead.


\subsection{Block Bitmaps}
\label{sec:dynsoar_block_bitmaps}
To find blocks or free memory quickly during object enumeration or object allocation, \soaalloc{} maintains three bitmaps of size $M$, where $M$ is the number of blocks on the heap.

\begin{itemize}
  \item The \emph{free block bitmap} indexes block locations that are not yet allocated. This bitmap is used to determine where new blocks are allocated. Bit $i$ is 1 iff block $i$ is free (uninitialized or invalidated). Initially, every bit is 1.
  \item There is one \emph{block allocation bitmap} for every type $T$. That bitmap indexes blocks of type $T$ and is used for enumeration of all objects. Blocks of subclasses are not included in bitmaps of the superclass. Initially, every bit is 0.
  \item There is one \emph{active block bitmap} for every type $T$, indexing allocated, non-full blocks. If a bit is 1, then the same bit in the block allocation bitmap must also be 1. This bitmap is used to find a block in which a new object can be allocated. Initially, every bit is 0.
\end{itemize}

Due to concurrent (de)allocations, block bitmaps cannot be kept consistent with the actual block states all the time, as indicated by object allocation bitmaps and type identifiers of blocks. However, we designed our algorithms in such a way that they can handle such inconsistencies and keep block states and block bitmaps \emph{eventually consistent} (Section~\ref{sec:concurrency_dynasoar}).

\subsection{Object Slot Allocation}
When a new object is created, \textsc{DynaSOAr} allocates memory and runs the constructor on the object pointer. Algorithm~\ref{alg:alloc_algo} shows how memory is allocated. This algorithm runs entirely on the GPU and is completely lock-free.

\textsc{DynaSOAr} tries to allocate memory in an already existing, active block. If no block could be found, it first initializes a new block at a location that is known to be free (\emph{slow path}). The state of the new block is \emph{allocated} and \emph{active}, so that the new block can also be found by other threads.

Once a block was selected, an object slot is reserved by atomically finding and flipping a bit from 0 to 1 in the object allocation bitmap (details in Algorithm~\ref{alg:block_allocate}). Based on the return value of the atomic operation, we know if this operation just allocated the last slot. In that case, the block is marked as \emph{inactive} (Line~12).

Since the allocator is used concurrently by many threads, we may select a block (Line~2) that is full or no longer exists when attempting to reserve an object slot (Line~8). If the block is full, object reservation fails and we retry by selecting a new active block. If the block no longer exists, we have to consider three cases\footnote{We give a more systematic correctness argument in Section~\ref{sec:concurrency_dynasoar}.}.

\begin{enumerate}
  \item There is currently no block at this location. In this case, object reserveration fails, because all slots are marked as allocated in the object allocation bitmap when a block is deleted. We call this process \emph{block invalidation} (Section~\ref{sec:safe_mem_reclam_invalidation}).
  \item The block was deleted and a new block of the same type was allocated at the same location. Such ABA problems are harmless and allocation will succeed.
  \item The block was deleted and there is now a block of different type at the same location. At this point, the constructor has not run yet, so no data in the data segment was corrupted. This is because all blocks have the same structure, i.e., the object allocation bitmap is always at the same location. We can safely rollback the allocation by running the deallocation routine.
\end{enumerate}


\SetKwComment{Comment}{$\triangleright$\ }{}
\SetAlgoVlined
\begin{algorithm}[t]
\small
 \Repeat(\Comment*[f]{\textsf{Infinite loop if OOM}}){false}{
  bid $\gets$ active[T].\emph{try\_find\_set}(); \hfill \Comment{\textsf{Find and return the position of any set bit.}}
  \If(\Comment*[f]{\textsf{Slow path}}){\emph{bid} = FAIL} {
    bid $\gets$ free.\emph{clear}();   \hfill \Comment{\textsf{Find and clear a set bit atomically, return position.}}
    \emph{initialize\_block}<T>(bid);  \hfill \Comment{\textsf{Set type ID, init. obj. alloc. bitmap. See Alg.~\ref{alg:block_init}.}}
    allocated[T].\emph{set}(bid)\;
    active[T].\emph{set}(bid)\;
  }
  alloc $\gets$ heap[bid].\emph{reserve}();  \hfill \Comment{\textsf{Reserve an object slot. See Alg.~\ref{alg:block_allocate}.}}
  \If{\emph{alloc} $\not=$ FAIL}{
    ptr $\gets$ \emph{make\_pointer}(bid, alloc.slot)\;
    t $\gets$ heap[bid].type; \hfill \Comment{\textsf{Volatile read}}
    \lIf{\emph{alloc.state} = FULL}{
      active[t].\emph{clear}(bid)
    }
    \lIf{t = T}{
      \Return ptr
    }{
      \emph{deallocate}<t>(ptr); \hfill \Comment{\textsf{Type of block has changed. Rollback.}}
    }
  }
 }
 \caption[\textsc{DynaSOAr}: DAllocatorHandle::allocate<T>]{DAllocatorHandle::allocate<T>() : T* \hfill \fbox{GPU}}
 \label{alg:alloc_algo}
\end{algorithm}
\SetKwComment{Comment}{$\triangleright$\ }{}
\begin{algorithm}[t]
\small
  bid $\gets$ \emph{extract\_block}(ptr)\;
  slot $\gets$ \emph{extract\_slot}(ptr)\;
  state $\gets$ heap[bid].\emph{deallocate}(slot); \hfill \Comment{\textsf{Opposite of slot reservation. See Alg.~\ref{alg:block_deleteslot}.}}
  \uIf(\Comment*[f]{\textsf{Deallocated first object of full block.}}){\emph{state} = FIRST}{
    active[T].\emph{set}(bid)\;
  }
  \ElseIf(\Comment*[f]{\textsf{Deallocated last object of block.}}){\emph{state} = EMPTY}{
    \If(\Comment*[f]{\textsf{Try to invalidate block. See Alg.~\ref{alg:block_invalidate}.}}){invalidate(bid)}{
      t $\gets$ heap[bid].type\;
      active[t].\emph{clear}(bid)\;
      allocated[t].\emph{clear}(bid)\;
      free.\emph{set}(bid)\;
    }
  }
 \caption[\textsc{DynaSOAr}: DAllocatorHandle::deallocate<T>]{DAllocatorHandle::deallocate<T>(T* ptr) : void \hfill \fbox{GPU}}
 \label{algo:dealloc_a}
\end{algorithm}

\subsection{Object Deallocation}
When an object is deallocated, \textsc{DynaSOAr} first determines its runtime type $T$ from the fake object pointer. Then, \textsc{DynaSOAr} runs the C++ destructor and deallocates the memory as shown in Algorithm~\ref{algo:dealloc_a}.

We first determine block and object slot IDs from the object pointer and deallocate the object slot by atomically flipping its bit in the object allocation bitmap from 1 to 0. Based on the return value of the atomic operation we know the fill level of the block right before the deallocation.

If this deallocation freed the first object slot (block previously full), we mark the block as active (Line~5), so that other threads can find it and allocate objects in it.

If this deallocation freed the last object slot (block now empty), we attempt to delete the block (Lines~7--11). Safe memory reclamation is known to be difficult in lock-free algorithms~\cite{Michael:2002:SMR:571825.571829}. The main problem is that one or more contending threads, in the course of their lock-free operations, may have selected the block that we are about to delete for new allocations.

To avoid the block from being modified by other threads, we \emph{invalidate} it (Section~\ref{sec:safe_mem_reclam_invalidation}). Block invalidation attempts to atomically flip all bits in the object allocation bitmap from 0 to 1. If this atomic operation failed to flip at least one bit from 0 to 1 (because it was already 1), another thread must have reserved an object slot in the meantime. In this case, we rollback the changes to the object allocation bitmap and abort block invalidation and deletion.

If invalidation was successful, the block is guaranteed to be empty and cannot be modified by other threads anymore because all bits in the object allocation bitmap are 1. The type of the block may have changed in the meantime (Line~8), but it is now safe to mark this block location as \emph{free}, so that a new block can be initialized at this location.




\subsection{Parallel Object Enumeration: \texttt{parallel\_do}}
\label{sec:par_do_all_sec3}
Parallel do-all is the foundation of SMMO applications. It launches a GPU kernel that runs a method \texttt{T::func} on all objects of a type $T$ (and subtypes). To avoid branch divergence, we run a separate kernel for each subtype. \texttt{T::func} may read and write fields of the object that it is bound to (\texttt{this}). The goal of parallel do-all is to assign objects to GPU threads in such a way that memory coalescing is maximized for those field accesses.

\begin{figure}
  \begin{minipage}[c]{0.52877004242\textwidth}
    \includegraphics[width=\textwidth]{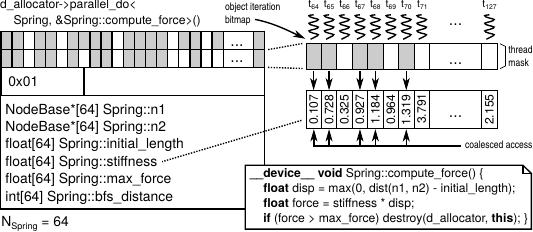}
  \end{minipage}\hfill
  \begin{minipage}[c]{0.43\textwidth}
  \caption[Thread assignment of \textsc{DynaSOAr} \texttt{parallel\_do}]{Thread assignment example. 64 threads with consecutive IDs are assigned to every allocated block of type \texttt{Spring}. Since not all object slots are in use, as indicated by the block iteration bitmap, some threads have no work to do. All other threads can benefit from memory coalescing when reading/writing fields of the object that they are assigned to.}
  \label{fig:thread_assignment_ex}
  \end{minipage}
\end{figure}

Memory coalescing is maximized when all threads of a warp access consecutive memory addresses at the same time (and addresses are properly aligned). In this case, all those memory accesses can be serviced by vector loads/writes. In CUDA, threads are identified by thread IDs. Each warp consists of a consecutive range of threads. E.g., warp 0 consists of threads $t_0, t_1, \ldots t_{31}$. Assuming a block capacity of $N_T$, \soaalloc{} assigns $N_T$ consecutive threads to the objects in a block (Figure~\ref{fig:thread_assignment_ex}). This leads to good memory coalescing on average. Perfect memory coalescing can be achieved if the following two conditions apply.

\begin{itemize}
  \item $N_T$ is a multiple of the warp size 32. If this is not the case, then there are warps whose threads process elements in two or more different blocks.
  \item Objects have good clustering, i.e., every block except for at most one is entirely full. Due to the way objects are allocated (only in active blocks), we expect a high fill level on average.
\end{itemize}

\soaalloc{} uses the block allocation bitmap to find blocks to which threads should be assigned. Assigning only one object to a thread is too inefficient if the number of objects is large. Therefore, a thread $t_\mathit{tid}$ may have to process an object slot in multiple blocks. $\mathit{num}_B(\mathit{tid})$ is the number of blocks that are assigned to a thread $t_\mathit{tid}$. Our scheduling strategy always assigns the same object slot position $\mathit{id}_O(\mathit{tid})$ to a thread, but in multiple blocks $\mathit{id}_B(\mathit{tid})$. In the following formulas, $R$ is an array of indices of all allocated blocks of type $T$, i.e., all blocks containing objects of type $T$. The total number of threads $n$ can be hand-tuned by the programmer. With those formulas, every thread can by itself determine the objects that it should process.
\begin{align*}
\mathit{id}_O(\mathit{tid}) = \mathit{tid} \,\,\, \% \,\,\, N_T \tag{\emph{assigned object slot index}}
\end{align*}
\begin{align*}
\mathit{id}_B(\mathit{tid}) = \left(R\hspace{-0.1cm}\left[\frac{\mathit{tid} + k \cdot n}{N_T}\right] \left.\right\vert k \in [0; \mathit{num}_B(\mathit{tid}))\right) \tag{\emph{assigned block indices}}
\end{align*}
\begin{align*}
\mathit{num}_B(\mathit{tid}) =  \left\lceil \frac{r\cdot N_T - \mathit{tid}}{n} \right\rceil \tag{\emph{number of assigned blocks}}
\end{align*}

\begin{figure}
  \begin{minipage}[c]{0.44187973217\textwidth}
    \includegraphics[width=\columnwidth]{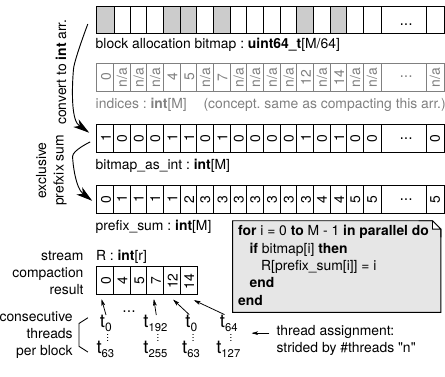}
  \end{minipage}\hfill
  \begin{minipage}[c]{0.5\textwidth}
    \caption[Bitmap compaction with prefix sum]{Example: Compacting block allocation bitmap indices and assigning $n=256$ threads to 6 allocated blocks with $N_T = 64$. The prefix sum retains the order of indices (i.e., $R$ is sorted), but this is not required for the correctness of our algorithms.
    } \label{fig:scan_enumeration}
  \end{minipage}
\end{figure}

The formulas $\mathit{id}_O$ and $\mathit{id}_B$ effectively implement a grid-stride loop (Section~\ref{sec:cuda_prog_model_backgr}). The array $R$ is required because every thread should by itself find the $\frac{\mathit{tid}}{N_T}$-th, $\frac{\mathit{tid}+n}{N_T}$-th, etc. allocated block of type $T$ quickly, without scanning the entire block allocation bitmap or communicating with other threads, which would be very slow.

\textsc{DynaSOAr} precomputes $R$ before every parallel do-all operation (Figure~\ref{fig:scan_enumeration}). Conceptually, this is an application of stream compaction~\cite{IJNC151}: Given a bitmap of size $M$ (size of heap in number of blocks), generate an \emph{indices} array of size $M$ containing $i$ at position $i$ if the $i$-th bit is set. Otherwise, store an \emph{invalid marker}. Now filter/compact the array to retain only valid values, resulting in an array $R$ of size $r$. Note that we do not care if the original ordering of indices is retained. This stream compaction could be implemented as follows. 

\begin{enumerate}
  \item Convert the bitmap into an integer array.
  \item Compute the exclusive prefix sum~\cite{Sengupta06awork-efficient, Billeter:2009:ESC:1572769.1572795} of the integer array.
  \item Compute the final result $R$: If a bit is set in the bitmap, store its index at the position that was computed by the exclusive prefix sum.
\end{enumerate}

Every step of this stream compaction implementation can run in parallel. Section~\ref{sec:hierarchical_bit} describes how this algorithm is further optimized with hierarchical bitmaps to avoid scanning empty bitmap parts.

\section{Optimizations}
\label{sec:optimizations}
This section describes performance optimizations that \soaalloc{} applies in addition to the SOA data layout to achieve good (de)allocation performance and to reduce memory fragmentation.

\subsection{Hierarchical Bitmaps}
\label{sec:hierarchical_bit}
\begin{figure}[t]
  \begin{minipage}[c]{0.44187973217\textwidth}
    \centering
    \includegraphics[width=0.90522368617\columnwidth]{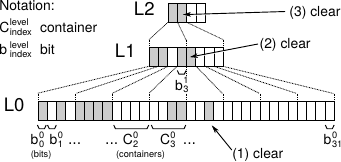}
  \end{minipage}\hfill
  \begin{minipage}[c]{0.5\textwidth}
    \caption[Hierarchical bitmap data structure]{Example: Hierarchical bitmap of size 32 with container size 4 (instead of 64). This example illustrates how (1) a \emph{clear}(18) operation triggers (2) a \emph{clear}(4) operation in the nested bitmap, which triggers (3) a \emph{clear}(1) operation in the next nested bitmap.} \label{fig:bitmap_clear_example}
  \end{minipage}
\end{figure}

\soaalloc{} uses bitmaps for finding blocks or free space for blocks. Since, with growing heap sizes, bitmaps can reach several megabytes in size, we use a hierarchy of bitmaps, such that \emph{set} bits (ones) can be found with a logarithmic order of memory accesses.

\begin{lstfloat}
\begin{lstlisting}[language=c++, numbers=none, morekeywords={uint64_t}, caption={[Structurally recursive C++ bitmap implementation]Structurally recursive C++ bitmap implementation (simplified)}, label={lst:struc_recurs_bitmap_cpp}]
template<int N, bool HasNested>
struct Bitmap;

template<int N>
struct Bitmap<N, /*HasNested=*/ false> {
  static const int kNumContainers = (N + 64 - 1) / 64;  // ceil(N / 64)
  uint64_t containers[kNumContainers];
};

template<int N>
struct Bitmap<N, /*HasNested=*/ true> {
  static const int kNumContainers = (N + 64 - 1) / 64;  // ceil(N / 64)
  static const bool kContinueHierarchy = kNumContainers > 1;

  uint64_t containers[kNumContainers];
  Bitmap<kNumContainers, kContinueHierarchy> nested;
};
\end{lstlisting}
\end{lstfloat}

Our hierarchical bitmaps are structurally recursive (nested bitmaps; Listing~\ref{lst:struc_recurs_bitmap_cpp}) and hide their hierarchy as an implementation detail from their interface. Such bitmaps are used in database systems~\cite{10.1007/978-3-540-39403-7_19} and garbage collectors~\cite{Ueno:2011:ENG:2034773.2034802}, but we do not know of any hierarchical bitmaps that support concurrent modifications.

\paragraph{Data Structure}
A hierarchical bitmap of size $N$ bits consists of two parts: an array of size $\lceil N/64\rceil$ of 64-bit \emph{containers} (\texttt{uint64\_t}), and a \emph{nested bitmap} of size $\lceil N/64 \rceil$ if $N > 64$. A container $C_i^l$ consists of bits $b_{64 \cdot i}^l$, ...,  $b_{64 \cdot i + 63}^l$ and is represented by one bit $b_{i}^{l+1}$ in the nested (higher-level) bitmap (Figure~\ref{fig:bitmap_clear_example}). That bit is set if at least one bit is set in the container.

\begin{align*}
b_i^{l+1} = \bigvee_{k=0}^{63} b_{64\cdot i + k}^{l} \tag{\emph{container consistency}}
\end{align*}

We chose a container size of 64 bits because C++ has a 64-bit integer type and CUDA (and most other architectures) provide atomic operations for modifying 64-bit values. Bits in a container are changed with atomic operations. Higher-level bits (and thus bitmaps) are \emph{eventually consistent} with their containers. Keeping both consistent all the time is impossible without locking, because two different memory locations cannot be changed together atomically. However, due to the design of the bitmap operations, the bitmap is guaranteed to be in a consistent state when all bitmap operations (of all threads) are completed, at the end of a GPU kernel. Bitmap operations retry or give up (\emph{FAIL}) to handle temporary inconsistencies. This is a key difference compared to other lock-free hierarchical data structures such as SNZI~\cite{Ellen:2007:SSN:1281100.1281106}, which have stronger runtime consistency guarantees and require more complex algorithms.

\paragraph{Operations}
All bitmap operations except for \emph{indices()} are device functions that run entirely on the GPU. All operations that modify memory are thread-safe and their semantics are atomic. Internally, they are all implemented with atomic memory operations. 

\begin{itemize}
  \item $\texttt{try\_clear}(\texttt{pos})$ atomically sets the bit at position $\texttt{pos}$ to 0. Returns \emph{true} if the bit was 1 before and \emph{false} otherwise.
  \item $\texttt{clear}(\texttt{pos})$ \emph{switches} the bit at position $\texttt{pos}$ from 1 to 0. Retries until the bit was actually changed by the current thread. This is identical to the following code snippet: \texttt{\textbf{while} (!try\_clear(pos)) \{\}}

  \item $\texttt{try\_set}(\texttt{pos})$ atomically sets the bit at position $\texttt{pos}$ to 1. Returns \emph{true} if the bit was 0 before and \emph{false} otherwise.
  \item $\texttt{set}(\texttt{pos})$ \emph{switches} the bit at position $\texttt{pos}$ from 0 to 1. Retries until the bit was actually changed by the current thread. This is identical to the following code snippet: \texttt{\textbf{while} (!try\_set(pos)) \{\}}

  \item $\texttt{try\_find\_set}()$ returns the position of an arbitrary bit that is set to 1 or \emph{FAIL} if none was found. This operation must be used with caution, because the returned bit position might already have changed when using the result.
  \item $\texttt{clear}()$ atomically clears and returns the position of an arbitrary set bit. This is a combination of \texttt{try\_find\_set} and \texttt{try\_clear} and identical to the following code snippet: \texttt{\textbf{while} ((i = try\_find\_set()) != FAIL \&\& try\_clear(i)) \{\}; \textbf{return} i;}
  \item $\texttt{get}(\texttt{pos})$ and $\texttt{operator[]}(\texttt{pos})$ return the value of the bit at position $\texttt{pos}$.
  \item $\texttt{indices}()$ returns an array of indices of all set bits. This is a host function and cannot be used in a GPU kernel.
\end{itemize}

\SetKwComment{Comment}{$\triangleright$\ }{}
\begin{algorithm}[t]
\small
 cid $\gets$ pos / 64; \hfill \Comment{\textsf{Container ID}}
 offset $\gets$ pos \% 64; \hfill \Comment{\textsf{Index of bit within the container}}
 mask $\gets$ 1 {<}{<} offset\;
 prev $\gets$ \emph{atomicAnd}(\&container[cid], $\sim$mask)\;
 success $\gets$ (prev \& mask) $\not=$ 0\;

\If{success $\wedge$ has\_nested $\wedge$ {popc}(\emph{prev})\hspace{-0.17cm}\tikz[remember picture] \node [] (d){};\hspace{-2pt} \emph{=\,1}}{
  nested.\emph{clear}(cid)\;
}
\begin{tikzpicture}[remember picture, overlay,
    every edge/.append style = { ->, thick, >=stealth,
                                  dashed, line width = 1pt },
    every node/.append style = { align = center,
                                 font=\sffamily\small, fill= gray!20},
                  text width = 2.75cm ]
  \node [notice={(-0.2,0.15)}, below right = 0.3cm and -1.5cm of d]  (D) {\mbox{\begin{varwidth}{2.75cm} \scriptsize \textbf{population count:} number of set bits\end{varwidth}}};

\end{tikzpicture} \hspace{-0.17cm} \textbf{return } success; \hfill \Comment{\textsf{Was the bit switched from 1 to 0?}}
 \caption[Hierarchical Bitmap: Bitmap::try\_clear]{Bitmap::try\_clear(pos) : void  \hfill \fbox{GPU}}
 \label{alg:clear}
\end{algorithm}

\begin{algorithm}[t]
\small
\uIf{has\_nested}{
  cid $\gets$ nested.\emph{try\_find\_set}(); \hfill \Comment{\textsf{Select container based on L+1 bitmap.}}
  \lIf{\emph{cid} = FAIL}{
    \Return \emph{FAIL}
  }
}\Else{
  cid $\gets$ 0; \hfill \Comment{\textsf{This is the top-most bitmap level. Only 1 container.}}
}
offset $\gets$ \emph{ffs}\tikz[remember picture] \node [] (d){};\hspace{-0.17cm}(container[cid])\;
\eIf{offset = \emph{NONE}}{
  \Return \emph{FAIL}\;
}{
\begin{tikzpicture}[remember picture, overlay,
    every edge/.append style = { ->, thick, >=stealth,
                                  dashed, line width = 1pt },
    every node/.append style = { align = center,
                                 font=\sffamily\small, fill= gray!20},
                  text width = 2.6cm ]
  \node [notice={(-0.35,0.075)}, below right = 0.0cm and 2.75cm of d]  (D) {\mbox{\begin{varwidth}{2.6cm} \scriptsize \textbf{find first set}: index of 1\textsuperscript{st} set bit\footnotemark\end{varwidth}}};

\end{tikzpicture}
  \Return 64*cid + offset\;
}

 \caption[Hierarchical Bitmap: Bitmap::try\_find\_set]{Bitmap::try\_find\_set() : int \hfill \fbox{GPU}}
 \label{alg:top_down_tr}
\end{algorithm}

\begin{algorithm}[t]
\small
\eIf{has\_nested}{
  selected $\gets$ nested.\emph{indices}(); \hfill \Comment{\textsf{Select containers based on L+1 bitmap.}}
}{
  selected $\gets$ [0]; \hfill \Comment{\textsf{This is the top-most bitmap level. Only 1 container.}}
}
R $\gets$ array(N)\;
r $\gets$ 0\;
\For(\Comment*[f]{\textsf{\fbox{GPU} (CUDA kernel)}}){$\mbox{cid} \in \mbox{selected}$ \emph{\textbf{in parallel}}}{
  c $\gets$ container[cid]\;
  s $\gets$ \emph{atomicAdd}(\&r, \emph{popc}(c))\;
  \For{$i \gets 0$ \KwTo \emph{\emph{popc}(c))}}{
    R[s + i] $\gets$ 64*cid + \emph{nth\_set\_bit}\tikz[remember picture] \node [] (d){};\hspace{-0.17cm}(c, i)\;
  }
}
 \begin{tikzpicture}[remember picture, overlay,
    every edge/.append style = { ->, thick, >=stealth,
                                  dashed, line width = 1pt },
    every node/.append style = { align = center,
                                 font=\sffamily\small, fill= gray!20},
                  text width = 2.95cm ]
  \node [notice={(-0.15,0.125)}, below left = 0.25cm and -2.0cm of d]  (D) {\mbox{\begin{varwidth}{2.95cm} \scriptsize idx. of $i$\textsuperscript{th} set bit in $c$\end{varwidth}}};

\end{tikzpicture} \hspace{-0.17cm} \Return R.\emph{subarray}(0, r)\;
 \caption[Hierarchical Bitmap: Bitmap::indices]{Bitmap::indices() : int[N] \hfill \fbox{CPU}}
 \label{alg:indices}
\end{algorithm}

\paragraph{Set and Clear with Atomic Operations}
As many other lock-free algorithms, our hierarchical bitmaps are based on a combination of atomic operations and retries~\cite{doi:10.1002/9781119332015.ch3}. The return value of an atomic operation indicates if a bit was actually changed and if it is this thread's responsibility to update the higher-level bitmap (Figure~\ref{fig:bitmap_clear_example}).

As an example, Algorithm~\ref{alg:clear} shows how to clear the bit at position \texttt{pos}. In Line~4, the respective container is bit-ANDed with a mask containing ones everywhere except for that position. This will clear the bit at position \texttt{pos} but leave all other bits unchanged. The current thread actually changed the bit if it is set in \texttt{prev} (Line~5). If this operation cleared the last bit (Line~6), then the bit in the higher-level bitmap must be cleared.\footnotetext{This is slightly different from CUDA, where\texttt{\_\_ffs(b)} returns the position of the 1\textsuperscript{st} set bit plus 1, or 0 if none exists. In this work, $\mathit{ffs}(b)$ returns the position of the 1\textsuperscript{st} set bit, or \emph{NONE} if none exists.}

Note that higher-level bits are always changed with \emph{clear(pos)}/\emph{set(pos)} and not with their respective \emph{try\_} versions, because other concurrently running bitmap operations may not have updated all higher-level bitmaps yet, leaving the data structure in a temporarily inconsistent state. If we were to use \emph{try\_} versions, a mandatory update of the higher-level bitmap could be accidentally dropped due to a bitmap inconsistency. However, \emph{clear(pos)}/\emph{set(pos)} ensure that the update is performed eventually by retrying (and spin-blocking the thread) until the update was successful.

\paragraph{Finding an Arbitrary Set Bit}
Instead of scanning the entire L0 bitmap, set bits can be found faster with a top-down traversal of the bitmap hierarchy, as shown in Algorithm~\ref{alg:top_down_tr}. A request is first delegated to the higher-level bitmap (Line~2) to select a container. When that call returns, a set bit is chosen in the selected container (Line~6).

This operation fails if there is no set bit in the entire bitmap data structure. However, even if the bitmap has set bits, this operation can fail if it reads an inconsistent combination of container values from different hierarchy levels. For example, consider the case where a container with exactly one set bit is chosen by the recursive call. However, before reaching Line~6, another thread clears that bit as part of a concurrent bitmap operation. In that case, \emph{try\_find\_set} fails even though there may be set bits in other containers.

\textsc{DynaSOAr} is affected by such bitmap inconsistencies when searching for active blocks (Algorithm~\ref{alg:alloc_algo}, Line~2) or free block positions (Line~4). While bitmap inconsistencies do not affect correctness, they can lead to higher fragmentation if \textsc{DynaSOAr} fails to find active blocks and instead initializes additional blocks, even though objects could have been accommodated in already existing blocks. We analyze the effect of bitmap inconsistencies in the benchmark section (Section~\ref{sec:dyna_detailed_analysis_wator}).

\paragraph{Enumerating Set Bit Indices}
\label{sec:enumerating_set_bit_indices}
Before launching a parallel do-all kernel, \soaalloc{} uses the \emph{indices} operation to generate a compact array of allocated block indices ($R$ in Figure~\ref{fig:scan_enumeration}). No GPU code is running at this time, so the bitmap is guaranteed to be in a consistent state. To ensure good scaling with increasing heap sizes, and thus increasing block bitmap sizes, \soaalloc{} utilizes the bitmap hierarchy to quickly skip containers without any set bits (Algorithm~\ref{alg:indices}).

First, an index array is generated for the higher-level bitmap (Line~2). This array is then processed in parallel; the \emph{for} loop in Line~7 is a GPU kernel and every thread processes one or multiple containers selected by the recursive call. If a container $C_i^l$ does not have any set bits, then its corresponding bit $b_i^{l+1}$ is in a cleared state in the higher-level bitmap and not included in \texttt{selected}. Every thread reserves space in the result array $R$ by increasing an atomic counter and fills its portion of the array with bit indices. This algorithm proved to be faster and requires less memory than a prefix sum algorithm, which needs multiple array copies/buffers per bitmap. Interestingly, related work has recently made the same observation with BFS algorithms on GPUs~\cite{Gaihre:2019:XER:3307681.3326606}. Note that, in contrast to the prefix sum-based implementation of Section~\ref{sec:par_do_all_sec3}, this algorithm does not necessarily retain the order of indices, i.e., the result array $R$ is not sorted.


\subsection{Reducing Thread Contention}
\label{sec:less_allocation_cont}
In Algorithms~\ref{alg:top_down_tr} and~\ref{alg:block_allocate}, threads are competing with each other for bits: Only one thread can reserve any given object slot and only a limited number of threads can succeed with allocations in a block. To guarantee correctness, our design is heavily based on atomic operations. These operations became considerably faster with recent GPU architectures~\cite{DeGonzalo:2019:AGW:3314872.3314884, warp_aggre}, but performance can still suffer when too many threads choose the same bit, because threads have to retry if allocation fails. \soaalloc{} employs two techniques to reduce such \emph{thread contention}. 

\paragraph{Allocation Request Coalescing}
Originally proposed by XMalloc~\cite{5577907}, \soaalloc{} combines memory allocation requests of the same type within a warp. One \emph{leader thread} reserves all object slots in a single block on behalf of all participating threads. If the selected active block does not have enough free object slots, \soaalloc{} reserves as many slots as possible and then chooses another active block for the remaining allocation requests. This reduces the number of atomic memory operations, because multiple bits in an object allocation bitmap are set in one operation. Furthermore, the constructor for newly allocated objects can run more efficiently, because field accesses are more likely coalesced.

Algorithm~\ref{alg:alloc_algo_with_coal} shows how allocation request coalescing is implemented in \textsc{DynaSOAr}. This algorithm is an improved version of Algorithm~\ref{alg:alloc_algo}. The following paragraph is intended for experienced CUDA programmers, as the algorithm takes advantage of CUDA warp-level primitives for intra-warp synchronization~\cite{cuda_website_using_warp_level}.

Out of all threads that are allocating an object of type $T$, the algorithm first select a leader thread. The leader attempts to reserve an object slot for every participating thread, similar to Algorithm~\ref{alg:alloc_algo}. The slot reservation algorithm was extended such that it can allocate multiple slots in one go: $n$=\emph{popc}(active) is the number of requested allocations and the algorithm attempts to reserve up to $n$ slots. The result is an \emph{allocation bitmap} with up to $n$ set bits, which were reserved in the object allocation bitmap of the block. Block ID and allocation bitmap are lane-shuffled to the other participating threads in the warp. Every thread locates the position of the $i$-th set bit in the allocation bitmap, where $i$ is the index of the thread in the warp among all participating threads (Line~20). If not enough bits were reserved for all threads (e.g., because the selected block was almost full), then these threads retry from the beginning, potentially with a new leader thread.

\SetKwComment{Comment}{$\triangleright$\ }{}
\SetAlgoVlined
\begin{algorithm}[t]
\small
 \Repeat(\Comment*[f]{\textsf{Infinite loop if OOM}}){false}{
  active $\gets$ \emph{\_\_activemask}(); \hfill \Comment{\textsf{Bitmap of active threads in warp}}
  leader $\gets$ \emph{ffs}(active); \hfill \Comment{\textsf{Leader = active thread with lowest ID}}
  rank $\gets$ \emph{\_\_lane\_id}(); \hfill \Comment{\textsf{Rank of this thread}}
  \If(\Comment*[f]{\textsf{This thread is the leader.}}){\emph{leader} = \emph{rank}}{
    bid $\gets$ active[T].\emph{try\_find\_set}()\;
    \If(\Comment*[f]{\textsf{Slow path}}){\emph{bid} = FAIL} {
      bid $\gets$ free.\emph{clear}()\;
      \emph{initialize\_block}<T>(bid)\;
      allocated[T].\emph{set}(bid)\;
      active[T].\emph{set}(bid)\;
    }
    alloc\_bitmap $\gets$ heap[bid].\emph{reserve\_multiple}(\,\emph{popc}(active)\,)\;
    \If{{popc}(\emph{alloc\_bitmap}) $>$ 0}{
      t $\gets$ heap[bid].type\;
      \lIf{\emph{alloc.state} = FULL}{
        active[t].\emph{clear}(bid)
      }
      \lIf{t $\not=$ T}{
        \emph{deallocate\_multiple}<t>(bid, alloc\_bitmap) 
      }
    }
  }
  alloc\_bitmap $\gets$ \emph{\_\_shfl\_sync}(active, alloc\_bitmap, leader)\;
  bid $\gets$ \emph{\_\_shfl\_sync}(active, bid, leader)\;
  id\_in\_active $\gets$ \emph{popc}(\emph{\_\_lanemask\_lt}() \& active)\;
  slot $\gets$ \emph{nth\_set\_bit}(alloc\_bitmap, id\_in\_active)\;
  \If(\Comment*[f]{\textsf{else: Not enough slots reserved}}){\emph{slot} $\not=$ FAIL} {
    \Return \emph{make\_pointer}(bid, slot)\;
  }
 }
 \caption[\textsc{DynaSOAr}: DAllocatorHandle::allocate<T>, with req. coalescing]{DAllocatorHandle::allocate<T>() : T* \hfill \fbox{GPU}}
 \label{alg:alloc_algo_with_coal}
\end{algorithm}

\paragraph{Bitmap Rotation}
Instead of a plain \emph{find first set} (ffs) in Algorithms~\ref{alg:top_down_tr} and~\ref{alg:block_allocate}, bitmaps are first rotating-shifted by a value depending on the warp ID and a seed that is changed with every retry. This increases the probability of threads choosing different active blocks for allocation and reduces the probability of threads trying to reserve the same object slots in a block. This is a key optimization technique that improved performance by an order of magnitude. 

\paragraph{Lookup Retry}
While bitmap traversals are relatively cheap, block initializations are expensive because in addition to initializing object bitmaps, bits in three different bitmaps (plus hierarachy) must be changed (slow path of Algorithm~\ref{alg:alloc_algo}). To avoid unnecessary block initializations, it proved beneficial to retry the search for active blocks (Line~2) a constant number of times before entering the slow path. This optimization resulted in lower fragmentation and improved performance.


\subsection{Efficient Bit Operations}
\SetKwComment{Comment}{$\triangleright$\ }{}
\begin{algorithm}[t]
\small
  \Repeat{\emph{success} $\vee$ \emph{block\_full}} {
    pos $\gets$ \emph{ffs}($\sim$bitmap)\;
    \lIf{\emph{pos} = NONE}{
      \Return \emph{FAIL}
    }
    mask $\gets$ 1 {<}{<} pos\;
    before $\gets$ \emph{atomicOr}(\&bitmap, mask)\;
    success $\gets$ (before \& mask)) = 0\;
    block\_full $\gets$ before = 0xFF...F\;
  }
  \If{{success}}{
    \eIf{popc(\emph{before}) \emph{= 63}}{
      \Return (pos, \emph{FULL})
    }{
      \Return (pos, \emph{REGULAR})
    }
  }
\Return \emph{FAIL}\;
 \caption[\textsc{DynaSOAr}: Block::reserve]{Block::reserve() : (int, state) \hfill $\triangleright$\ \textsf{Assuming block size 64.} \,\,\,\fbox{GPU}}
 \label{alg:block_allocate}
\end{algorithm}

\SetKwComment{Comment}{$\triangleright$\ }{}
\begin{algorithm}[t]
\small
    mask $\gets$ 1 {<}{<} pos\;
    before $\gets$ \emph{atomicAnd}(\&bitmap, $\sim$mask)\;
    \textcolor{gray}{success $\gets$ (before \& mask)) $\not=$ 0\;}
    \textcolor{gray}{\emph{assert}(success);} \hfill \Comment{\textcolor{gray}{\textsf{Precondition.\,\,\,\,\,\,\,\,}}}

    \uIf{popc(\emph{before}) \emph{= 1}}{
      \Return \emph{EMPTY}\;
    }
    \uElseIf{popc(\emph{before}) \emph{= 64}}{
      \Return \emph{FIRST}\;
    }
    \Else{
      \Return \emph{REGULAR}\;
    }
 \caption[\textsc{DynaSOAr}: Block::deallocate]{Block::deallocate(pos) : state \hfill $\triangleright$\ \textsf{Assuming block size 64.} \,\,\,\fbox{GPU}}
 \label{alg:block_deleteslot}
\end{algorithm}

\soaalloc{} is taking advantage of efficient bitwise operations such as \emph{ffs} (``find first set'') and \emph{popc} (``population count''). Modern CPU and GPU architectures have dedicated instructions for such operations. As an example, Algorithm~\ref{alg:block_allocate} shows how a single object slot is reserved\footnote{This algorithm was extended for allocation request coalescing (not shown here).}. Instead of checking all bits in a loop, \emph{ffs} in Line~2 is used to find a free slot (index of a cleared bit) in the object allocation bitmap and \emph{popc} in Line~10 counts the number of previously allocated slots (number of set bits) to decide if this request filled up the block. Similarly, to decide if a deallocation freed the first or last slot, \emph{popc} counts the number of bits before modifying the object allocation bitmap (Algorithm~\ref{alg:block_deleteslot}).

As another example, due to allocation request coalescing, every thread must now extract its reserved object slot from a set of allocations performed by a leader thread on behalf of the entire warp. This boils down to finding the $i$-th set bit in a 64-bit bitmap $b$ of newly reserved object slots, where $i$ is the rank of a thread among all allocating threads in the warp (Algorithm~\ref{alg:alloc_algo_with_coal}, Line~20). Instead of checking every bit in $b$ one-by-one (loop with 64 iterations in the worst case) and keeping track of the number of set bits seen so far, we clear the first $i-1$ bits with bitwise AND and then calculate \emph{ffs}(\emph{b}) (Algorithm~\ref{alg:nth_set_bit_al}).

\SetKwComment{Comment}{$\triangleright$\ }{}
\begin{algorithm}[t]
\small
    \For(\Comment*[f]{\textsf{Clear first $n-1$ bits.}}){$i\gets0$ \KwTo $n - 1$}{
      bitmap $\gets$ bitmap \& (bitmap - 1)\;
    }
    \Return \emph{ffs}(bitmap)\;
 \caption[Bit Operations: nth\_set\_bit]{nth\_set\_bit(bitmap, n) : int\hfill \fbox{GPU}}
 \label{alg:nth_set_bit_al}
\end{algorithm}

\SetKwComment{Comment}{$\triangleright$\ }{}
\begin{algorithm}[t]
\small
  heap[bid].type $\gets$ T; \hfill \Comment{\textsf{Volatile write.}}
  \emph{\_\_threadfence}()\;
  heap[bid].bitmap $\gets$ 0; \hfill \Comment{\textsf{Volatile write, assuming block capacity 64.}}
 \caption[\textsc{DynaSOAr}: DAllocatorHandle::initialize\_block<T>]{DAllocatorHandle::initialize\_block<T>(int bid) : void \hfill \fbox{GPU}}
 \label{alg:block_init}
\end{algorithm}



\section{Concurrency and Correctness}
\label{sec:concurrency_dynasoar}

CUDA has a weak consistency model for global memory access (Section~\ref{sec:cuda_prog_model_backgr}). Writes to memory performed by one thread are \emph{not} guaranteed to become visible to other threads in the same order. However, atomic writes \emph{have} that property (\emph{sequential consistency}). Furthermore, \emph{thread fences} can be used between two memory writes to enforce sequential consistency, if necessary.

Moreover, global memory reads/writes may be buffered in registers/caches, without a global memory load/store. Thus, memory writes by one thread may not become visible to other threads until the next GPU kernel, unless reads/writes are \texttt{volatile} or performed with atomic operations.

All bitmap operations are sequentially consistent and do not suffer from load/ store buffering because they are based on atomic memory operations.

\subsection{Object Slot Reservation/Freeing}
\label{sec:details_alloc_dealloc}
Inside a block, object allocations are tracked with an object allocation bitmap. Every object allocation bitmap has 64 bits, regardless of the block capacity. If a block's capacity is smaller than 64, then the last $64-N$ bits are set to 1 during block initialization to prevent threads from reserving these slots during object allocation.

Object slots are reserved/freed with atomic operations. These bypass the incoherent L1 caches and are thread-safe: E.g., based on their return value, we know if the current thread reserved a selected slot or if another contending thread was faster (Algorithm~\ref{alg:block_allocate}, Line~5). Based on their return value, we also know if the current thread reserved the last slot (Line~10), in which case the block should be marked as inactive by the allocation algorithm.

\paragraph{Slot Reservation}
\texttt{Block::reserve()} (Algorithm~\ref{alg:block_allocate}) reserves a single object slot in the block. Our actual implementation (\texttt{Block::reserve\_multiple(n)}; not discussed here) may reserve multiple slots at once due to allocation request coalescing.

\begin{enumerate}
  \item \textbf{Preconditions:} Block was initialized at least once. (Calling this method on invalidated blocks or full blocks is OK. This function will simply return FAIL.)
  \item \textbf{Postconditions:} If the result is different from FAIL, then the resulting slot position is reserved for this thread (and no other thread).
  \item \textbf{Return Value:} Success indicator, atomically reserved slot position, block state.
  \item \textbf{Linearization Point:} Atomic OR operation (Line~5).
\end{enumerate}

\paragraph{Slot Freeing}
\texttt{Block::deallocate(pos)} (Algorithm~\ref{alg:block_deleteslot}) frees a single object slot in the block. To support allocation request coalescing, we have a modified version of this algorithm that can rollback multiple slots at once (not discussed here).

\begin{enumerate}
  \item \textbf{Preconditions:} Bit \texttt{pos} is set to 1 in the object allocation bitmap. (Deleting an object multiple times or trying to delete an arbitrary pointer is illegal.)
  \item \textbf{Postconditions:} Bit \texttt{pos} is set to 0 in the object allocation bitmap.
  \item \textbf{Return Value:} Block state.
  \item \textbf{Linearization Point:} Atomic AND operation (Line~2).
\end{enumerate}

\subsection{Safe Memory Reclamation with Block Invalidation}
\label{sec:safe_mem_reclam_invalidation}
Safe memory reclamation (SMR) in lock-free algorithms is notoriously difficult. An SMR problem arises in \textsc{DynaSOAr} when deleting blocks. A block should be deleted as soon as its last object has been deleted. This by itself is easy to detect with atomic operations (Algorithm~\ref{alg:block_deleteslot}, Line~6). However, a contending thread may already have selected the now empty block in the course of its own concurrent allocate operation, before the block is actually deleted. Now it is no longer safe to delete the block, but the deleting thread is not aware of that.

Elaborate techniques for SMR such as hazard pointers~\cite{1291819} and epoch-based reclamation~\cite{UCAM-CL-TR-579} have been proposed in previous work~\cite{Brown:2015:RML:2767386.2767436, Michael:2002:SMR:571825.571829}. \textsc{DynaSOAr} is able to exploit a key characteristic of its data structure to solve this SMR problem in a simple way: Since all blocks have the same size and structure, object allocation bitmaps are always located at the same position. Therefore, we can optimistically proceed with bitmap modifications and rollback changes if necessary.

Our solution to SMR is \emph{block invalidation}. Before deleting a block, a thread tries to \emph{invalidate} (atomically set to 1) all bits in the object allocation bitmap. Bits that were already 1 are not considered invalidated because those object slots are in use. After successful invaldation, bits remain invalidated until a new block is initialized at the same location. Other threads may still be able to find the block in the active block bitmap for a while, but slot reservations can no longer succeed. 

\SetKwComment{Comment}{$\triangleright$\ }{}
\begin{algorithm}[t]
\small
  bitmap\_ptr $\gets$ \&heap[bid].bitmap\;
  before $\gets$ \emph{atomicOr}(bitmap\_ptr, 0xFF...F); \hfill \Comment{\textsf{Invalidate (set) all obj. alloc. bitmap bits.}}
  \If(\Comment*[f]{\textsf{$\geq$ 1 bit was invalidated.}}){before $\not=$ 0xFF...F}{
    t $\gets$ heap[bid].type\;
    \uIf(\Comment*[f]{\textsf{All 64 bits invalidated by this \emph{atomicOr}.}}){\emph{before = 0}}{
      \Return \emph{true}\;
    }
    \Else(\Comment*[f]{\textsf{Not all bits invalidated. Rollback.}}){
      before\_rollback $\gets$ \emph{atomicAnd}(bitmap\_ptr, before)\;
      \If(\Comment*[f]{\textsf{Other thread cleared a bit.}}){\emph{before\_rollback $\not=$ 0xFF...F}}{
            active[t].\emph{clear}(bid); \hfill \Comment{\textsf{Other thread expects an inactive block.}}
          }
      \If(\Comment*[f]{\textsf{Empty again. Retry invalidation.}}){\emph{(before\_rollback \& before) = 0}}{
        \Return \emph{invalidate}(bid)\;
      }
    }
  }
\Return \emph{false}\;
 \caption[\textsc{DynaSOAr}: DAllocatorHandle::invalidate]{DAllocatorHandle::invalidate(int bid) : bool \hfill \fbox{GPU}}
 \label{alg:block_invalidate}
\end{algorithm}

Allocating threads can detect changes in the block type. Before a previously invalidated block becomes available for allocations again (by initializing its object allocation bitmap), we update the block type. We put a thread fence between both writes to ensure that threads see the new block type before they see free slots in the bitmap (Algorithm~\ref{alg:block_init}). Threads allocate objects optimistically and rollback changes should they detect a different block type (Algorithm~\ref{alg:alloc_algo}, Line~14; also see Section~\ref{sec:cor_obj_alloc}).

\paragraph{Details}
Block invalidation\footnote{For presentation reasons, we assume a block capacity of 64 in this and other algorithms.} (Algorithm~\ref{alg:block_invalidate}) fails if a thread is unable to invalidate at least one bit. In that case, if at least one bit was changed through invalidation, this change must be rolled back (Line~8): In \texttt{before} exactly those bits are zero that were invalidated by the thread.

While a thread is running an invalidation operation, other threads may continue to concurrently reserve/free object slots in the same block, unaware of the fact that a thread is trying to invalidate the block. Those threads will update block state bitmaps based on the object allocation bitmap state that they are seeing. Therefore, block invalidation must update block bitmaps, as every invalidated bit appears to be an allocated object slot to other threads.

Since block invalidation fills up a block, the block's \emph{active[t]} state should be removed after Line~7, because, if we enter this \emph{else} branch, the thread just \emph{filled up} the block by reserving the remaining object slots (however, not all 64~slots, otherwise, we would be in  the \emph{then} branch of Line~5). However, we defer this step, as an invalidation rollback would likely have to mark the same block as \emph{active[t]} again. Unless, another thread concurrently freed an object slot in-between invalidation and invalidation rollback. For such a thread it will seem as if its deallocation just freed the first slot, causing it to activate the block (Algorithm~\ref{algo:dealloc_a}, Line~5). However, since we deferred block deactivation, this \emph{set(bid)} operation will spin until we deactivate the block (Algorithm~\ref{alg:block_invalidate}, Line~10). If invalidation rollback empties the block again, we try to invalidate the block one more time\footnote{Our actual implementation is iterative instead of recursive.}.

Note that block invalidation is independent of the type of a block. After invalidating at least one bit, the block type is fixed until invalidation rollback or block initialization, since other threads do not change invalidated bits. As such, the block cannot be deleted or reinitialized to another type by another thread. Other threads can, however, delete and initialize a block with different type after invalidation rollback. It is, nevertheless, safe to assume a block type of \emph{t} in Line~10, since this is merely an execution of a deferred operation that should have happened earlier when the block type was known to be \emph{t}.



\subsection{Object Allocation}
\label{sec:cor_obj_alloc}
The critical parts during allocations (Algorithm~\ref{alg:alloc_algo}) are \emph{block selection} (Line~2) and \emph{object slot reservation} (Line~8). Both operations by themselves are atomic, but not together. Block selection returns the index of an active block of type $T$, so we expect that after Line~8, we reserved an object slot in a block of type $T$. However, due to concurrent operations of other threads, some of these assumptions may be violated.

\begin{description}
  \item[Block Full] An active block was selected by \texttt{try\_find\_set} but the block filled up before making an allocation (i.e., the block is no longer active). In this case, object slot reservation will fail. Whenever allocation fails, it will restart from the beginning.
  \item[Block Deallocated] A block was selected by \texttt{try\_find\_set} but deallocated before reserving a slot. In this case, slot reservation will fail because the block is now in an invalidated state.
  \item[Block Replaced (ABA)] A block was selected by \texttt{try\_find\_set} but deallocated and reinitialized to a block of same type $T$. This is harmless: We do not care about block identity.
  \item[Block Replaced (Different Type)] A block was selected by \texttt{try\_find\_set} but deallocated and reinitialized to another type\footnote{Block initialization (Algorithm~\ref{alg:block_init}) has a thread fence between setting the block type and resetting the object allocation bitmap, so threads are guaranteed to read the correct type $t$ after an allocation succeeded.} $t \not= T$. In this case, the allocation must be rolled back (Line~14). All blocks have the same basic structure, so no object data can be overwritten accidentally during bitmap updates. Note that the rollback may trigger additional block bitmap updates.
  \textcolor{gray}{\item[Active Block Not Selected] A block becomes active shortly after \texttt{try\_find\_set} fails. Or, due to bitmap hierarchy inconsistencies, \texttt{try\_find\_set} fails to find an active block even though active blocks exist. This is harmless: No assumption is violated. A new block will be initialized, which merely increases fragmentation.}
\end{description}

Note that a block cannot be deallocated after an object slot was already reserved, because block invalidation would fail. Thus, the type of a block can also no longer change.


\subsection{Object Deallocation}
\label{sec:dynasoar_obj_dealloc_discussion}
The critical part during deallocations (Algorithm~\ref{algo:dealloc_a}) is consistency between \emph{object slot deallocation} (Line~3) and \emph{block state updates}. If the current thread deallocated the first object (i.e., the block was full), then the block's bit in the active block bitmap must be set. If the current thread deallocated the last object (i.e., the block is empty), then the block must be deleted. The problem is that object slot deallocation and the corresponding block state update together are not atomic.

\begin{description}
  \item[Allocate After Delete-First] A thread $t_1$ deleted the first object of a block. However, before marking the block as active (Line~6), another thread $t_2$ allocated this slot again; the block should be inactive. In this case, $t_2$ reserved the last slot, so it will mark the block as inactive (Algorithm~\ref{alg:alloc_algo}, Line~12). This operation expects the bit to be in a set state and it will retry until $t_1$ sets the bit.
  \item[Block Deleted after Delete-First] A thread $t_1$ deleted the first object of a block. However, before marking the block as active, other threads deallocated all other objects and a thread $t_2$ deleted the block. This is not possible because $t_2$ expects the block to be active (Line~9), i.e., bit set to 1, and blocks until then.
  \item[Block Replaced after Delete-First] A thread $t_1$ deleted the first object of a block. However, before marking the block as active, the block was reinitialized to another type. This is not possible because only deleted blocks can be reinitialized (see previous point).
  \item[Allocate after Delete-Last] A thread $t_1$ deleted the last object of a block. However, before deleting the block, another thread $t_2$ allocated an object again, so it is unsafe to delete the block now. This case in handled by block invalidation.
  \item[Block Deleted after Delete-Last] A thread $t_1$ deleted the last object of a block. However, before deleting the block, another thread $t_2$ allocated an object and yet another thread $t_3$ deleted that object, rendering the block empty again and deleting it. Now the block is already deleted when $t_1$ is trying to delete the block. In this case, block invalidation of $t_1$ will fail because the block is still in an invalidated state and $t_1$ fails to invalidate all object slot bits.
  \item[Block Replaced after Delete-Last] Same as before, but yet another thread $t_4$ reininitializes the block to a different type. Now $t_1$ will invalidate and delete a new block whose type is different. This is OK. Block invalidation will succeed only if the block is empty. Both block invalidation and block deletion are independent of and do not assume a certain block type.
\end{description}

\paragraph{Divergent Branch Scheduling}
The ordering of \emph{if} statement branches in Algorithm~\ref{algo:dealloc_a} is crucial. Let us assume that a thread $t_1$ deletes the first object of a full block (Line~4) and another thread $t_2$ of the \emph{same warp} deletes the last object of the same block and succeeds with block invalidation (Line~7). The \emph{clear} operation in Line~9 will block until the respective bit was actually changed to 0. Therefore, to avoid a deadlock, it is crucial that the \emph{set} operation of Line~5 is scheduled to execute before the \emph{clear} operation. Since both $t_1$ and $t_2$ belong to the same warp, both \emph{if} branches are executed sequentially (\emph{thread divergence}). Only if the CUDA compiler schedules the branch of Line~4 before the branch of Line~6 in the compiled PTX, is this implementation free of deadlocks. While CUDA leaves the order of execution of divergent branches undefined (Section~\ref{sec:cuda_prog_model_backgr}), to the best of our knowledge, all recent CUDA versions schedule branches in the order in which they appear in the source code.

\SetKwComment{Comment}{$\triangleright$\ }{}
\SetKwRepeat{Do}{do}{while}
\begin{algorithm}[t]
\small
  bid $\gets$ \emph{extract\_block}(ptr)\;
  slot $\gets$ \emph{extract\_slot}(ptr)\;
  state $\gets$ heap[bid].\emph{deallocate}(slot)\;
  active\_op $\gets$ allocated\_op $\gets$ free\_op $\gets$ 0\;
  \uIf{\emph{state} = FIRST}{
    t $\gets$ T\;
    active\_op $\gets$ 1\;
  }
  \ElseIf{\emph{state} = EMPTY}{
    \If(\Comment*[f]{\textsf{invalidate() must be inlined and changed similarly.}}){invalidate(bid)}{
      t $\gets$ heap[bid].type\;
      active\_op $\gets$ -1\;
      allocated\_op $\gets$ -1\;
      free\_op $\gets$ 1\;
    }
  }
  \While{active\_op + allocated\_op + free\_op $\not=$ 0}{
    \If{active\_op = 1}{
      \lIf{\emph{active[t].}try\_set\emph{(bid)}} {
        active\_op $\gets$ 0
      }
    }
    \If{active\_op = -1}{
      \lIf{\emph{active[t].}try\_clear\emph{(bid)}} {
        active\_op $\gets$ 0
      }
    }

    \If{allocated\_op = 1}{
      \lIf{\emph{allocated[t].}try\_set\emph{(bid)}} {
        allocated\_op $\gets$ 0
      }
    }
    \If{allocated\_op = -1}{
      \lIf{\emph{allocated[t].}try\_clear\emph{(bid)}} {
        allocated\_op $\gets$ 0
      }
    }

    \If{free\_op = 1}{
      \lIf{\emph{free[t].}try\_set\emph{(bid)}} {
        free\_op $\gets$ 0
      }
    }
    \If{free\_op = -1}{
      \lIf{\emph{free[t].}try\_clear\emph{(bid)}} {
        free\_op $\gets$ 0
      }
    }
  }
 \caption[\textsc{DynaSOAr}: DAllocatorHandle::deallocate<T>, w/o deadlock]{DAllocatorHandle::deallocate<T>(T* ptr) : void \hfill \fbox{GPU}}
 \label{algo:dealloc_a_improved}
\end{algorithm}

Algorithm~\ref{algo:dealloc_a_improved} is an alternative implementation of Algorithm~\ref{algo:dealloc_a} that does not depend on the ordering of \emph{if} branches. Block bitmaps are not changed immediately. Instead, we maintain helper variables that indicate if and how a bit should be changed: -1 indicates \emph{clear}, 0 indicates \emph{no change}, 1 indicates \emph{set}. At first, all participating threads in a warp determine the necessary changes, potentially causing thread divergence between Line~5 and Line~13. However, after Line~13, all threads converge again and modify their respective bits with \emph{try\_} bitmap operations. These operations return immediately, regardless of their success. Only if a thread performed all necessary bitmap changes can a thread exit the loop.

Note that block invalidation (Line~9) also modifies block bitmaps. Therefore, block invalidation must be incorporated into Algorithm~\ref{algo:dealloc_a_improved} and changed in a similar way. Only for presentation reasons, we keep block invalidation in a separate algorithm.

\subsection{Correctness of Hierarchical Bitmap Operations}
A container $C_i^l$ consists of bits $b_{64 \cdot i}^l$, ...,  $b_{64 \cdot i + 63}^l$ and is represented by one bit $b_{i}^{l+1}$ in the nested (higher-level) bitmap. That bit is set if and only if at least one bit is set in the container.

\begin{definition}[Consistency]
\label{def:consistency_crit}
A bit in level $b_{i}^{l+1}$ is \textbf{consistent} with its corresponding container $C_i^l$ in the lower-level bitmap if and only if:

\begin{align*}
b_i^{l+1} = \bigvee_{k=0}^{63} b_{64\cdot i + k}^{l}\,\, \textcolor{gray}{= \mathds{1}\left(\sum C_{\lfloor i/64 \rfloor}^l > 0\right)}
\end{align*}
\end{definition}

In Definition~\ref{def:consistency_crit}, $\mathds{1}$ is the indicator function and $\sum C_{i}^l$ is the number of set bits in $C_i^l$, as computed by \texttt{popc} (population count) in the algorithms.

We say that the $L_{l+1}$ bitmap is in a consistent state with the $L_l$ bitmap if all bits $b_i^{l+1}$ in the $L_{l+1}$ bitmap satisfy the consistency criterion. The bitmap data structure as a whole is in a consistent state if all bitmap levels $L_i$ satisfy the consistency criterion.

We do not make any consistency guarantees during the execution of a GPU program. Instead, we guarantee eventual consistency: If the bitmap data structure is in a consistent state before a GPU kernel launch, it is guaranteed to be in a consistent state after the GPU kernel finished running.

\begin{definition}[Semantics of Bitmap Operations]
\label{def:semant_bitmap_ops}
Every bitmap $L_l$ provides operations for setting and clearing bits (Section~\ref{sec:hierarchical_bit}). These operations may update bits in the higher-level bitmap $L_{l+1}$ if they \textbf{s}et the \textbf{f}irst bit ($\mathit{SF}_{\lfloor i/64 \rfloor}^l$) or \textbf{c}lear the \textbf{l}ast bit ($\mathit{CL}_{\lfloor i/64 \rfloor}^l$) of a container $C^l_i$, respectively:

\begin{align*}
\underbrace{\textsf{set}(b^l_i) \text{ and } \mathds{1}\left(\sum C_{\lfloor i/64 \rfloor}^l = 0\right)}_{\text{set-first: } \mathit{SF}_{\lfloor i/64 \rfloor}^l} & \text{ then } \textsf{set}(b_{\lfloor i/64 \rfloor}^{l+1}) & \,\,\,\,\,\, \forall i \in [0; \mathit{num\_bits}) \\
\underbrace{\textsf{clear}(b^l_i) \text{ and } \mathds{1}\left(\sum C_{\lfloor i/64 \rfloor}^l = 1\right)}_{\text{clear-last: } \mathit{CL}_{\lfloor i/64 \rfloor}^l} & \text{ then } \textsf{clear}(b_{\lfloor i/64 \rfloor}^{l+1}) & \,\,\,\,\,\, \forall i \in [0; \mathit{num\_bits})
\end{align*}
\end{definition}

We would like to show that, assuming that a bitmap data structure is initially in a consistent state and given a multiset of bitmap operations $O_0$ on the $L_0$ bitmap, the entire bitmap data structure is in a consistent state after executing all operations. In contrast to other concurrent data structures~\cite{Ellen:2007:SSN:1281100.1281106}, we do not guarantee consistency during runtime.

\begin{definition}[Legal Bitmap Operations]
\label{def:correctness_multiset_op}
Let $\#\textsf{set}(b^l_i)$ and $\#\textsf{clear}(b^l_i)$ be the number of set and clear operations of $b^l_i$ in a multiset of bitmap operations $O_l$. We call $\mathds{S}(b^l_i) = \#\textsf{set}(b^l_i) - \#\textsf{clear}(b^l_i)$ the \textbf{set-surplus} of $b^l_i$. $O_l$ is \textbf{legal} if it satifies the following conditions for every bit $b^l_i$ in $L_l$.

\begin{enumerate}
  \item Overall bit operation is clear, remain or set: $\mathds{S}(b^l_i) \in \{-1, 0, 1\}$.
  \item Bit is in a cleared or set state afterwards: $b^l_i + \mathds{S}(b^l_i) \in \{0, 1\}$ (denoted by $b'^l_i$).
\end{enumerate}
\end{definition}

Intuitively, an already cleared bit cannot be cleared again and an already set bit cannot be set again\footnote{\textsf{try\_clear} and \textsf{try\_set} allow clearing/setting already cleared/set bits.}. For example:

\begin{itemize}
  \item $C^l_0 = \texttt{01001...}$, $O_l = \{\textsf{set}(b^l_2), \textsf{set}(b^l_2), \textsf{clear}(b^l_2)\}$. $\mathds{S}(b^l_2) = 2 - 1 = 1$. $b'^l_i = 0 + 1 = 1$. Therefore, $O_l$ is legal.
  \item $C^l_0 = \texttt{01001...}$, $O_l = \{\textsf{set}(b^l_2), \textsf{set}(b^l_2), \textsf{clear}(b^l_2), \textsf{set}(b^l_2)\}$. $O_l$ is illegal because $\mathds{S}(b^l_2) = 3 - 1 = 2$.
\end{itemize}

Note that illegal bitmap operations deadlock in our implementation because \textsf{set} and \textsf{clear} spin-block and retry until they acutally changed the bit. If a legal bitmap operations multiset is executed fully concurrent (i.e., one thread per operation), then there is always a thread/operation that can make progress.

\begin{hypothesis}
Let us assume that a multiset of bitmap operations $O_l$ on the $L_l$ bitmap is legal according to Definition~\ref{def:correctness_multiset_op} for an arbitrary $l$ and that $L_l$ is initially consistent with $L_{l+1}$.
\end{hypothesis}

\begin{lemma}
Under the induction hypothesis, the bitmap operations multiset $O_{l+1}$ that is generated by the operations in $O_l$ according to Definition~\ref{def:semant_bitmap_ops} is also legal. Furthermore, after executing $O_l$, $L_l$ is still consistent with $L_{l+1}$.
\end{lemma}

\begin{proof}
Let us first consider the bitmap operations of a single container $C^l_i$. Let $\#\mathit{SF}_i^l$ be the number of times a first bit is set in the container and $\#\mathit{CL}_i^l$ be the number of times a last bit is cleared in the container. Then, according to Definition~\ref{def:semant_bitmap_ops}, $b_{\lfloor i/64 \rfloor}^{l+1}$ is set $\#\mathit{SF}_i^l$ times and cleared $\#\mathit{CL}_i^l$ times. We have to prove that the set-surplus $\mathds{S}(b^{l+1}_{\lfloor i/64 \rfloor}) = \#\mathit{SF}_i^l - \#\mathit{CL}_i^l$ satisfies the legality criteria of Definition~\ref{def:correctness_multiset_op}.

Without loss of generality, let us assume that all set-first and clear-last operate on the same bit $b^l_k$. Then, $\mathds{S}(b^{l+1}_{\lfloor i/64 \rfloor}) = \mathds{S}(b^l_k) \in \{-1, 0, 1\}$. Hence, the generated bitmap operations $O_{l+1}$ for any bit on the $L_{l+1}$ bitmap satisfy the first legality condition of Definition~\ref{def:correctness_multiset_op}.

Now we have to show that also the second legality condition holds and that $b^{l+1}_{\lfloor i/64 \rfloor}$ is consistent with $C^l_i$ after executing $O_l$. We consider two cases.

\begin{enumerate}
  \item The higher-level bit is initially in a cleared state: $b^{l+1}_{\lfloor i/64 \rfloor} = 0$. Therefore, due to initial consistency, $\sum C_{\lfloor i/64 \rfloor}^l = 0$. Therefore, $\#\mathit{SF}_i^l - \#\mathit{CL}_i^l \in \{0, 1\}$, otherwise, $O_l$ would not be legal. Therefore, $b^{l+1}_{\lfloor i/64 \rfloor} + \mathds{S}(b^{l+1}_{\lfloor i/64 \rfloor}) \in \{0, 1\}$.
  \begin{enumerate}
    \item If $\#\mathit{SF}_i^l - \#\mathit{CL}_i^l = 0$, then $\vee_{k=0}^{63} b_{64\cdot i + k}^{l} = 0$ after $O_l$. At the same time, $\mathds{S}(b^{l+1}_{\lfloor i/64 \rfloor}) = 0$, so $b^{l+1}_{\lfloor i/64 \rfloor} = 0$ after $O_l$, which is consistent with the state of $C^l_i$ after $O_l$.
    \item If $\#\mathit{SF}_i^l - \#\mathit{CL}_i^l = 1$, then $\vee_{k=0}^{63} b_{64\cdot i + k}^{l} = 1$ after $O_l$. At the same time, $\mathds{S}(b^{l+1}_{\lfloor i/64 \rfloor}) = 1$, so $b^{l+1}_{\lfloor i/64 \rfloor} = 1$ after $O_l$, which is consistent with the state of $C^l_i$ after $O_l$.
  \end{enumerate}
  \item The higher-level bit is initially in a set state: $b^{l+1}_{\lfloor i/64 \rfloor} = 1$. Therefore, due to initial consistency, $\sum C_{\lfloor i/64 \rfloor}^l > 0$. Therefore, $\#\mathit{SF}_i^l - \#\mathit{CL}_i^l \in \{-1, 0\}$, otherwise, $O_l$ would not be legal. Therefore, $b^{l+1}_{\lfloor i/64 \rfloor} + \mathds{S}(b^{l+1}_{\lfloor i/64 \rfloor}) \in \{0, 1\}$.
  \begin{enumerate}
    \item If $\#\mathit{SF}_i^l - \#\mathit{CL}_i^l = -1$, then $\vee_{k=0}^{63} b_{64\cdot i + k}^{l} = 0$ after $O_l$. At the same time, $\mathds{S}(b^{l+1}_{\lfloor i/64 \rfloor}) = -1$, so $b^{l+1}_{\lfloor i/64 \rfloor} = 0$ after $O_l$, which is consistent with the state of $C^l_i$ after $O_l$.
    \item If $\#\mathit{SF}_i^l - \#\mathit{CL}_i^l = 0$, then $\vee_{k=0}^{63} b_{64\cdot i + k}^{l} = 1$ after $O_l$. At the same time, $\mathds{S}(b^{l+1}_{\lfloor i/64 \rfloor}) = 0$, so $b^{l+1}_{\lfloor i/64 \rfloor} = 1$ after $O_l$, which is consistent with the state of $C^l_i$ after $O_l$.
  \end{enumerate}
\end{enumerate}

If all containers in $L_l$ are consistent with their respective bits in $L_{l+1}$, then the entire $L_l$ bitmap is consistent with the $L_{l+1}$ bitmap. Futhermore, all generated bitmap operations $O_{l+1}$ are legal because they satisfy both legality criteria.
\end{proof}

\begin{basecase}
The bitmap data structure is initially in a consistent state. Furthermore, $O_0$ is legal. Otherwise, programmers use the bitmap data structure incorrectly.
\end{basecase}

\section{Related Work}
\label{sec:related_work}
\begin{table}[b]
\caption[Comparison of dynamic memory allocators]{Comparison of allocators. \emph{Coal.} means \emph{allocation request coalescing}.}
\label{fig:baselines}
\begin{tabularx}{\columnwidth}{Xcccc}
\hline \hline
 \footnotesize \narrowstyle \textbf{Allocator} & \footnotesize \narrowstyle \textbf{Coal.} & \footnotesize \textbf{SOA} & \footnotesize \textbf{Container} & \footnotesize \textbf{Finding Free Memory} \\
\Xhline{2\arrayrulewidth}
 \footnotesize \narrowstyle \soaalloc{} & \footnotesize\ding{51} & \footnotesize\ding{51} & \footnotesize Block & \footnotesize Hierarchical Bitmap \\ \hline
 \footnotesize\narrowstyle \soaalloc{}-NoCoal & \footnotesize \ding{55} & \footnotesize \ding{51} & \footnotesize Block & \footnotesize \footnotesize Hierarchical Bitmap \\ \hline
 \footnotesize \narrowstyle BitmapAlloc & \footnotesize\ding{55} & \footnotesize\ding{55} & \footnotesize\ding{55} & \footnotesize Hierarchical Bitmap \\ \hline
 \textcolor{gray}{\footnotesize\narrowstyle CircularMalloc} & \textcolor{gray}{\footnotesize \ding{55}} & \textcolor{gray}{\footnotesize \ding{55}} & \textcolor{gray}{\footnotesize \ding{55}} & \textcolor{gray}{\footnotesize Linked List, Ring Buffer} \\ \hline
 \footnotesize\narrowstyle Default CUDA Allocator & \footnotesize \ding{55} & \footnotesize \ding{55} & \footnotesize (Unknown) & \footnotesize (Unknown) \\ \hline
 \textcolor{gray}{\footnotesize \narrowstyle FDGMalloc} & \textcolor{gray}{\footnotesize \ding{51}} & \textcolor{gray}{\footnotesize \ding{55}} &  \textcolor{gray}{\footnotesize Priv. Heap, Superblock} &  \textcolor{gray}{\footnotesize Linked List}\\ \hline
 \footnotesize \narrowstyle Halloc & \footnotesize \ding{55} & \footnotesize \ding{55} & \footnotesize Slab & \footnotesize Bitmap, Hashing \\ \hline
 \footnotesize \narrowstyle mallocMC (ScatterAlloc) & \footnotesize \ding{51} & \footnotesize \ding{55} & \footnotesize Superblock, Region, Page & \footnotesize Hashing \\ \hline
 \textcolor{gray}{\footnotesize \narrowstyle XMalloc} & \textcolor{gray}{\footnotesize \ding{51}} & \textcolor{gray}{\footnotesize \ding{55}} & \textcolor{gray}{\footnotesize (4 block hierarchies)} & \textcolor{gray}{\footnotesize Lock-free Free Lists}  \\
 \hline
\hline
\end{tabularx}
\normalstyle
\end{table}

CUDA provides an on-device dynamic memory allocator, but it is unoptimized and slow. To solve this issue, multiple custom allocators have been developed in the last years (Table~\ref{fig:baselines}). These allocators achieve good performance by exploiting an allocation pattern that many applications on massively parallel SIMD architectures exhibit: Most allocations are small in size and due to mostly regular control flow, many allocations have the same byte size.

Halloc~\cite{hallocweb} is one of these allocators. It is a slab allocator and can allocate only a few dozen predetermined byte sizes between 16~bytes and 3~KB. This is fast but can lead to internal fragmentation. \textsc{DynaSOAr} can avoid such internal fragmentation because allocation sizes are determined from compile-time type information of the application. A slab in Halloc contains same-size allocations and tracks allocations with a bitmap. To avoid scanning large bitmaps, a hash function determines which bits to check during allocations. Only one slab can be active per allocation size and if the active slab becomes too full, it is replaced with a new one. In contrast, more than one block per type can be active in \textsc{DynaSOAr} and blocks are filled up entirely.

XMalloc~\cite{5577907} was the first allocator with allocation request coalescing, which was adopted by many other allocators, including \textsc{DynaSOAr}. Coalesced requests are served from \emph{basicblocks}, which are organized in one of multiple lock-free free lists depending on their size.

FDGMalloc maintains a private heap for every warp~\cite{Widmer:2013:FDM:2458523.2458535}, similar to Hoard~\cite{Berger:2000:HSM:378993.379232}. It does not have a general \emph{free} operation and can only deallocate entire heaps, so it is not suitable for SMMO applications.

CircularMalloc (CMalloc)~\cite{Vinkler:2015:RED:3071494.3071506} allocates memory in a ring buffer. Every allocation has a pointer to the next allocation or free chunk, wrapping around at the end of the buffer. CMalloc traverses the linked list for free chunks during allocations. To reduce allocation contention, every multiprocessor starts its traversal at a different location. This is similar to \textsc{DynaSOAr}'s \emph{bitmap rotation} technique for reducing thread contention.

ScatterAlloc~\cite{6339604} hashes allocation requests to memory \emph{pages} depending on their allocation size and the multiprocessor ID. Pages hold allocations of the same size, but slightly smaller requests can be accommodated, leading to internal fragmentation. While \soaalloc{} uses hierarchical bitmaps, ScatterAlloc uses hashing with linear probing for finding pages during allocations. For benchmarks, we use mallocMC~\cite{eckert_carlchristian_helmut_johannes_2014_34461}, a reimplementation of ScatterAlloc that is still maintained.


Both Halloc and ScatterAlloc maintain fill levels to quickly skip congested memory areas that are above a certain threshold, because the performance of any hashing technique degrades with an increasing number of collisions. In \soaalloc{}, temporary inconsistencies in bitmap hierarchies increase with the number of concurrent allocations, but \soaalloc{} can dynamically adapt to such cases by initializing additional blocks.

\section{Benchmarks}
\label{sec:benchmark}
We evaluated \soaalloc{} with multiple real-world SMMO applications that exhibit different memory allocation patterns (Table~\ref{tab:smmo_apps}). We describe the SMMO structure of these applications in detail in Chapter~\ref{chap:smmo_examples}. All benchmarks were run on a computer with an Intel Core i7-5960X CPU, 32~GB main memory and an NVIDIA TITAN Xp GPU (12~GB device memory), and compiled with nvcc (-O3) from the CUDA Toolkit~9.1 on Ubuntu~16.04.4.

\begin{table}[!htp]
\caption{Description of SMMO benchmark applications}
\label{tab:smmo_apps}
\begin{tabularx}{\textwidth}{cXllllll}
\hline\hline
 \narrowstyle & \footnotesize \textbf{Benchmark Description} & \footnotesize \rotatebox[origin=l]{90}{\narrowstyle\textbf{\#par. do-all}} & \footnotesize \rotatebox[origin=l]{90}{\textbf{\narrowstyle\#classes}\,\,\,\,} & \footnotesize \rotatebox[origin=l]{90}{\narrowstyle\textbf{alloc./dealloc.}} & \footnotesize \rotatebox[origin=l]{90}{\narrowstyle\textbf{smallest class\,\,\,\,}} & \footnotesize \rotatebox[origin=l]{90}{\narrowstyle\textbf{largest class}} \\
\Xhline{2\arrayrulewidth}
\narrowstyle \raisebox{-0.95\totalheight}{\includegraphics[width=0.165\textwidth]{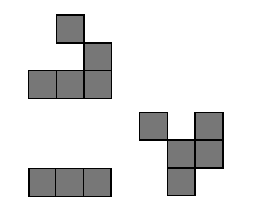}} &
\footnotesize \textbf{Game of Life:} A CA due to J. H. Conway. This version has a time complexity of O(\#alive cells) instead of the standard O(\#cells) algorithm. Cells can be dead, alive or alive-candidates. Alive-candidates are dead cells that may become alive in the next iteration. Only alive-candidates and alive cells are processed.
& \footnotesize\rotatebox[origin=r]{90}{4 / iteration} & \footnotesize\rotatebox[origin=r]{90}{4 (2 dyn.)} & \footnotesize\rotatebox[origin=r]{90}{\ding{51} / \ding{51}} & \footnotesize\rotatebox[origin=r]{90}{5B, 2 fields} & \footnotesize\rotatebox[origin=r]{90}{8B, 1 field} \\
\hline
\narrowstyle\raisebox{-0.9\totalheight}{\includegraphics[width=0.165\textwidth]{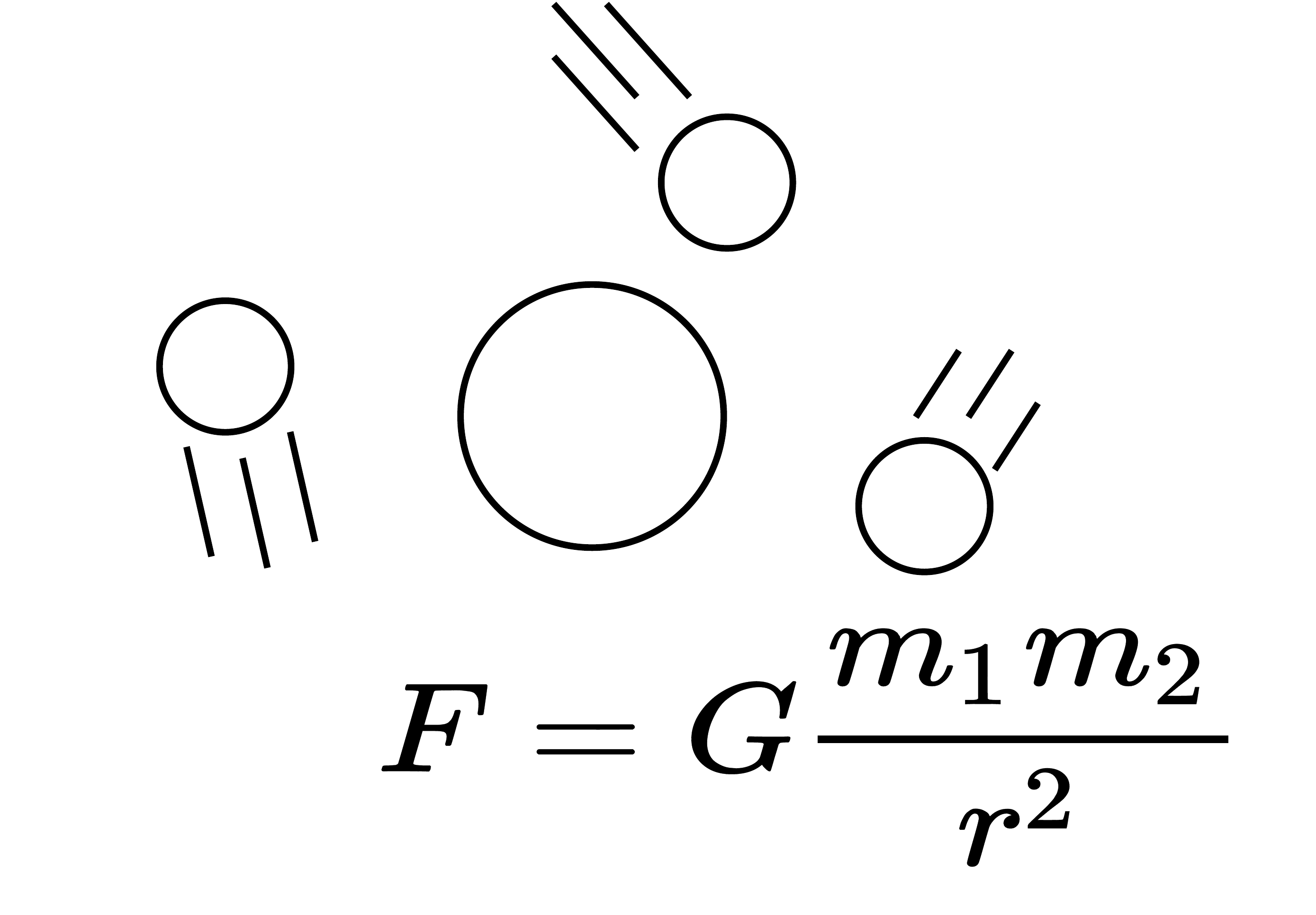}} & \footnotesize \textbf{N-Body:} Simulates the movement of particles according to gravitational forces. A \texttt{device\_do} operation is required to calculate (and then sum up) the gravitational force between every pair of particles. This benchmark has no dynamic object (de)allocation. & \footnotesize\rotatebox[origin=r]{90}{2 / iteration}  & \footnotesize \rotatebox[origin=r]{90}{1 (0 dyn.)} & \footnotesize\rotatebox[origin=r]{90}{\ding{55} / \ding{55}} & \footnotesize\rotatebox[origin=r]{90}{\,\,\,28B, 7 fields} & \footnotesize\rotatebox[origin=r]{90}{(same)} \\
\hline
\narrowstyle\raisebox{-1\totalheight}{\includegraphics[width=0.165\textwidth]{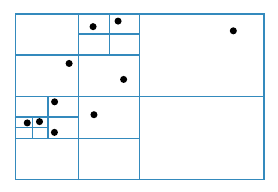}} & \footnotesize \textbf{Barnes-Hut:} An extension of N-Body in which bodies are stored in a quad tree~\cite{BURTSCHER201175}, to evaluate \textsc{DynaSOAr} with dynamic tree data structures. The running time is dominated by the construction and maintenance (i.e., frequent node inserts and removals) of the quad tree. Tree nodes are dynamically (de)allocated.
& \footnotesize\rotatebox[origin=r]{90}{10 / iteration} & \footnotesize\rotatebox[origin=r]{90}{3 (1 dyn.)} & \footnotesize\rotatebox[origin=r]{90}{\ding{51} / \ding{51}} & \footnotesize\rotatebox[origin=r]{90}{68B, 9 fields} & \footnotesize\rotatebox[origin=r]{90}{\,\,\,102B, 12 fields}\\
\hline
\raisebox{-0.925\totalheight}{\includegraphics[width=0.165\textwidth]{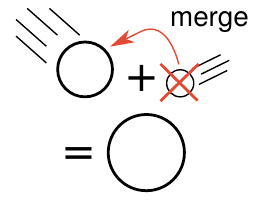}} & \footnotesize \textbf{Particle Collisions:} Similar to N-Body, but particles are merged according to perfectly inelastic collision when they are getting too close. The number of particles decreases gradually. This benchmark has dynamic object deallocation but no dynamic object allocation. & \footnotesize\rotatebox[origin=r]{90}{6 / iteration} & \footnotesize\rotatebox[origin=r]{90}{1 (1 dyn.)} & \footnotesize\rotatebox[origin=r]{90}{\ding{55} / \ding{51}} & \footnotesize\rotatebox[origin=r]{90}{\,\,38B, 10 fields} & \footnotesize\rotatebox[origin=r]{90}{(same)} \\
\hline
\raisebox{-1.05\totalheight}{\includegraphics[width=0.165\textwidth]{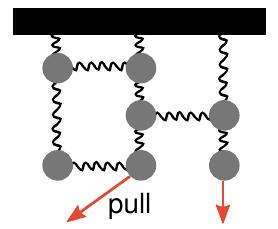}} &\footnotesize \textbf{Structure:} Simulates a fracture in a composite material, modeled as an FEM. Intuitively, the mesh is a graph and edges between nodes are springs. When pulling the mesh on one side, the material starts to break eventually. Isolated nodes are detected with a BFS~\cite{Harish:2007:ALG:1782174.1782200} and removed. Literature describes extensions that would benefit from dynamic allocation~\cite{LU2018240}. & \footnotesize\rotatebox[origin=r]{90}{3 / iteration} & \footnotesize\rotatebox[origin=r]{90}{5 (4 dyn.)} & \footnotesize\rotatebox[origin=r]{90}{\ding{55} / \ding{51}} & \footnotesize\rotatebox[origin=r]{90}{32B, 6 fields} & \footnotesize\rotatebox[origin=r]{90}{46B, 7 fields} \\
\hline
\raisebox{-0.975\totalheight}{\includegraphics[width=0.165\textwidth]{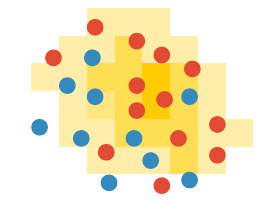}} & \footnotesize \textbf{Sugarscape:} An agent-based social simulation~\cite{RePEc:mtp:titles:0262550253}. Agents inhabit a 2D grid and can move to neighboring cells. Cells contain sugar which is consumed by agents. Sugarscape can simulate a variety social dynamics (e.g., trade, war, environmental pollution). We simulate resource consumption, ageing and mating.  & \footnotesize\rotatebox[origin=r]{90}{12 / iteration} & \footnotesize\rotatebox[origin=r]{90}{4 (2 dyn.)} &  \footnotesize\rotatebox[origin=r]{90}{\ding{51} / \ding{51}} & \footnotesize\rotatebox[origin=r]{90}{52B, 7 fields} & \footnotesize\rotatebox[origin=r]{90}{74B, 11 fields} \\
\hline
\raisebox{-0.95\totalheight}{\includegraphics[width=0.165\textwidth]{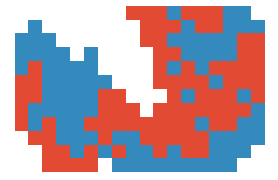}} & \footnotesize \textbf{Wa-Tor:} An agent-based predator-prey simulation~\cite{10.2307/24969495}. Fish/sharks occupy a 2D grid of cells and can move to neighboring cells. Fish and sharks reproduce after some iterations. Fish die when they are eaten and sharks die when they run out of food. & \footnotesize\rotatebox[origin=r]{90}{8 / iteration} & \footnotesize\rotatebox[origin=r]{90}{4 (2 dyn.)} & \footnotesize\rotatebox[origin=r]{90}{\ding{51} / \ding{51}} & \footnotesize\rotatebox[origin=r]{90}{60B, 4 fields} & \footnotesize\rotatebox[origin=r]{90}{\,\,64B, 5 fields} \\
\hline
\raisebox{-0.93\totalheight}{\includegraphics[width=0.165\textwidth]{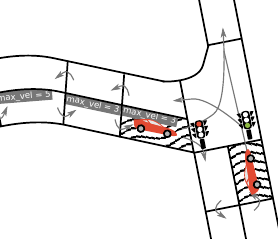}} & \footnotesize \textbf{Nagel-Schreckenberg:} A traffic flow simulation on a street network~\cite{nagel_schr} that can reproduce traffic jams and other phenomena. Streets are modeled as a network of cells, with at most one vehicle per cell. New vehicles are continuously added to the simulation and existing vehicles are removed at their final destination. & \footnotesize\rotatebox[origin=r]{90}{3 / iteration}  & \footnotesize\rotatebox[origin=r]{90}{4 (1 dyn.)} & \footnotesize\rotatebox[origin=r]{90}{\ding{51} / \ding{51}} & \footnotesize\rotatebox[origin=r]{90}{97B, 10 fields} & \footnotesize\rotatebox[origin=r]{90}{124B, 6 fields}\\
\hline
& \footnotesize \textbf{Linux Scalability:} Not an SMMO application. This microbenchmark allocates, then deallocates a fixed number of same-size objects in each thread, without accessing the memory~\cite{Lever:2000:MPM:1267724.1267780}. & \footnotesize \rotatebox[origin=r]{90}{n/a} & \footnotesize\rotatebox[origin=r]{90}{1 (1 dyn.)} & \footnotesize\rotatebox[origin=r]{90}{\ding{51} / \ding{51}} & \footnotesize\rotatebox[origin=r]{90}{\,\,\,4B, 1 field} & \footnotesize\rotatebox[origin=r]{90}{(same)} \\
\hline\hline
\end{tabularx}
\normalstyle \normalsize
\end{table}

We compare the running time with different allocators. If possible, we also measured the running time of \emph{baseline} implementations that do not use any dynamic memory management.

\paragraph*{Benchmark Applications}
Our benchmarks are from different domains and fall into four categories.

\begin{enumerate}
  \item Objects allocated up front, no deallocation: \textsf{nbody}
  \item Objects allocated up front, then only deallocation: \textsf{collision}, \textsf{structure}
  \item Cellular automaton (CA) with static cells network: \textsf{sugarscape}, \textsf{traffic}, \textsf{wa-tor}
  \item Other: \textsf{barnes-hut}, \textsf{game-of-life}
\end{enumerate}

Baselines (SOA/AOS) are application variants without any dynamic memory allocation. Baselines of category (1) are trivial to implement with static allocation. In category (2) applications, every object has a boolean \texttt{active} flag to prevent deleted objects from being enumerated in the future. In category (3) applications, classes are merged with the underlying static cell data structure, which wastes memory in case of empty cells (Section~\ref{sec:bench_mem_usage_sec}). Category (4) applications cannot be implemented with only static allocation, unless the application is changed fundamentally.

\paragraph*{Parallel Do-All in Custom Allocators}
Other allocators do not provide do-all operations, which are required for SMMO applications. To compare \soaalloc{} with other allocators, we developed standalone \texttt{parallel\_do} and \texttt{device\_do} implementations that can be used with any allocator.


These implementations maintain arrays for allocated and deleted object pointers of each type. Pointers are added to these arrays with atomic operations. At the end of a parallel do-all operation, deleted pointers are removed from the array of allocated pointers. This process is non-trivial because the same memory location/pointer could be allocated and deleted multiple times throughout a parallel do-all operation. After all deleted pointers were removed, the array of allocated pointers is compacted with a prefix sum operation (same as Figure~\ref{fig:scan_enumeration}).

Depending on the number of (de)allocations, this mechanism may take a long time. A better allocator-specific mechanism could likely be developed with some reverse engineering. For that reason, we break down running times into \emph{enumeration} time and remaining time. Enumeration time should not be taken into account when comparing the performance of different allocators.

\paragraph*{BitmapAlloc}
To analyze the performance of pure bitmap-based object allocation without SOA layout, blocks and fake object pointers, we developed a second allocator \emph{BitmapAlloc}. This allocator treats the entire heap as one large object array, whose slots are managed by hierarchical bitmaps, similarly to \soaalloc{}: one \emph{allocation bitmap} per type and one \emph{free slot bitmap}. Allocation bitmaps are also used in \texttt{parallel\_do} and \texttt{device\_do} implementations.

The main downside of BitmapAlloc is its inefficient memory usage. It supports only a single allocation size, potentially leading to high internal fragmentation.

\subsection{Performance Overview}
Figure~\ref{fig:bench_overview} shows the running time of all benchmarked SMMO applications. We gave each allocator some extra memory to avoid memory scarcity slowdowns: The heap size is 8~GiB, at least 4 times bigger than the maximum amount of all allocated memory at any point throughout the program execution. \textsc{DynaSOAr} achieves superior performance over other allocators due to the SOA layout, a dense object allocation policy and an efficient parallel do-all operation.

All applications except for \textsf{structure} see a speedup by switching from AOS to SOA (compare baselines). In \textsf{structure}, most fields are used together, so SOA does not pay off for this benchmark.

Despite having no dynamic (de)allocation during the benchmark, \textsf{nbody} can see a slight speedup with dynamic memory allocation. This is likely due to fewer cache associativity collisions compared to a denser allocation within an array~\cite{7853809}.

In \textsf{collision}, \soaalloc{}/BitmapAlloc enumerate objects with a bitmap scan (\texttt{device\_do}; 1 bit/object). This is more efficent than in other allocators, which must read object pointers from an array (8 bytes/object). The baseline versions must read an \texttt{active} flag (1 byte/object) from every object, including deleted ones. 

\textsf{game-of-life} and \textsf{wa-tor} are applications that (de)allocate a large number of objects, so enumeration time dominates the running time with custom allocators. \soaalloc{} and BitmapAlloc have much more efficient parallel do-all operations than other allocators.

\textsf{sugarscape} and \textsf{wa-tor} have a 2D grid structure of cells. Baseline versions take advantage of this geometric structure, leading to more coalesced memory accesses. In contrast, programmers have no control over where dynamic allocators place objects in memory. For this reason, the baseline versions are faster than the versions with dynamic memory management. 

In general, in applications with dynamic memory management, objects are always referred to with 64-bit object identifiers/pointers, while all baseline versions use 32-bit integer indices. This especially penalizes benchmarks with small objects; they grow considerably just by switching from 32-bit integers indices to 64-bit pointers.

\begin{figure}
  \includegraphics[width=\textwidth]{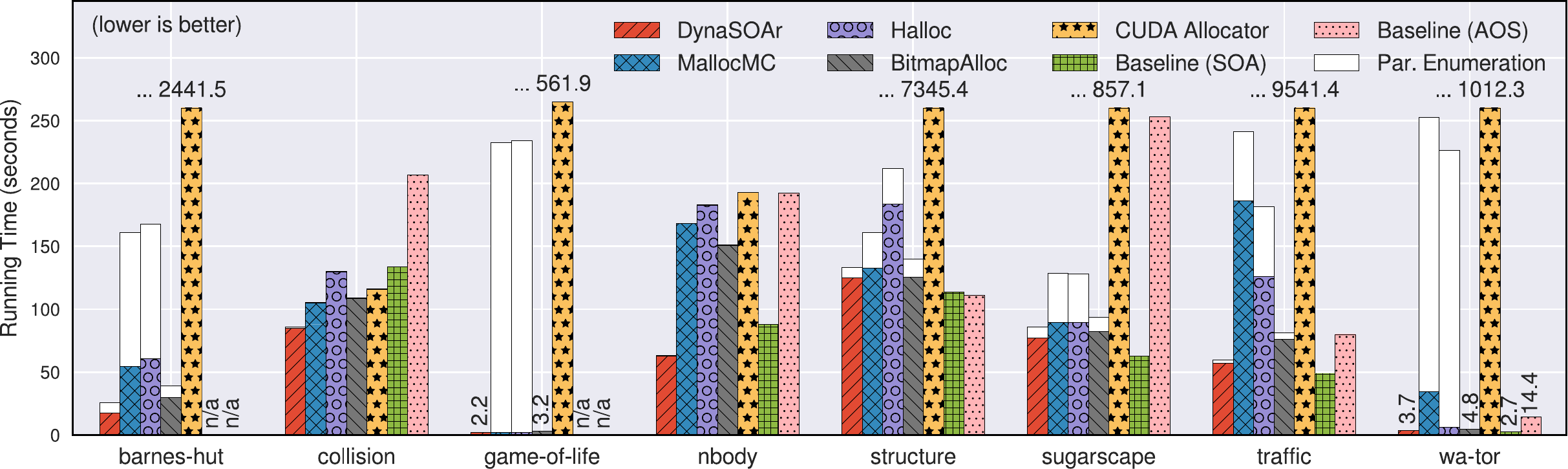}
  \caption[Running time of SMMO applications]{Running time of SMMO application benchmarks}
  \label{fig:bench_overview}
  \vspace{25pt}

  \includegraphics[width=\textwidth]{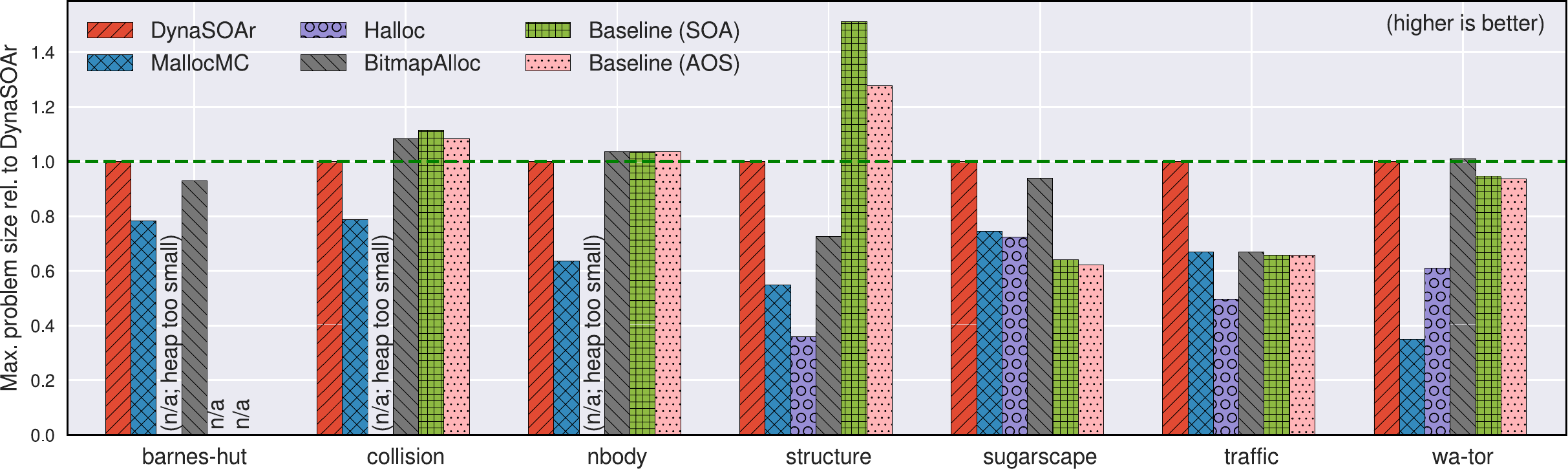}
  \caption[Space efficiency of SMMO applications]{Space efficiency of SMMO application benchmarks (without object enumeration arrays, relative to \textsc{DynaSOAr}}
  \label{fig:memusage}
  \vspace{10pt}

    \subfloat[Comparison with other allocators]{\includegraphics[width=0.485\textwidth]{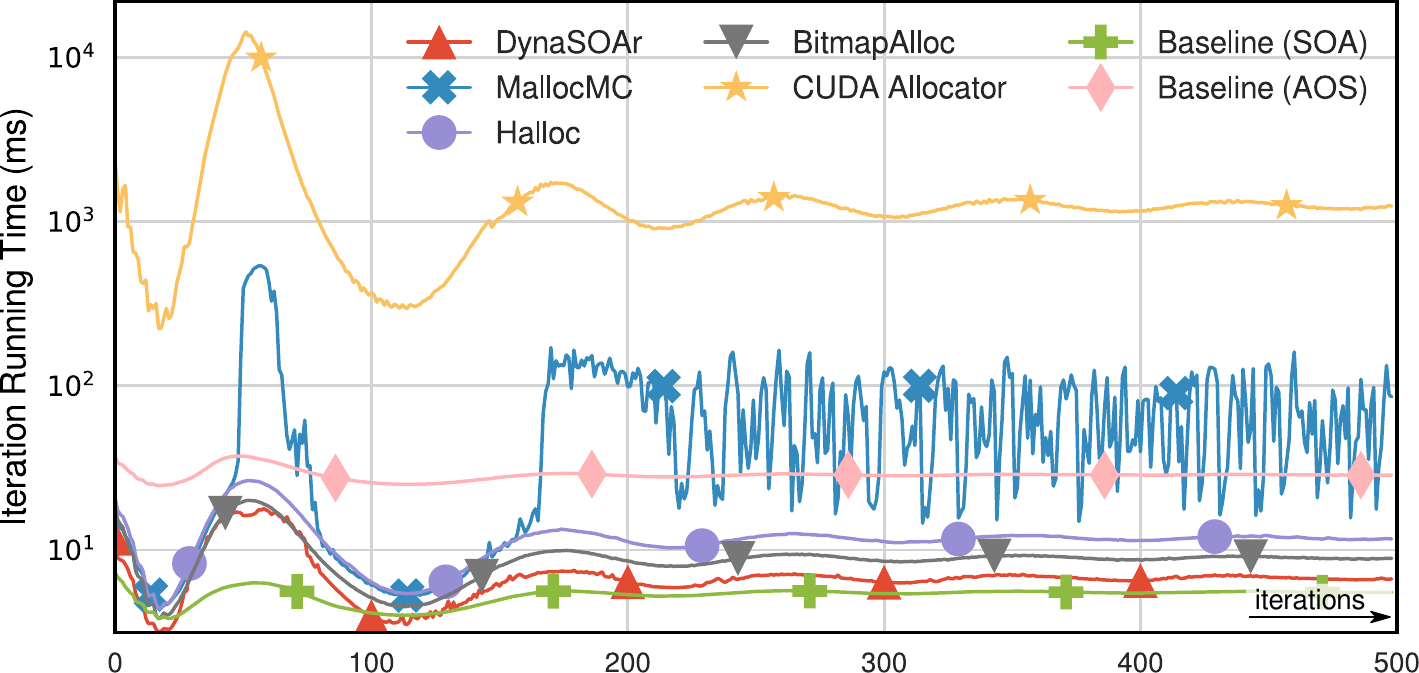}}\hfill
    \subfloat[Fixed heap size, increasing problem size]{\includegraphics[width=0.485\textwidth]{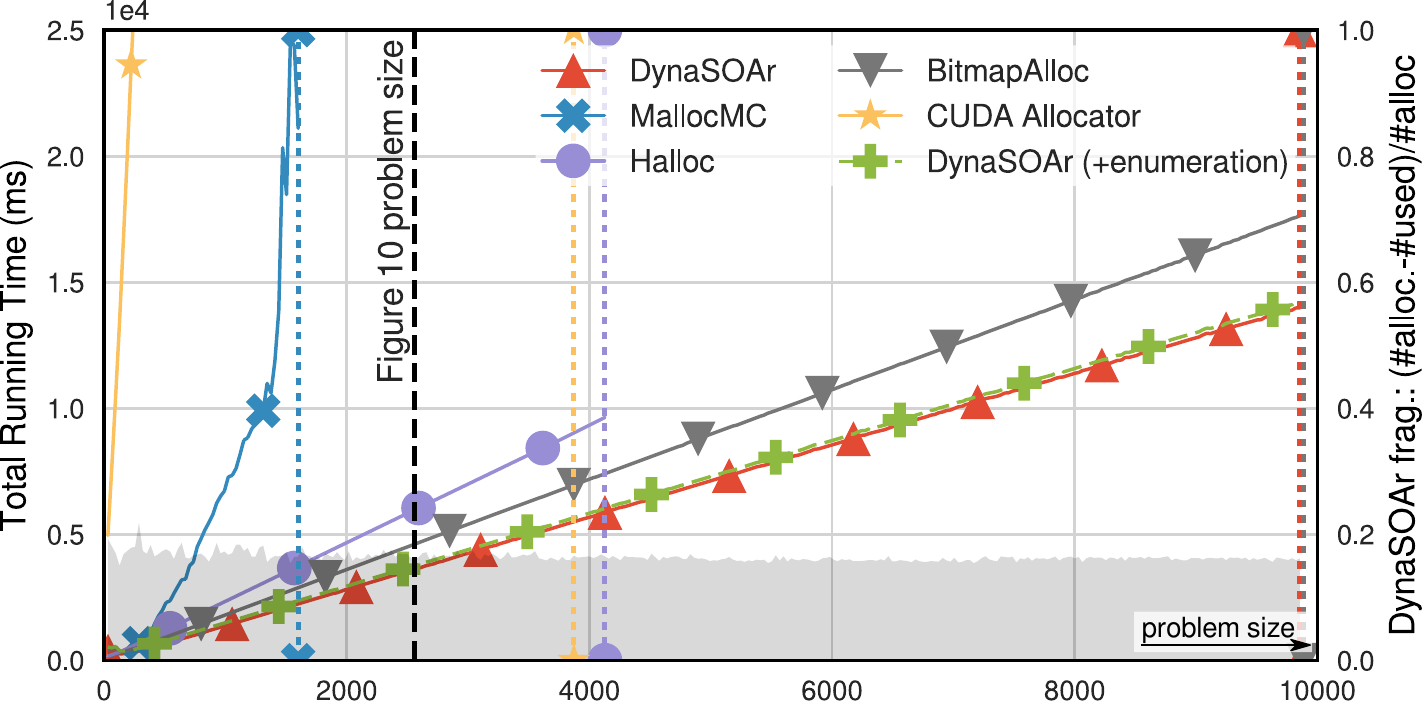}}\\
    
    \subfloat[Isolating single \soaalloc{} optimizations]{\includegraphics[width=0.485\textwidth]{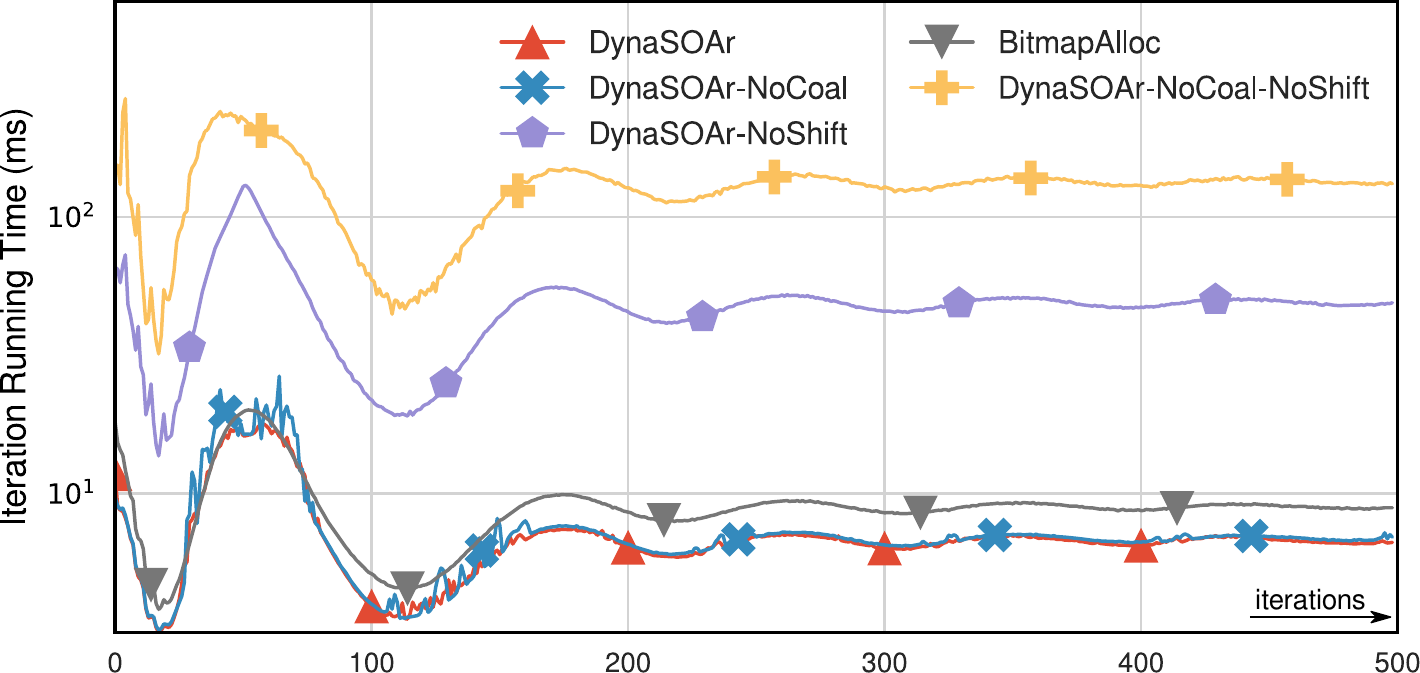}}\hfill
    \subfloat[Number of (de)allocations and fragmentation]{\includegraphics[width=0.485\textwidth]{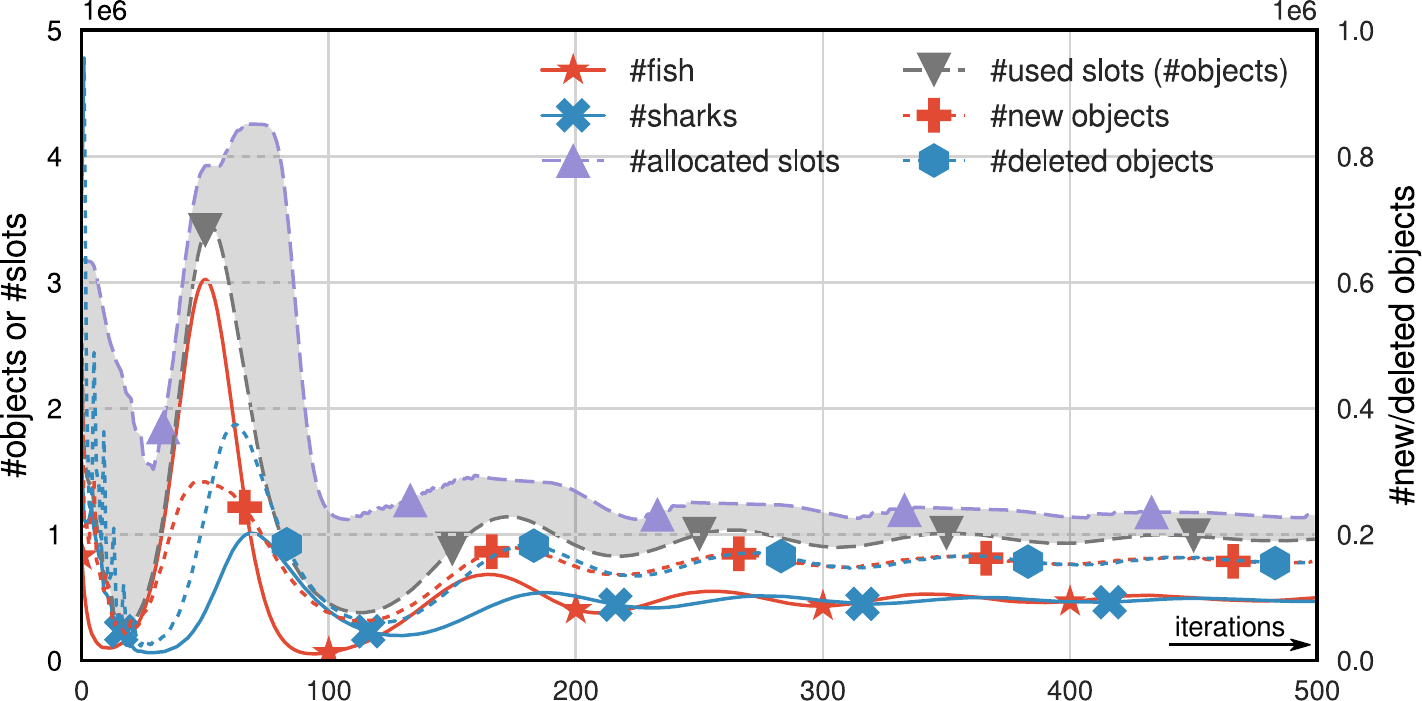}}

  \caption[Detailed analysis of \textsf{wa-tor} benchmark]{Detailed analysis of \textsf{wa-tor} (without enumeration time, unless indicated)}\label{fig:detailed_wa_tor}

  \label{fig:benchmark_scaling}
\end{figure}


\subsection{Space Efficiency}
\label{sec:bench_mem_usage_sec}
To evaluate how efficiently allocators manage memory, we gave them the same heap size and experimentally determined the maximum problem size before running out of memory (Figure~\ref{fig:memusage}).

For category (1) and category (2) applications that allocate all memory during startup (\textsf{collision}, \textsf{nbody}, \textsf{structure}), the baseline versions are more space-efficient. The exact number of objects per type is known ahead of time, so placing objects in memory is trivial. However, even though category (2) applications delete objects throughout their runtime, the memory consumption of the baseline versions does not decrease over time. This is a problem even for \textsc{DynaSOAr} because blocks can only be deleted when they are entirely empty, which can take some time. This problem can be solved with memory defragmentation (Chapter~\ref{sec:chap_gpu_mem_defrag}).

Category (3) applications (\textsf{sugarscape}, \textsf{traffic}, \textsf{wa-tor}) exhibit a fixed grid/network structure of cells, upon which a dynamic set of agents is moving. The baseline versions allocate the fields of agents directly inside cells. Classes for agents are combined with the respective cell class and some fields have \texttt{null} values (or garbage) if they are not used. This wastes memory because not all cells are occupied by agents all the time. In those applications, \soaalloc{} is not as fast as optimized SOA baseline implementations, but it can handle significantly larger problem sizes.

Out of all dynamic memory allocators, \soaalloc{} is most space-efficient. MallocMC and Halloc are based on a hashing approach. With rising heap fill levels, it becomes increasingly difficult to find free memory for allocations, so they fail to use the entire heap memory. \soaalloc{} and BitmapAlloc can avoid this problem with bitmaps, which act as an index for free memory.

Albeit negligible in these benchmarks, \soaalloc{} and Baseline (SOA) also benefit from slightly smaller object sizes: Only SOA arrays must be aligned/padded and not every single object.

\subsection{Detailed Analysis of \textsf{wa-tor}}
\label{sec:dyna_detailed_analysis_wator}
\textsf{wa-tor} (Section~\ref{sec:smmo_ex_wa_tor_sec}) is a particularly interesting benchmark. It exhibits a massive number of (de)allocations in waves, until an equilibrium between fish and sharks is reached. This allows us to measure the performance at a massive and at a lower number of concurrent (de)allocations. For a fair comparison of allocators, we do not include the time spent on enumeration in this section, unless indicated.

Figure~\ref{fig:detailed_wa_tor}\textsc{a} shows that \soaalloc{} always provides superior performance compared to other allocators; during (de)allocation spikes (around iteration 50), as well as if fewer concurrent (de)allocations take place. The performance of mallocMC degrades after a few iterations and does not recover, possibly due to a highly fragmented heap.

In Figure~\ref{fig:detailed_wa_tor}\textsc{b}, all allocators were given a heap size of 1~GB and the problem size increases gradually on the x-axis. mallocMC performs well at first, but its performance drops rapidly as soon as the heap starts filling up. \soaalloc{} can handle much larger problem sizes, given the same amount of heap memory. The running time grows linearly with the problem size, showing that recent GPU architectures can handle atomic operations quite well.

Fragmentation in \soaalloc{} is different from other allocators: \soaalloc{} does not have internal or external fragmentation by design, but memory within allocated blocks is only available for a certain type. This sort of fragmentation decreases with better clustering. In \textsc{DynaSOAr}, fragmentation $F$ is the relative number of unused objects slots among all allocated blocks $\mathit{Blocks}$ (gray area in Figure~\ref{fig:detailed_wa_tor}\textsc{b}, \textsc{d}).

\begin{align*}
F = \frac{\sum_{b \in \mathit{Blocks}} (N_{\mathit{type}(b)} - \mathit{used}(b))}{\sum_{b \in \mathit{Blocks}}  N_{\mathit{type}(b)}}  \approx \frac{1}{\mbox{\#blocks}} \sum_{b \in \mathit{Blocks}} \frac{\mbox{\#free slots}(b)}{\mbox{\#slots}(b)} \tag{\emph{fragmentation}}
\end{align*}

At iterations~60--80 in Figure~\ref{fig:detailed_wa_tor}\textsc{d}, \soaalloc{} has high fragmentation because many fish objects were deallocated. However, a block can only be deallocated when \emph{all} of its objects were deallocated. The fragmentation level decreases gradually because more fish/shark objects are deallocated over time and new allocations are performed in existing (active) blocks. Therefore, new blocks are rarely allocated and there is a chance that an active block will eventually run empty. As can be seen in Figure~\ref{fig:detailed_wa_tor}\textsc{b}, fragmentation is independent of the problem size and constant at around 18\% (gray area) after 500 Wa-Tor iterations. 

We implemented multiple \soaalloc{} variants to pinpoint the source of \soaalloc{}'s speedup over other allocators (Figure~\ref{fig:detailed_wa_tor}\textsc{c}). The most important optimization is the rotation-shifting of bitmaps. Without shifting (\textsf{*-NoShift} variants), performance degrades severely due to high thread contention, because (1) all threads are trying to allocate objects in the same few active blocks and (2) all threads compete for the same few free block locations instead of choosing free block locations in different parts of the heap. Allocation request coalescing is another optimization that reduces thread contention significantly (compare \textsf{DynaSOAr-NoCoal-NoShift} and \textsf{DynaSOAr-NoShift}), but it cannot improve performance much further if we are already rotation-shifting bitmaps (compare \textsf{DynaSOAr} and \textsf{DynaSOAr-NoCoal}).

In Figure~\ref{fig:frag_param}, we experiment with the number of active block lookup attempts before entering the slow path, which strongly affects fragmentation. By default, \textsc{DynaSOAr} attempts to locate an active block five times ($r=5$) before initializing a new active block. This is close to the lowest achievable fragmentation level (i.e., without any thread contention). Due to unfortunate allocate-deallocate patterns, a fragmentation rate of 0\% is not achievable without manually relocating objects through memory defragmentation or predicting future (de)allocations.

\begin{figure}
  \begin{minipage}[c]{0.485\textwidth}
    \includegraphics[width=\textwidth]{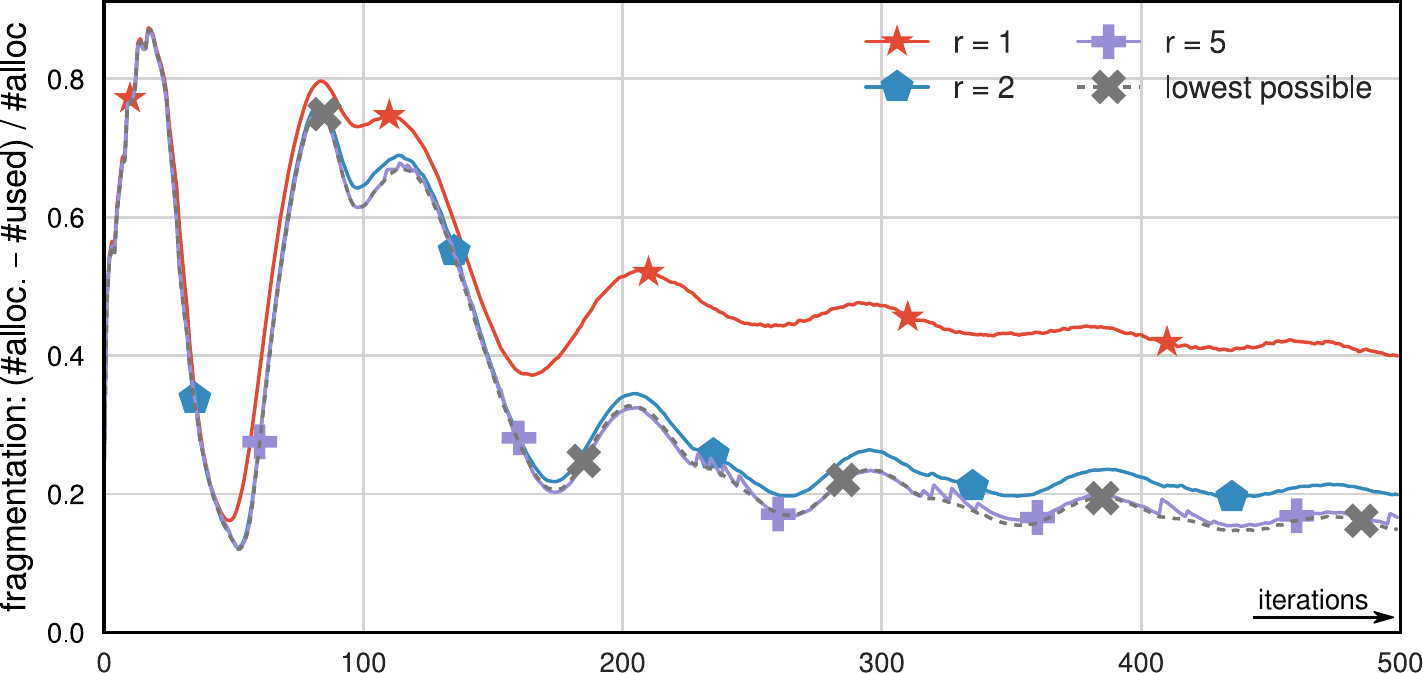}
  \end{minipage}\hfill
  \begin{minipage}[c]{0.475\textwidth}
    \caption[Memory fragmentation experiment]{Memory fragmentation of \textsf{wa-tor} by the number of active block lookup attemps $r$ (Algorithm~\ref{alg:alloc_algo}, Line~2). The x-axis denotes iterations and the y-axis denotes the fragmentation rate. With only 1 retry ($r=2$), fragmentation is reduced by 50\%.} \label{fig:frag_param}
  \end{minipage}
\end{figure}

\subsection{Raw Allocation Performance}
The \emph{Linux Scalability} microbenchmark~\cite{Lever:2000:MPM:1267724.1267780} measures the raw (de)allocation time of memory allocators. We set the heap size to 1~GiB and one CUDA kernel allocates $n$ 64-byte objects in each of the 16,384 threads. A second CUDA kernel deallocates all objects. In Figure~\ref{fig:benchmark_scaling_2}\textsc{a}, the x-axis denotes the number of allocations $n$ per thread and the y-axis shows the total benchmark running time divided by $n$.

We chose the size of the heap such that it can hold exactly $16384 \times n$ objects with $n=1024$ (100\% heap utilization). No allocator can reach perfect utilization because some memory is used for internal data structures such as bitmaps.


Halloc is the fastest allocator. Both Halloc and mallocMC fail to allocate more than 510 objects (49.8\% utilization). This is better than in some other benchmarks, probably because only objects of one size are allocated. \soaalloc{} (96.9\% utilization), BitmapAlloc (98.4\% utilization) and Halloc scale almost perfectly with the number of allocations.

\begin{figure}
    \subfloat[Linux Scalability: Increasing \#allocations]{\includegraphics[width=0.485\textwidth]{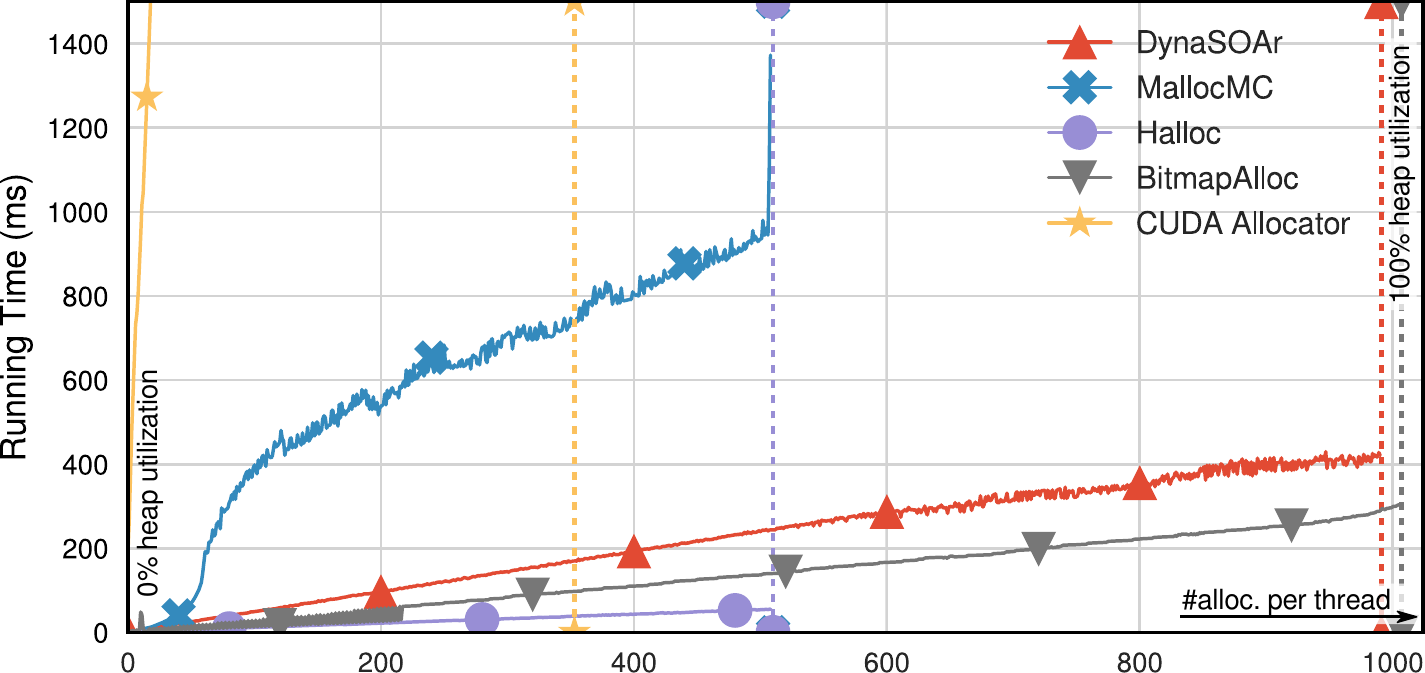}}\hfill
    \subfloat[Scaling study: Heap size (\textsf{wa-tor})]{\includegraphics[width=0.485\textwidth]{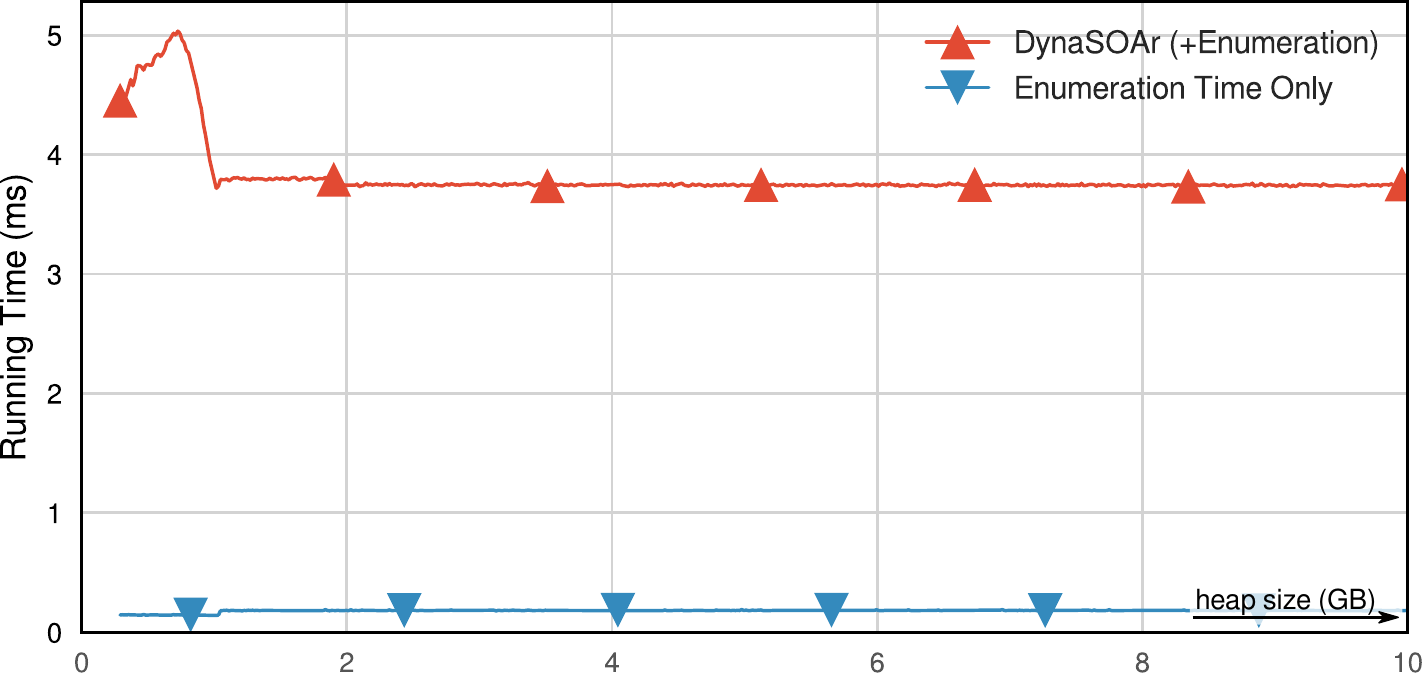}}\\
  \caption[Scaling study with different memory allocators]{Scaling study: Number of allocations and heap size}
  \label{fig:benchmark_scaling_2}
\end{figure}

\subsection{Parallel Object Enumeration}
The overhead of object enumeration in \textsc{DynaSOAr} is negible in most benchmarks (Figure~\ref{fig:bench_overview}, Figure~\ref{fig:benchmark_scaling}\textsc{b}). In Figure~\ref{fig:benchmark_scaling_2}\textsc{b}, the problem size is fixed but the heap size increases on the x-axis. \textsc{DynaSOAr}'s performance (and that of object enumeration) is independent of the size of the heap, if enough memory is available for the application. This shows that our hierarchical bitmaps work well with various heap sizes.



\section{Conclusion}
\label{sec:conclusion}
We presented \soaalloc{}, a new dynamic object allocator for SIMD architectures. The main insight of our work is that memory allocators should not only aim for good raw (de)allocation performance, but also optimize the usage of allocated memory. \soaalloc{} was designed for GPUs, but its basic ideas are applicable to other architectures and systems with good or guaranteed vectorization such as the Intel SPMD compiler (ispc)~\cite{6339601}.

\soaalloc{} achieves good memory access performance by controlling (a) memory allocation and (b) memory access with a parallel do-all operation. \textsc{DynaSOAr}'s main speedup over other allocators is due to an SOA-style object layout, which can benefit memory bandwidth utilization (through coalesced memory access) and cache utilization. To allow for dynamic (de)allocation of objects, \textsc{DynaSOAr} allocates objects in blocks instead of a plain SOA layout. \textsc{DynaSOAr} utilizes hierarchical bitmaps for fast and compact allocations with low fragmentation.

Our benchmarks show that \soaalloc{} can achieve significant speedups over state-of-the-art allocators of more than 3x in SMMO application code due to better memory access performance. \soaalloc{} also has a significantly lower memory footprint than other allocators, mainly because by design \textsc{DynaSOAr} has no internal fragmentation and is not based on hashing. Our work also shows how an SOA layout can support class inheritance without wasting memory: by allocating objects in blocks and encoding block sizes in object pointers.


\chapter{GPU Memory Defragmentation}
\label{sec:chap_gpu_mem_defrag}
Memory fragmentation is a challenging problem of dynamic memory allocators and has been widely studied on single-core and multi-core CPU systems (MIMD architectures). On such systems, dynamic memory allocators can achieve low memory fragmentation with good allocation policies~\cite{Johnstone:1998:MFP:286860.286864} and compacting garbage collectors.

However, memory fragmentation has not been studied thoroughly on SIMD architectures, including GPUs. In this chapter, we present the design and implementation of \textsc{CompactGpu}, a memory defragmentation system for the \textsc{DynaSOAr} dynamic GPU memory allocator.

\minitoc

\paragraph{Outline}
This chapter is organized as follows. Section~\ref{sec:why_mem_defrag} describes the effect of memory fragmentation on GPUs. Section~\ref{sec:dynasoar_overview} gives a high-level overview of the \textsc{DynaSOAr} memory allocator (see Chapter~\ref{sec:chapter_dynasoar} for details). Section~\ref{sec:memory_defrag} describes the design and implementation of \textsc{CompactGpu}. Section~\ref{sec:alternative_designs_sec} discusses alternative designs that utilize CUDA shared memory. Section~\ref{sec:evaluation} evaluates the defragmentation quality and performance impact of \textsc{CompactGpu} with real and synthetic benchmarks. Finally, Sections~\ref{sec:related} and~\ref{sec:conclusion_compactgpu} explore related work and conclude this chapter.

\paragraph{Overview}
Despite the recent popularity of massively parallel single-instruction multiple-data (SIMD) architectures, the memory fragmentation problem has not been studied thoroughly on such architectures. This is because dynamic memory allocators for SIMD architectures such as GPUs have just been developed recently~\cite{5577907, 6339604, Gelado:2019:TGM:3293883.3295727} and not been around long enough yet. 

We need to study memory (de)fragmentation on massively parallel SIMD architectures because allocations follow different patterns on such architectures. Most allocations are small in size\footnote{If thousands of threads were to request large allocations, a GPU would run out of memory immediately.} and due to mostly regular control flow, many allocations have the same byte size. Such patterns are reflected in the design of state-of-the-art GPU allocators. For example, Halloc~\cite{hallocweb}, one of the fastest GPU allocators can allocate only a few dozen predetermined byte sizes between 16~bytes and 3~KB. Such specialties must be exploited by memory defragmentation systems to achieve good performance.

In this chapter, we present \textsc{CompactGpu}, an incremental, fully parallel, in-place memory defragmentation system for GPUs. \textsc{CompactGpu} is implemented as an extension to \textsc{DynaSOAr}, but it could also be implemented in other systems. GPUs/SIMD architectures are predominantly programmed in a C++ dialect (e.g., CUDA, OpenCL, ispc~\cite{6339601}, Sierra~\cite{LeiBa:2014:SSE:2568058.2568062}) and memory management in C++ is manual, so we cannot rely on a garbage collector to collect metainformation for us. 



\textsc{CompactGpu} is an efficient GPU memory defragmentation system that is optimized for GPU-specific allocation patterns. It is fully parallel and in many cases the performance gain of defragmentation is much larger than the defragmentation overhead. This is due to careful design and engineering efforts: \textsc{CompactGpu} is based on parallel block merging, utilizes bitmaps to speed up pointer rewriting, exhibits mostly uniform control flow and requires no synchronization between GPU threads.

We evaluated \textsc{CompactGpu} with synthetic and real benchmarks. \textsc{CompactGpu} can improve application performance by up to 16\% and reduce the overall memory consumption of an application, while incurring minimal runtime overheads.

\paragraph{Publications}
This chapter is in part based on the following papers.
\begin{itemize}
  \item Matthias Springer, Hidehiko Masuhara. \textbf{``Massively Parallel GPU Memory Compaction.''} In: \emph{Proceedings of the ACM SIGPLAN International Symposium on Memory Management.} ISMM 2019. ACM, 2019, pp.~14--26. \\ \texttt{\doi{10.1145/3315573.3329979}}
  \item Matthias Springer. \textbf{``CompactGpu: Massively Parallel Memory Defragmentation on GPUs.''} Extended Abstract. In: \emph{ACM Student Research Competition at PLDI 2019.} 3 pages. (reviewed, no formal proceedings)
\end{itemize}

\section{Why GPU Memory Defragmentation?}
\label{sec:why_mem_defrag}
Before introducing the design of our system, we review important performance characteristics of SIMD architectures.

\paragraph{Effects of Fragmentation}
Fragmentation measures the degree of scattering of allocations across the heap and is caused by unfortunate allocate-deallocate patterns. High fragmentation leads to three main disadvantages.

\begin{itemize}
  \item \textbf{Premature Out-of-Memory:} Large allocations cannot be accommodated even if there is enough free memory overall (\emph{external fragmentation}).
  \item \textbf{Low Cache Hit Rate:} Poor data locality causes poor cache performance~\cite{Grunwald:1993:ICL:155090.155107}, because fragmented data occupies more cache lines.
  \item \textbf{Low Vector Load/Store Efficiency:} SIMD vector load/store instructions are less efficient, because accessing fragmented data requires more vector transactions than accessing compact data.
\end{itemize}

While the first two points are well-established and apply to most architectures, the specific effects on SIMD architectures have received little attention.


\paragraph{Effect of Fragmentation on Vectorized Access}
SIMD architectures achieve parallelism by executing instructions on a vector register. However, vector load/store operations are less efficient with higher fragmentation. When threads in a GPU application simultaneously access different memory addresses, the GPU \emph{coalesces} accesses from the same SIMD work group (\emph{warp} in CUDA, every 32 consecutive threads) into one physical memory transaction if the addresses are on the same 128-byte cache line (Section~\ref{sec:background_mem_coal}). More fragmentation leads to more scattered memory addresses, resulting in poorer performance due to a higher number of memory transactions.


Memory fragmentation can greatly affect vectorized access, even if data is stored in a Structure of Arrays (SOA) data layout. Recent NVIDIA architectures coalesce simultaneous accesses of consecutive memory addresses into 128-byte vector transactions. Accessing fragmented data requires more vector transactions than accessing the same amount of dense data. This reduces the overall performance of memory-bound applications because memory bandwidth is limited~\cite{nvidia_bound}.





\paragraph{Memory Defragmentation}
To optimize the memory access of global memory, we are developing a memory defragmentation system for GPUs in this chapter. In essence, every memory defragmentation system has to solve four basic problems.
\begin{enumerate}
  \item Determine which parts of the heap are fragmented.
  \item Based on that information, decide which objects\footnote{We use the term \emph{object} instead of \emph{allocation} throughout this chapter because we are focusing on object-oriented systems.} to move (relocate) and where to move them.
  \item Physically relocate objects in memory.
  \item Find and rewrite pointers to relocated objects. (Alternative: Ensure that objects can still be accessed through their old pointers.)
\end{enumerate}
Most memory defragmentation systems are part of a garbage collector. Since garbage collectors have to scan large parts of the heap anyway, they can gather additional metainformation almost for free. This information can be used to select memory areas for compaction~\cite{Ossia:2004:MCC:1029873.1029877, Kermany:2006:CCI:1133981.1134023} or to determine which parts of the heap contains pointers that must be rewritten~\cite{Veldema:2012:PMD:2247684.2247693}.

\section{Heap Layout and Data Structures}
\label{sec:dynasoar_overview}
A variety of dynamic memory allocators for GPUs have been developed in recent years. \textsc{CompactGpu} is implemented in \textsc{DynaSOAr}, but its basic ideas can be adapted to other dynamic memory allocators as long as they follow a few basic design requirements.

\begin{itemize}
  \item The heap is divided into fixed-size \textbf{memory blocks}, in which objects are allocated. Every \textsc{DynaSOAr} block has a constant block capacity (based on its type), regardless of the number of allocated objects.
  \item A block contains only \textbf{objects of the same size}. This requirement is crucial. Same-size objects can be compacted much more easily than objects of different size. Every \textsc{DynaSOAr} block contains objects of only one type, so all objects of a block have the same size.
  \item Blocks of the same object size have the \textbf{same capacity}. All \textsc{DynaSOAr} blocks have the same size in bytes, so all blocks of the same type/object size have the same capacity. Section~\ref{sec:block_capacity_dynasoar_s5} describes how exactly block capacities are determined in \textsc{DynaSOAr}.
  \item The allocator maintains \textbf{fill levels} for each block. The fill level of every \textsc{DynaSOAr} block can be determined by counting the number of set bits in the object allocation bitmap.
\end{itemize}

\textsc{DynaSOAr}, Halloc~\cite{hallocweb} and UAlloc~\cite{Gelado:2019:TGM:3293883.3295727} are three examples of allocators that satisfy these requirements. \textsc{DynaSOAr} is the only memory allocator with an SOA data layout, which allows for efficient vectorized memory access of allocated memory. While all allocators would benefit from better cache performance and more space-efficient memory usage, memory defragmentation in \textsc{DynaSOAr} additionally leads to more efficient vectorized memory accesses and thus better memory bandwidth utilization.

\subsection{Running Example}
We use a simple fish-and-sharks simulation (\textsf{wa-tor}) as a running example to describe the data structures of \textsc{CompactGpu}. Fish and sharks inhabit a 2D grid of cells in a predator-prey relationship. This application has four classes: \texttt{Cell}, \texttt{Agent}, \texttt{Fish} and \texttt{Shark}. The last two classes are subclasses of the abstract class \texttt{Agent}. We describe \textsf{wa-tor} in more detail in Section~\ref{sec:smmo_ex_wa_tor_sec}.

\textsf{wa-tor} exhibits a large number of allocations and deallocations in GPU code, which leads to memory fragmentation. By reducing memory fragmentation, \textsc{CompactGpu} can reduce the overall memory consumption of \textsf{wa-tor}.

\subsection{Overview of the \textsc{DynaSOAr} Allocator}
Chapter~\ref{sec:chapter_dynasoar} described the \textsc{DynaSOAr} memory allocator in detail. In the following paragraphs, we give a simplified summary of \textsc{DynaSOAr} and describe which parts of the allocator had to be modified to implement \textsc{CompactGpu}.

\begin{figure}
  \centering
  \includegraphics[width=\textwidth]{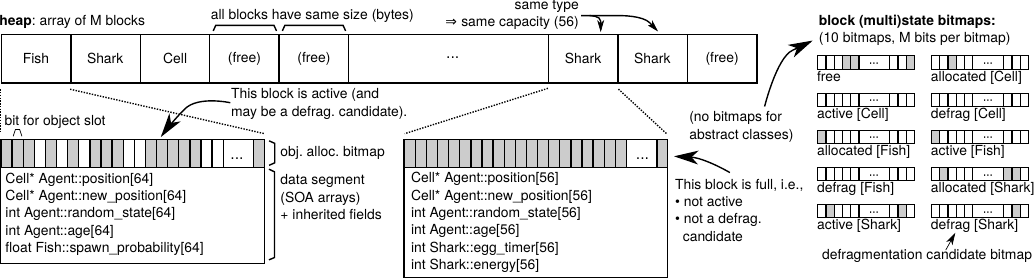}
  \caption[\textsc{DynaSOAr} heap layout of \textsf{wa-tor}]{Example: Heap layout for \textsf{wa-tor}. The heap consists of equally sized blocks. Up to 64 objects can be stored in a block, as indicated by the \emph{object allocation bitmap}. A block can be in one or multiple of 10 possible (multi)states, as indicated by the state bits shown for every block. There are \emph{allocated}, \emph{active} and \emph{defrag} states for the three classes \texttt{Fish}, \texttt{Shark} and \texttt{Cell}, but not for class \texttt{Agent} because it is an abstract class.}
  \label{fig:heap_layout_defrag}
\end{figure}

\begin{figure}
  \centering
  \includegraphics[width=0.7\columnwidth]{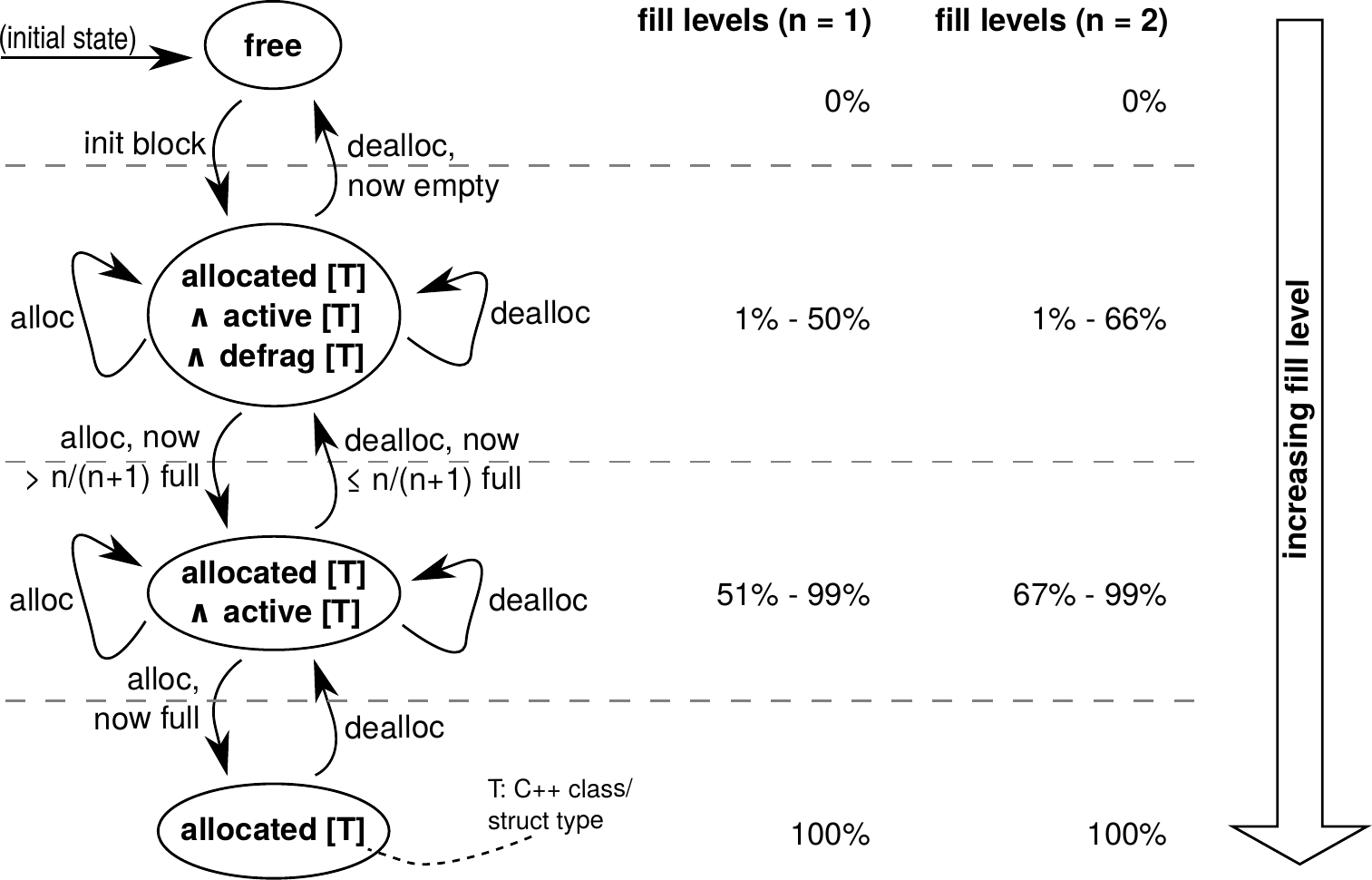}
  \caption[\textsc{CompactGpu} block state transitions]{Block states. Initially, every block is \emph{free}. New objects are allocated in \emph{active} blocks of the corresponding type. We introduced a new state \emph{defrag} to indicate defragmentation candidates. Only objects from such blocks are relocated during defragmentation.}
  \label{fig:block_states}
\end{figure}

\textsc{DynaSOAr} is a slab allocator~\cite{Bonwick:1994:SAO:1267257.1267263}. It divides the heap into $M$ blocks of equal byte size, each of which can contain up to 64 objects (\emph{capacity}) of the same C++ class/struct type, depending on the size of the type (Figure~\ref{fig:heap_layout_defrag}). A position where an object can be stored is called an \emph{object slot}. A 64-bit \emph{object allocation bitmap} keeps track of allocations. Objects are stored in the \emph{data segment} in an SOA data layout: one \emph{SOA array} per field.

\paragraph{Block States}
A block can be in one or more \emph{multistates}. There are $3 \times \#\mathit{types} + 1$ possible states: one global \emph{free} state and three states for each type $T$ in the system (Figure~\ref{fig:block_states}).
\begin{itemize}
  \item \textbf{free:} The block is empty and does not contain any objects. No type is specified for this block.
  \item \textbf{allocated[T]:} The block may contain objects only of type $T$. No other objects can be stored in the block.
  \item \textbf{active[T]:} The block contains objects of type $T$. It is not full yet, i.e., it has space for at least one more object. \emph{active[T]} $\Rightarrow$ \emph{allocated[T]}.
  \item \textbf{defrag[T]:} The block is considered for defragmentation. We call such a block a \emph{defragmentation candidate}. We introduced this state to support defragmentation in \textsc{DynaSOAr} and will describe its purpose in the next section. \\ \emph{defrag[T]} $\Rightarrow$ \emph{allocated[T]} $\wedge$ \emph{active[T]}.
\end{itemize}

Allocation, deallocation and defragmentation routines frequently lookup blocks by state. For that reason, block states are indexed by hierarchical bitmaps of size $M$, as described in Section~\ref{sec:dynsoar_block_bitmaps}: one bitmap per state. The bitmap hierarchy is currently not utilized by \textsc{CompactGpu}, so \emph{defrag[T]} does not necessarily have to be hierarchical.



\paragraph{Fragmentation}
Our definition of fragmentation $F$ differs from other systems. We define it as the fraction of allocated but unused memory. If a block is in an \emph{allocated} state, we consider all of its object slots as \emph{allocated}. However, only object slots that actually contain an object, as indicated by the object allocation bitmap, are \emph{used}. Fragmentation is defined as the average \emph{free level} among all allocated blocks.

\begin{align*}
F = \frac{1}{\mbox{\#blocks}} \sum_{b \in \mathit{Blocks}} \frac{\mbox{\#free slots}(b)}{\mbox{\#slots}(b)} \tag{\emph{fragmentation}}
\end{align*}

Our goal is to reduce $F$ as much as possible. Zero fragmentation means that all allocated blocks are 100\% full and the other blocks are empty (\emph{free}). In that case, vectorized memory access is most efficient. Conversely, a vector load on a block that is 60\% full will on average read 40\% garbage. Moreover, unused memory in an allocated block is not available for objects of other types. This leads to less space-efficient memory usage.

The blocks themselves may be widely \emph{scattered} in the heap. For example, in Figure~\ref{fig:heap_layout_defrag}, three allocated blocks are stored at the beginning of the heap and two are stored towards the end. This does not affect cache utilization or vector load/store efficiency, because with up to 64 objects per blocks, most SOA arrays are much larger than a cache line or the size of a vector load/store (128~bytes on NVIDIA GPUs). Neither can it lead to external fragmentation and premature out-of-memory errors, because all blocks have the same byte size.




\paragraph{Object Allocation}
To reduce fragmentation, even without active memory defragmentation, \textsc{DynaSOAr} allocates new objects of type $T$ always in \emph{active[T]} blocks. These are blocks that have space for at least one more object. Only if no active block could be found, \textsc{DynaSOAr} locates a free block and turns it into an \emph{allocated[T]} and \emph{active[T]} block (\emph{slow path}).

\textsc{DynaSOAr} then reserves an object slot inside the block by atomically flipping a bit in the object allocation bitmap from 0 to 1. If this operation was successful, the block state may have changed, so we may have to update the respective bits in the block state bitmaps.

If the number of objects of a type drops, fragmentation can increase, because a block is deallocated only if all of its objects are deallocated. This kind of fragmentation can be eliminated with \textsc{CompactGpu}.

\paragraph{Programming Interface}
\textsc{DynaSOAr} provides an embedded C++ DSL for defining classes/fields. Through this DSL, \textsc{CompactGpu} can programmatically \emph{reflect} on the classes/fields that are defined in an application, somewhat similar to the Java Reflection API or metaobject protocols~\cite{Chiba:1995:MPC:217838.217868}. This functionality is used in the pointer rewriting step to restrict heap scans to a smaller part of the heap (Section~\ref{sec:rewriting_ptrs}) by excluding parts of the heap that are guaranteed to be free of pointers that must be rewritten.

\section{Defragmentation with \textsc{CompactGpu}}
\label{sec:memory_defrag}
\textsc{CompactGpu} is a memory defragmentation system for GPUs, implemented as a \textsc{DynaSOAr} extension. \textsc{CompactGpu} is:

\begin{itemize}
  \item \textbf{Configurable:} The desired target fragmentation rate can be tuned with parameters. Better defragmentation can lead to more space savings and better memory access performance, but also has a higher defragmentation overhead.
  \item \textbf{Incremental:} A single defragmentation pass is very fast and compacts only a fraction of the heap. Compacting the entire heap requires multiple passes.
  \item \textbf{In-place:} No auxiliary storage is necessary and the entire heap remains usable.
  \item \textbf{A stop-the-world approach:} A defragmentation pass can run only when no other GPU code is running. This is because, in current GPU architectures, there is no efficient way of interrupting a kernel to run a defragmentation pass, should the allocator run out of memory during the kernel. Many GPU programs (including all SMMO examples in Section~\ref{chap:smmo_examples}) are a sequence of GPU kernel invocations~\cite{shen2015study}, so there are usually plenty of opportunities to run a defragmentation pass in-between.
  \item \textbf{Fully parallel:} Every step is implemented as a perfectly parallel CUDA kernel. No synchronization among threads is necessary for defragmentation.
  \item \textbf{Not order preserving:} After defragmentation, objects are likely arranged in a different order on the heap.
\end{itemize}

Programmers initiate defragmentation manually, typically after a parallel do-all operation, and specify the C++ type that should be defragmented. \textsc{CompactGpu} extends \textsc{DynaSOAr} with an additional host allocator handle function for initiating defragmentation.

\begin{itemize}
  \item \texttt{HAllocatorHandle::parallel\_defrag<T, k1, k2>()}: Initiate memory defragmentation for objects of type $T$. Internally, this function may run multiple defragmentation passes. $k_1$ and $k_2$ are parameters that control the number of defragmentation passes and are described later.
\end{itemize}

\noindent \textsc{CompactGpu} is based on three fundamental ideas.

\begin{itemize}
  \item \textbf{Block Merging:} The heap is defragmented by moving/relocating objects from source blocks to target blocks.
  \item \textbf{Forwarding Pointers:} After relocating objects, pointers to the new object locations are stored in source blocks.
  \item \textbf{Bitmaps:} To speed up pointer rewriting, bitmaps are utilized to quickly decide whether a pointer must be rewritten.
\end{itemize}

We considered various alternative designs (Section~\ref{sec:alternative_designs_sec}), but the combination of forwarding pointers with bitmaps proved to be most performant on GPUs.

\paragraph{Block Merging}
\textsc{CompactGpu} compacts the heap by merging blocks of the same type. Blocks of the same type have the same capacity, so \emph{a source block can be merged into a target block of the same type if both blocks are no more than 50\% full}. This is to ensure that all objects of the source block fit into the target block. While it would be sufficient to require that both fill levels \emph{together} are no more than 100\%, this techique is easier to implement and more space-efficient, because the fill level of a block can then be encoded in a single bit (i.e., $1$ means $\leq 50\%$ and $0$ means $> 50\%$).

We can extend this idea to higher fill levels: A source block can be merged into two target blocks if none of the three blocks is more than 66\% full. Or in general: A source block can be merged into $n$ target blocks if none of the $n+1$ blocks is more than $\frac{n}{n+1}$ full. In each case, the source block is eliminated and the number of allocated blocks is reduced by one.

\paragraph{Defragmentation Factor}
We call $n$ the \emph{defragmentation factor}. This value is problem-specific and must be chosen by the programmer at compile time. Blocks that are no more than $\frac{n}{n+1}$ full are \emph{defragmentation candidates}. Only those blocks are considered during defragmentation. Higher defragmentation factors increase the overhead of memory defragmentation but can lead to a lower final fragmentation rate. Whether a higher defragmentation factor pays off depends on the application.

During defragmentation, all objects from a source block are moved to one or multiple target blocks. The source block is deleted. Target blocks lose their defragmentation candidate state \emph{defrag[T]} if they are now more than $\frac{n}{n+1}$ full. They also lose their active state \emph{active[T]} if they are now entirely full.

Given a defragmentation factor of $n$, \textsc{CompactGpu} is guaranteed to bring down fragmentation to $1 - \frac{n}{n+1} = \frac{1}{n+1}$, if all defragmentation candidates are eliminated. This may require multiple defragmentation passes, as will be described later. For example, for $n=2$, all blocks with $\leq 66\%$ fill level are eliminated during defragmentation\footnote{Since every source block must be matched with $n=2$ target blocks, up to two defragmentation candidates may be left over.}. Only blocks with a higher fill level are left over. Consequently, the final fragmentation level is guaranteed to be less than $1 - 66\% = 33\%$. 

\subsection{Defragmentation Candidate Bitmaps}
\label{sec:compact_gpu_extend_defrag}
A defragmentation pass must be able to quickly find all defragmentation candidates in order to choose source and target blocks efficiently. \textsc{CompactGpu} extends object allocation and deallocation routines of \textsc{DynaSOAr} to keep track of defragmentation candidates. \textsc{CompactGpu} maintains \emph{defrag[T]} bitmaps (one per type) in which a bit is set if the corresponding block is a defragmentation candidate. Every defragmentation candidate is by definition also an active block, because it is not entirely full. Similar to other block state bitmaps, defragmentation candidate bitmaps may have to be updated whenever block fill levels change. An alternative implementation could scan the object allocation bitmaps of every block and generate \emph{defrag[T]} on demand.

\SetKwComment{Comment}{$\triangleright$\ }{}
\SetAlgoVlined
\begin{algorithm}[t]
\small
 \Repeat(\Comment*[f]{\textsf{Infinite loop if OOM}}){false}{
  bid $\gets$ active[T].\emph{try\_find\_set}(); \hfill \Comment{\textsf{Find and return the position of any set bit.}}
  \If(\Comment*[f]{\textsf{Slow path}}){\emph{bid} = FAIL} {
    bid $\gets$ free.\emph{clear}();   \hfill \Comment{\textsf{Find and clear a set bit atomically, return position.}}
    \emph{initialize\_block}<T>(bid)\;
    allocated[T].\emph{set}(bid)\;
    \fbox{defrag[T].\emph{set}(bid)\;} \\
    active[T].\emph{set}(bid)\;
  }
  alloc $\gets$ heap[bid].\emph{reserve}();  \hfill \Comment{\textsf{Reserve an object slot. See Alg.~\ref{alg:block_allocate_defrag_ext}.}}
  \If{\emph{alloc} $\not=$ FAIL}{
    ptr $\gets$ \emph{make\_pointer}(bid, alloc.slot)\;
    t $\gets$ heap[bid].type\;
    \fbox{\lIf{\emph{alloc.state} = LEQ}{
      defrag[t].\emph{clear}(bid)
    }} \\
    \lIf{\emph{alloc.state} = FULL}{
      active[t].\emph{clear}(bid)
    }
    \lIf{t = T}{
      \Return ptr
    }{
      \emph{deallocate}<t>(ptr); \hfill \Comment{\textsf{Type of block has changed. Rollback.}}
    }
  }
 }
 \caption[\textsc{CompactGpu} Extension: DAllocatorHandle::allocate<T>]{DAllocatorHandle::allocate<T>() : T* \hfill \fbox{GPU}}
 \label{alg:alloc_algo_defrag_extended}
\end{algorithm}

\SetKwComment{Comment}{$\triangleright$\ }{}
\begin{algorithm}[t]
\small
  \vspace{0.05cm}
  \nosemic\nonl \textsf{(same as Algorithm~\ref{alg:block_allocate}, Lines 1--8)}\;
  \vspace{0.1cm}
  \setcounter{AlgoLine}{8}
  \If{{success}}{
    \uIf{popc(\emph{before}) \emph{= 63}}{
      \Return (pos, \emph{FULL})
    }
    \uElseIf(\Comment*[f]{\textsf{E.g., \emph{popc}(before) = 33 for $n=2$}}){\fbox{popc(\emph{before}) \emph{= $64 \cdot \frac{n}{n+1} + 1$}}}{
      \fbox{\Return (pos, \emph{LEQ})}
    }
    \Else{
      \Return (pos, \emph{REGULAR})
    }
  }
\Return \emph{FAIL}\;
 \caption[\textsc{CompactGpu} Extension: Block::reserve]{Block::reserve() : (int, state) \hfill $\triangleright$\ \textsf{Assuming block size 64.} \,\,\,\fbox{GPU}}
 \label{alg:block_allocate_defrag_ext}
\end{algorithm}

\SetKwComment{Comment}{$\triangleright$\ }{}
\begin{algorithm}[t]
\small
  bid $\gets$ \emph{extract\_block}(ptr)\;
  slot $\gets$ \emph{extract\_slot}(ptr)\;
  state $\gets$ heap[bid].\emph{deallocate}(slot)\;
  \uIf(\Comment*[f]{\textsf{Deallocated first object of full block.}}){\emph{state} = FIRST}{
    active[T].\emph{set}(bid)
  }
  \uIf(\Comment*[f]{\textsf{Now $\leq \frac{n}{n+1}$ full}}){\fbox{\emph{state} = LEQ}}{
    \fbox{defrag[T].\emph{set}(bid)\;}\\
  }
  \ElseIf(\Comment*[f]{\textsf{Deallocated last object of block.}}){\emph{state} = EMPTY}{
    \If{invalidate(bid)}{
      t $\gets$ heap[bid].type\;
      active[t].\emph{clear}(bid)\;
      \fbox{defrag[t].\emph{clear}(bid)\;}\\
      allocated[t].\emph{clear}(bid)\;
      free.\emph{set}(bid)\;
    }
  }
 \caption[\textsc{CompactGpu} Extension: DAllocatorHandle::deallocate<T>]{DAllocatorHandle::deallocate<T>(T* ptr) : void \hfill \fbox{GPU}}
 \label{algo:dealloc_a_extended_defrag}
\end{algorithm}

Algorithms~\ref{alg:alloc_algo_defrag_extended}, \ref{alg:block_allocate_defrag_ext} and~\ref{algo:dealloc_a_extended_defrag} are extended versions of Algorithms~\ref{alg:alloc_algo}, \ref{alg:block_allocate} and~\ref{algo:dealloc_a}. They show which parts of object (de)allocation\footnote{We are not taking into account allocation request coalescing or block invalidation here.} we changed to keep track of defragmentation candidates. Recall that the ordering of \emph{if} branches is crucial to avoid deadlocks in Algorithm~\ref{algo:dealloc_a_extended_defrag} (Section~\ref{sec:dynasoar_obj_dealloc_discussion}). Assuming that CUDA schedules branches in the order in which they appear in the source code, the \emph{if} branches must be ordered from \emph{high fill level} to \emph{low fill level}.

\subsection{Defragmentation Pass}
A single defragmentation pass consists of four main steps. Every step is implemented as a CUDA kernel and runs in parallel.

\begin{enumerate}
  \item Copy objects from source to target locations.
  \item Store forwarding pointers in source blocks.
  \item Scan the heap and rewrite pointers to source locations.
  \item Update block state bitmaps.
\end{enumerate}

\noindent In the following sections, we describe each of these steps.

\subsection{Copying Objects}
\textsc{CompactGpu} copies objects from source to target blocks. Those blocks must be defragmentation candidates, to ensure that all objects of a source block fit into the corresponding target blocks.

\paragraph{Choosing Source/Target Blocks}
We utilize the defragmentation candidate bitmap to quickly find and assign target blocks to source blocks (Figure~\ref{fig:compacting_prefix_defrag}). We first generate an \emph{indices} array of size $M$ that contains $i$ at position $i$ if the $i$-th bit is set. Otherwise, we store an \emph{invalid marker}. Now we filter/compact the array to retain only valid values, resulting in array $R$ of size $r$. This \emph{stream compaction}~\cite{IJNC151} is implemented with a parallel prefix sum operation (CUB library~\cite{nvidiaprefixsum}).

\begin{figure}
  \begin{minipage}[c]{0.48755605987\textwidth}
    \centering
    \includegraphics[width=0.906315742\textwidth]{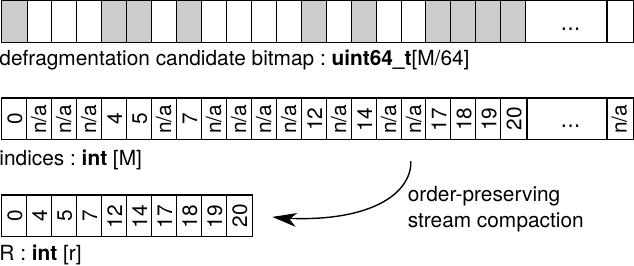}
  \end{minipage}\hfill
  \begin{minipage}[c]{0.47\textwidth}
    \caption[Bitmap compaction with prefix sum (simplified)]{Example: Compacting a bitmap of defragmentation candidates (\emph{defrag[T]}). There is such a bitmap for every type. See Figure~\ref{fig:scan_enumeration} for full details. Must preserve order of indices.} \label{fig:compacting_prefix_defrag}
  \end{minipage}
\end{figure}

\begin{figure}
  \begin{minipage}[c]{0.48755605987\textwidth}
    \includegraphics[width=\textwidth]{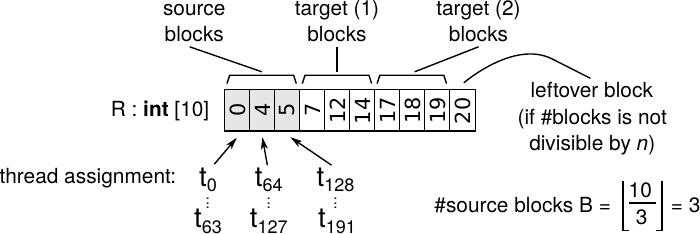}
  \end{minipage}\hfill
  \begin{minipage}[c]{0.47\textwidth}
    \caption[\textsc{CompactGpu}: Assigning source and target blocks]{Example: Assigning source and target blocks ($n=2$). E.g., objects from source block 5 are moved into blocks 14 and 19. One block is left over.} \label{fig:assign_source_target_defrag}
  \end{minipage}
\end{figure}

Based on array $R$, each GPU thread can later by itself (without synchronization) efficiently determine its assigned source block and corresponding target blocks (Figure~\ref{fig:assign_source_target_defrag}). Given a defragmentation factor $n$, the $B=\left\lfloor\frac{r}{n+1}\right\rfloor$ blocks with indices $R[0]$ through $R[B - 1]$ are source blocks. Given a source block $R[\mathit{s\_rid}]$, its corresponding target blocks are:

\begin{align*}
\big\{\,R[\mathit{s\_rid} + i \cdot B] \,\,\,\big|\,\,\, i \in 1...n \,\big\} \tag{\emph{target block IDs}}
\end{align*}

Note that the first $B$ blocks in $R$ are source blocks. This fact will be used later to optimize pointer rewriting.

\paragraph{Copying Objects}
Objects are copied in parallel. We assign 64 consecutive threads to every source block (Figure~\ref{fig:assign_source_target_defrag}) and every thread copies at most one object. Some threads will have no work to do, because the capacity of the block may be less than 64 and because not all slots in a source block are occupied. However, by assigning 64 threads, we are assigning exactly two full warps, which is reduces thread divergence.

\begin{figure}
  \begin{minipage}{0.59375\columnwidth}
    \centering
    \includegraphics[width=1.21879627077\columnwidth]{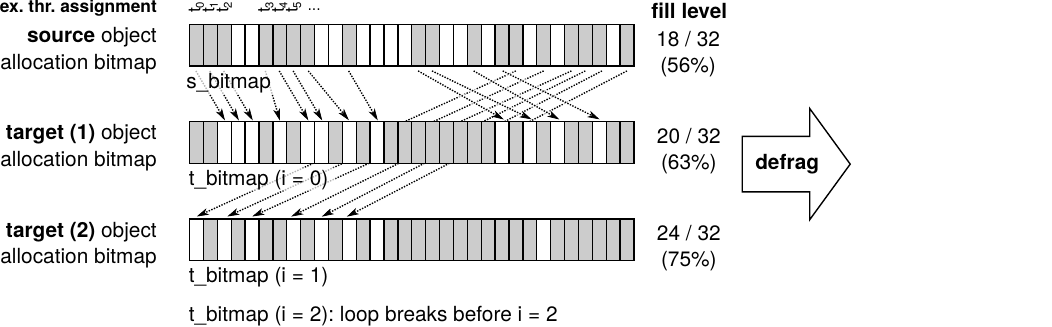}
    \caption*{\footnotesize \textbf{(a)} before relocation}   
  \end{minipage}
  \begin{minipage}{0.37867314781\columnwidth}
    \centering
    \includegraphics[width=\columnwidth]{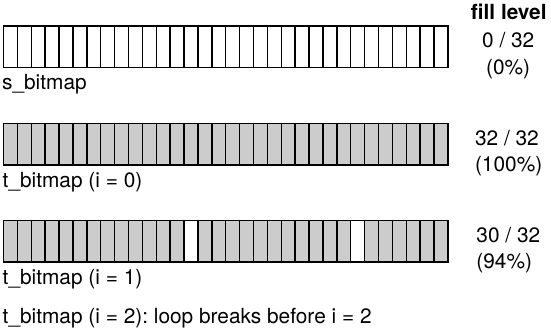}
    \caption*{\footnotesize \textbf{(b)} after relocation} 
  \end{minipage}\\
  \begin{minipage}{\columnwidth}
  \vspace{0.5cm}
  \caption[\textsc{CompactGpu}: Relocating objects]{Example: Relocating objects ($n=3$). Assuming block size 32 instead of 64. All 18 objects fit into the first two selected target blocks, so no third block is needed. The first 12 objects fit into the first target block. The remaining 6 objects fit into the second one.}
  \label{fig:move_objects}
  \end{minipage}
\end{figure}

Algorithm~\ref{alg:move_objects} shows how objects are copied (Figure~\ref{fig:move_objects}). There are 64 threads for every source block. A thread $t_{\mathit{tid}}$ copies the ($\textsf{s\_loc} = \mathit{tid}\,\,\%\,\, 64$)-th object of the source block (ID \textsf{s\_bid}). Let \textsf{s\_oid} be the slot ID of this object. The target slot is the \textsf{s\_loc}-th free slot among all target blocks. Let \textsf{t\_oid} and \textsf{t\_bid} be the slot ID and block ID of that slot. To determine the target slot, we may have to examine the object allocation bitmap of multiple or all $n$ target blocks (Line~7). This causes some \emph{thread divergence} because the number of \emph{for} loop iterations differs among threads, but $n$ is usually small. Furthermore, note that no synchronization is required among threads.

On GPUs, memory accesses have a much higher latency than arithmetic instructions~\cite{Volkov:EECS-2016-143}. Therefore, we have to keep the number of extra accesses in addition to the field copies low. Object copies cannot be avoided without changing the defragmentation strategy and these extra accesses represent the overhead of our copy phase implementation. Since all threads in a warp copy from/to the same blocks, we require at most $1+n$ read transactions from the array $R$ and the same number read transactions of object allocation bitmaps per warp. Since all threads in a warp have the same source/target blocks, they access the same array slot of $R$ and the same object allocation bitmaps, so these accesses can be coalesced. 

\paragraph{Better Source/Target Choices?}
The number of object copies (and pointer rewritings) could be reduced by selecting less full defragmentation candidates as source blocks. Such an optimization does not pay off for three reasons.

First, selecting source blocks becomes much more difficult. How would a GPU thread know which block is less full? We would either have to sort the array $R$ with a comparator function that counts the set bits in each block's object allocation bitmap (a random memory access!). Or we would have to maintain additional \emph{defrag[T]} bitmaps for various fill levels. Both variants would greatly reduce performance.

Second, since memory is accessed in 128-byte vector transactions, reading/writing a slightly lower number of scalar values (that are likely scattered within an SOA array) is unlikely to reduce the number of memory transactions.

And third, there would be more threads without work, but since all threads in a warp must execute the same instructions, these threads nevertheless have to \emph{wait} for the copying threads in the warp (\emph{warp divergence}).

\LinesNumbered
\SetAlgoVlined
\SetKwComment{Comment}{$\triangleright$\ }{}
\begin{algorithm}[t]
\small
\For(\textcolor{black}{\Comment*[f]{\textsf{\fbox{GPU} CUDA kernel}}}){$\mathit{tid} \gets0$ \KwTo $64 \cdot B$  \emph{\textbf{in parallel}}}{
  s\_rid $\gets$ tid / 64;  \,\,  s\_bid $\gets R[\mathit{s\_rid}]$; \hfill \textcolor{black}{\Comment{\textsf{Source block ID}}}
  s\_loc $\gets$ tid $\%$ 64; \hfill \textcolor{black}{\Comment{\textsf{This thread copies s\_loc\textsuperscript{th} obj.}}}
  s\_bitmap $\gets$ heap[s\_bid].bitmap\;
  \uIf{s\_loc $<$ popc(s\_bitmap)}{
    t\_loc $\gets$ s\_loc; \hfill \textcolor{black}{\Comment{\textsf{This thread copies to t\_loc\textsuperscript{th} free slot.}}}
    \For(\textcolor{black}{\Comment*[f]{\textsf{Thr. divergence increases with $n$.}}}){$i \gets 0$ \KwTo $n$}{
      t\_rid $\gets$ s\_rid + $i \cdot B$\;
      t\_bid $\gets R[\mathit{t\_rid}]$; \hfill \textcolor{black}{\Comment{\textsf{Target block ID}}}
      t\_bitmap $\gets$ $\sim$heap[t\_bid].bitmap\;
      t\_slots $\gets$ \emph{popc}(t\_bitmap)\;
      \eIf{t\_loc $<$ t\_slots}{
        \textbf{break}; \hfill \textcolor{black}{\Comment{\textsf{Target block t\_bid determined.}}}
      }{
        t\_loc $\gets$ t\_loc $-$ t\_slots\;
      }
    }
    s\_oid $\gets$ \emph{nth\_set\_bit}(s\_bitmap, s\_loc)\;
    t\_oid $\gets$ \emph{nth\_set\_bit}(t\_bitmap, t\_loc)\;
    s\_ptr $\gets$ \emph{make\_pointer}(s\_bid, s\_oid)\;
    t\_ptr $\gets$ \emph{make\_pointer}(t\_bid, t\_oid)\;
    *t\_ptr $\gets$ *s\_ptr; \hfill \textcolor{black}{\Comment*[f]{\textsf{Copy all fields.}}}
  }
  \textbf{end} \hfill \textcolor{black}{\Comment*[f]{\textsf{else: No work for this thread.}}}
}
 \caption[\textsc{CompactGpu}: move\_objects<T>]{move\_objects<T>() : void \hfill \fbox{CPU}}
 \label{alg:move_objects} \normalstyle
\end{algorithm}

\LinesNumbered
\SetKwComment{Comment}{$\triangleright$\ }{}
\begin{algorithm}[t]
\small
\For(\textcolor{black}{\Comment*[f]{\textsf{\fbox{GPU} CUDA kernel}}}){$\mathit{tid} \gets0$ \KwTo $64 \cdot B$  \emph{\textbf{in parallel}}}{
  \nosemic\nonl \textsf{(same as in \emph{move\_objects}<T>(), lines 2--19)}\;
  \setcounter{AlgoLine}{19}\mbox{\,\,\,}|\,\,\,\,\,\,heap[s\_bid].data.forwarding\_ptr[s\_oid] $\gets$ t\_ptr; \;
  \textbf{end} \hfill \textcolor{black}{\Comment*[f]{\textsf{else: No work for this thread.}}}
}
 \caption[\textsc{CompactGpu}: place\_forwarding\_ptrs<T>]{place\_forwarding\_ptrs<T>() : void \hfill \fbox{CPU}}
 \label{alg:place_fwd_ptr} \normalstyle
\end{algorithm}

\LinesNumbered
\SetKwComment{Comment}{$\triangleright$\ }{}
\begin{algorithm}[t]
\small
  s\_bid $\gets$ \emph{extract\_block\_id}(ptr)\;
  \eIf{s\_bid < $R[B]$ $\wedge$ defrag[T][s\_bid]}{
    s\_oid $\gets$ \emph{extract\_object\_id}(ptr)\;
    \Return heap[s\_bid].data.forwarding\_ptr[s\_oid]\;
  }{
    \Return n/a\;
  }
 \caption[\textsc{CompactGpu}: rewrite\_pointer<T>]{rewrite\_pointer<T>(T* ptr) : T* \hfill \fbox{GPU}}
 \label{alg:rewrite_ptr}
\end{algorithm}

\subsection{Storing Forwarding Pointers}
After copying objects, pointers to the old memory location must be updated (\emph{rewritten}). Most memory defragmentation systems do this with \emph{forwarding pointers}: A pointer to the object's new memory location is stored at its old location.

We extended \textsc{DynaSOAr} to store forwarding pointers inside blocks. Every block may contain either a data segment or forwarding pointers. Listing~\ref{lst:block_str_defrag_fish} shows the data structure of a block of type \texttt{Fish} (also see lower left part of Figure~\ref{fig:heap_layout_defrag}).

\begin{lstfloat}
\begin{lstlisting}[caption={[Example: \textsc{CompactGpu} block structure for class of \textsf{wa-tor}]Example: Block structure for class \texttt{Fish}}, language=c++, label={lst:block_str_defrag_fish}, morekeywords={uint64_t}]
template<> struct Block<Fish> {
  uint64_t bitmap;

  union {
    struct {
      Cell* position[64];
      Cell* new_position[64];
      int random_state[64];
      int age[64];
      float spawn_probability[64];
    } data_segment;  /* SOA arrays */

    Fish* forwarding_ptr[64];
  } data;
};
\end{lstlisting}
\end{lstfloat}

Algorithm~\ref{alg:place_fwd_ptr} shows how the forwarding pointers array is populated. This algorithm is identical to Algorithm~\ref{alg:move_objects}, except for Line~20. These two algorithms cannot be merged into a single kernel because a forwarding pointer would then overwrite parts of the data segment\footnote{If the size of the first field is 8 bytes (same as a pointer), as in the example here, this is harmless. Otherwise, a thread would store a forwarding pointer into a memory location that is copied by another thread (race condition).}.

\subsection{Rewriting Pointers}
\label{sec:rewriting_ptrs}
The heap is now scanned for pointers that must be rewritten. If a pointer points to a location with a forwarding pointer, it is replaced with the forwarding pointer. We utilize the defragmentation candidate bitmap to quickly decide if a pointer must be rewritten, without reading the memory at the pointer location. \emph{We read that location only if we are sure that it contains a forwarding pointer.} This is a key difference compared to other defragmentation systems.

Given a pointer \textsf{ptr}, Algorithm~\ref{alg:rewrite_ptr} returns the corresponding forwarding pointer or \emph{n/a} if no forwarding pointer exists for \textsf{ptr}. We first extract the block ID \textsf{s\_bid} of the object that \textsf{ptr} points to. This block is a source block if it is a defragmentation candidate (i.e., bit set in the defragmentation candidate bitmap) and if the block ID is smaller than $R[B]$ (Line~2). Recall that the blocks with IDs $R[i]$ with $i \in [0; B-1]$ are source blocks and $R$ is sorted (Figure~\ref{fig:assign_source_target_defrag}). Large parts of the \emph{defrag[T]} bitmap will likely be cached by the L1/L2 caches, so we expect these bitmap lookups to be fast.

If the block is a source block, we extract the object ID \textsf{s\_oid} from \textsf{ptr} and return the corresponding forwarding pointer. It is crucial that Algorithm~\ref{alg:rewrite_ptr} is efficient because it is executed for every pointer that is found during a heap scan.

\paragraph{Limiting Heap Scans}
Recent GPUs have up to 32~GB of memory, so scanning all allocated memory for pointers to rewrite is expensive. However, since classes and fields of application code are defined with \textsc{DynaSOAr}'s DSL, \textsc{CompactGpu} can reflect on application classes and determine which parts of the heap may contain pointers that must be rewritten. As such, only a small part of the heap is scanned.

For example, when defragmenting \texttt{Fish} objects, we can avoid looking into blocks of type \texttt{Fish} or \texttt{Shark}, because those classes do not have fields of type \emph{pointer to \texttt{Fish}} or \emph{pointer to a superclass of \texttt{Fish}}. Only class \texttt{Cell} has a field of type \texttt{Agent*}, so we only scan the corresponding SOA arrays in the data segment of allocated blocks of type \texttt{Cell}. These are the pointers that are rewritten according to Algorithm~\ref{alg:rewrite_ptr}.

In general, when defragmenting objects of type $T$, we first determine all classes $U$ that have at least one field of type \emph{pointer to S}, where $S :> T$ is a supertype of $T$ or equal to $T$\footnote{We also consider fields of type \emph{array of pointer to S} etc. For simplicity, we mention only simple pointer types here.}. For every such type $U$, we scan allocated blocks of type $U$. For every allocated block, we scan the SOA arrays of fields of type \emph{pointer to S}. Only these pointers are scanned and rewritten. This can exclude more than 95\% of the heap. Moreover, reading the values from an SOA array is fast because those field reads are coalesced by the GPU.

Most other allocators have limited information about the structure of their allocations. If we were to implement \textsc{CompactGpu}'s technique in such allocators, we cannot restrict the search space as described here. Previous work describes alternative techniques for limiting the search space based on intermediate results from a garbage collector~\cite{Veldema:2012:PMD:2247684.2247693}.

\subsection{Updating Block State Bitmaps}
\label{sec:compactgpu_updating_block_st}
Finally, the state of source and target blocks must be updated in the corresponding block state bitmaps. Source blocks lose their \emph{active}, \emph{allocated} and \emph{defrag} states and become \emph{free}. Target blocks may lose their \emph{active} and/or \emph{defrag} states depending on their new fill level.

Moreover, the object allocation bitmaps of all target blocks must be updated. If $m$ objects were relocated to a given target block, one thread flips the $m$ first cleared bits to 1, using an algorithm that is similar to Algorithm~\ref{alg:nth_set_bit_al}.

\subsection{Multiple Defragmentation Passes}
A single defragmentation pass is guaranteed to delete all source blocks, i.e., $\frac{1}{n+1}$ of all defragmentation candidates. In addition, some target blocks may lose their defragmentation candidate state. However, a fragmentation level of $\frac{1}{n+1}$ can be achieved only if all candidates were eliminated (Section~\ref{sec:memory_defrag}). This may require multiple defragmentation passes.

The efficiency of defragmentation passes decreases with a decreasing number of defragmentation candidates. For example, for $n=1$, a single pass is guaranteed to eliminate 500 out of 1,000 total candidates. However, the next pass is only guaranteed to eliminate 250 out of 500 remaining candidates. Moreover, too much defragmentation can make allocations more expensive because \textsc{DynaSOAr} has to (re)initialize new blocks if there are not enough active blocks. To reduce runtime overheads, defragmentation should stop before eliminating all defragmentation candidates.

\textsc{CompactGpu} runs multiple defragmentation passes until all but $k_1$ defragmentation candidates were eliminated. The value of $k_1$ can be configured. Since defragmentation passes are very fast, the value of $k_1$ matters only in cases with a large number defragmentation passes or a massive number of (de)allocations.

\paragraph{Worst-case Analysis}
Let $d$ be the number of defragmentation candidates. A single defragmentation pass reduces $d$ at least by a fraction of $\frac{1}{n+1}$, so no more than $\frac{n}{n+1}$ of defragmentation candidates are left over. We can bound the number of defragmentation passes that are necessary in the worst case to eliminate all defragmentation candidates by:

\begin{align*}
\log_{\frac{n + 1}{n}} d  \tag{\emph{worst-case number of defrag. passes}}
\end{align*}

If all but $k_1 \geq 1$ defragmentation candidates should be eliminated, we can bound the number of defragmentation passes by:

\begin{align*}
\log_{\frac{n + 1}{n}} d - \log_{\frac{n + 1}{n}} k_1 = \log_{\frac{n + 1}{n}} \frac{d}{k_1} \tag{\emph{keep $k_1$ defrag. candidates}}
\end{align*}

In reality, the number of required defragmentation passes is usually much lower. We experimentally analyze the actual number of defragmentation passes in Section~\ref{sec:number_of_defrag_passes_eval}.

\subsection{Defragmentation Frequency}
Memory defragmentation must be initiated by the programmer explicitly. Once initiated, \textsc{CompactGpu} may run multiple defragmentation passes depending on $n$, $k_1$ and the number of defragmentation candidates. We suggest one of the two following defragmentation policies for initiating defragmentation.

\begin{description}
  \item[Every $m$ Iterations] Many GPU programs run a number of CUDA kernels iteratively in a loop. This policy initiates defragmentation every $m$ iterations.
  \item[After Massive Deallocations] Initiate defragmentation if there are at least $k_2$ many defragmentation candidates\footnote{There is no global object counter. The number of defragmentation candidates approximates the number of allocated but unused object slots.}, where $k_2$ should be a large enough value. \textsc{CompactGpu} provides a helper function that lets programmers specify this threshold as an absolute number or as a percentage of the heap size and then initiates defragmentation if necessary. Internally, \textsc{CompactGpu} scales $k_2$ by $\frac{n}{n+1}$ to account for the fact that a larger value of $n$ usually leads to more defragmentation candidates.
\end{description}

As a rule of thumb, we use the first policy for applications that experience a speedup from defragmentation, because even small compactions can lead to a performance gain. The second policy is useful for applications that  mainly benefit from better space efficiency or see a slowdown from defragmentation. Future work will investigate how to automate defragmentation (choosing policies, parameters, etc.).

\section{Pointer Rewriting Alternatives}
\label{sec:alternative_designs_sec}
Pointer rewriting is the most time-consuming step in applications with a large object set. In Algorithm~\ref{alg:rewrite_ptr}, reading the forwarding pointer in Line~4 is a random memory access that cannot be coalesced and thus the most expensive operation of the algorithm. To get rid of this memory access, we implemented two alternatives that recompute forwarding pointers on-the-fly.

\begin{figure}
  \centering
  \includegraphics[width=0.9\columnwidth]{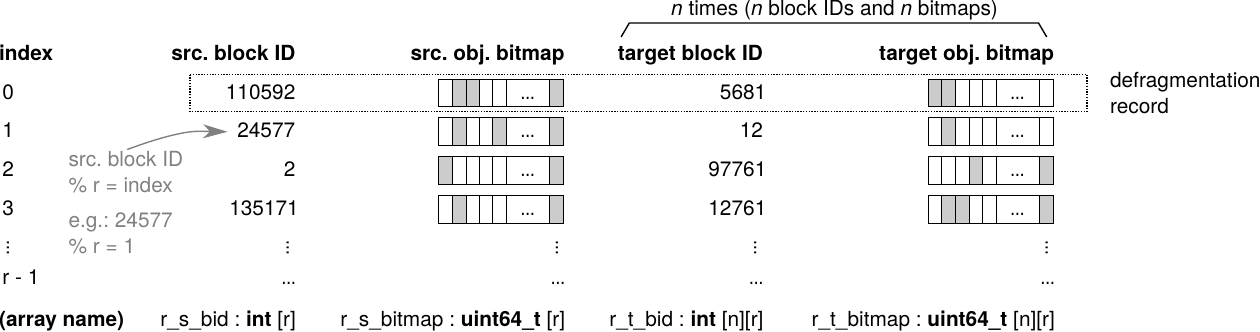}
  \caption[\textsc{CompactGpu}: Defrag. records for pointer rewriting alternatives]{Example: Defragmentation records. Stored in SOA layout.}
  \label{fig:defrag_records}
\end{figure}

\subsection[\textsc{Recompute-Global}]{\textsc{Recompute-Global}: Recompute Forwarding Pointers}
Instead of storing forwarding pointers in objects, we maintain \emph{defragmentation records} (Figure~\ref{fig:defrag_records}). A defragmentation record is a tuple of source block ID, source object bitmap (i.e., object allocation bitmap), target block IDs and target object bitmaps. Based on a defragmentation record, all forwarding pointers for objects from the respective source block can be recomputed. Defragmentation records are stored in SOA layout (4 arrays for $n=1$).

Source block IDs are stored in shared memory\footnote{Shared memory is an explicitly programmable part of the L1 cache.} and the remaining 3 arrays are stored in global memory. Current NVIDIA GPUs have 48~KB of shared memory, so we can have at most $r =$ 48~KB / \textsf{sizeof(\textbf{int})} = 12288 defragmentation records per defragmentation pass.

\begin{align*}
r = \frac{48 \cdot 1024 \mbox{ bytes}}{4 \mbox{ bytes}} = 12288 \tag{\emph{number of defrag. record slots}}
\end{align*}

\paragraph{Choosing Source/Target Blocks}
During source/target block selection, we hash as many defragmentation records to defragmentation record array slots as possible and proceed with object copying. The defragmentation records data structure is effectively a hash table with the source block ID as key.

\begin{align*}
h(\mathit{s\_bid}) = \mathit{s\_bid} \,\,\,\%\,\,\, r \tag{\emph{hash function}}
\end{align*}

The main challenge of selecting source/target blocks is to avoid hash collisions. Only a few defragmentation candidates have a suitable block ID \emph{s\_bid}, such that the block can be a source block stored in a given defragmentation record slot \emph{rid}: Namely, those blocks whose hash code $h(\mathit{s\_bid})$ is \emph{rid}. If all of those blocks are assigned as target blocks, then this defragmentation record slot remains unused. Too many unused defragmentation record slots increase the number of required defragmentation passes.

\LinesNumbered
\SetAlgoVlined
\SetKw{KwBy}{by}
\SetKwComment{Comment}{$\triangleright$\ }{}
\begin{algorithm}[t]
\small
max\_blocks $\gets \left\lfloor\frac{d}{n+1}\right\rfloor$; \hfill \Comment{\textsf{Leave enough blocks for target blocks.}}
\For(\Comment*[f]{\textsf{\fbox{GPU} CUDA kernel}}){$\mathit{tid} \gets0$ \KwTo $r$  \emph{\textbf{in parallel}}}{
  s\_bid $\gets$ \emph{n/a}\;
  \For(\Comment*[f]{\textsf{Invariant: bid \% r = tid}}){$\mathit{bid} \gets \mathit{tid}$ \KwTo $m$ \KwBy $r$}{
    \If{\emph{defrag[T][bid]}}{
      s\_bid $\gets$ bid; \hfill \Comment{\textsf{Choose this block as source.}}
      \textbf{break}\;
    }
  }
  \eIf{\emph{s\_bid $\not=$ \emph{n/a}}}{
    \eIf{atomicSub\emph{(\&max\_blocks, 1) > 0}}{
      r\_s\_bid[tid] $\gets$ s\_bid\;
      r\_s\_bitmap[tid] $\gets$ heap[s\_bid].bitmap\;
      defrag[T].\emph{clear}(s\_bid); \hfill \Comment{\textsf{Block cannot be a target.}}
    }{
      r\_s\_bid[tid] $\gets$ \emph{n/a}\;
    }
  }{
      r\_s\_bid[tid] $\gets$ \emph{n/a}\;
    }
}
 \caption[\textsc{CompactGpu} Ptr. Rewr. Alternative: choose\_source\_blocks<T>]{choose\_source\_blocks<T>() : void \hfill \fbox{CPU}}
 \label{alg:alternative_choose_src_blocks}
\end{algorithm}

\LinesNumbered
\SetAlgoVlined
\SetKw{KwBy}{by}
\SetKwComment{Comment}{$\triangleright$\ }{}
\begin{algorithm}[t]
\small
\For(\Comment*[f]{\textsf{\fbox{GPU} CUDA kernel}}){$\mathit{tid} \gets0$ \KwTo $r$  \emph{\textbf{in parallel}}}{
  \If(\Comment*[f]{\textsf{\textbf{else:} Slot not in use.}}){\emph{r\_s\_bid[tid] $\not=$ \emph{n/a}}}{
    \For(\Comment*[f]{\textsf{Choose $n$ target blocks.}}){$i \gets 0$ \KwTo $n$}{
      r\_t\_bid[i][tid] $\gets$ defrag[T].\emph{clear}()\;
      r\_t\_bitmap[tid] $\gets$ heap[r\_t\_bid[i][tid]].bitmap\;
    }
  }
}
 \caption[\textsc{CompactGpu} Ptr. Rewr. Alternative: choose\_target\_blocks<T>]{choose\_target\_blocks<T>() : void \hfill \fbox{CPU}}
 \label{alg:alternative_choose_target_blocks}
\end{algorithm}

To utilize as many defragmentation record slots as possible and to make best use of the limited amount of shared memory, we first assign source blocks and later target blocks. Algorithm~\ref{alg:alternative_choose_src_blocks} shows how source blocks are assigned. We run a CUDA kernel with one GPU thread for each defragmentation record array slot \emph{tid} (thread ID). Each thread iteratively checks in the defragmentation candidate bitmap the bits of all blocks whose hash code equal the thread's assigned array slot. The algorithm chooses the first suitable block (Line~6). We should ensure that no more than $\frac{d}{n+1}$ are selected, where $d$ is the number of defragmentation candidates; otherwise, there would not be enough target blocks for some source blocks. To that end, we maintain an atomic counter (Line~9). Finally, after storing the source block ID and bitmap in the defragmentation records data structure, we clear the block's defragmentation candiate bit\footnote{In contrast to our original forwarding pointer technique, we clear this bit during source block selection instead of in a separate step (Section~\ref{sec:compactgpu_updating_block_st}).}.

Finally, we assign target blocks (Algorithm~\ref{alg:alternative_choose_target_blocks}). We assign $n$ blocks to each source block by atomically finding and clearing a set bit in the defragmentation candidate bitmap. If this block is still a defragmentation candidate after this defragmentation pass, the bit must be set again.

\paragraph{Rewriting Pointers}
\LinesNumbered
\SetKwComment{Comment}{$\triangleright$\ }{}
\begin{algorithm}[t]
\small
  s\_bid $\gets$ \emph{extract\_block\_id}(ptr)\;
  \eIf(\Comment*[f]{\textsf{Matching defrag. record in shared memory.}}){\emph{r\_s\_bid[\emph{h}(s\_bid)] = s\_bid}}{
    s\_oid $\gets$ \emph{extract\_object\_id}(ptr)\;
    s\_bitmap $\gets$ r\_s\_bitmap[\emph{h}(s\_bid)]\;
    s\_loc $\gets$ \emph{popc}(((1 {<}{<} s\_oid) - 1) \& s\_bitmap); \hfill \Comment{\textsf{Bit s\_oid is the s\_loc-th set bit.}}
      t\_loc $\gets$ s\_loc\;
      \For{$i \gets 0$ \KwTo $n$}{
        t\_bid $\gets$ r\_t\_bid[i][\emph{h}(s\_bid)]\;
        t\_bitmap $\gets$ $\sim$r\_t\_bitmap[i][\emph{h}(s\_bid)]\;
        t\_slots $\gets$ \emph{popc}(t\_bitmap)\;
        \eIf{\emph{t\_loc $<$ t\_slots}}{
          \textbf{break}; \hfill \Comment{\textsf{Target block t\_bid determined.}}
        }{
          t\_loc $\gets$ t\_loc $-$ t\_slots\;
        }
      }
      t\_oid $\gets$ \emph{nth\_set\_bit}(t\_bitmap, t\_loc)\;
      t\_ptr $\gets$ \emph{make\_pointer}(t\_bid, t\_oid)\;
      \Return t\_ptr\;
  }{
    \Return n/a\;
  }
 \caption[\textsc{CompactGpu} Ptr. Rewr. Alternative: rewrite\_pointer<T>]{rewrite\_pointer<T>(T* ptr) : T* \hfill \fbox{GPU}}
 \label{alg:rewrite_ptr_alternative}
\end{algorithm}

We place no forwarding pointers in blocks. During pointer rewriting (Algorithm~\ref{alg:rewrite_ptr_alternative}), we check in shared memory if there is a defragmentation record for the block of \textsf{ptr} instead of checking the bit in the defragmentation candidate bitmap: We traded a global memory access for a much faster shared memory access. However, this approach limits the number of source blocks per defragmentation pass and, therefore, may increase the number of required defragmentation passes. Furthermore, we now have to recompute the forwarding pointer, which requires additional global memory accesses for reading the remaining defragmentation record values in case \texttt{ptr} must be rewritten. The computation of target pointers in Algorithm~\ref{alg:rewrite_ptr_alternative} is similar to Algorithm~\ref{alg:move_objects} and follows the same notation and variable names.


\subsection[\textsc{Recompute-Shared}]{\textsc{Recompute-Shared}: Defrag. Records in Shared Memory}
To further reduce the number of global memory accesses, we modified \textsc{Recompute-Global} to store the entire defragmentation records data structure in shared memory. The size of a defragmentation record is then $12 \cdot (n+1)$ bytes (4-byte block IDs and 8-byte bitmaps), so the shared memory can hold only 2048 defragmentation records in shared memory for $n=1$ (and even less for larger $n$). This further increases the number of required defragmentation passes, but also reduces the number of global memory accesses during pointer rewriting.

\begin{align*}
r = \frac{48 \cdot 1024 \mbox{ bytes}}{(n+1) \cdot 12 \mbox{ bytes}} \tag{\emph{number of defrag. record slots}}
\end{align*}

\section{Evaluation}
\label{sec:evaluation}
We evaluated \textsc{CompactGpu} with an NVIDIA TITAN Xp GPU (12 GB device memory). We compiled the programs with nvcc (-O3) from the CUDA Toolkit 10.1 on Ubuntu 16.04.4.

\subsection{Defragmentation Quality}
We first investigate how much fragmentation \textsc{CompactGpu} can eliminate. As described in Section~\ref{sec:memory_defrag}, given a defragmentation factor $n$, the fragmentation level is guaranteed to be less than $\frac{1}{n+1}$ after defragmentation. However, in reality the fragmentation level is even lower.

\begin{figure}
  \subfloat[Defrag. quality by initial fragmentation]{\includegraphics[width=0.48\columnwidth]{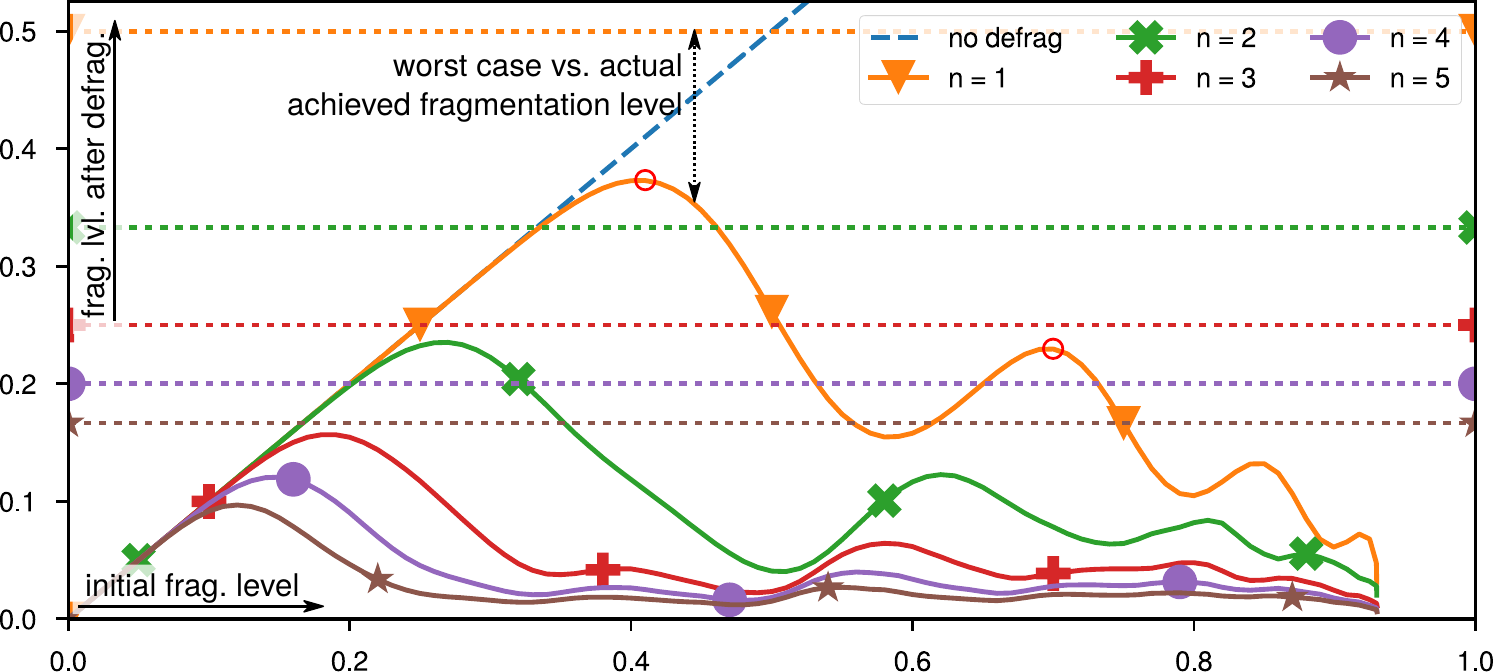}}\hfill
  \subfloat[Number of object relocations]{\includegraphics[width=0.48\columnwidth]{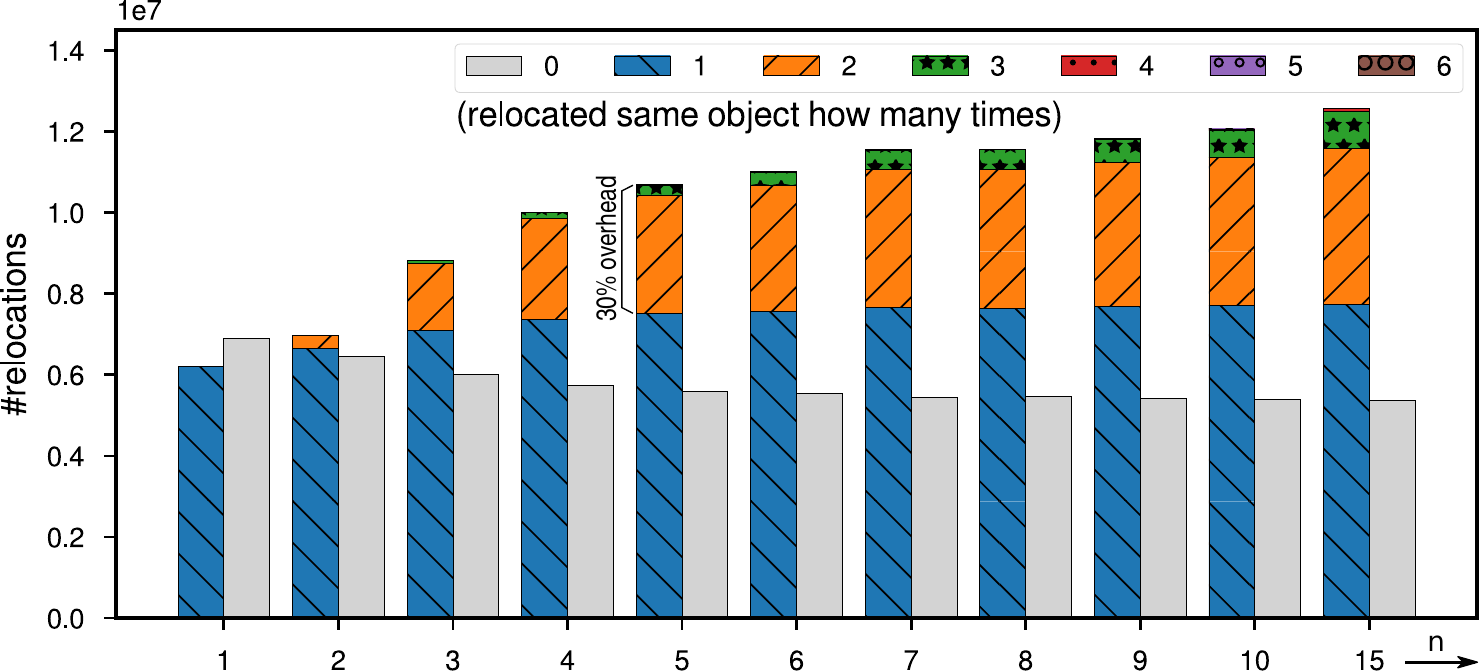}}

  \caption[\textsc{CompactGpu} experiment: Achieved fragmentation level]{Achieved fragmentation level and number of object relocations}
  \label{fig:frag_by_init}
\end{figure}
\begin{figure}
  \includegraphics[width=\textwidth]{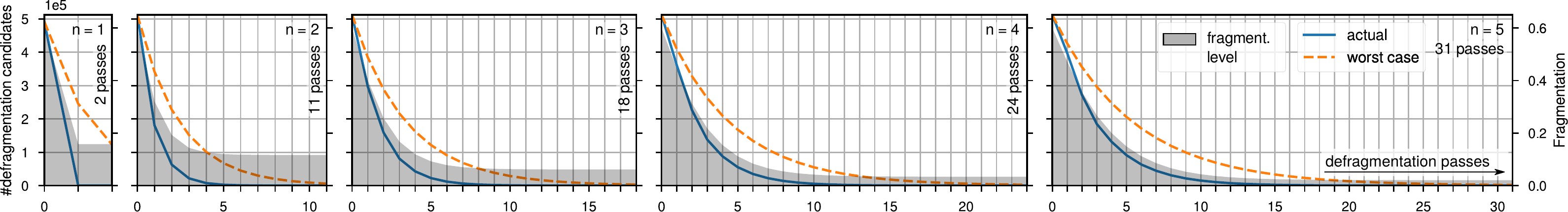}
  \caption[\textsc{CompactGpu} experiment: Number of defragmentation passes]{Number of defragmentation candidates by number of defragmentation passes}
  \label{fig:bench_num_passes}
\end{figure}

We ran a synthetic benchmark that first allocates a very large number of objects and then randomly deallocates some objects. \textsc{CompactGpu} then defragments the heap with $k_1=0$. We measured the fragmentation level after defragmentation for different values of $n$. In Figure~\ref{fig:frag_by_init}\textsc{a}, the x-axis denotes the initial heap fragmentation level (i.e., the percentage of deallocated objects) and the y-axis denotes the fragmentation level after defragmentation. Lower is better. The dotted lines indicate guaranteed fragmentation levels after defragmentation (fragmentation level $\frac{1}{n+1}$).

\textsc{CompactGpu} achieves its worst defragmentation quality at an initial fragmentation level that is slightly smaller than $\frac{n}{n+1}$ (e.g, 45\% for $n=1$). In this case, many blocks are at the boundary of becoming defragmentation candidates. There are a few more points with bad defragmentation quality. For example, around 70\% for $n=1$. In this case, a number of defragmentation candidates were eliminated in the first defragmentation pass, but the resulting blocks have unfortunate fill levels at the boundary of becoming a defragmentation candidate.

\subsection{Number of Defragmentation Passes}
\label{sec:number_of_defrag_passes_eval}
We now investigate the number of defragmentation passes that are necessary to reach good fragmentation levels. The number of defragmentation passes is bounded by $\log_{\frac{n+1}{n}} d$, but in reality fewer passes are needed because some target blocks lose their state as defragmentation candidates.

We ran the same synthetic benchmark, but fixed the initial fragmentation level at 60\%. Figure~\ref{fig:bench_num_passes} shows the number of defragmentation candidates and the fragmentation level after every defragmentation pass. The length of the x-axis indicates the number of defragmentation passes required to eliminate all defragmentation candidates. The y-axis shows the number of remaining defragmentation candidates and the fragmentation level (gray area). Only a few passes bring down fragmentation to very low levels. Moreover, the number of required passes to eliminate all candidates is significantly lower than the theoretical upper bound.

Figure~\ref{fig:frag_by_init}\textsc{b} shows the total number of object relocations for the synthetic benchmark (60\% initial fragmentation level) at various defragmentation factors (x-axis). \textsc{CompactGpu} runs multiple defragmentation passes, so some objects may be relocated multiple times. The stacked bars classify relocations by the number of times an object is relocated (e.g., \emph{0} = object not relocated, \emph{1} = object relocated for the first time, etc.). All bars above ``1'' indicate an overhead of \textsc{CompactGpu} and could potentially be avoided by choosing different source/target blocks or with a different defragmentation strategy.

E.g., for $n=5$, in 30.0\% of all object relocations, an object was copied already for the second time or even more often. Even though 31 defragmentation passes are required for $n=5$, no object was relocated more than 5 times.

\subsection{Benchmark Applications}
\label{sec:benchmark_application}
We evaluated \textsc{CompactGpu} with four SMMO applications (with $k_1=16$). Since dynamic memory allocation is not widely used on GPUs yet, there are no suitable standard benchmark suites. Our benchmarks are a subset of the \textsc{DynaSOAr} benchmarks and exhibit varying allocation patterns.

We measured the defragmentation quality and running time with different defragmentation factors. The defragmentation factor must be smaller than the capacity of a block, so every problem has a different maximum defragmentation factor. The dashed red lines in the running time graphs are baseline running times without defragmentation.

For every application, we also show a memory profile. The shaded area indicates the number of allocated objects (used object slots). Different colors indicate different C++ classes in the application. The lines indicate the actual memory usage (allocated object slots). The gap between the shaded area and a line is memory that is wasted due to fragmentation.

All applications except for \textsf{wa-tor} experience a speedup with memory defragmentation. \textsf{wa-tor} experiences a slowdown, but space savings.

\paragraph{\textsf{collision}}
\begin{figure}
  \includegraphics[width=0.48\columnwidth]{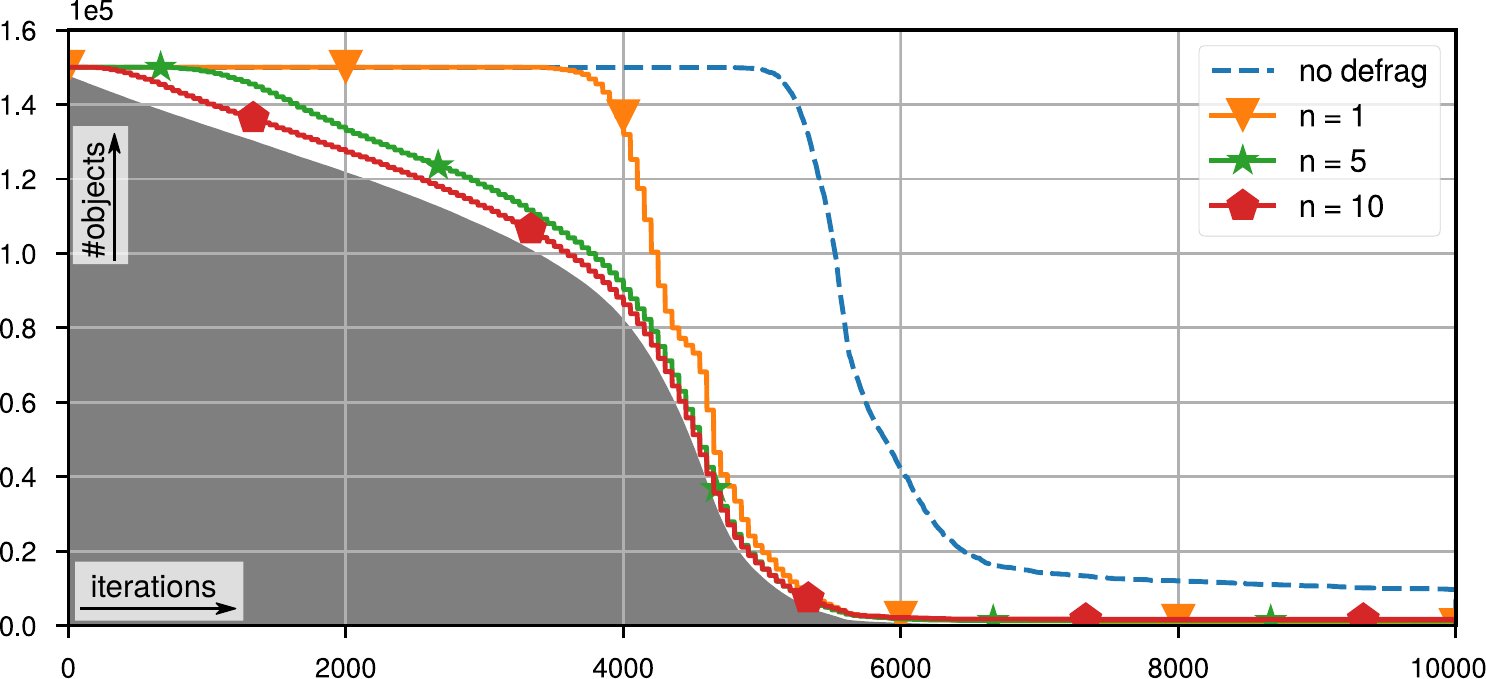}\hfill
  \includegraphics[width=0.48\columnwidth]{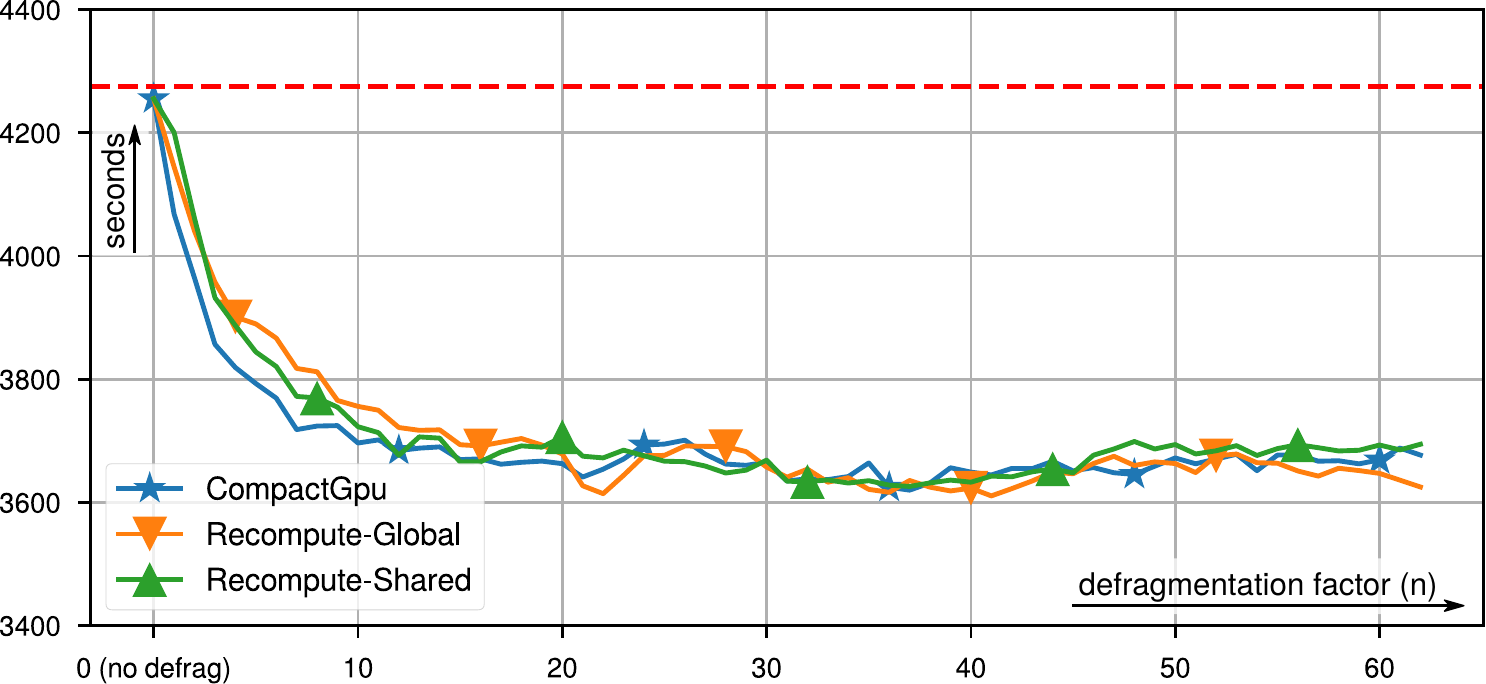}
  \centering
  \caption[\textsc{CompactGpu} benchmark: \textsf{collision}]{\textsf{collision}: N-body simulation with collisions}
  \label{fig:nbody}
\end{figure}

This is an n-body simulation with collisions (Figure~\ref{fig:nbody}). A large number of body objects is allocated at the beginning. No other objects are allocated. When two body objects collide, they are merged and one object is deallocated. The fragmentation level increases gradually with every deallocated body. The worst fragmentation is reached around iteration 5,000, when most objects were already deallocated but most blocks are still allocated due to a few remaining objects in each block.

We initiated defragmentation every 50 iterations. Defragmentation had a very small overhead and led to a performance improvement of 12.2\% for $n=36$. This is because of more efficient vector load/store instructions (more coalescing) and due to better cache utilization. The initial dataset size is 5.7~MB. Towards the end of the simulation, only few objects remain, and if they are stored in a dense way, they fit into GPU caches.

\paragraph{\textsf{structure}}
\begin{figure}
  \includegraphics[width=0.48\columnwidth]{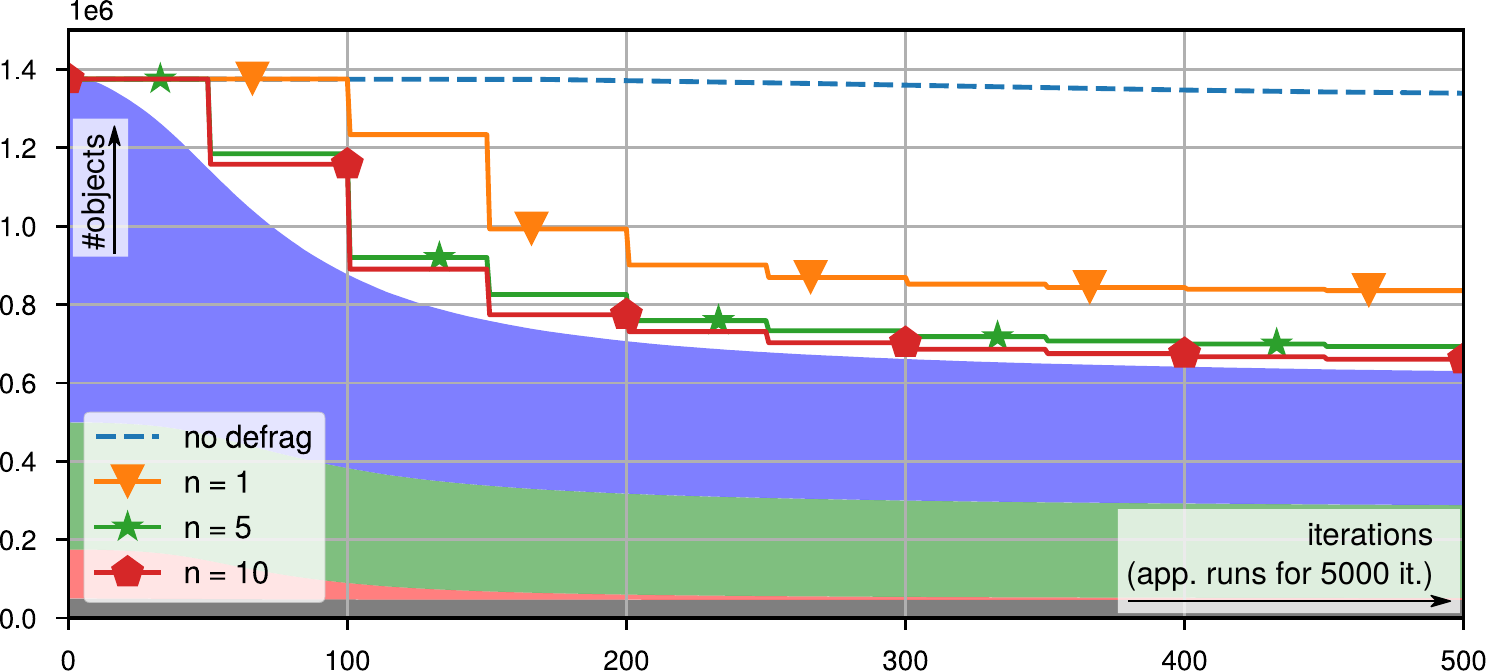}\hfill
  \includegraphics[width=0.48\columnwidth]{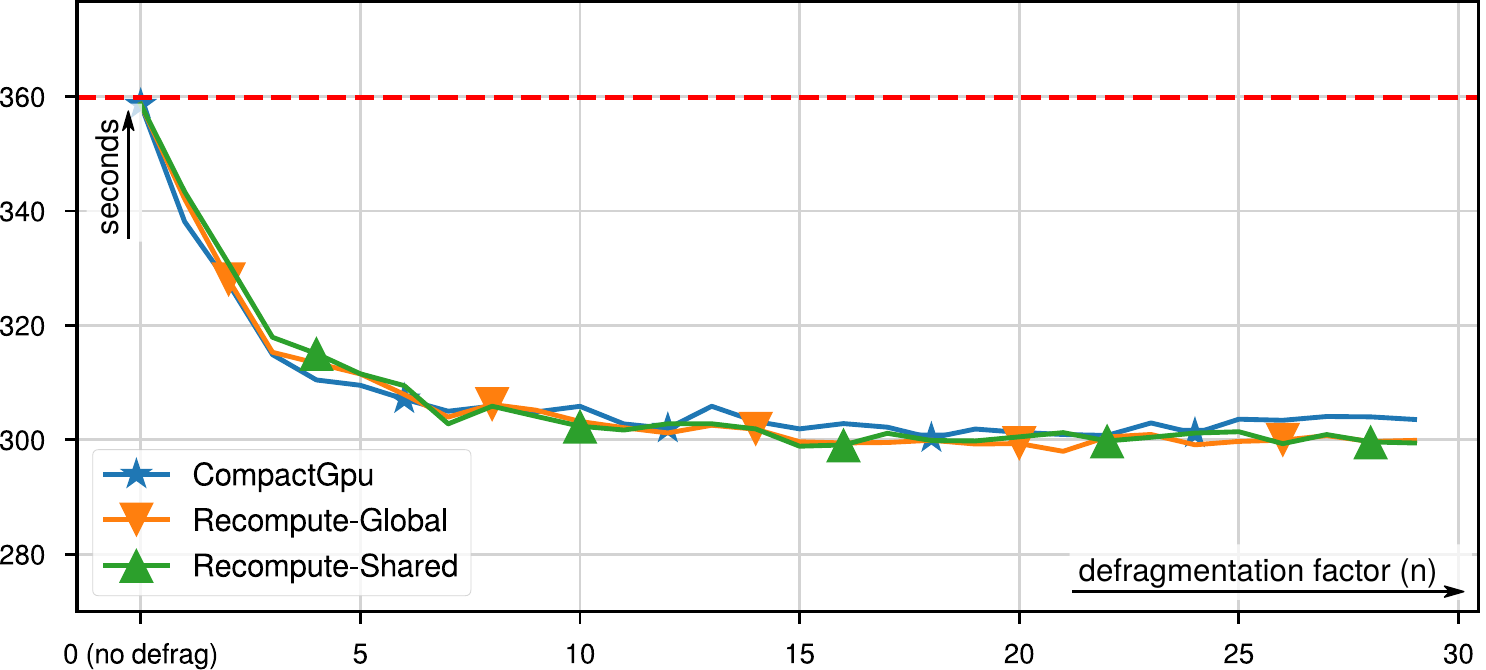}
  \centering
  \caption[\textsc{CompactGpu} benchmark: \textsf{structure}]{\textsf{structure}: Simulation of a fracture in a composite material (FEM)}
  \label{fig:structure}
\end{figure}

This is a simulation of a fracture in a composite material (Figure~\ref{fig:structure}), modeled as a mesh of finite elements. The simulation exerts a force on some elements and connections between two elements break if the force between them exceeds a certain threshold. A BFS pass identifies elements that are disconnected from the remaining simulation and deallocates them.

Similar to \textsf{collision}, this simulation exhibits only deallocations. However, this simulation has four classes. Even though many objects are already deallocated at the end of the simulation, most blocks are still allocated and overall memory consumption has barely decreased.

We initiated defragmentation every 50 iterations. Defragmentation achieved a peak speedup of 16.3\% for $n=18$.

\paragraph{\textsf{generation}}
\begin{figure}
  \includegraphics[width=0.48\columnwidth]{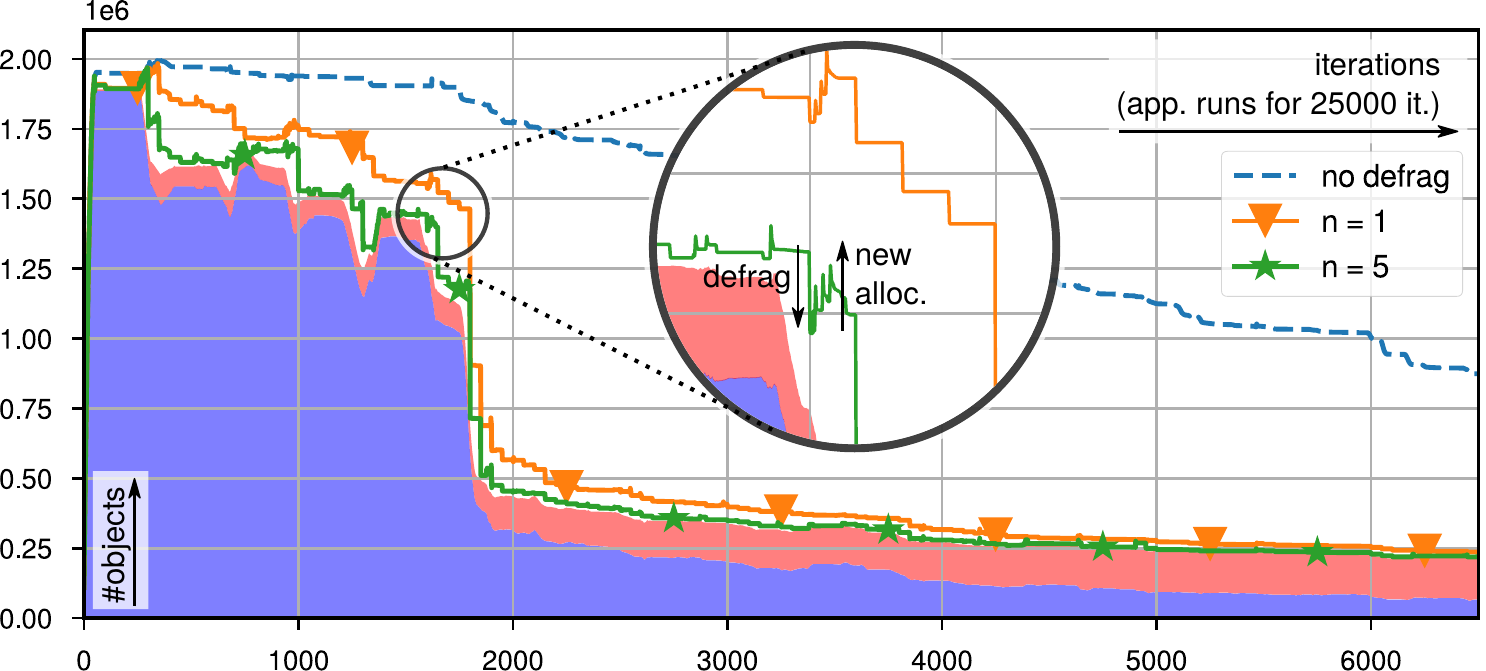}\hfill
  \includegraphics[width=0.48\columnwidth]{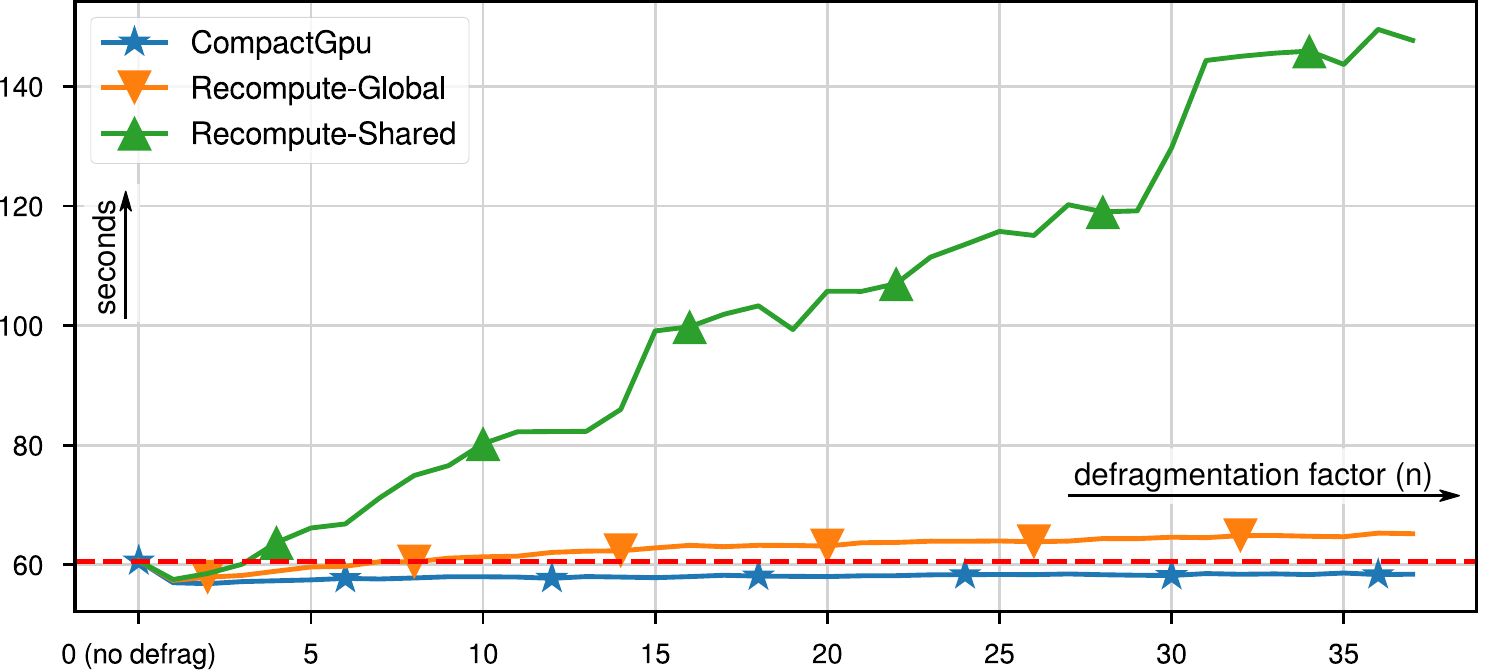}
  \centering
  \caption[\textsc{CompactGpu} benchmark: \textsf{generation}]{\textsf{generation}: Generational cellular automaton}
  \label{fig:genca}
\end{figure}

This is an adaptation of Game of Life (rule 0235678/3468/255, similar to ``Burst''~\cite{gol_burst}; Figure~\ref{fig:genca}). After a cell dies, it stays around for a few more iterations and blocks the cell. This implementation simulates only alive cells by allocating objects for alive cells and cells that may become alive in the next iteration.

This simulation contains both allocation and deallocation of objects. After iteration 2,000, most cells are dead and the simulation converges into a mostly static pattern.

We initiated defragmentation every 50 iterations for a speedup of 6.3\% ($n=2$). Higher defragmentation factors led to \emph{overfitting}: E.g., the line for $n=5$ follows the number of allocated objects very closely, even into small local minima. This does not give any additional performance benefit.

We chose $k_1 = 16$, i.e., 16 defragmentation candidates are excluded from defragmentation. It is important to retain a few fragmented (non-full) blocks because \textsc{DynaSOAr} first looks for active (non-full) blocks during allocations (\emph{fast path}) and has to initialize a new block if none were found (\emph{slow path}). Too small values of $k_1$ led to \emph{overcompaction}: Consider the enlarged part of the memory profile in Figure~\ref{fig:genca}. At first, defragmentation lowers the overall memory usage. However, allocations in the next few iterations immediately increase the fragmentation level again due to new block initializations, bringing it almost back to the initial fragmentation level. A higher value of $k_1$ would likely speed up allocations and increase the performance of the overall application a little bit.

\paragraph{\textsf{wa-tor}}
\begin{figure}
  \includegraphics[width=0.32\textwidth]{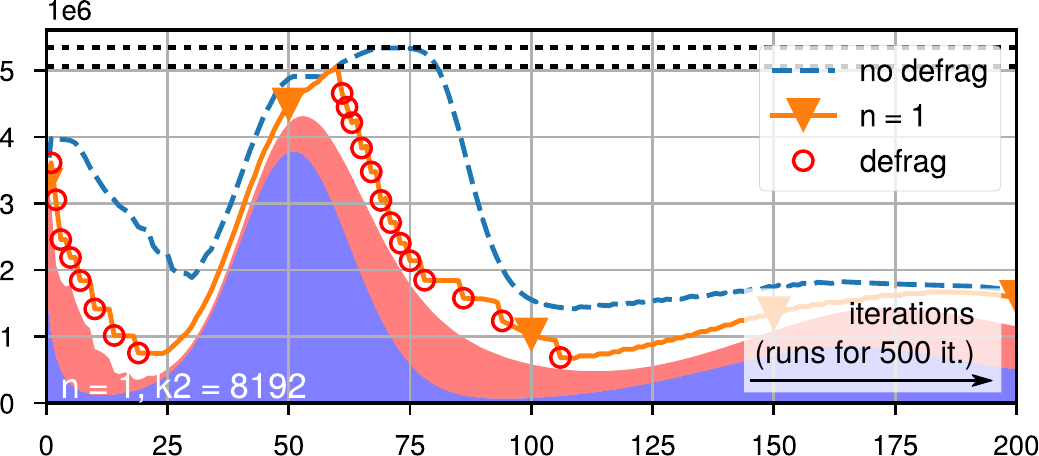}\hfill
  \includegraphics[width=0.32\textwidth]{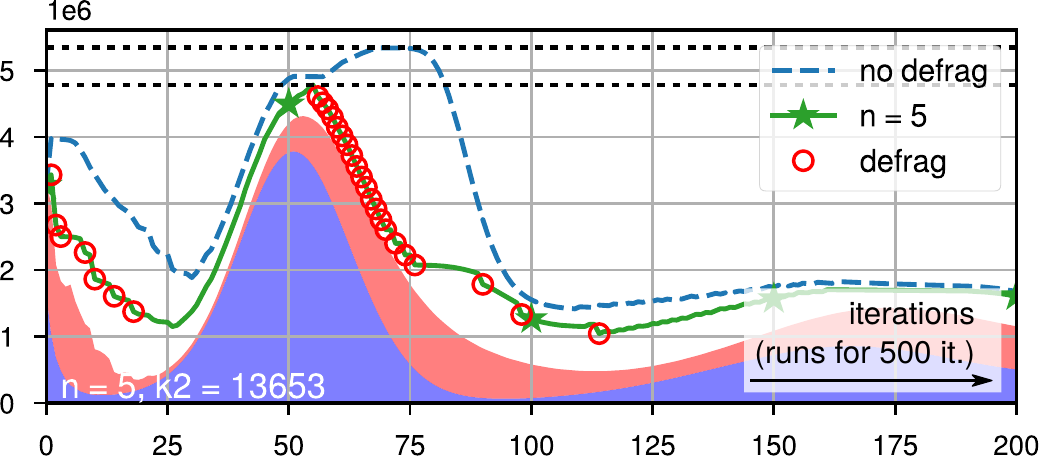}\hfill
  \includegraphics[width=0.32\textwidth]{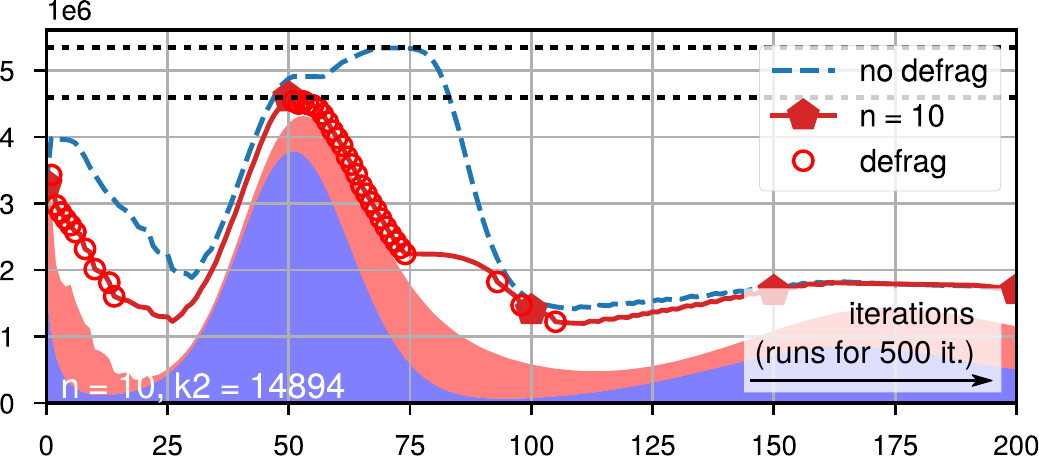}
  \centering
  \caption[\textsc{CompactGpu} benchmark: \textsf{wa-tor} (memory profile)]{\textsf{wa-tor}: An agent-based fish-and-sharks simulation}
  \label{fig:wator_mem}
\end{figure}

\begin{figure}
  \includegraphics[width=0.48\columnwidth]{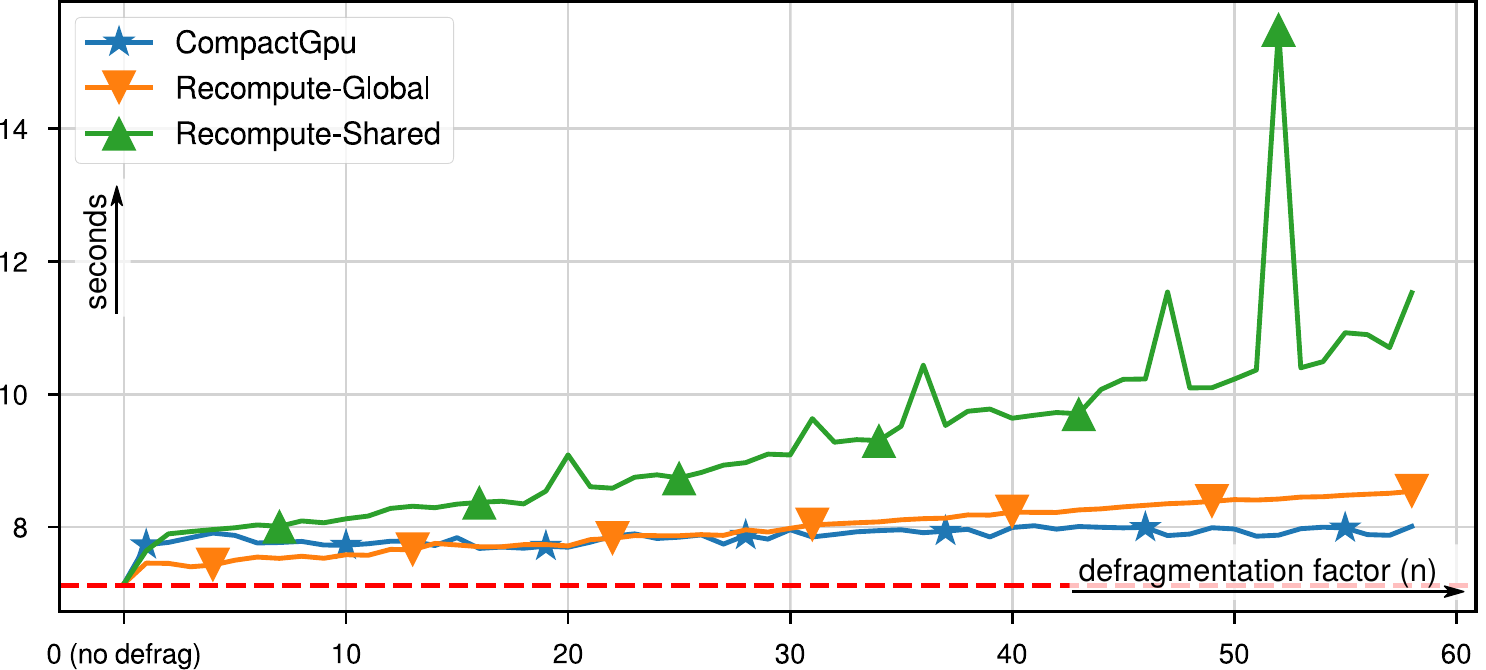}\hfill
  \includegraphics[width=0.48\columnwidth]{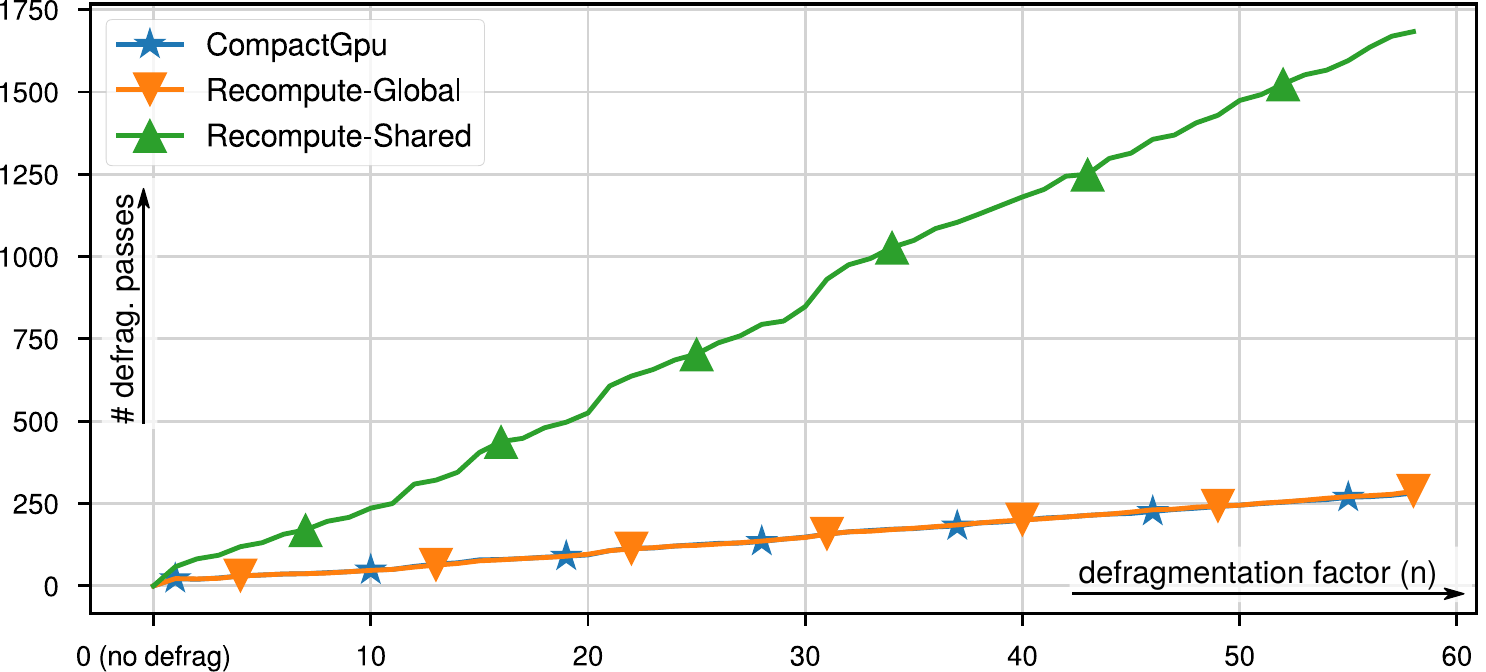}
  \caption[\textsc{CompactGpu} benchmark: \textsf{wa-tor} (running time)]{Running time and number of defragmentation passes for \textsf{wa-tor}}
  \label{fig:wator_slow}
\end{figure}

This is the fish-and-sharks running example. Fish and sharks appear in waves until an equilibrium is reached~\cite{10.2307/24969495}.

This simulation experiences a slowdown from defragmentation (Figure~\ref{fig:wator_slow}). At $n=10$, the slowdown is 8.1\%. This slowdown is \emph{not} due to defragmentation runtime overhead. The main reason is that \textsc{CompactGpu} is not order-preserving: Objects from a source block are scattered into multiple target blocks. This leads to less coalesced memory accesses when certain fields are accessed. Similar slowdowns have been reported on certain benchmarks in CPU systems~\cite{Abuaiadh:2004:EPH:1028976.1028995}. We will further investigate this effect in future work.

The benefit of defragmentation in \textsf{wa-tor} is a lower memory footprint. In Figure~\ref{fig:wator_mem}, the dotted lines indicate the maximum memory usage throughout the simulation. Circles indicate defragmentation runs. For $n=10$, the overall memory consumption is reduced by 14\%, so that programmers can run larger problem sizes on the same hardware.

We use the \emph{After Massive Deallocations} heuristic to initiate defragmentation. Initiating defragmentation every few iterations would incur a higher slowdown. To save memory, defragmentation is most important around iteration 50. At that time, many fish objects (red area) are deallocated and a large number of shark objects are allocated. New shark objects can reuse deallocated memory locations of fish objects as soon as the corresponding blocks are deallocated. Defragmentation eliminates many non-full fish blocks by compaction.

\paragraph{Pointer Rewriting Alternatives}
\textcolor{black}{\textsc{CompactGpu} is faster than \textsc{Recompute} variants in most cases. The high number of bitwise operations for recomputing a memory pointer (Algorithm~\ref{alg:rewrite_ptr_alternative}), combined with divergent execution (some pointer are rewritten, some are not) led to a high slowdown compared to our forwarding pointer method.}

\textcolor{black}{Furthermore, with increasing $n$, \textsc{Recompute-Shared} requires a larger number of defragmentation passes (Figure~\ref{fig:wator_slow}) because the shared memory is very small, limiting the number of defragmentation candidates per pass.}

\textcolor{black}{For \textsf{collision} and \textsf{structure}, this slowdown is negible because very little time is spent on defragmentation overall, but we can see clear difference for \textsf{generation} and \textsf{wa-tor}.}

\subsection{Runtime Overhead}
\begin{table*}
\setlength\tabcolsep{3pt}
\footnotesize \narrowstyle 
    \begin{tabularx}{\textwidth}{| X | r | r | r | r | r || r | r | r | r | r |}
        \hline \hline
        \footnotesize \narrowstyle  \textbf{Benchmark} & \begin{tabular}{@{}c@{}} \footnotesize \narrowstyle  \textbf{Alloc.} \\ \footnotesize \narrowstyle  \textbf{Size}\end{tabular} & \begin{tabular}{@{}c@{}} \footnotesize \narrowstyle  \textbf{\#Rewr.} \\ \footnotesize \narrowstyle  \textbf{Fields}\end{tabular} & \emph{\textbf{n}} & \textbf{\#Defrag} & \textbf{\#Passes} & \begin{tabular}{@{}c@{}}\textbf{Total} \\ \textbf{Runtime}\end{tabular} & \textbf{Defrag} & \textbf{Scan} & \textbf{Copy} & \textbf{Rewrite}     \\ \hline
        Synthetic (60\% frag.) & 2,097.2 MB & 1 & 3 & 1 & 18 & n/a & 44.4 & 4.0 & 6.7 & 33.3 \\ \hline
        \textsf{collision}  & 5.7 MB & 1 & 10 & 200 & 186 & 3,698,945 & 36 & 17 & 7 & 8 \\ \hline
        \textsf{generation} & 57.4 MB & 1 & 2 & 500 & 537 & 56,830 & 191 & 80 & 17 & 85 \\ \hline
        \textsf{structure}   & 58.9 MB & 3 & 10 & 100 & 368 & 305,846 & 140 & 54 & 16 & 65 \\ \hline
        \textsf{wa-tor}       & 1,107.6 MB & 1 & 9 & 38 & 43 & 7,729 & 49 & 7 & 14 & 20 \\ \hline \hline
    \end{tabularx}
    \normalstyle
    \caption[\textsc{CompactGpu} benchmark characteristics and running time]{Benchmark characteristics and running time (right side; milliseconds) for selected defragmentation factors}
    \vspace{-0.35cm}
    \label{fig:overheads}
    \setlength\tabcolsep{6pt} \normalstyle
\end{table*}
To evaluate the efficiency of our implementation, we measured the runtime overhead of \textsc{CompactGpu}. There are two kind of overheads.

First, \textsc{CompactGpu} extends (de)allocation procedures of \textsc{DynaSOAr} to maintain \emph{defrag[T]} bitmaps (Section~\ref{sec:compact_gpu_extend_defrag}). To measure this overhead, we compare the running time without defragmentation (red dashed line) and \emph{no defrag} values in the running time graphs. In \emph{no defrag}, we maintain a defragmentation candidate bitmap but never initiate defragmentation. There is almost no measurable overhead for maintaining these bitmaps.

Second, \textsc{CompactGpu} has three potentially expensive steps: (a) Generating/ compacting an indices array $R$ from a defragmentation candidate bitmap (\emph{scan}), (b) copying objects and placing forwarding pointers (\emph{copy}) and (c) scanning the heap and rewriting pointers (\emph{rewrite}). In Table~\ref{fig:overheads}, we show the time spent in each step, as well as the overall time spent on defragmentation (\emph{defrag}), which includes additional overheads such as block state updates. If \emph{\#Passes $<$ \#Defrag}, the programmer initiated defragmentation but there were not enough defragmentation candidates to start a defragmentation pass.

In every benchmark, defragmentation takes only a very small fraction of the overall application running time. \textsf{wa-tor} has the largest overhead: The application spends 0.6\% of its running time in defragmentation.

\paragraph{Synthetic Benchmark}
The synthetic benchmark isolates the runtime overhead for one defragmentation. We added a second class to the benchmark and made objects of both classes point to each other randomly. There are initially 32,768,000 objects of each class (object size 32~bytes). The benchmark deletes 60\% of the objects of one class and then initiates defragmentation.

The performance of the \emph{scan} phases mainly depends on the efficiency of the prefix sum operations (CUB library) and can thus not be further optimized.

The \emph{copy} phases copy (read+write) 282.1~MB of object data and write 70.5~MB of forwarding pointers in 6.7 milliseconds (94.7~GB/s). This is 17.3\% of the global memory bandwidth of our TITAN Xp GPU.

The \emph{rewrite} step is most time consuming: \textsc{CompactGpu} has to check 32,768,000 pointers (262.1~MB) per pass (18 $\times$ 262.1~MB = 4,717.2~MB in total). Out of these pointers, \textsc{CompactGpu} rewrites only a small part (read+write 70.3~MB of forwarding pointers) because many objects were deleted. \textsc{CompactGpu} finishes the rewrite step in 33.3 milliseconds, resulting in a memory transfer rate of 145.9~GB/s (not taking into account other memory accesses). This is 26.6\% of the global memory bandwidth of our TITAN Xp GPU. The Nvidia Profiler shows that 32\% of all global memory accesses in this step hit the L1 cache (64~KB) and 63\% hit the L2 cache (3,072~KB), indicating that defragmentation candidate bitmaps are largely cached.

\textsc{CompactGpu} achieves a high performance because most memory reads/writes have good coalescing. Overall, our benchmark results show that \textsc{CompactGpu} is highly optimized with little room for improvement.

\section{Related Work}
\label{sec:related}
A vast number of memory defragmentation systems have been developed for CPU systems in the past. A main difference on GPU architectures is that it is easier to decide where to relocate objects to, because there are only a small number of object sizes. This pattern is reflected in the design of many GPU dynamic memory allocators: Many allocators maintain containers for objects of the same size~\cite{hallocweb,Gelado:2019:TGM:3293883.3295727}. On CPU systems, there are typically many different allocation sizes.

The only existing GPU memory defragmentation system was developed by Veldema and Philippsen~\cite{Veldema:2012:PMD:2247684.2247693}. Their work consists of an allocator and a defragmentation system\footnote{The source code of this system is not available, so we could not compare it with \textsc{CompactGpu}.}. To compact the memory, their defragmentation system selects 10\% of all memory regions that are less than 75\% full as source regions. They use their memory allocator to allocate a target location in another region. This is problematic because allocation is expensive and requires some sort of synchronization between threads. Runtime overheads of their defragmentation system range from 0.5\% to 33\%, higher than the overhead of \textsc{CompactGpu}.

Veldema and Philippsen also propose a technique for limiting the search space during pointer rewriting based on additional data collected by a garbage collector. This technique could be used in \textsc{CompactGpu} instead of relying on class structure metainformation of \textsc{DynaSOAr}.

To the best of our knowledge, there are no other defragmentation systems for GPUs. We believe that this is because of limited support for dynamic memory allocation. The default CUDA dynamic memory allocator is known to be slow and unreliable~\cite{6339604}, so most programmers avoid dynamic memory management entirely. It is still a common practice to allocate a large chunk of memory statically and manage it manually. Out of the few custom memory allocators that exist, many (e.g., ScatterAlloc~\cite{6339604}, Halloc~\cite{hallocweb}) use a hashing approach to scatter allocations in the heap almost randomly, in order to avoid collisions among allocating threads. Not only do they miss important opportunities for vectorization (e.g., SOA layout), but they are also known to incur the negative effects of high fragmentation~\cite{hallocweb}.


Many efficient CPU memory defragmentation systems divide the heap into two areas: Objects are copied from a \emph{from-space} to a \emph{to-space}~\cite{Kermany:2006:CCI:1133981.1134023, Lang:1987:IIC:29650.29677}. Both spaces are swapped before every defragmentation pass. In such an approach, only half of the memory space is usable by the allocator. This is acceptable on virtual memory architectures because the virtual memory space is much larger than the physical memory space. Current GPU architectures do not have virtual memory and even the amount of physical memory is much smaller than on CPU systems. Cutting the available memory by half would be unacceptable on GPUs.

\paragraph{Pointer Rewriting without Forwarding Pointers}
Some memory defragmentation systems use data structures other than forwarding pointers~\cite{Abuaiadh:2004:EPH:1028976.1028995, Kermany:2006:CCI:1133981.1134023}. For example, the \emph{Compressor} uses a markbit vector to recompute forwarding pointers on-the-fly during pointer rewriting~\cite{Kermany:2006:CCI:1133981.1134023}. A markbit vector is a bit vector where bits for the first and last heap word of an allocated object are set. Since forwarding pointers are not read from memory, only two accesses (read pointer, replace with new pointer) are required to rewrite a pointer, assuming the markbit vector is cached. We experimented with similar techniques (\textsc{Recompute-Global}, \textsc{Recompute-Shared}; Section~\ref{sec:alternative_designs_sec}), but they did not lead to a performance improvement.




\section{Conclusion}
\label{sec:conclusion_compactgpu}
We presented \textsc{CompactGpu}, a memory defragmentation system for GPUs. \textsc{CompactGpu} is able to (a) speed up applications through better cache utilization and vector load/store efficiency on allocated memory and (b) lower the overall memory consumption of an application.

\textsc{CompactGpu} achieves low runtime overheads through careful SIMD-friendly design considerations and implementation efforts: \textsc{CompactGpu} utilizes bitmaps to select source/ target blocks and to quickly decide if a pointer must be rewritten. Furthermore, \textsc{CompactGpu} exhibits mostly regular control flow, accesses memory in coalescing-friendly patterns and requires no synchronization between threads.

Our main takeaways are that (a) memory defragmentation on GPUs is feasible and able to deliver speedups, (b) too much defragmentation does not pay off (due to overfitting and overcompaction) and can even be detrimental to performance due to less efficient allocations, and (c) careful design considerations are necessary to achieve good performance on GPUs; many good CPU designs, such as recomputing forwarding pointers on-the-fly, are not efficient on GPUs.

\chapter{SMMO Examples}
\label{chap:smmo_examples}
In this chapter, we present examples of Single-Method Multiple-Objects (SMMO) applications. We used these applications to evaluate \textsc{DynaSOAr} and \textsc{CompactGpu}. We show the data structure of each application and highlight their SMMO structure, i.e., which parallel do-all operations they consist of.

\setcounter{minitocdepth}{1}
\minitoc

\paragraph{Benefits of Object-oriented Programming}
To analyze the performance of \textsc{DynaSOAr}, we implemented most SMMO applications with and without dynamic memory allocation (\emph{baseline} versions; Section~\ref{sec:benchmark}). SOA baselines store objects in a hand-written SOA data layout (Listing~\ref{lbl:soa_layput}). We no longer consider such a layout as \emph{object-oriented} because C++ abstractions for object-oriented programming cannot be used. We noticed that missing OOP abstractions made the development of the SOA baseline versions more tedious than their counterpart versions that utilize the \textsc{DynaSOAr} data layout DSL and dynamic memory management, in particular:

\begin{figure}
\centering
\includegraphics[scale=0.75]{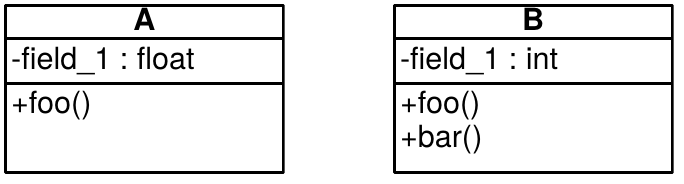}
\caption{Example: Dummy classes}

\vspace{10pt}

\begin{lstlisting}[language=c++, numbers=none, caption={Example: Dummy classes with hand-written SOA layout (no OOP)}, label={lst:dummy_soa_ex}]
float A_field_1[10000];
int B_field_1[10000];

void A_foo(int id) { printf("A: %f\n", A_field_1[id]); }
void B_foo(int id) { printf("B: %d\n", B_field_1[id]); }
void B_bar(int id) { printf("B_bar: %d\n", B_field_1[id]); }

int get_A() { /* Return ID of an object of type A */ }

int main() {
  A_foo(get_A());   // OK
  B_foo(get_A());   // Typo: B_foo instead of A_foo. But code compiles and runs.
  B_bar(get_A());   // Calling opeation of B on A. But code compiles and runs.
}
\end{lstlisting}

\begin{lstlisting}[language=c++, numbers=none, caption={Example: Dummy classes in AOS layout (with OOP)}, label={lst:dummy_aos_ex}]
class A {
 public:
  float field_1;
  void foo() { printf("A: %f\n", field_1); }
};
A objects_A[10000];

class B {
 public:
  int field_1;
  void foo() { printf("A: %d\n", field_1); }
  void bar() { printf("B_bar: %d\n", field_1[id]); } 
};
B objects_B[10000];

A* get_A() { /* Return pointer to an object of type A */ }

int main() {
  get_A()->foo();   // OK
  get_A()->foo();   // It is impossible to make the same mistake as in Listing <@\ref{lst:dummy_soa_ex}@>.
  get_A()->bar();   // Compile error
}
\end{lstlisting}
\end{figure}

\begin{description}
\item[No Type Safety] With respect to typing, the implementation of SOA baselines felt like programming in an untyped/dynamically-typed language because all object references are of type \emph{integer}. This has two disadvantages. First, many programming errors are not caught by the type checker at compile time, which slowed down the development process. To make matters worse, some programming errors do not even cause the program to crash at runtime, but simply produce a wrong result (Listing~\ref{lst:dummy_soa_ex} and~\ref{lst:dummy_aos_ex}). Second, types are also a form of code documentation and their absence makes code harder to read and understand~\cite{Endrikat:2014:ADS:2568225.2568299}.
\item[No Dynamic Memory Allocation] Many applications are more difficult to implement without dynamic allocation. Category~2 applications require an additional boolean field to keep track of active/allocated objects (Section~\ref{sec:benchmark}). Category~3 applications required changes to the data structures within the application, which breaks abstractions.
\item[Notation of Chained Field Accesses] Some applications exhibit a pattern of chained field accesses in their source code, e.g.: \texttt{obj->f1->f2}. The order in which field name tokens appear in hand-written SOA code is inversed and counter-intuitive: \texttt{B\_f2[A\_f1[obj]]}. Moreover, this notation requires programmers to repeat the type/class of objects/entities (\texttt{A} and \texttt{B} in this example).
\end{description}

In this chapter, we present the design and implementation of various SMMO applications with \textsc{DynaSOAr}. These implementations are object-oriented and do not suffer from the above mentioned shortcomings. To further highlight the benefits of object-oriented programming, we discuss interesting design and implementation choices of SOA baseline versions for certain SMMO applications.

\paragraph{Publications}
This chapter is in part based on the following papers.
\begin{itemize}
  \item Matthias Springer, Hidehiko Masuhara. \textbf{``DynaSOAr: A Parallel Memory Allocator for Object-oriented Programming on GPUs with Efficient Memory Access (Artifact).''} In: \emph{Dagstuhl Artifacts Series.} Vol. 5, Iss. 2, Art. 2. Leibniz-Zentrum f{\"u}r Informatik, Dagstuhl Publishing, 2019. \texttt{\doi{10.4230/DARTS.5.2.2}}.
\end{itemize}

\section{\textsf{nbody}: N-body Simulation}
\label{sec:smmo_nbody_sec71}
\textsf{nbody} is a 2D particle system simulation. Such simulations are used by astronomers to simulate the collision of galaxies or the formation of planets~\cite{ALEXANDER1998113}. This example is a very simple SMMO application with only one class. The entire application consists of only around 100 lines of code.

\textsf{nbody} simulates a large number of bodies. Each body has a position, velocity and mass. According to Newton's theory of gravity, two bodies with a masses $m_1$ and $m_2$ and a distance of $r$ pull each other closer with a gravitational force $F$.

\begin{align*}
F = G \frac{m_1 m_2}{r^2} \tag{\emph{gravitational force}}
\end{align*}

The goal of \textsf{nbody} is to calculate the future position and velocity of all bodies. \textsf{nbody} is an iterative algorithm: Each iteration advances the simulation by a small time step $\Delta t$. We assume that the gravitational forces remain constant during an iteration, so we can compute each body's new velocity and position as follows.

\begin{align*}
v \gets v + \frac{F}{m} \cdot \Delta t \tag{\emph{simulation time step}} \\
x \gets x + v \cdot \Delta t
\end{align*}

\begin{figure}
  \subfloat[Data structure]{\includegraphics[scale=0.75]{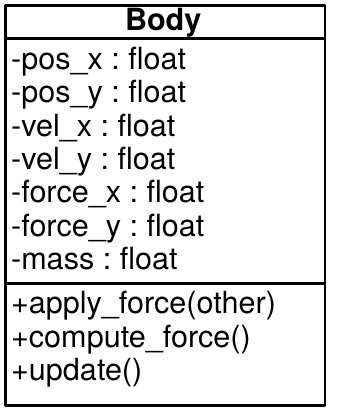}}\qquad
  \subfloat[Screenshot]{\includegraphics[scale=0.25]{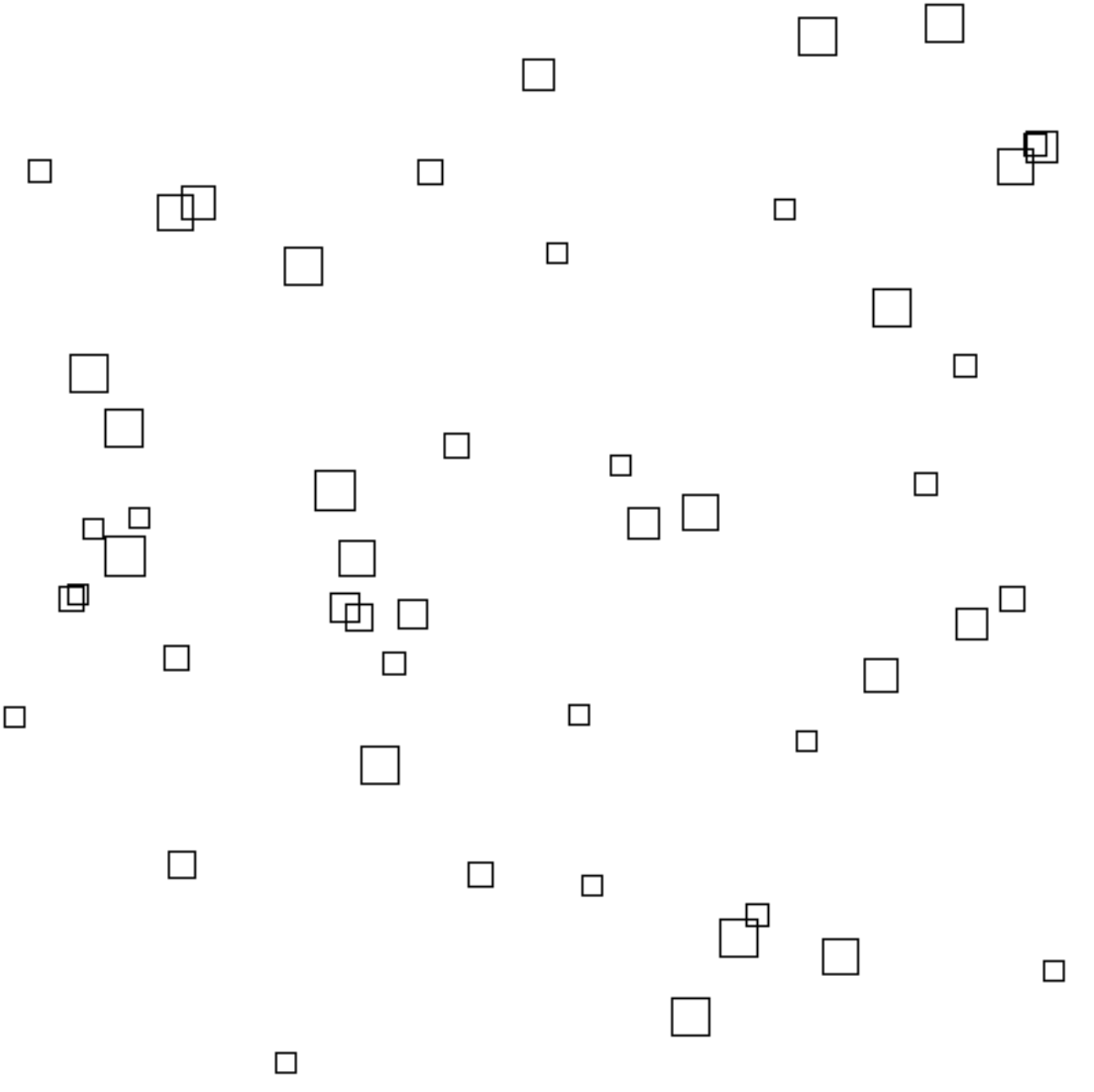}}
  \centering
  \caption[\textsf{nbody}: Data structure and screenshot]{Data structure and screenshot of \textsf{nbody}. Bodies pull each other closer with graviational force. Apart from the gravitational force, there is no interaction between bodies, making this example the simplest one in this chapter.}
  \label{fig:data_s_nbody}
\vspace{10pt}

\begin{lstlisting}[language=c++, caption={[\textsf{nbody}: Data structure]Data structure of \textsf{nbody}}, label={lst:header_nbody}, morekeywords={__device__}]
#include "dynasoar.h"

class Body;
using AllocatorT = SoaAllocator<kNumObjects, Body>;

class Body : public AllocatorT::Base {
 public:
  declare_field_types(
      Body,
      float,          // pos_x_
      float,          // pos_y_
      float,          // vel_x_
      float,          // vel_y_
      float,          // force_x_
      float,          // force_y_
      float)          // mass_

 private:
  Field<Body, 0> pos_x_;
  Field<Body, 1> pos_y_;
  Field<Body, 2> vel_x_;
  Field<Body, 3> vel_y_;
  Field<Body, 4> force_x_;
  Field<Body, 5> force_y_;
  Field<Body, 6> mass_;

 public:
  __device__ Body(float pos_x, float pos_y, float vel_x, float vel_y, float mass);

  // Constructor for parallel_new.
  __device__ Body(int index);

  __device__ void apply_force(Body* other);

  __device__ void compute_force();

  __device__ void update();
};
\end{lstlisting}
\end{figure}

\subsection{Data Structure}
Every body is an object of class \texttt{Body} (Figure~\ref{fig:data_s_nbody}, Listing~\ref{lst:header_nbody}). This class has fields for position, velocity, mass, as well as helper variables for accumulating the force that acts on the body.

\subsection{Application Implementation}
This application consists of two parallel do-all operations and one parallel new operation (Listing~\ref{lst:code_nbody}). All member functions of \texttt{Body} are device functions (annotated with CUDA keyword \texttt{\_\_device\_\_}) because they are executed on the GPU.

\begin{enumerate}
  \item \textbf{Initialization:} Create a few random \texttt{Body} objects. Parallel new: \texttt{Body}.
  \item \textbf{Iterative Algorithm:} Each iteration consists of the following steps.
  \begin{enumerate}
    \item \textbf{Computing Forces:} Compute forces between all pairs of bodies with a nested (sequential) do-all iteration. Parallel do-all: \texttt{Body::compute\_force}.
    \item \textbf{Acceleration and Movement:} Compute a new velocity and position for each body. Parallel do-all: \texttt{Body::update}.
  \end{enumerate}
\end{enumerate}

In the following paragraphs, we describe each step of this simulation in more detail.

\paragraph{Allocator Initialization}
Before the actual application code begins, we have to create an allocator. We first create a new \emph{allocator handle} which can be used from host (CPU) code. This internally allocates a large chunk of global GPU memory which contains the heap of our allocator. We then store a device allocator pointer in variable \texttt{device\_allocator}. This variable is a device variable, so it is only visible from GPU code. We use this pointer to interact with \textsc{DynaSOAr} from GPU code. Every \textsc{DynaSOAr} program initializes the memory allocator in this way, so we will skip over this part in the remaining examples.

\paragraph{Step~1: Simulation Initialization}
At the beginning, we initialize the simulation with 65,536 random \texttt{Body} objects. The parallel new operation invokes the second constructor of \texttt{Body} in parallel, passing an index value between 0 and 65,535 as argument. This constructor initializes the position, velocity and mass of the \texttt{Body} object with random values using the cuRAND library.

\paragraph{Step~2: Iterative Algorithm}
The simulation invokes two parallel do-all operations iteratively. The first method \texttt{compute\_force()} computes the total force that is acting on a body and stores it in the fields \texttt{force\_x\_} and \texttt{force\_y\_}. This requires a nested \emph{for} loop because we have to compute and sum the forces between all pairs of bodies. In \textsc{DynaSOAr}, we express such a nested \emph{for} loop with \texttt{device\_do}, which is similar to \texttt{parallel\_do} but runs sequentially.

Listing~\ref{lst:nbody_nested_loops} illustrates this nested loop structure. The outer loop is a CUDA kernel and runs in parallel. The inner loop runs sequential. While the implementation of \texttt{b2->apply\_force(b1)} could conceptually sum the force in either \texttt{b1} or \texttt{b2} (the same force acts in both directions), it has to modify \texttt{b1} to avoid race conditions. If \texttt{apply\_force} were to modify \texttt{b2}, then there would be multiple GPU threads modifying \texttt{b2} concurrently, because the outer loop runs in parallel. To avoid race conditions, this would require an atomic add operation. Instead, \texttt{apply\_force} modifies the force of \texttt{b1}, which is now only accessed by one GPU thread.

\begin{lstfloat}
\begin{lstlisting}[language=c++, caption={[\textsf{nbody}: Conceptual nested loop structure]Nested loop structure (conceptually)}, label={lst:nbody_nested_loops}]
for (Body* b1 : AllocatorT::all_objects<Body>) {    // parallel_do
  for (Body* b2 : AllocatorT::all_objects<Body>) {  // device_do
    // Compute gravitational force between b1 and b2. Accumulate the total
    // gravitational force (of all bodies) in force_x_ and force_y_ values.
    // But should apply_force modify b1 or b2?
    b2->apply_force(b1);  // Or: b1->apply_force(b2)
  }
}
\end{lstlisting}
\end{lstfloat}

The second method \texttt{update()} computes a new velocity and position based on the force value that was computed by the first method. If a body goes out of bounds of the simulation space, we invert its velocity. This only for visualization reasons; a real n-body simulation would not do this.

Both methods run in separate parallel do-all operations (and thus separate CUDA kernels) to ensure that the simulation is free of race conditions and to ensure that the application produces the same result on every run, assuming the random number generator is initialized with the same seed. If both methods were to run in one CUDA kernel, the implemention of \texttt{apply\_force} may read an updated or not-yet updated position of another body, depending on the scheduling of threads.

\begin{lstfloat}
\begin{lstlisting}[language=c++, caption={[\textsf{nbody}: Application logic]Application logic of \textsf{nbody}}, label={lst:code_nbody}, morekeywords={__device__}]
// Allocator handles.
AllocatorHandle<AllocatorT>* allocator_handle;
__device__ AllocatorT* device_allocator;

__device__ Body::Body(float pos_x, float pos_y,
                      float vel_x, float vel_y, float mass)
    : pos_x_(pos_x), pos_y_(pos_y), vel_x_(vel_x), vel_y_(vel_y), mass_(mass) {}

__device__ Body::Body(int index) {
  curandState rand_state;
  curand_init(kSeed, index, 0, &rand_state);
  // Initialize with random float between -1 and 1.
  pos_x_ = 2 * curand_uniform(&rand_state) - 1;
  /* Similarly for the pos_y_, vel_x_, vel_y_, mass_... */
}

__device__ void Body::apply_force(Body* other) {
  if (other != this) {
    // To avoid race conditions: Update other instead of this.
    float dx = pos_x_ - other->pos_x_;
    float dy = pos_y_ - other->pos_y_;
    float dist = sqrt(dx*dx + dy*dy);
    float F = kGravityConstant * mass_ * other->mass_ / (dist * dist);
    other->force_x_ += F*dx / dist;
    other->force_y_ += F*dy / dist;
  }
}

__device__ void Body::compute_force() {
  force_x_ = force_y_ = 0.0f;
  device_allocator->device_do<Body>(&Body::apply_force, this);
}

__device__ void Body::update() {
  vel_x_ += force_x_ * kDt / mass_;
  vel_y_ += force_y_ * kDt / mass_;
  pos_x_ += vel_x_ * kDt;
  pos_y_ += vel_y_ * kDt;

  // Bounce off the walls.
  if (pos_x_ < -1 || pos_x_ > 1) { vel_x_ = -vel_x_; }
  if (pos_y_ < -1 || pos_y_ > 1) { vel_y_ = -vel_y_; }
}

int main() {
  // Create new allocator.
  allocator_handle = new AllocatorHandle<AllocatorT>();
  AllocatorT* dev_ptr = allocator_handle->device_pointer();
  cudaMemcpyToSymbol(device_allocator, &dev_ptr, sizeof(AllocatorT*), 0,
                     cudaMemcpyHostToDevice);

  // Create 65536 new Body objects.
  allocator_handle->parallel_new<Body>(65536);

  // Simulation loop.
  for (int i = 0; i < kNumIterations; ++i) {
    allocator_handle->parallel_do<Body, &Body::compute_force>();
    allocator_handle->parallel_do<Body, &Body::update> 
  }
}
\end{lstlisting}
\end{lstfloat}

\subsection{Further Optimizations}
Related work describes three additional techniques to further optimize this n-body simulation~\cite{Nguyen:2007:GG:1407436}. To keep this benchmark simple, these optimizations are not implemented in our n-body simulation.

\begin{itemize}
  \item \textbf{Shared Memory:} Since we compute forces between all pairs of bodies, every CUDA thread reads the position and mass fields of all bodies. To reduce the amount of data read from global memory, we can read those fields only once per CUDA block and store the values in shared memory.
  \item \textbf{Nested Parallelism:} Instead of parallelizing only the outer \emph{for} loop, also partly parallelize the inner \emph{for} loop. This technique is particular useful for improving occupancy if the number of bodies is small.
  \item \textbf{Loop Unrolling:} Especially on older GPU architectures, unrolling the inner \emph{for} loop can improve instruction-level parallelism (ILP) and thus better utilize the GPU's resources (e.g., through \emph{dual-issue}, Section~\ref{sec:backg_parall_exec}).
\end{itemize}

To be able to implement these optimizations, we have to extend \textsc{DynaSOAr}'s API in future work. In particular, there is currently no way of unrolling \texttt{device\_do} iterations. This could be simplified if range-based \emph{for} loops could be used in lieu of \texttt{device\_do} iterations (Section~\ref{sec:cpp_range_baed_for}). Moreover, the current API makes it difficult to partly parallelize object enumeration. In essence, what we require is a mixture of \texttt{parallel\_do} and \texttt{device\_do}. Such an API and its implementation could be inspired by virtual warp-centric programming~\cite{Hong:2011:ACG:1941553.1941590}.

\section{\textsf{collision}: N-Body Simulation with Collisions}
\label{sec:nody_with_coll}
We now extend the n-body simulation of Section~\ref{sec:smmo_nbody_sec71} with collisions. Two bodies are merged into one large body if their distance is below a certain threshold. This example is interesting because it exhibits more complex object interactions and stores a pointer to another object in a field, as opposed to only primitively-typed values in the n-body simulation of Section~\ref{sec:smmo_nbody_sec71}.

\begin{figure}
  \centering
  \subfloat[Data structure]{\includegraphics[scale=0.75]{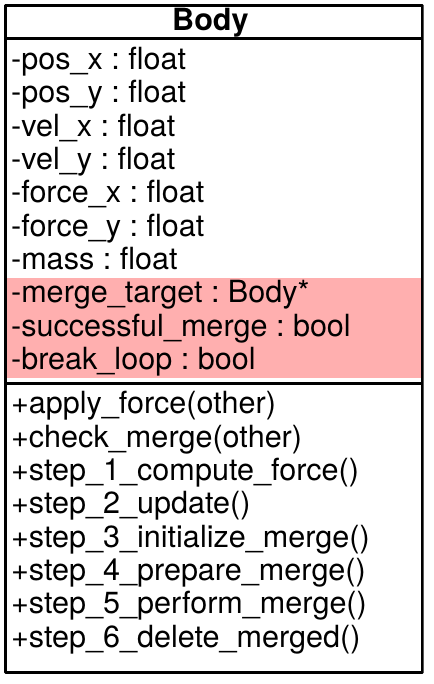}} \qquad
  \subfloat[Screenshot]{\includegraphics[scale=0.35]{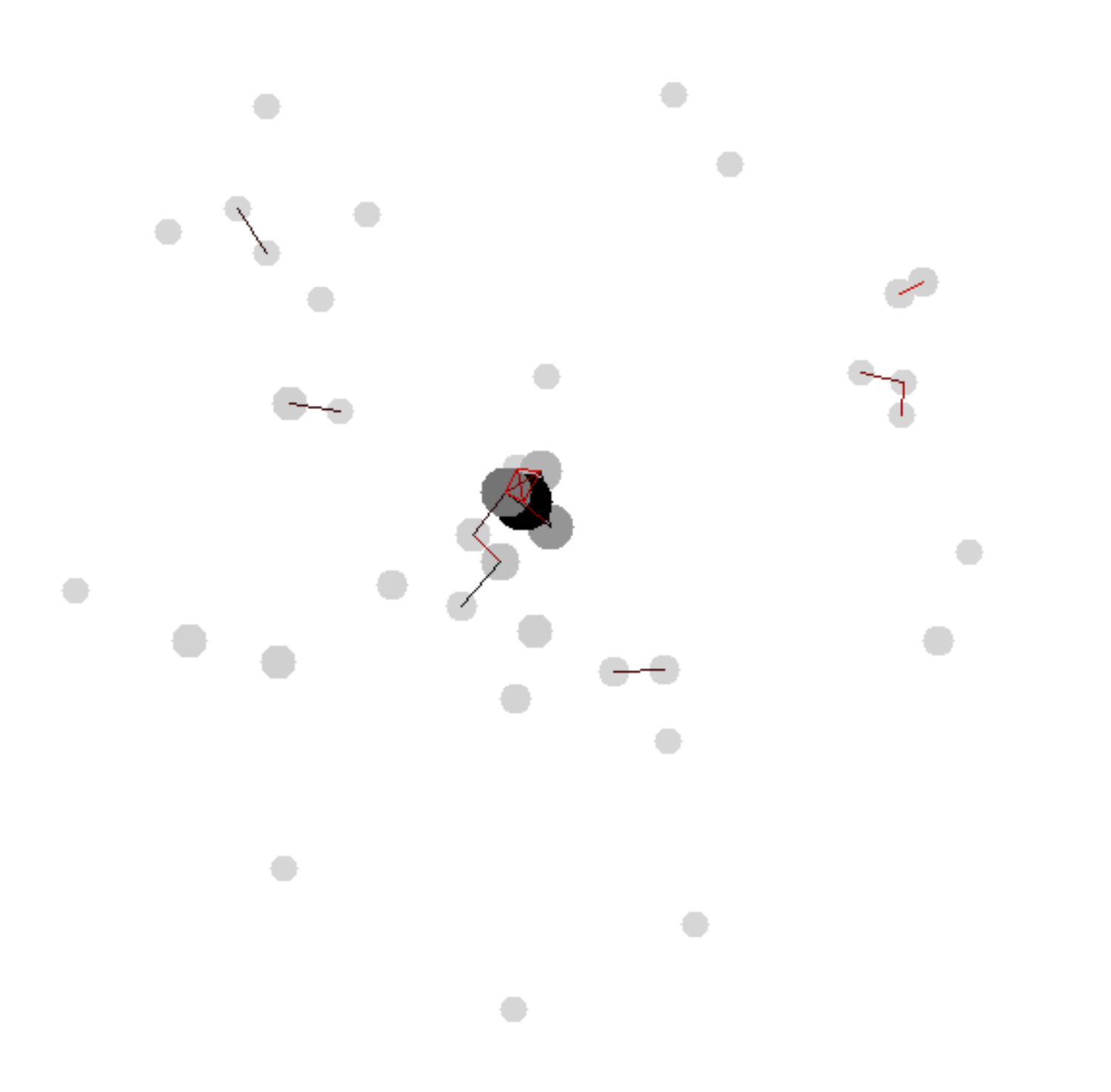}}
    \caption[\textsf{collision}: Data structure and screenshot]{Data structure and screenshot of \textsf{collision}, an extension of \textsf{nbody}. In the screenshot, the darkness of a body indicates its mass. When bodies are getting close, we connect them with lines. A red line indicates that two bodies are almost within merging distance.} \label{fig:data_s_collision}
\end{figure}

This simulation is an extension of the previous n-body simulation. It computes forces between bodies and accelerates/moves bodies in the same way. However, two bodies are merged according to the physical law of perfectly inelastic collision if they become too close. In that case, the lighter body $b_2$ is merged into the heavier one $b_1$ ($m_1 > m_2$). According to the law of perfectly inelastic collision, the new velocity of the heavier body $b_1$ is the weighted sum of the velocities of both bodies.

\begin{align*}
v_1 & \gets \frac{v_1 \cdot m_1 + v_2 \cdot m_2}{m_1 + m_2} \\
x_1 & \gets \frac{x_1 + x_2}{2}  \tag{\emph{perfectly inelastic collision}} \\
m_1 & \gets m_1 + m_2
\end{align*}

\subsection{Data Structure}
\label{sec:cpp_range_baed_for}
We added three fields to class \texttt{Body} to implement body merging semantics (Figure~\ref{fig:data_s_collision}, red color). \texttt{merge\_target} points to the \texttt{Body} object into which a given body should be merged. \texttt{successful\_merge} is set to \emph{true} if a body was successfully merged and can be deleted. \texttt{break\_loop} is another boolean flag for breaking out of a \texttt{device\_do} iteration. This flag is necessary because the C++ \texttt{break} keyword, which is used to break out of loops, cannot be used with \texttt{device\_do} iterations.

We will provide a C++ iterator for enumerating allocated objects in future versions of \textsc{DynaSOAr}. Programmers can than use range-based \emph{for} loops (Listing~\ref{lst:code_collision_for_loop}) instead of \texttt{device\_do}. In that case, the field \texttt{break\_loop} will no longer be necessary.

\subsection{Application Implementation}
The initialization of the simulation is identical to the previous n-body simulation and omitted in this section. The actual simulation invokes six parallel do-all operations iteratively. The first two parallel do-all operations are identical to the previous n-body simulation.

\begin{enumerate}
  \item \textbf{Compute Forces:} Compute and accumulate the force exerted by all other bodies on a given body. Parallel do-all: \texttt{Body::step\_1\_compute\_force}.
  \item \textbf{Acceleration and Movement:} Accelerate and move a body. Parallel do-all: \texttt{Body::step\_2\_update}.
  \item \textbf{Reset Merge Fields:} Initialize the three new fields that will be used during merging. Parallel do-all: \texttt{Body::step\_3\_initialize\_merge}.
  \begin{enumerate}
    \item \texttt{merge\_target} $\gets$ \texttt{nullptr}: No merge target was selected yet for this body.
    \item \texttt{successful\_merge} $\gets$ \texttt{false}: This body was not merged yet.
    \item \texttt{break\_loop} $\gets$ \texttt{false}
  \end{enumerate}
  \item \textbf{Prepare Merge:} For a given body, check if and which other body can be merged into it. This step implements a \emph{pull} semantics: Instead of looking for body into which a given body can be merged, we are looking for a body that can be merged into a given body. Parallel do-all: \texttt{Body::step\_4\_prepare\_merge}.
  \item \textbf{Perform Merge:} If a merge target was selected for a given body in the previous step, perform the perfectly inelastic collision. This step implements a \emph{push} semantics. Parallel do-all: \texttt{Body::step\_5\_update\_merge}.
  \item \textbf{Delete Body:} If a given body was merged in the previous step, delete it. Parallel do-all: \texttt{Body::step\_6\_delete\_merged}.
\end{enumerate}

Steps 4 and 5 (Listing~\ref{lst:code_collision}) are the interesting ones and we will describe them in more detail in the following paragraphs.

\begin{lstfloat}
\begin{lstlisting}[language=c++, caption={[\textsf{collision}: Application logic]Additional application logic of \textsf{collision}}, label={lst:code_collision}, morekeywords={__device__}]
__device__ void Body::check_merge(Body* other) {
  // Only merge into larger body.
  if (!other->break_loop_ && mass_ < other->mass_) {
    float dx = pos_x_ - other->pos_x_;
    float dy = pos_y_ - other->pos_y_;
    float dist_square = dx * dx + dy * dy;

    if (dist_square < kMergeThreshold * kMergeThreshold) {
      // Try to merge this into other. There is a race condition here:
      // Multiple threads may try to merge this body. Only one can win.
      this->merge_target_ = other;
      other->break_loop_ = true;
    }
  }
}

__device__ void Body::step_4_prepare_merge() {
  device_allocator->device_do<Body>(&Body::check_merge, this);
}

__device__ void Body::step_5_update_merge() {
  Body* m = merge_target_;
  if (m != nullptr) {
    if (m->merge_target_ == nullptr) {
      // Perform merge.
      float new_mass = mass_ + m->mass_;
      float new_vel_x = (vel_x_ * mass_ + m->vel_x_ * m->mass_) / new_mass;
      float new_vel_y = (vel_y_ * mass_ + m->vel_y_ * m->mass_) / new_mass;
      m->mass_ = new_mass;
      m->vel_x_ = new_vel_x;
      m->vel_y_ = new_vel_y;
      m->pos_x_ = (pos_x_ + m->pos_x_) / 2;
      m->pos_y_ = (pos_y_ + m->pos_y_) / 2;

      successful_merge_ = true;
    }
  }
}

__device__ void Body::step_6_delete_merged() {
  if (successful_merge_) { destroy(device_allocator, this); }
}
\end{lstlisting}

\begin{lstlisting}[language=c++, caption={[\textsf{collision}: Range-based \emph{for} loop instead of \texttt{device\_do}]Alternative: C++ range-based \emph{for} loop instead of \texttt{device\_do}}, label={lst:code_collision_for_loop}, morekeywords={__device__}]
__device__ void Body::step_4_prepare_merge() {
  for (Body& other : device_allocator->iterator<Body>()) {
    if (other.mass_ < mass_) {
      float dx = other.pos_x_ - pos_x_;
      float dy = other.pos_y_ - pos_y_;
      float dist_square = dx * dx + dy * dy;

      if (dist_square < kMergeThreshold * kMergeThreshold) {
        other.merge_target_ = this;
        break;
      }
    }
  }
}
\end{lstlisting}
\end{lstfloat}

\paragraph{Step 4: Prepare Merge}
This method determines if another body should be merged into a given body. A body $b_y$ can be merged into $b_x$ if the distance between $b_y$ and $b_x$ is below a certain threshold and if $m_x > m_y$. This method operates from the perspective of a \emph{receiving} body $b_x$ (\emph{pull} semantics) and iterates over all other bodies (\texttt{device\_do}) to find a body that satisfies these requirements. If such a body $b_y$ was found, the merge target of $b_y$ is set to $b_x$. Furthermore, the \texttt{break\_loop} flag of $b_x$ is set. This does not really stop the loop, but future iterations of \texttt{check\_merge} will immediately return.

\begin{figure}
    \subfloat[$b_5$ selects $b_4$ to merge.]{\includegraphics[width=0.45\textwidth]{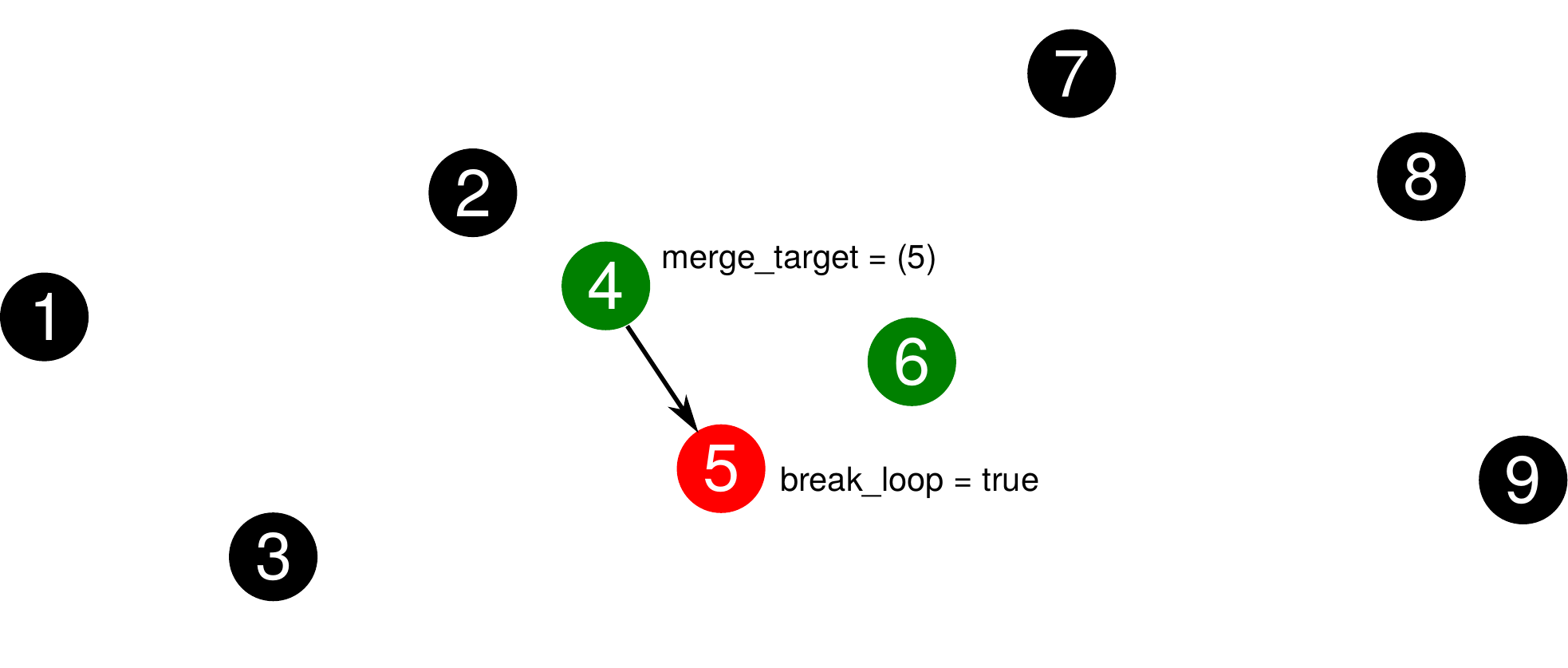}}\hfill
    \subfloat[Afterwards, $b_2$ also selects $b_4$.]{\includegraphics[width=0.45\textwidth]{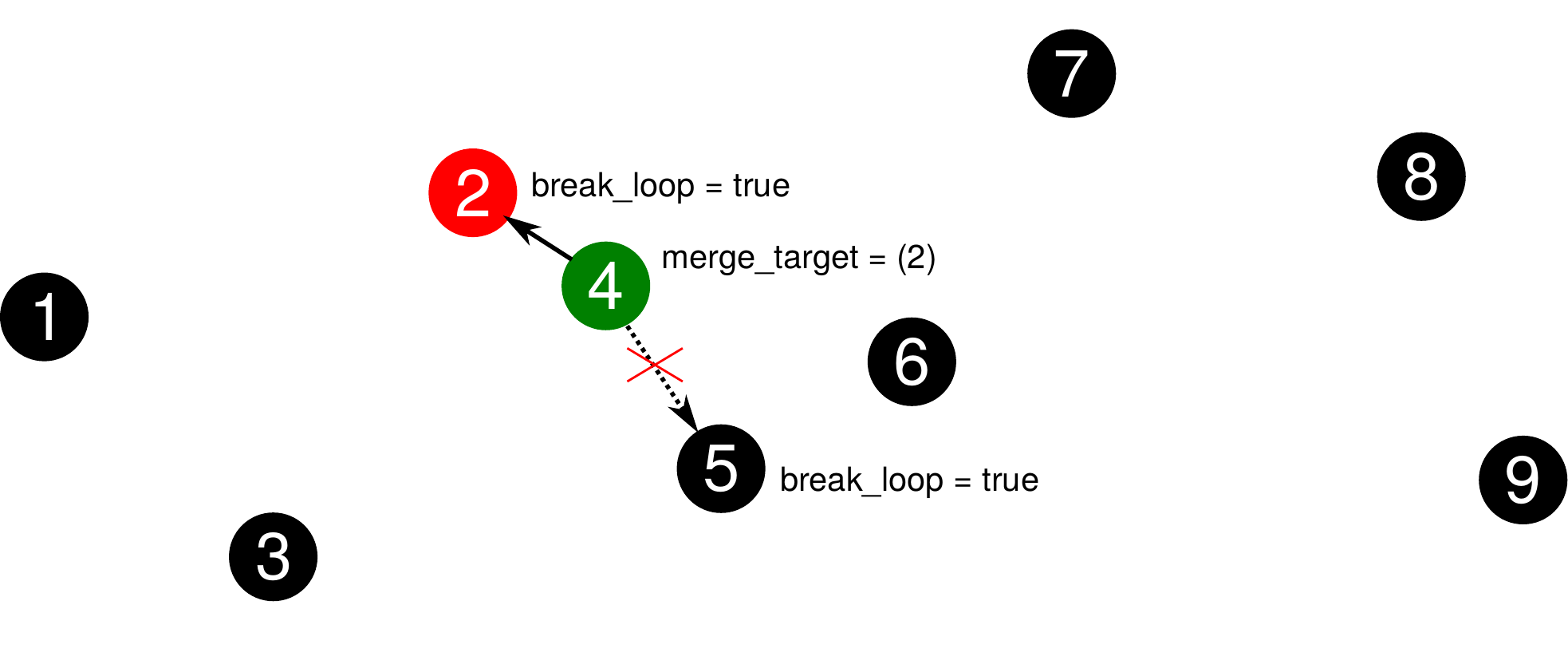}}
    \caption[\textsf{collision}: Preparing a body merge operation]{Preparing a body merge operation}
    \label{fig:ex_coll_prepare_merge}

  \includegraphics[width=0.45\textwidth]{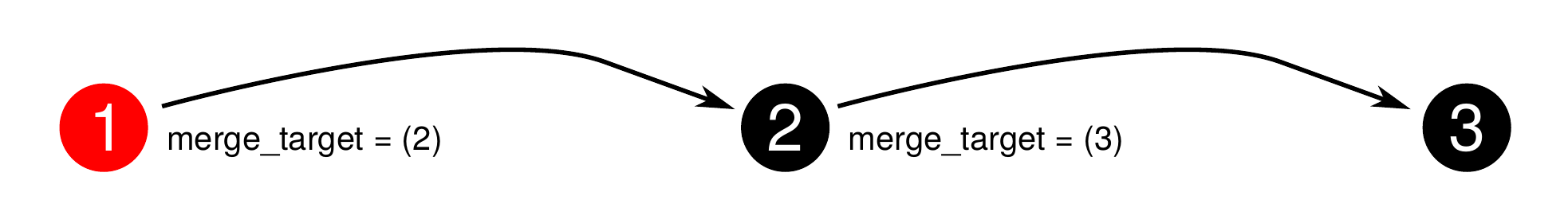}
  \centering
  \caption[\textsf{collision}: Detecting a merge race condition]{Merge target being merged itself}
  \label{lst:col_merge_merged_itself}
\end{figure}

Consider the example in Figure~\ref{fig:ex_coll_prepare_merge}a. We observe the process from the perspective of the red body $b_5$, i.e., \texttt{step\_4\_prepare\_merge} is bound to $b_5$. The green bodies are within merging range. $b_5$ checks $b_4$ and decides to merge it into itself. The \texttt{device\_do} loop breaks at this point and $b_5$ does not even consider $b_6$, which is also in range.

Since \texttt{step\_4\_prepare\_merge} runs as a parallel do-all operation, multiple bodies are concurrently looking for merge partners. In Figure~\ref{fig:ex_coll_prepare_merge}b, $b_2$ selects the same body $b_4$ as $b_5$ and writes a pointer to itself into \texttt{merge\_target}. This is a race condition: Which body $b_4$ will be merged into depends on thread scheduling. In this example, $b_2$ wins. Note that $b_5$ cannot select a new body now because it already decided to break from the loop. It would be difficult to let $b_5$ select a new body in such cases. First, $b_5$ may already be done executing the \texttt{device\_do}, so this would potentially require another \texttt{device\_do} iteration. Second, $b_5$ is not even guaranteed to notice that \texttt{merge\_target} was overwritten by another thread because GPU caches are not coherent.

We accept the race condition in this step. If the merging threshold is sufficiently small, then such race conditions are less likely to appear. Furthermore, even if a body is not merged due to such a race condition, it is likely getting merged in one of the next iterations, because a body will likely spend multiple iterations within the merging threshold.

\paragraph{Step 5: Perform Merge}
This method performs a merge of a given body in case a merge target was selected in the previous step. It operates from the perspective of a \emph{giving} body (\emph{push} semantics).

There is another potential race condition in this method. Consider the example in Figure~\ref{lst:col_merge_merged_itself}. $b_1$ should be merged into $b_2$ and $b_2$ should be merged into $b_3$. Since \texttt{step\_5\_perform\_merge} runs as a parallel do-all operation, both bodies are merged concurrently. This is problematic because the thread of $b_1$ writes fields of $b_2$ and the thread of $t_2$ reads those fields concurrently. Depending on the thread scheduling, the thread of $t_2$ may read the original or updated values. This race condition is more problematic than the one in the previous step because a body could effectively disappear from the simulation.

To avoid this race condition, we merge $b_1$ into $b_2$ only if $b_2$ does not have a merge target itself. After the merge, the \texttt{successful\_merge} flag of $b_1$ is set to \emph{true}, so that it will be deleted from the simulation in \texttt{step\_6\_delete\_merged}.

\subsection{Benefits of Object-oriented Implementation}
We would like to highlight two benefits of object-oriented programming in this implementation.

\begin{itemize}
  \item \textbf{Field Access Notation:} C++'s object field access notation is more readable compared to a hand-written SOA layout. For example, Listing~\ref{lst:nbody_col_field_not} shows Line~27 of Listing~\ref{lst:code_collision} in a hand-written SOA layout. There are two problems with this notation: First, we cannot use C++'s member access (arrow) operator. Second, we have to specify an index for each SOA array access, whereas methods are bound to an object and field accesses without an explictly specified object refer to the bound object (\texttt{this}/\texttt{id}).
  \item \textbf{Active Flag:} Objects that were deleted (Listing~\ref{lst:code_collision}, Line~41) are no longer enumerated by subsequent parallel do-all operations. To implement the same semantics without dynamic memory allocation (baselines AOS/SOA), we had to add an extra field \texttt{active} to class \texttt{Body}, which is initially \emph{true} but later set to \emph{false} if the \texttt{Body} object was merged. CUDA kernels process only active objects. 
\end{itemize}

\begin{lstfloat}
\begin{lstlisting}[language=c++, numbers=none, caption={[\textsf{collision:} Field access notation in hand-written SOA layout]Field access notation in hand-written SOA layout}, label={lst:nbody_col_field_not}]
float new_vel_x = (Body_vel_x[id] * Body_mass[id] + Body_vel_x[m] * Body_vel_x[m]) / new_mass;
\end{lstlisting}
\end{lstfloat}

\subsection{Further Optimizations}
The most time-consuming steps of this n-body simulation are \emph{computing forces} (Step~1) and \emph{finding merge partners} (Step~4). Both steps are implemented with a \texttt{device\_do} iteration, resulting in $N^2$ body-body computations, where $N$ is the number of bodies. Step~1 can be approximated with the Barnes-Hut algorithm (Section~\ref{sec:example_barnes_hut_sec}). Step~4 is a common problem in physical simulations and could be optimized with spatial subdivision~\cite{Nguyen:2007:GG:1407436coll}.

\section{\textsf{barnes-hut}: Approximating N-Body with a Quad Tree}
\label{sec:example_barnes_hut_sec}
\begin{figure}
  \centering
  \subfloat[Quad tree structure]{\includegraphics[width=0.7\textwidth]{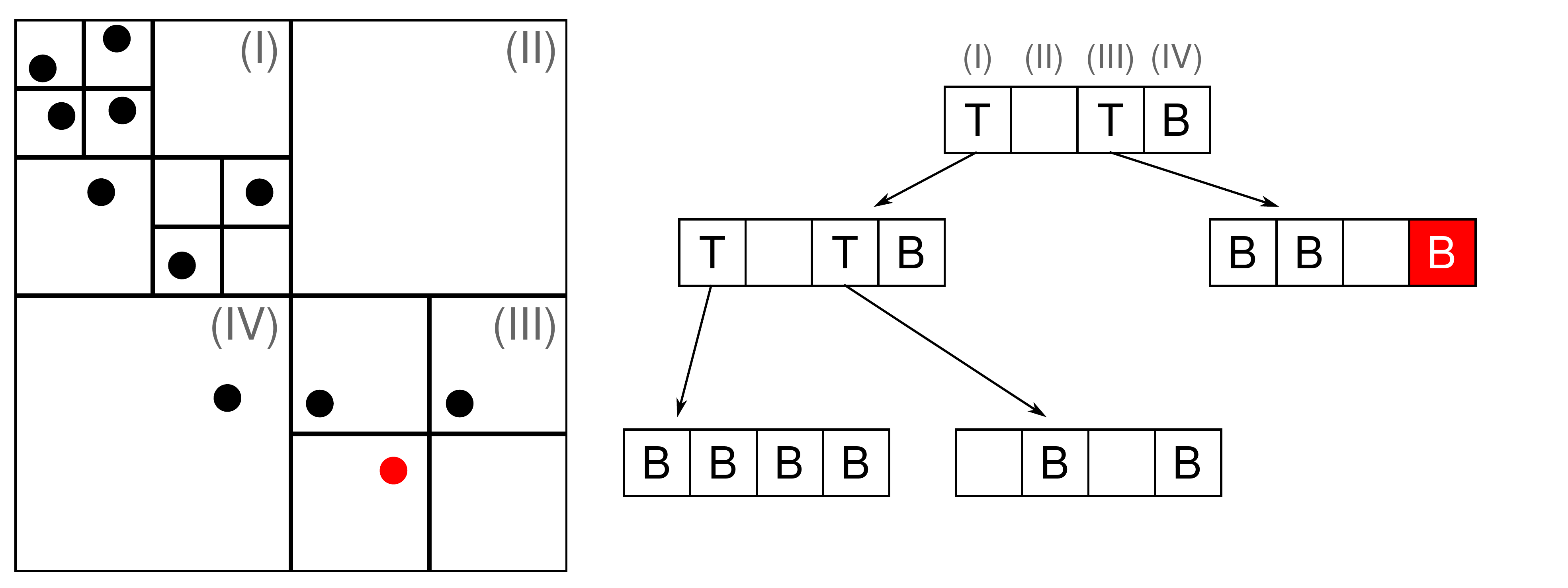}}\hfill\subfloat[Screenshot]{\includegraphics[width=0.27\textwidth]{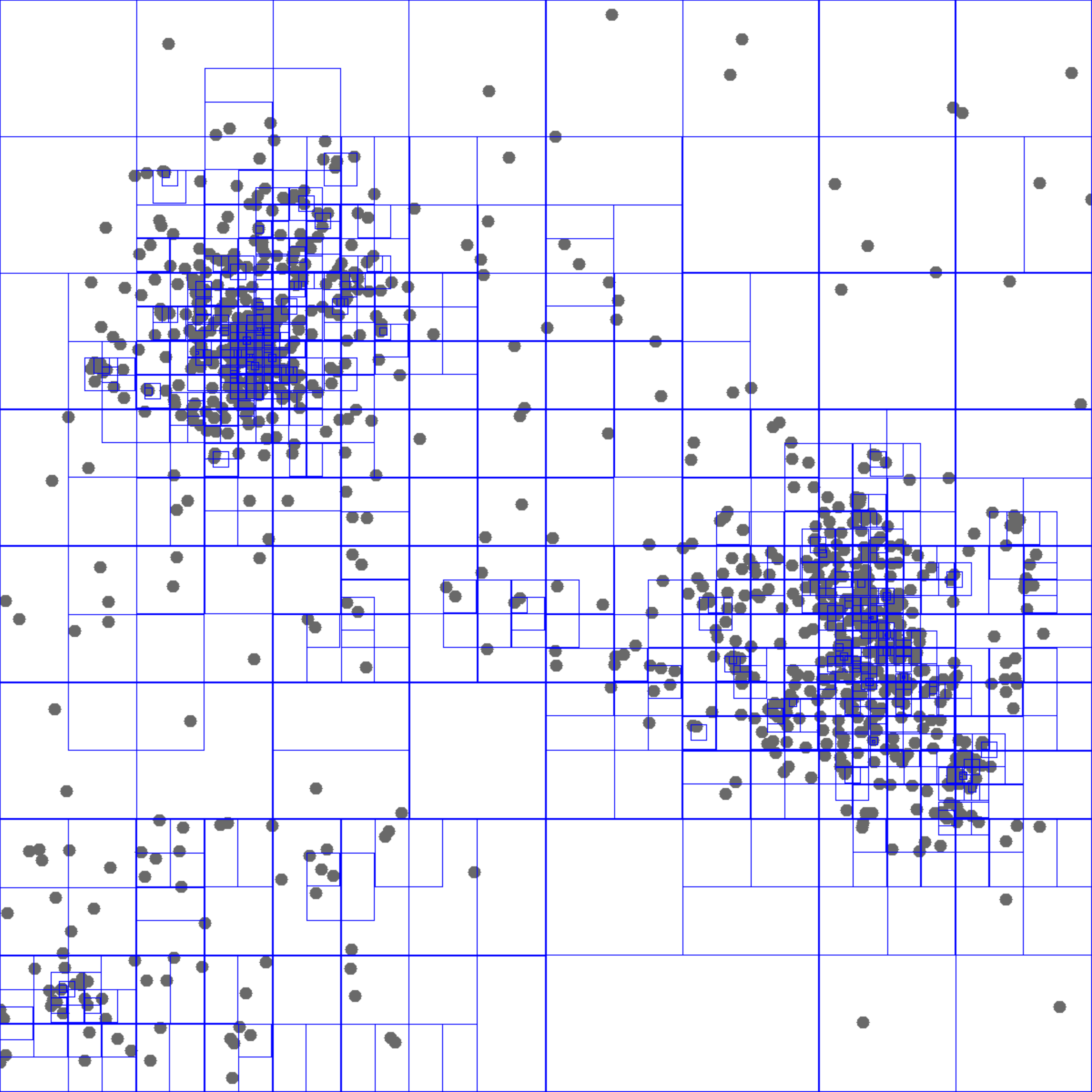}}
  \caption[\textsf{barnes-hut}: Quad tree structure and screenshot]{Barnes-Hut quad tree}
  \label{fig:ex_barnes_hut_quad_tree}
\end{figure}

Barnes-Hut~\cite{Barnes1986} is an approximation of the n-body simulation in Section~\ref{sec:smmo_nbody_sec71}. N-body is computationally expensive because it computes gravitational forces between all pairs of \texttt{Body} objects. Barnes-Hut recursively divides the simulation space into quadrants by building a quad tree (Figure~\ref{fig:ex_barnes_hut_quad_tree}). If a body is far enough away from a quadrant, the force of the quadrant on the body can be approximated by treating the entire quadrant as a single larger body instead of computing exact forces with every body inside the quadrant. The \texttt{device\_do} iteration of the original n-body simulation is then replaced with a quad tree traversal. We are designing a 2D n-body simulation in this section. A 3D n-body simulation would utilize an octree instead of a quad tree.

\subsection{Data Structure}
This application has three classes. An abstract class \texttt{NodeBase} and two subclasses \texttt{BodyNode} and \texttt{TreeNode}.

\begin{figure}
  \includegraphics[scale=0.75]{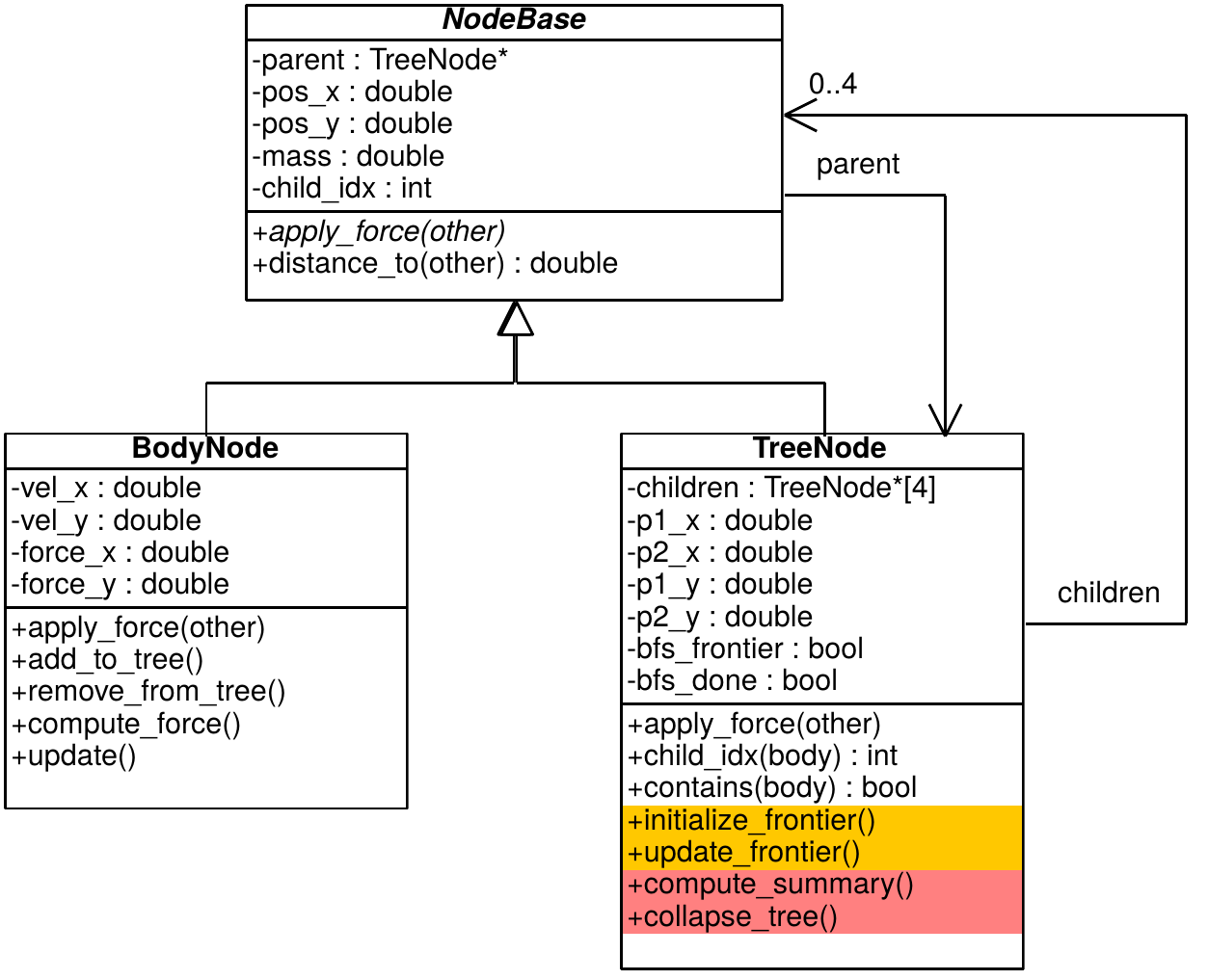}
  \centering
  \caption[\textsf{barnes-hut}: Data structure]{Data structure of \textsf{barnes-hut}}
  \label{lst:data_s_barnes_hut}
\end{figure}

The tree is made up of \texttt{TreeNode}s. Every \texttt{TreeNode} stores up to four children in the \texttt{children} array field. This array has one slot per quadrant (e.g., quadrant 1 = index 0). A child can be either another \texttt{TreeNode} or a \texttt{BodyNode}. Furthermore, every node stores a pointer to its parent \texttt{TreeNode}. \texttt{child\_idx} is the position of the node within its parent's \texttt{children} array. The root of the tree has no parent (\texttt{nullptr}).

Every node has a position and a mass. In case of a \texttt{BodyNode}, this is the position and mass of the body. In case of a \texttt{TreeNode}, this is the center of gravity and accumulated mass of all bodies in the subtree (\emph{summary}). As the position of bodies changes throughout the simulation, their position within the tree also changes. Therefore, tree summaries must be recomputed from time to time.

\texttt{TreeNode}s store two additional coordinates: \texttt{p1} is the left upper corner and \texttt{p2} is the lower right corner of the node. Moreover, \texttt{TreeNode}s have two boolean flags \texttt{bfs\_frontier} and \texttt{bfs\_done} that are used to implement bottom-up tree traversals, as described later. These field names start with \texttt{bfs\_} because the tree traversal is conceptually similar to the parallel frontier-based breadth-first search algorithm.

\subsection{Application Implementation}
This simulation consists of a larger number of parallel do-all operations. We first give a high-level overview of the simulation and then discuss selected do-all operations in more detail in the following paragraphs.

\begin{enumerate}
  \item \textbf{Tree Initialization:} Create the root of the tree (\texttt{TreeNode} object) and store it in a top-level variable \texttt{tree}. The tree has no children yet at this point.
  \item \textbf{Body Initialization:} Create a few random \texttt{BodyNode} objects with a parallel new operation, similar to the previous simulations. These bodies are not yet part of the tree, as indicated by \texttt{parent = nullptr}.
  \item \textbf{Inserting Bodies:} Insert those bodies into the tree that are not yet part of the tree (all bodies at this point). Parallel do-all: \texttt{BodyNode::add\_to\_tree}.
  \item \textbf{Iterative Algorithm:} Each iteration consists of the following steps.
  \begin{enumerate}
    \item \textbf{Computing Summaries:} Compute \texttt{TreeNode} summaries with a bottom-up tree traversal. Multiple parallel do-all operations.
    \item \textbf{Computing Forces:} Compute forces between bodies with top-down tree traversal. Parallel do-all: \texttt{BodyNode::compute\_force}.
    \item \textbf{Acceleration and Movement:} Compute a new velocity and position for each body. This step is identical to the second step of the original n-body simulation. Parallel do-all: \texttt{BodyNode::update}.
    \item \textbf{Removing Bodies:} Remove bodies from the tree if they moved into a different quadrant. Parallel do-all: \texttt{BodyNode::remove\_from\_tree}.
    \item \textbf{Reinserting Bodies:} Insert those bodies back into the tree that were just removed. Parallel do-all: \texttt{BodyNode::add\_to\_tree}.
    \item \textbf{Collapsing Tree:} Remove empty \texttt{TreeNode}s and collapse \texttt{TreeNode}s with only one \texttt{BodyNode} child. This is another bottom-up tree traversal and implemented with multiple parallel do-all operations.
  \end{enumerate}
\end{enumerate}

In the remainder of this section, we describe the design and implementation of all steps that require tree traversals or tree modifications. The latter ones are difficult to implement because multiple threads may be concurrently modifying the tree.

\paragraph{Step 4a: Computing Tree Summaries}
Before gravitational forces can be computed, we have to ensure that tree summaries are up to date. Step~4a and Step~4f are both implemented with a bottom-up tree traversal. Listing~\ref{lst:barnes_bottom_up_trav} shows the structure of such traversals.

Many parallel tree/graph traversals (e.g., BFS) are iterative, frontier-based algorithms. In such a traversal, each node has a boolean \emph{frontier flag} which indicates if a node is part of the frontier. Only frontier nodes are processed in an iteration. When all frontier nodes finished processing, the frontier is advanced, usually based on the results of the processing step. Alternatively, instead of a boolean flag, the frontier is sometimes defined as all nodes that have a certain property; e.g., all nodes with a certain distance (Section~\ref{sec:inner_arrays_perf_eval}).

In our bottom-up tree traversal, every node has two flags: A boolean frontier flag \texttt{bfs\_frontier} and a boolean \emph{done flag} \texttt{bfs\_done} which is set to \emph{true} after processing to ensure that a node is not processed multiple times. Our traversal algorithm visits only \texttt{TreeNodes}. All \texttt{TreeNode}s without \texttt{TreeNode} children (tree leaves) are part of the initial frontier. A processing step computes tree summaries for all frontier nodes and sets their done flag to \emph{true}. If a \texttt{TreeNode} has not been processed yet but all of its children are done processing, it will become part of the next frontier. As a side note, this is a common pattern of parallel graph processing and fits well with the vertex-based \emph{bulk-synchronous model}~\cite{Valiant:1990:BMP:79173.79181}, which is the foundation of graph processing frameworks such as Gunrock~\cite{Wang:2016:GHG:2851141.2851145}.

\begin{figure}
    \subfloat[\texttt{TreeNode::initialize\_frontier()}]{\includegraphics[width=0.45\textwidth]{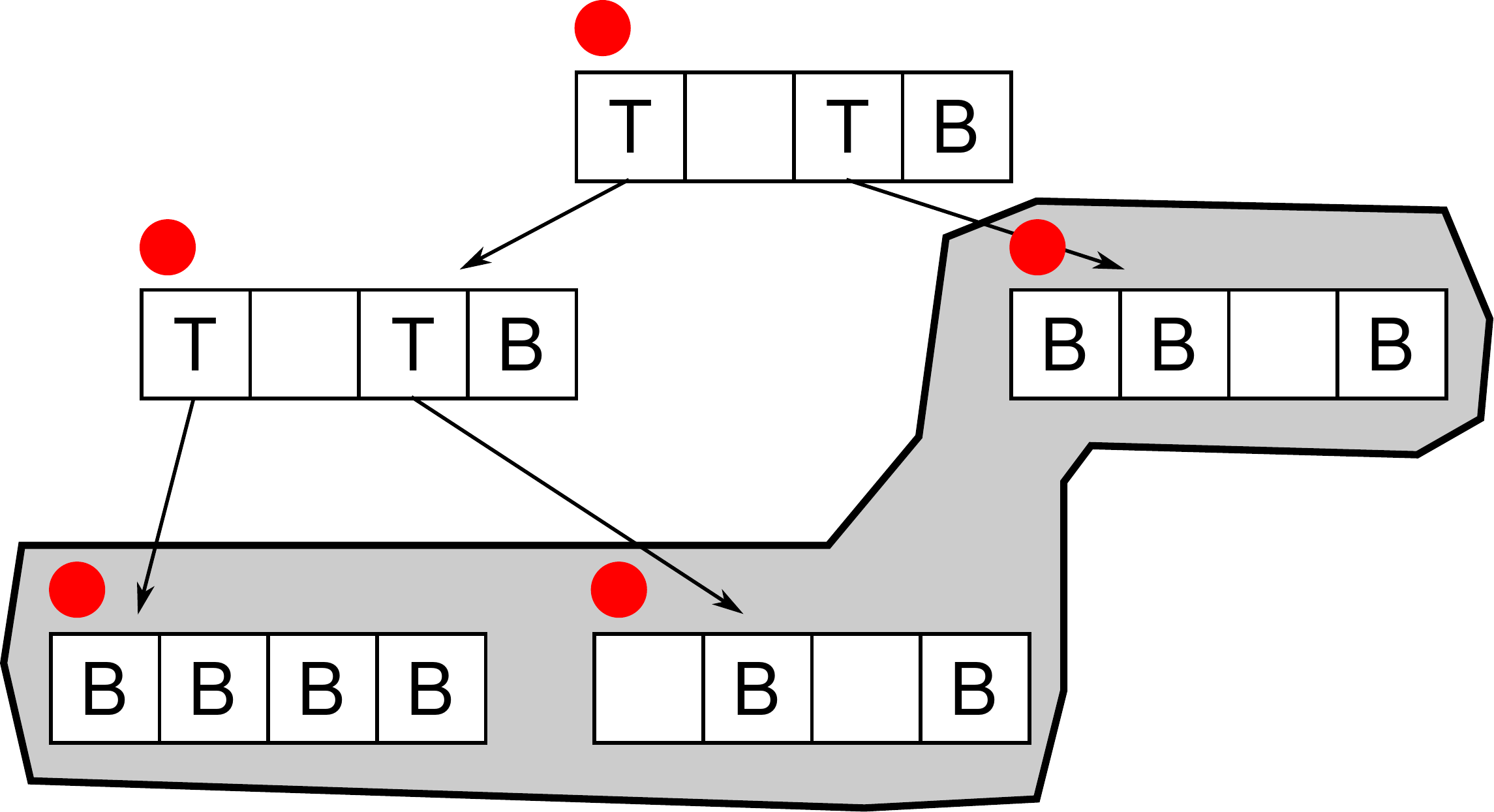}}\hfill
    \subfloat[\texttt{TreeNode::compute\_summary()}]{\includegraphics[width=0.45\textwidth]{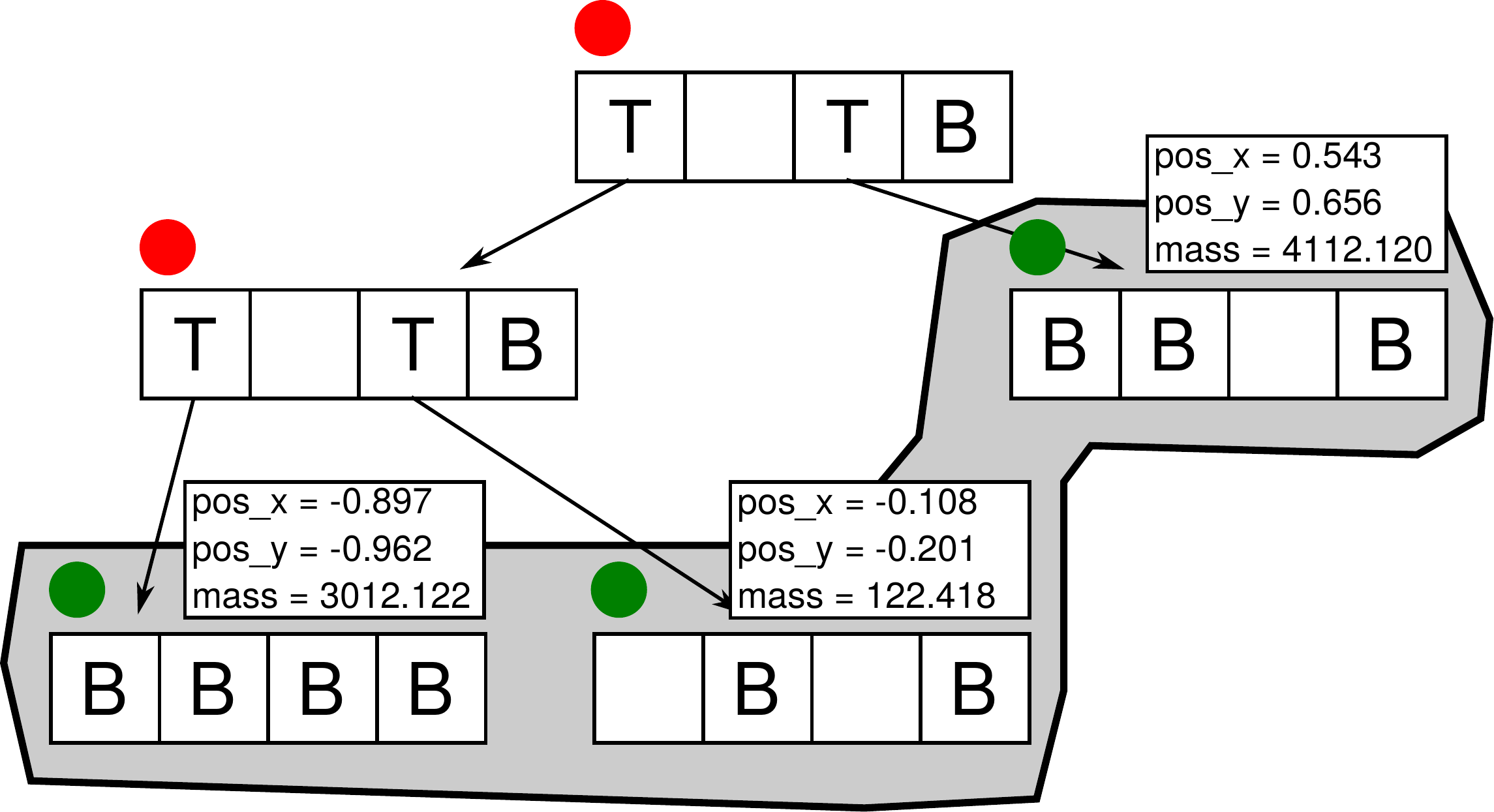}} \\
    \subfloat[\texttt{TreeNode::advance\_frontier()}]{\includegraphics[width=0.45\textwidth]{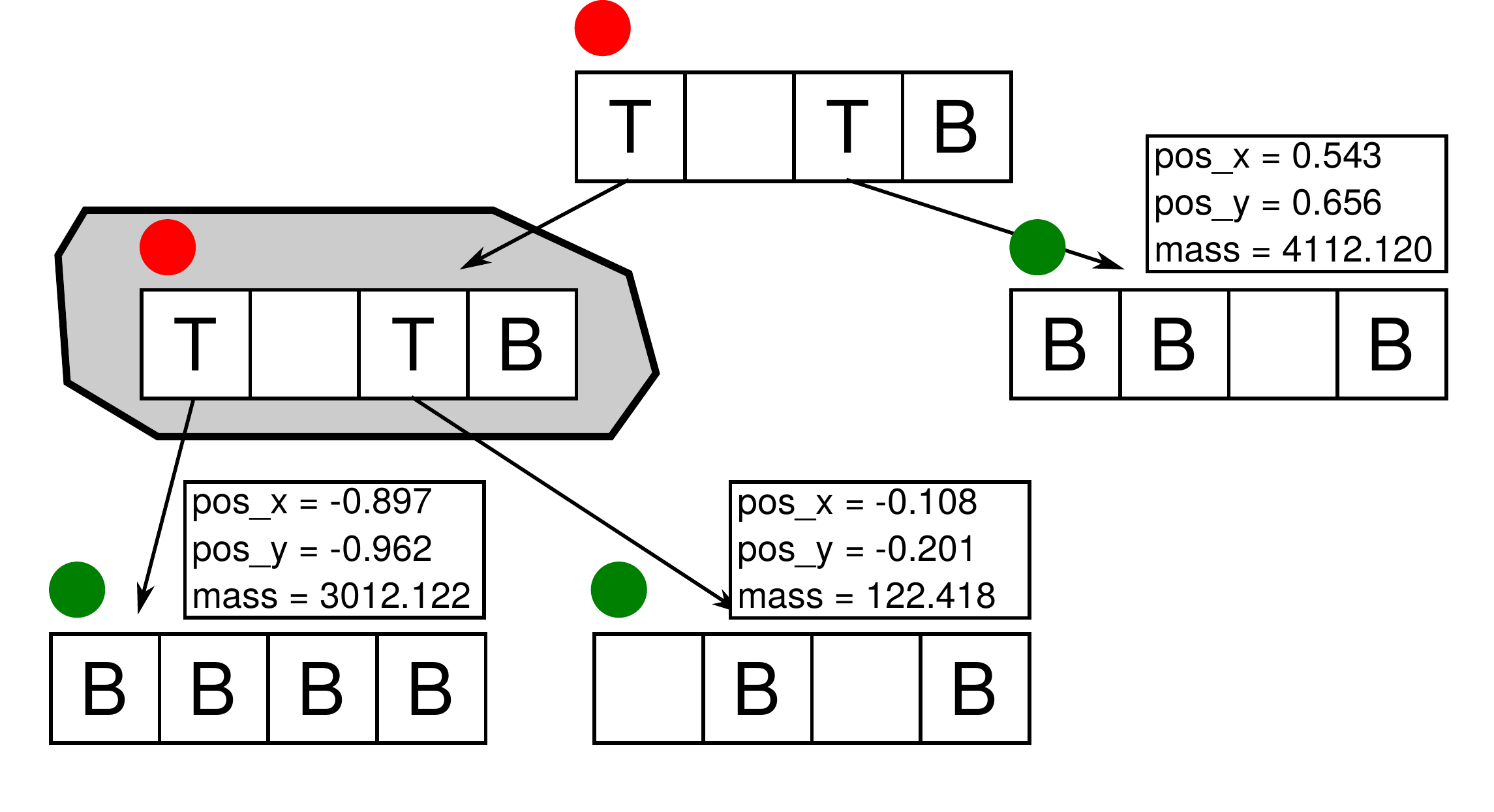}}\hfill
    \subfloat[\texttt{TreeNode::compute\_summary()}]{\includegraphics[width=0.45\textwidth]{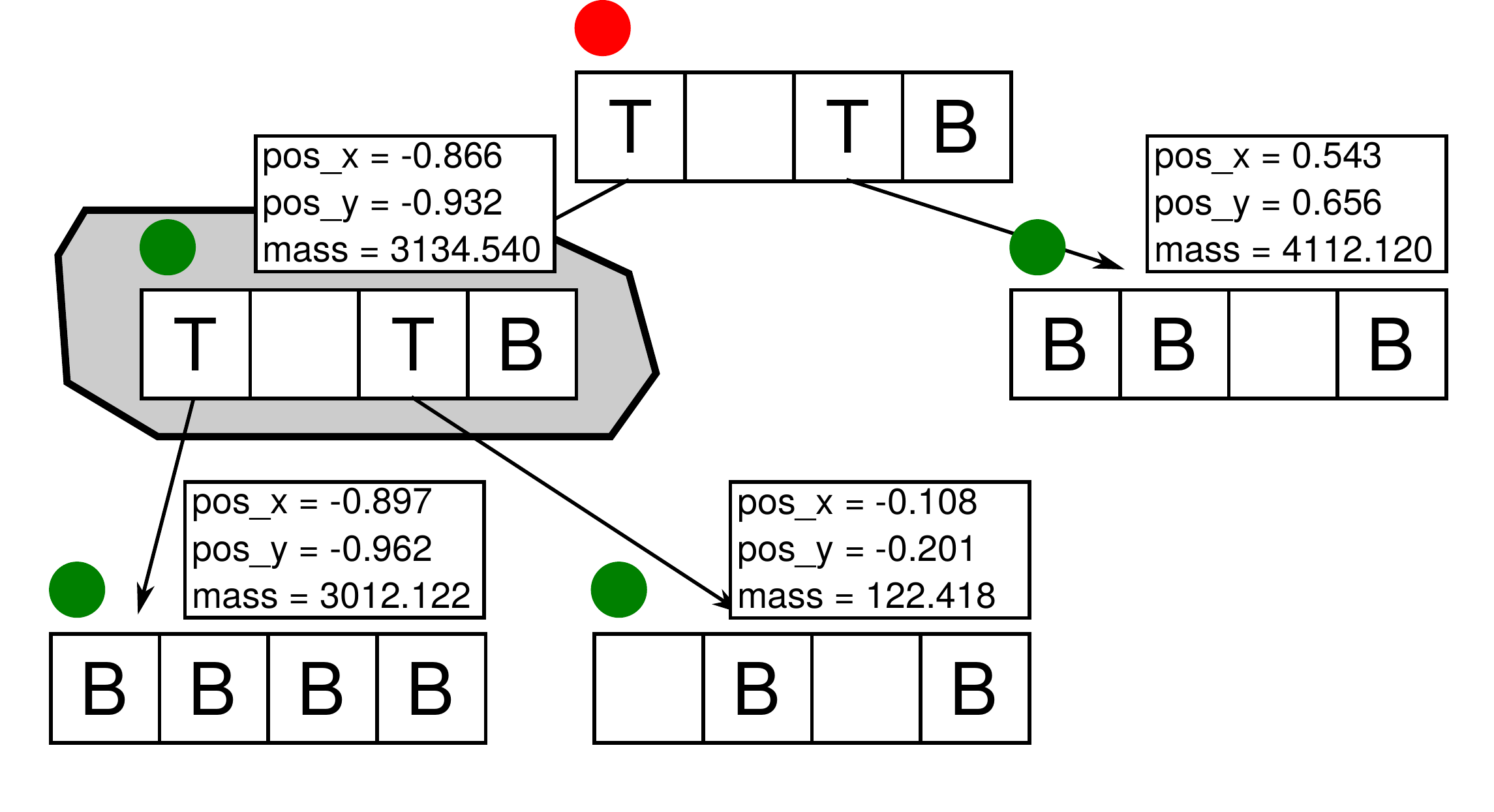}} \\
    \subfloat[\texttt{TreeNode::advance\_frontier()}]{\includegraphics[width=0.45\textwidth]{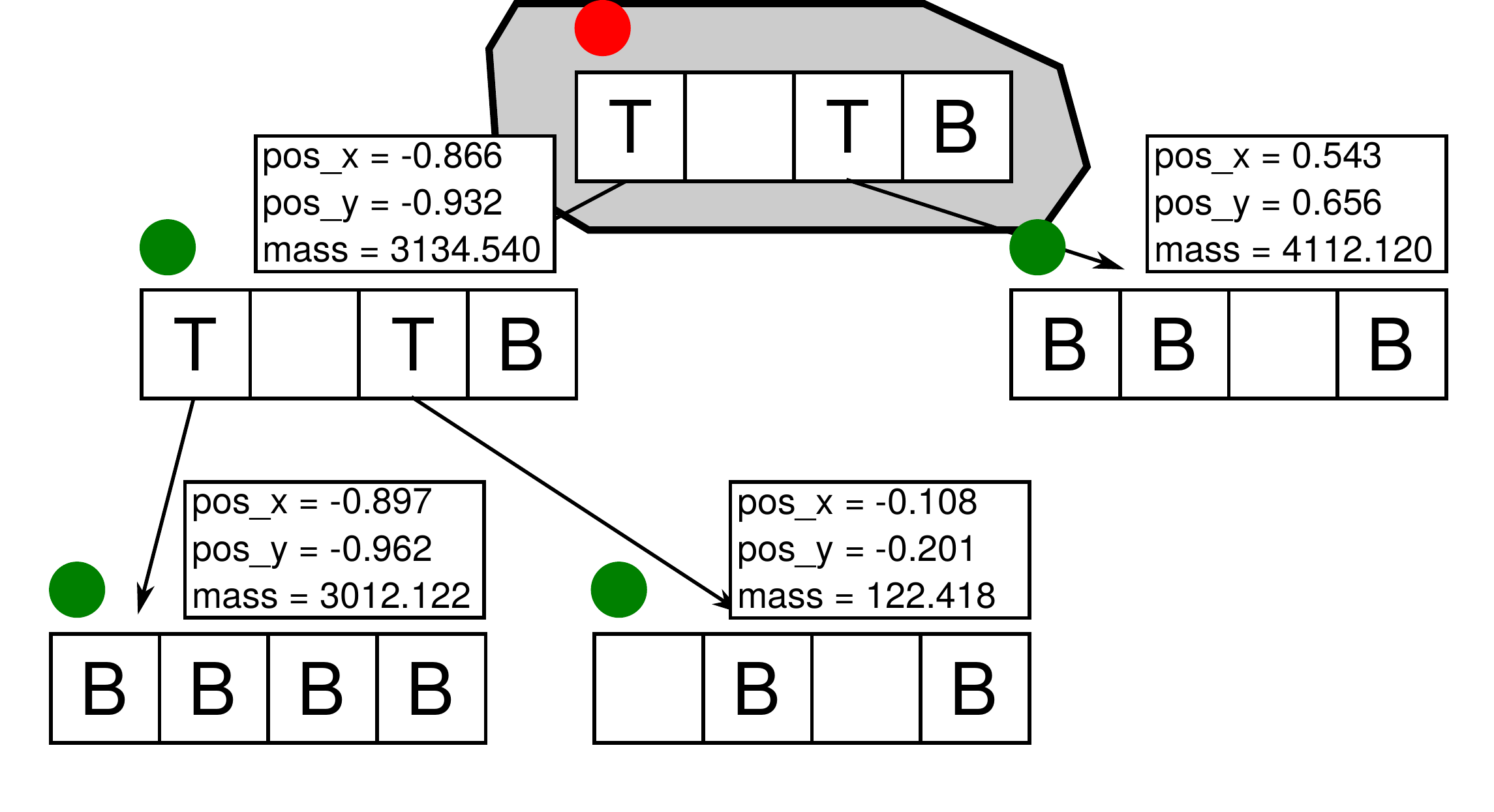}}\hfill
    \subfloat[\texttt{TreeNode::compute\_summary()}]{\includegraphics[width=0.45\textwidth]{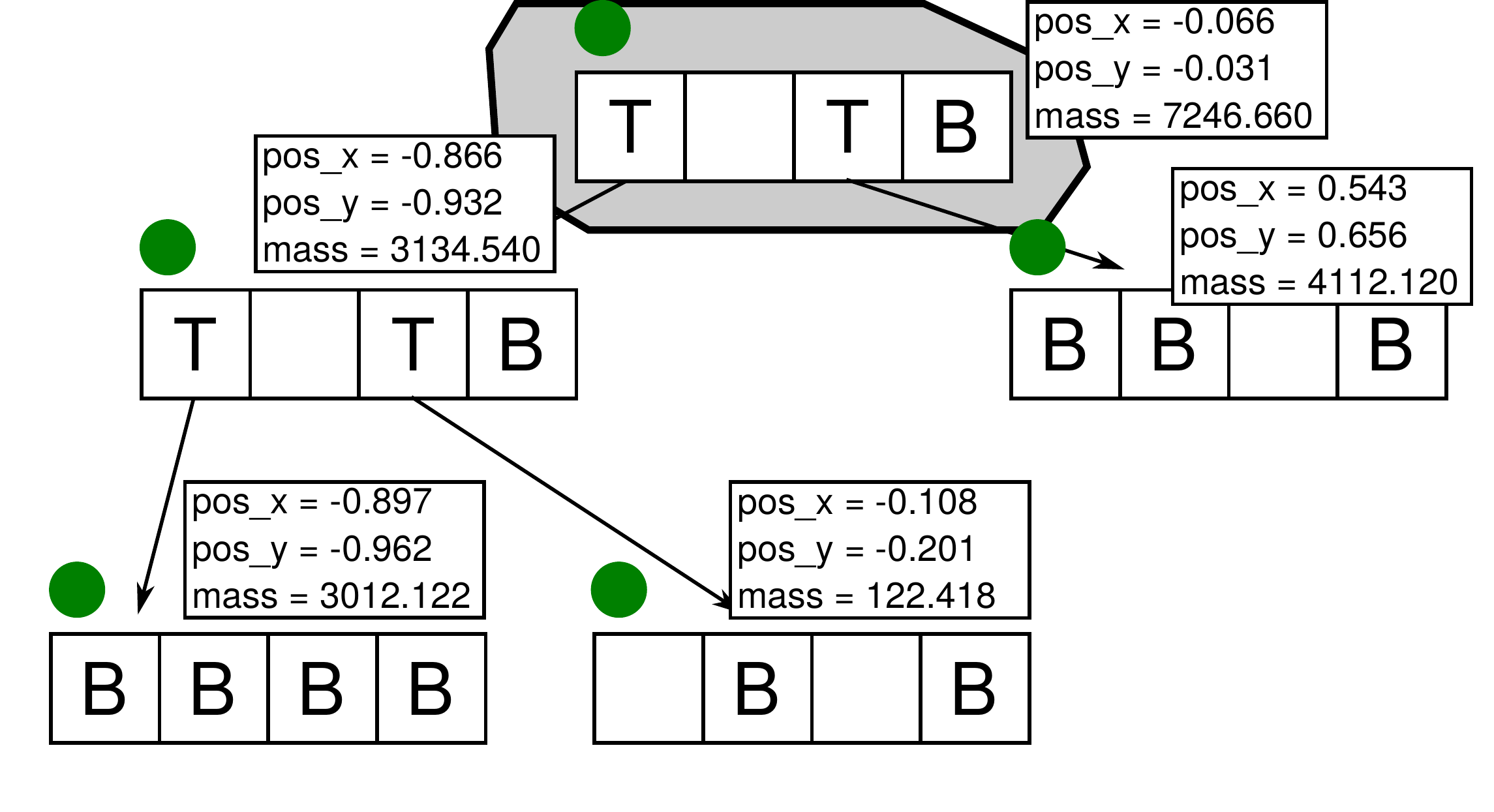}}
    \caption[\textsf{barnes-hut}: Computing quad tree summaries]{Compute summaries of quad tree}
    \label{fig:compute_summaries_of_quad}

\begin{lstlisting}[language=c++, caption={[\textsf{barnes-hut}: Bottom-up quad tree traversal]Bottom-up quad tree traversal pattern}, label={lst:barnes_bottom_up_trav}, morekeywords={__device__, nullptr}, numbers=none]
__device__ void TreeNode::initialize_frontier() {
  bfs_frontier_ = true;  bfs_done_ = false;
  for (int i = 0; i < 4; ++i) {
    if (children_[i]->cast<TreeNode>() != nullptr) { bfs_frontier_ = false; break; }
  }
}

__device__ void TreeNode::advance_frontier() {
  if (!bfs_done_) {
    for (int i = 0; i < 4; ++i) {
      TreeNode* child = children_[i]->cast<TreeNode>();
      if (child != nullptr && !child->bfs_done_) return;  // Unprocessed child found
    }
    bfs_frontier_ = true;
  } else { bfs_frontier_ = false; }
}

void bottom_up_traversal() {
  allocator_handle->parallel_do<TreeNode, &TreeNode::initialize_frontier>();
  do {
    // Do something: e.g., compute TreeNode summary or collape TreeNode
    allocator_handle->parallel_do<TreeNode, &TreeNode::advance_frontier>();
  } while (!host_tree->bfs_done_);
}
\end{lstlisting}
\end{figure}

Figure~\ref{fig:compute_summaries_of_quad} illustrates this algorithm with an example. The circles indicate the state of the \texttt{bfs\_done} flag (red = \emph{false}, green = \emph{true}). Initially, only nodes without \texttt{TreeNode} children are part of the frontier (Subfigure~\textsc{a}), as indicated by the shaded area. The first processing step computes tree node summaries of those nodes and sets their done flag (Subfigure~\textsc{b}). Then, the frontier is advanced: Only the left middle node becomes part of the next frontier because it is the only unprocessed node whose children are all processed (Subfigure~\textsc{c}). The algorithm proceeds in this pattern until the tree root was processed (Subfigure~\textsc{f}).

\paragraph{Step 4b: Computing Forces}
Barnes-Hut is faster than a regular n-body simulation because the computation of gravitational forces has a lower computational complexity. In n-body, the total force of all other bodies acting on a given body is computed with a \texttt{device\_do} operation (complexity $\theta(n)$, where $n$ is the number of bodies). In Barnes-Hut, this force is computed with a top-down tree traversal, which has the same worst-case complexity but a lower expected complexity. Our implementation uses a \emph{preorder depth-first traversal}, but other traversals would also work. 

\begin{lstfloat}
\begin{lstlisting}[language=c++, caption={[\textsf{barnes-hut}: Force computation via top-down quad tree traversal]Force computation via quad tree traversal}, label={lst:barnes_hut_apply_force}, morekeywords={__device__, nullptr}]
__device__ TreeNode* tree;  // Pointer to tree root
TreeNode* host_tree;        // Pointer to tree root in host memory

__device__ void BodyNode::apply_force(BodyNode* other) {
  // Same as Body::apply_force in n-body.
}

__device__ void TreeNode::apply_force(BodyNode* other) {
  if (contains(other) || distance_to(other) <= kDistThreshold) {
    // Too close or inside. Recurse.
    for (int i = 0; i < 4; ++i) {
      if (children_[i] != nullptr) { children_[i]->apply_force(other); }
    }
  } else {
    // Far enough away to use approximation. Same as BodyNode::apply_force.
  }
}

__device__ void BodyNode::compute_force() {
  force_x_ = force_y_ = 0.0f;
  tree->apply_force(this);
}
\end{lstlisting}
\end{lstfloat}

Listing~\ref{lst:barnes_hut_apply_force} shows how gravitational forces are computed with a quad tree traversal. We assume that all \texttt{TreeNode} summaries (i.e., fields \texttt{pos\_x}, \texttt{pos\_y}, \texttt{mass}) are up to date. Instead of a \texttt{device\_do}, we start the tree traversal at the root of the tree (Line~21).

\texttt{NodeBase::apply\_force} is a pure virtual function. Let \texttt{this} be the object that the function is bound to. If \texttt{this} is a \texttt{BodyNode}, the function adds the gravitational force induced by the body to the \texttt{force\_} fields of \texttt{other}, similar to the original n-body simulation. If \texttt{this} is a \texttt{TreeNode}, there are two possibilites:

\begin{enumerate}
  \item If \texttt{other} is contained in \texttt{this} or if \texttt{other} is close to the \texttt{this}'s center of gravity, we have to recurse and do a more accurate computation with every child of \texttt{this}. Using a Barnes-Hut approximation would result in a too large error.
  \item If \texttt{other} is not within the bounds of \texttt{this} and far enough away from its center of gravity, we can approximate the gravitational force induced by the entire \texttt{TreeNode} using the \texttt{this}'s summary.
\end{enumerate}

Since \textsc{DynaSOAr} does not support virtual function calls yet, we had to implement this function with a hand-written \emph{switch-case} statement (Section~\ref{sec:virtual_func_calls_barnes}). We plan to auto-generate such code in future versions of \textsc{DynaSOAr}.

\paragraph{Step 4d: Removing Bodies}
We are now removing \texttt{BodyNode}s from the tree if they moved into a different quadrant of their parent \texttt{TreeNode} or if they moved into a different \texttt{TreeNode}. This is easy to detect based on the position of the body and corner points \texttt{p1}, \texttt{p2} of the \texttt{TreeNode}. Such bodies are removed from the \texttt{children} array and their \texttt{parent} pointer is reset to \texttt{nullptr}.

\begin{figure}
  \centering
  \includegraphics[width=0.7\textwidth]{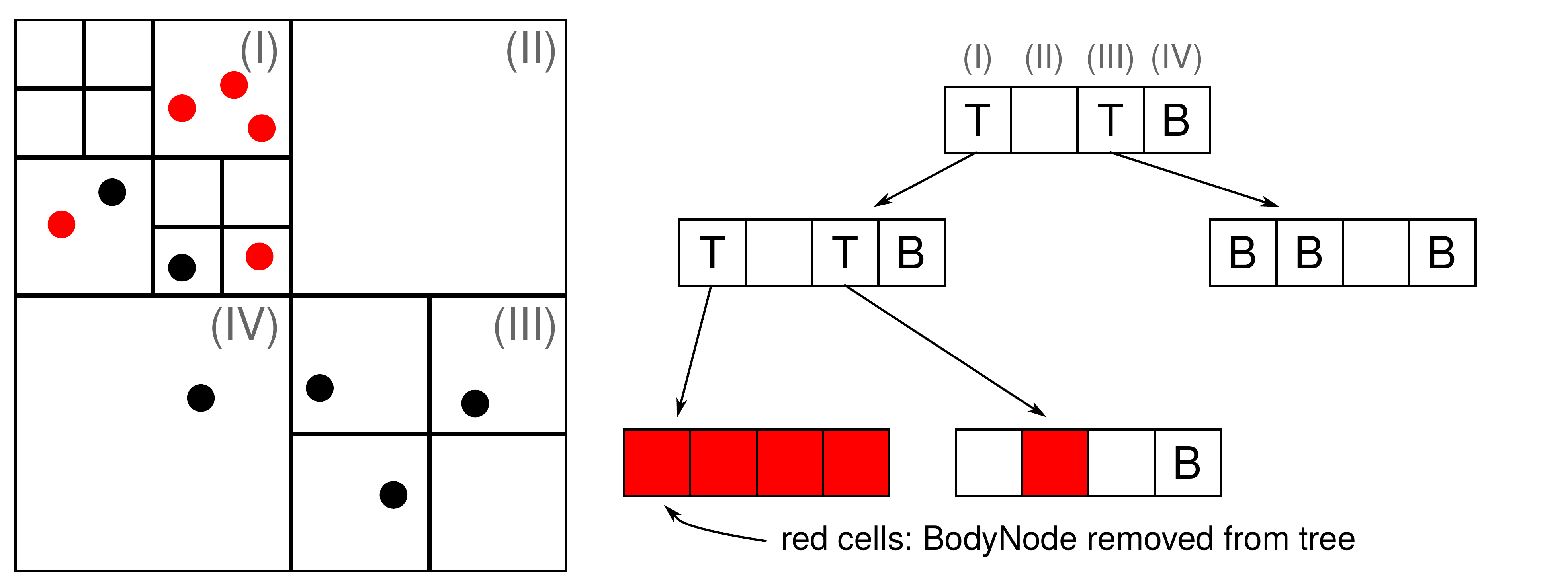}
  \caption[\textsf{barnes-hut}: Removing bodies from the quad tree]{Removing \texttt{BodyNode}s from the quad tree}
  \label{fig:ex_barnes_hut_quad_tree_remove}
\end{figure}

Figure~\ref{fig:ex_barnes_hut_quad_tree_remove} shows the quad tree of Figure~\ref{fig:ex_barnes_hut_quad_tree}\textsc{a} after bodies were moved (Step~4c). The red bodies moved into a different quadrant or parent and are removed from the tree. This may leave some \texttt{TreeNode}s empty.

\paragraph{Step 3 / Step 4e: Inserting Bodies}
We are now (re)inserting \texttt{BodyNode}s into the quad tree which have no parent. In Step~3, these are all newly created bodies. In Step~4e, these are all bodies that were removed in the previous step. This step is the most challenging part of \textsf{barnes-hut}, because the structure of the quad tree is modified concurrently by multiple threads.

Listing~\ref{lst:barnes_insertion_ex} shows the insertion algorithm. Only \texttt{BodyNode}s without a parent are processed. The algorithm maintains a \texttt{current} pointer to the \texttt{TreeNode} into which the body should be inserted. This pointer is initialized to the root of the tree. The \emph{while} loop of Line~4 traverses the tree structure by determining the \texttt{children} array slot into which the body should be inserted (Line~5) and updating the \texttt{current} pointer if necessary. At this point, we have to consider three cases.

\begin{enumerate}
  \item The selected \textbf{array slot is empty} (Line~7, Figure~\ref{fig:ex_barnes_hut_quad_tree_insert1}). In this case, we can directly insert the body into the array. However, since multiple threads may be attempting to insert into the same slot, we insert the body with an atomic compare-and-swap operation, such that only one thread can succeed. Afterwards, we set the \texttt{parent} pointer with an atomic exchange operation. This is to ensure that the modification to the \texttt{parent} pointer is guaranteed to become visible to other threads in the CUDA kernel (\emph{volatile write}). If a thread fails to insert a body because another thread succeeded with a contending compare-and-swap operation, the thread retries with another \emph{while} loop iteration.
  \item The selected \textbf{array slot contains a \texttt{TreeNode}} (Line~12). In this case, we try to insert the body into that \texttt{TreeNode} in the next \emph{while} loop iteration.
  \item The selected \textbf{array slot contains a \texttt{BodyNode}}, denoted by \texttt{other} (Line~14, Figure~\ref{fig:ex_barnes_hut_quad_tree_insert2}). In this case, we have to insert a new \texttt{TreeNode} into the tree. The method \texttt{make\_child\_ tree\_node} creates a new \texttt{TreeNode} with the correct \texttt{p1}, \texttt{p2} and \texttt{parent} values, but does not insert it into \texttt{current} yet. We first insert \texttt{other} into the newly created (empty) \texttt{TreeNode}\footnote{The thread fence in Line~21 ensures that other threads see \texttt{other}'s slot in the \texttt{children} array of the new \texttt{TreeNode} as occupied before the new \texttt{TreeNode} becomes visible through the CAS in Line~23.}. Now we swap in the new \texttt{TreeNode} with an atomic compare-and-swap operation. Similar to Case~1, only one thread can succeed in modifying the respective \texttt{children} array slot. If successful, we update the \texttt{parent} pointer of \texttt{other}. However, this body may still be in the process of insertion (Case~1 or \texttt{other} in Case~3) and its \texttt{parent} pointer may not have been set yet. Therefore, we require another \emph{while} loop to retry until \texttt{parent} was updated. Now we can insert the original body into the new \texttt{TreeNode} by running another iteration of the outer \emph{while} loop.
\end{enumerate}

Note that the outer \emph{while} loop of Listing~\ref{lst:barnes_insertion_ex} maintains a boolean flag \texttt{is\_true} instead of an infinite loop with a break or return statement. This is to avoid deadlocks, since the structure of this loop is similar to the one of the critical section implementation of Section~\ref{sec:cuda_prog_model_backgr}.

\begin{figure}
  \centering
  \includegraphics[width=0.7\textwidth]{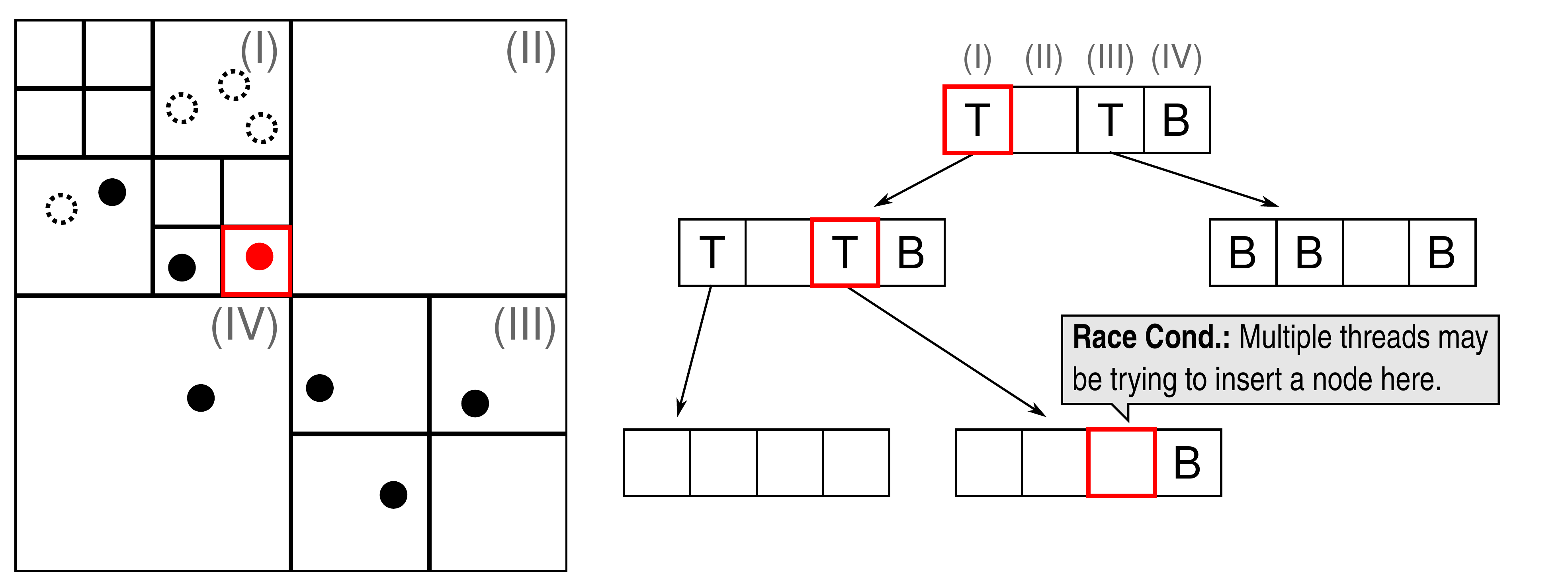}
  \caption[\textsf{barnes-hut}: Inserting a body in an empty slot]{Inserting a \texttt{BodyNode} into an empty slot}
  \label{fig:ex_barnes_hut_quad_tree_insert1}
  \vspace{10pt}

  \centering
  \includegraphics[width=0.7\textwidth]{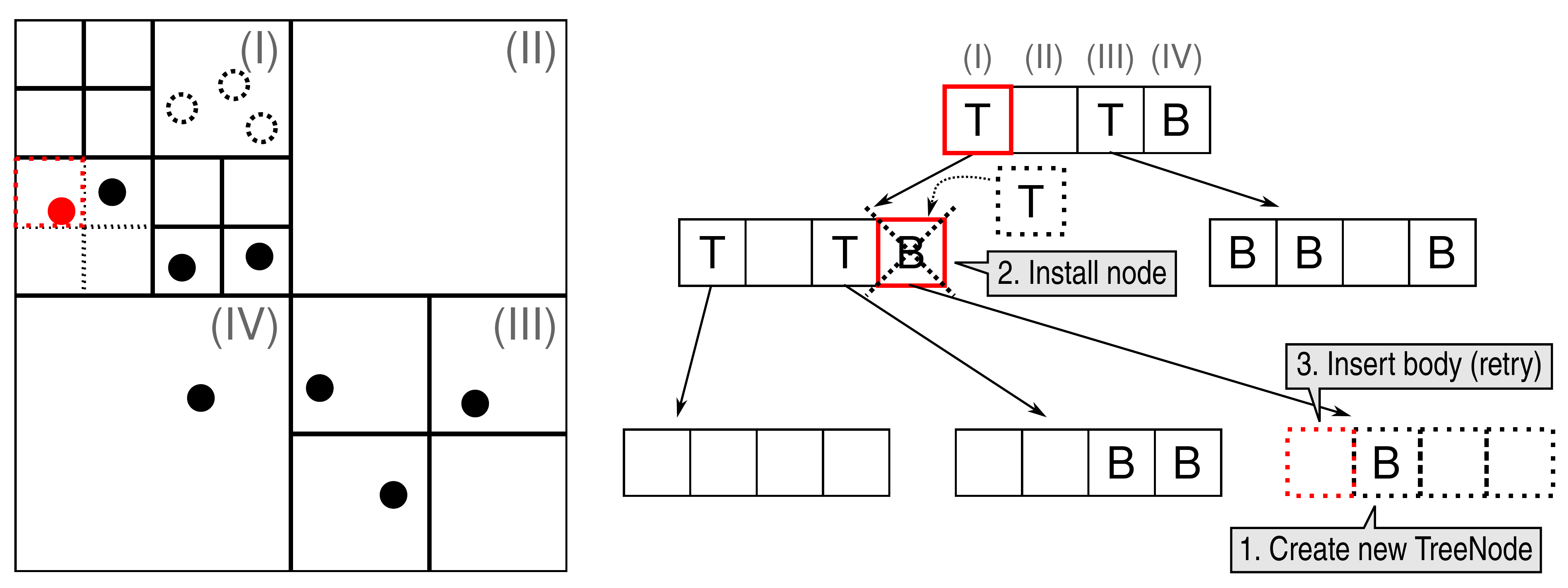}
  \caption[\textsf{barnes-hut}: Creating a new quad tree node]{Inserting a \texttt{BodyNode} into a slot with another \texttt{BodyNode}}
  \label{fig:ex_barnes_hut_quad_tree_insert2}

\begin{lstlisting}[language=c++, caption={[\textsf{barnes-hut}: Inserting a body into the quad tree]Inserting a body into the quad tree}, label={lst:barnes_insertion_ex}, morekeywords={__device__, nullptr}]
__device__ void BodyNode::add_to_tree() {
  if (parent_ == nullptr) {
    TreeNode* current = tree;  bool is_done = false;
    while (!is_done) {  // Check where to insert in this node.
      int c_idx = current->child_idx(this);
      auto*& child_ptr = current->children_[c_idx];
      if (child == nullptr) {  // Slot not in use.
        if (atomicCAS<NodeBase>(&child, nullptr, this) == nullptr) {
          atomicExch(&parent_, current);
          child_index_ = c_idx;  is_done = true;  // Done inserting.
        }
      } else if (child->cast<TreeNode>() != nullptr) {  // There is a subtree here.
        current = static_cast<TreeNode*>(child);
      } else {  // There is a body "other" here.
        BodyNode* other = static_cast<BodyNode*>(child);
        // Replace BodyNode with TreeNode.
        auto* new_node = current->make_child_tree_node(c_idx);
        // Insert other into new node.
        int other_c_idx = new_node->child_idx(other);
        new_node->children_[other_c_idx] = other;
        __threadfence();
        // Try to install the new TreeNode. (Retry.)
        if (atomicCAS<NodeBase>(&child, other, new_node) == other) {
          // It may take a while until we see the correct parent, because
          // another may not be done inserting this node yet.
          TreeNode* parent_before = nullptr;
          do {
            parent_before = atomicCAS<TreeNode>(&other->parent_, current, new_node);
          } while (parent_before != current);
          other->child_index_ = other_c_idx;
          current = new_node;  // Now insert body into new_node.
        } else { destroy(device_allocator, new_node); }  // Rollback.
      }
    }
  }
}
\end{lstlisting}
\end{figure}

\paragraph{Step 4f: Collapsing Tree}
As the final step of a \textsf{barnes-hut} iteration, we clean up the quad tree. Empty \texttt{TreeNode}s are removed and \texttt{TreeNode}s with only one child, which is a \texttt{BodyNode}, are collapsed with the parent node. This step is implemented as a bottom-up tree traversal, similar to Step~4a. Listing~\ref{lst:barnes_tree_collapse} shows the \emph{processing} step of the tree traversal. Only frontier \texttt{TreeNode}s $f$ are processed. We now have to consider four cases.

\begin{enumerate}
  \item The node $f$ has \textbf{zero children} (Line~16). In this case, we unregister $f$ from its parent and delete it.
  \item The node $f$ has \textbf{one child $c$ which is a \texttt{BodyNode}} (Line~24). In this case, we store the $c$ in the \texttt{children} array slot of the parent where $f$ is currently stored. We then delete $f$.
  \item The node $f$ has \textbf{one child $c$ which is a \texttt{TreeNode}}. We cannot collapse this node. The subtree $c$ must contain more than one \texttt{BodyNode}; otherwise, it would have been collapsed already by the previous bottom-up iteration. Since every \texttt{TreeNode} level divides the space into four equally-sized quadrants, collapsing $f$ would require changes to the \texttt{TreeNode}s within $c$. Moreover, at least two \texttt{BodyNode}s would end up as direct children in the same quadrant of some \texttt{TreeNode} within $c$, which is forbidden.
  \item The node $f$ has \textbf{more than one child}. We cannot collapse this node.
\end{enumerate}

Note that removing nodes from a tree is much simpler than adding nodes. No atomic operations or other synchronization primitives are necessary.

\begin{lstfloat}
\begin{lstlisting}[language=c++, caption={[\textsf{barnes-hut}: Collapsing the quad tree]Collapsing a \texttt{TreeNode}}, label={lst:barnes_tree_collapse}, morekeywords={__device__, nullptr}]
__device__ void TreeNode::collapse_tree() {
  if (bfs_frontier_) {
    bfs_frontier_ = false;

    // Count children.
    int num_children = 0;
    NodeBase* single_child = nullptr;

    for (int i = 0; i < 4; ++i) {
      if (children_[i] != nullptr) {
        ++num_children;
        single_child = children_[i];
      }
    }

    if (num_children == 0) {  // Remove node without children.
      if (parent_ != nullptr) {  // Do not remove the root.
        parent_->children_[child_idx_] = nullptr;
        destroy(device_allocator, this);
        return;
      }
    } else if (num_children == 1) {  // Collapse TreeNodes with 1 child...
      if (parent_ != nullptr) {
        if (single_child->cast<BodyNode>() != nullptr) {
          // ... but only if the node is a body.
          single_child->parent_ = parent_;
          single_child->child_index_ = child_index_;
          parent_->children_[child_index_] = single_child;

          destroy(device_allocator, this);
          return;
        }  // else: TreeNode child cannot be collapsed.
      }
    }

    // Done processing this node.
    bfs_done_ = true;
  }
}
\end{lstlisting}
\end{lstfloat}

\subsection{Virtual Function Calls}
\label{sec:virtual_func_calls_barnes}
Unfortunately, \textsc{DynaSOAr} does not yet support C++ abstractions for virtual functions. Therefore, we have to implement the dispatch logic of the virtual function \texttt{NodeBase::apply\_force} in Step~4b by hand. Our implementation consists of an explicit runtime type check (\texttt{cast<T>}) and an \emph{if-then-else} statement that dispatches to the correct method implementation (Listing~\ref{lst:handwritten_virt_meth_call}). \texttt{cast<T>} is a method provided by \textsc{DynaSOAr}'s data layout DSL. Its semantics are similar to C++'s \texttt{dynamic\_cast<T*>}, but it determines the runtime type of an object from its address (fake pointer; Section~\ref{sec:dynasoar_cpp_datalayout_dsl_fp}). This implementation is more efficient than C++'s \texttt{dynamic\_cast<T*>} because it does not have to read a vtable pointer from memory.

\begin{lstfloat}
\begin{lstlisting}[language=c++, numbers=none, caption={[\textsf{barnes-hut}: Handwritten virtual method call]Handwritten virtual method call}, label={lst:handwritten_virt_meth_call}, morekeywords={__device__, assert, nullptr}]
__device__ void NodeBase::apply_force(BodyNode* other) {
  if (cast<BodyNode>() != nullptr) {
    static_cast<BodyNode*>(this)->apply_force(other);
  } else {
    assert(cast<TreeNode>() != nullptr);
    static_cast<TreeNode*>(this)->apply_force(other);
  }
}
\end{lstlisting}
\end{lstfloat}

A \emph{switch-case} statement-based implementation of a virtual function call is compiled to code that is more efficient than a vtable-based implementation, because it allows compilers to inline virtual method calls. This results in more efficient GPU binary code.

GPU programs usually consist of a single compilation unit. Therefore, it is possible to enumerate and inline all possible runtime types (subtypes of the static receiver type) in the source code. We plan to automate this process in the future by auto-generating \emph{switch-case} statements as part of the compilation process.

\subsection{Benefits of Object-oriented Programming}
Implementing \textsf{barnes-hut} without object-oriented programming is tedious. One particular problem is the dynamic allocation of new \texttt{TreeNode}s in Line~17 of Listing~\ref{lst:barnes_insertion_ex}. Our implementation optimistically alloates a new \texttt{TreeNode} and tries to insert it with an atomic compare-and-swap operation. If this operation fails, we delete the object again. Related work describes an alternative implementation without dynamic memory allocation but with a lightweight lock that prevents two threads from inserting a node into the same location~\cite{BURTSCHER201175}. While such an implementation may be a bit faster than our implementation, locking is problematic on GPUs and such locks must be carefully designed and implemented to avoid deadlocks (Section~\ref{sec:cuda_prog_model_backgr}).

Implicit type information is another benefit of object-oriented programming in this application. A non-OOP implementation would have to store an explicit type identifier in addition to an object ID for each child in \texttt{TreeNode::children}. This is necessary because a child could be a \texttt{BodyNode} or another \texttt{TreeNode}. We do not have native support for virtual functions in \textsc{DynaSOAr} yet, but virtual function calls would be another benefit of an OOP-based implementation.

Finally, consider the field access in Line~28 of Listing~\ref{lst:barnes_tree_collapse}. Listing~\ref{lst:field_acc_not_barnes_hut} shows the same functionality in a hand-written SOA layout. Without OOP abstractions, such code is much harder to write/maintain, especially due to the nested array accesses.

\begin{lstfloat}
\begin{lstlisting}[language=c++, numbers=none, caption={\textsf{barnes-hut}: Field access notation in hand-written SOA layout}, label={lst:field_acc_not_barnes_hut}]
TreeNode_children[TreeNode_child_idx[id]][TreeNode_parent[id]] = single_child;
\end{lstlisting}
\end{lstfloat}

\subsection{Further Optimizations}
Related work describes an optimization to speed up Step~4b, which computes forces with a quad tree traversal. The performance of this step can be improved by assigning spatially local bodies to the same warp~\cite{BURTSCHER201175}. This could be achieved by physically rearranging/sorting \texttt{Body} objects. Bodies which are spatially local traverse similar parts of the quad tree, which improves memory accesses (threads of a warp access similar addresses) and reduces warp divergence (similar tree traversals result in similar control flow). Apart from \textsc{CompactGpu}'s memory defragmentation, \textsc{DynaSOAr} does currently not allow programmers to rearrange allocations in memory. Future work could extend \textsc{DynaSOAr} in that direction.

Another problem of our \textsf{barnes-hut} implementation is the recursive nature of Step~4b. If the mass and position of a \texttt{TreeNode} cannot be approximated, we have to recursively visit up to 4 child nodes (Listing~\ref{lst:barnes_hut_apply_force}, Line~12). Unfortunately, this pattern is \emph{not} tail recursive. Therefore, the compiler has to allocate stack frames for the recursive method calls in local memory, which resides in the (slow) global memory. However, the only state that we conceptually have to maintain is a pointer to the current \texttt{TreeNode} and an array index into the \texttt{children} array. This information is sufficient for implementing a preorder tree traversal, so it should be possible to transform this recursion into an iterative implementation that does not allocate new stack frames.

\section{\textsf{structure}: Finite Element Method}
\label{sec:example_structure_sec7}
\textsf{structure} is a finite element method (FEM), inspired by a problem in material science: Simulating the formation of a crack in a composite material~\cite{LU2018240}. \textsf{structure} simulates a mesh of elements, which can be seen as an undirected graph. Every graph node is connected by springs with up to three other nodes.

The simulation has three types of nodes: Ordinary \emph{nodes}, \emph{anchor nodes} and \emph{pull nodes}. Pull nodes move into a certain direction with a constant velocity and \emph{anchor nodes} are fixed at a certain position. Ordinary nodes can move freely. Springs have an initial length and a stiffness $k$. If a spring is stretched beyond its initial length, it exerts a pulling force $F$ on both node endpoints according to Hooke's law.

\begin{align*}
F = k \cdot \Delta x \tag{\emph{Hooke's law}}
\end{align*}

Pull nodes stretch the entire mesh in a certain direction. As soon as the force $F$ between two nodes exceeds a certain threshold, the spring breaks, forming a crack in the material.

\subsection{Data Structure}
\begin{figure}
  \centering
  \includegraphics[scale=0.75]{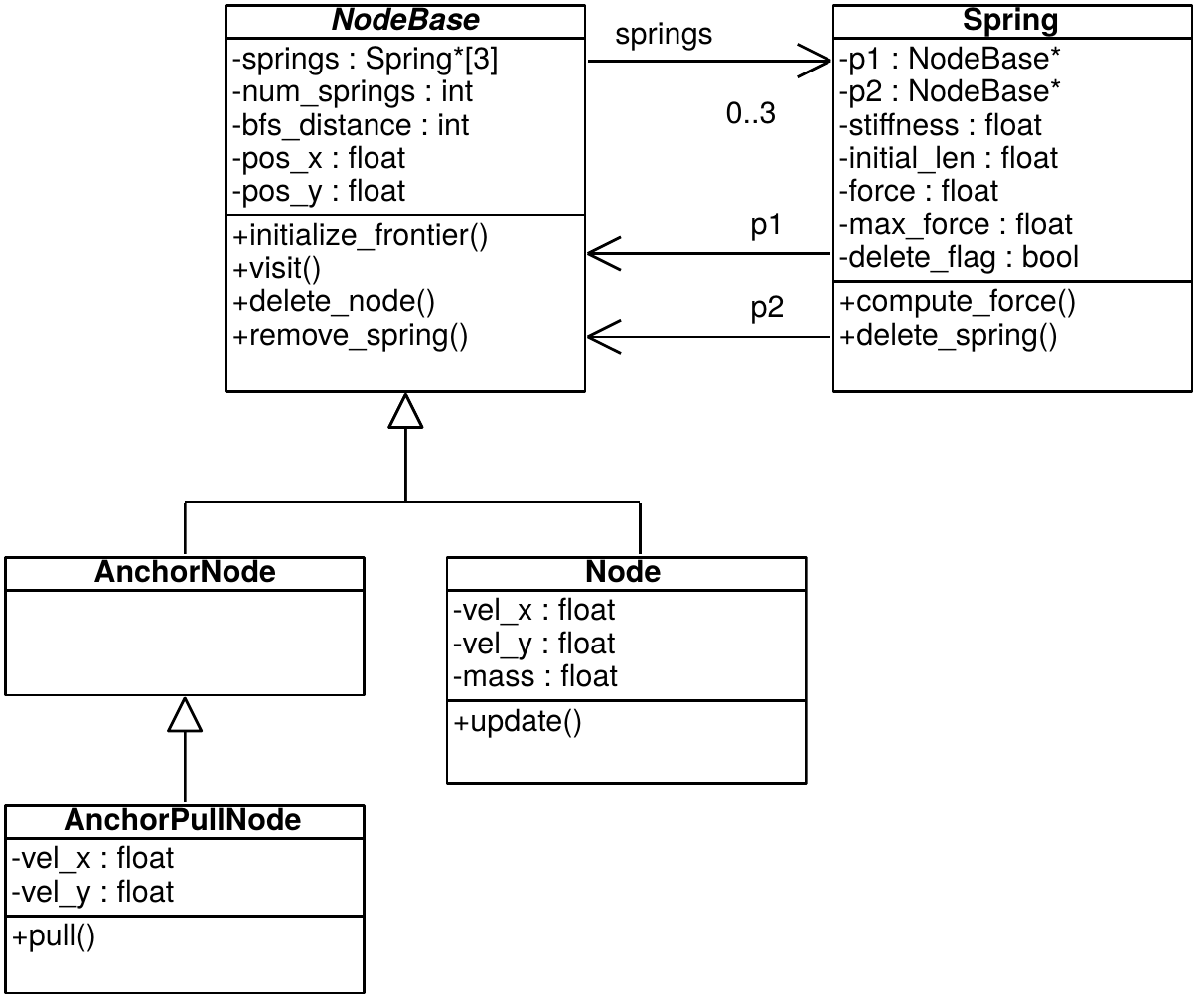}
  \caption[\textsf{structure}: Data structure]{Data structure of \textsf{structure}}
  \label{lst:data_s_structure}
\end{figure}

This application consists of five classes (Figure~\ref{lst:data_s_structure}): An abstract class \texttt{NodeBase} with three subclasses \texttt{AnchorNode}, \texttt{AnchorPullNode} and \texttt{Node}, and a class \texttt{Spring}.

Every node is connected with up to three springs. These are the edges of the graph and stored in an array \texttt{springs} (adjacency list), along with the actual number of springs \texttt{num\_springs}. Every node has a 2D position. In addition, \texttt{AnchorPullNode}s and \texttt{Node}s have a velocity. In the former case, this velocity is constant. In the latter case, the velocity changes based on the forces induced by the springs.

\texttt{Spring}s store pointers to their node endpoints. If at any point the spring is stretched so far that \texttt{force} $>$ \texttt{max\_force}, the spring breaks.

\subsection{Application Implementation}
Before running the main simulation loop, we have to load or generate a mesh network. Our current implementation generates a random graph with a configurable percentage of each node type. \textsf{structure} is an iterative algorithm and consists of the following steps.

\begin{enumerate}
  \item \textbf{Pull Nodes:} Update the position of pull nodes based on their velocity. Parallel do-all: \texttt{AnchorPullNode::pull}.
  \item \textbf{Compute Steps:} Repeat 40 times.
  \begin{enumerate}
    \item \textbf{Compute Forces:} Compute the force of each spring based on Hooke's law. If this force exceeds a spring's threshold, delete the spring from the simulation. Parallel do-all: \texttt{Spring::compute\_force}.
    \item \textbf{Update Velocity/Position:} Sum computed forces exerted by springs for each node. Accelerate and move nodes based on this force. Parallel do-all: \texttt{Node::update}.
  \end{enumerate}
  \item \textbf{Remove Disconnected Nodes:} Find disconnected nodes with a BFS and remove them from the simulation. Implemented with multiple parallel do-all operations.
\end{enumerate}

Step~2 is somewhat similar to an n-body simulation. We are computing forces between elements, however, according to Hooke's law instead of Newton's gravitational law. Furthermore, we only consider direct neighbors in the graph structure. 

\begin{lstfloat}
\begin{lstlisting}[language=c++, caption={[\textsf{structure}: Removing disconnected nodes with BFS]Removing disconnected nodes with BFS}, label={lst:smmo_ex_bfs_structure}, morekeywords={__device__, nullptr}]
__device__ bool continue_bfs;

__device__ void NodeBase::initialize_frontier() {
  distance_ = this->cast<AnchorNode>() == nullptr ? kMaxDistance : 0;
}

__device__ void NodeBase::visit(int distance) {
  if (distance == distance_) {
    continue_bfs = true;
    for (int i = 0; i < kMaxDegree; ++i) {
      auto* spring = springs_[i];
      if (spring != nullptr) {
        // Neighboring vertex.
        auto* n = spring->p1() == this ? spring->p2() : spring->p1();
        // Set distance on neighboring vertex if unvisited.
        if (n->distance_ == kMaxDistance) { n->distance_ = distance + 1; }
      }
    }
  }
}

__device__ void NodeBase::delete_node() {
  if (distance_ == kMaxDistance) {
    for (int i = 0; i < 3; ++i) {
      if (springs_[i] != nullptr) { springs_[i]->delete_flag_ = true; }
    }
  }
}

__device__ void Spring::delete_spring() {
  if (delete_flag_) {
    p1_->remove_spring(this);
    p2_->remove_spring(this);
    destroy(device_allocator, this);
  }
}

__device__ void NodeBase::remove_spring(Spring* s) {
  for (int i = 0; i < 3; ++i) {
    if (springs_[i] == s) {
      springs_[i] = nullptr;
      if (atomicSub(&num_springs_, 1) == 1) { destroy(device_allocator, this); }
      return;
    }
  }
}

void delete_disconnected_nodes() {
  allocator_handle->parallel_do<NodeBase, &NodeBase::initialize_frontier>();

  for (int dist = 0; continue_bfs /*read with cudaMemcpyFromSymbol*/; ++dist) {
    continue_bfs = false;  // write with cudaMemcpyToSymbol
    allocator_handle->parallel_do<NodeBase, int, &NodeBase::visit>(dist);
  }

  allocator_handle->parallel_do<NodeBase, &NodeBase::delete_node>();
  allocator_handle->parallel_do<Spring, &Spring::delete_spring>();
}
\end{lstlisting}
\end{lstfloat}

\paragraph{Step~3: Removing Disconnected Nodes}
This step is the most challenging part. \texttt{Node}s and \texttt{AnchorPullNode}s are removed from the simulation if they are no longer connected to an \texttt{AnchorNode}. \texttt{Spring}s whose endpoint(s) were deleted are also removed from the simulation.

This step is based on a parallel breadth-first graph traversal. The traversal starts at \texttt{AnchorNode}s (excl. \texttt{AnchorPullNode}s) and marks every visited node. Unvisited nodes are removed after the traversal. We use the standard frontier-based BFS algorithm, which is known to work well on GPUs (Section~\ref{sec:inner_arrays_perf_eval}).

Listing~\ref{lst:smmo_ex_bfs_structure} shows how disconnected nodes are detected and removed. Nodes do not have a boolean frontier flag in our implementation. Instead, every node has a \texttt{distance} field, which is initialized to infinity (\texttt{kMaxDistance}). In BFS iteration~$i$, the frontier consists of the nodes with \texttt{distance}~$i$.

In the beginning, all \texttt{AnchorNode}s are initialized to a distance of 0. A BFS iteration (\texttt{NodeBase::visit}) iterates over the neighboring nodes of frontier nodes and updates their distance if they are still unvisited (Line~16). We keep iterating until no vertices were visited in an iteration, as indicated by the \texttt{continue\_bfs} flag.

Finally, we process all unvisited nodes with \texttt{NodeBase::delete\_node}. This method does not delete those nodes yet, but sets a flag on springs that are connected to them. Springs are deleted in \texttt{Spring::delete\_spring}. Before a spring deletes itself, it unregisters itself from its endpoints (\texttt{remove\_spring}). As part of this process, we atomically decrement the spring counter of the node. Once the counter reaches zero, the node is deleted. This process requires an atomic operation, because multiple GPU threads may be concurrently unregistering springs from a node.

\section{\textsf{traffic}: Traffic Flow Simulation}
\label{sec:smmo_traf_flow_sia}
Traffic flow simulations are important tools in transportation planning~\cite{2005physics...7127M}. They can guide the design and construction of city street networks. \textsf{traffic} is an agent-based microsimulation that simulates single vehicles (\emph{agents}) which move on a street network. It is based on the Nagel-Schreckenberg model~\cite{nagel_schr}, a simple model based on cellular automata which can reproduce real-world traffic phenomena~\cite{10.1007/978-1-4471-1281-5_17, WAHLE2001719} such as traffic jams. Compared to other SMMO applications, \textsf{traffic} is quite complex: It does not only model agents, but also more advanced street features such as traffic lights or yield signs. We describe this application only on a high level. Related work describes many variantions and how to implement them~\cite{primer_street}.

By default, this application generates a random street network upon startup. Alternatively, real-world street networks can be imported from OpenStreetMap (OSM) dumps in GraphML file format. Such dumps include street properties such as position, shape, connections to other streets, speed limits and OSM street type.

\subsection{Data Structure}
In the Nagel-Schreckenberg model, a street (\emph{link}) is divided into equally-sized \texttt{Cell}s, each of which can contain up to one agent (Figure~\ref{fig:traffic_cells_and_inters}). An agent (class \texttt{Car}; Figure~\ref{fig:traffic_architecture}) can move onto a neighboring cell only if it is free. Every agent has a velocity, measured in cells per iteration. Both agents and cells have a maximum velocity; the latter one can be used to model speed limits on streets.

Agents precompute their \texttt{path} of movement per iteration. The \texttt{path} array contains the next \texttt{velocity} many cells onto which the agent is about to move. All of these cells must be empty. This is to ensure that the agent does not crash into other agents.

\paragraph{Intersections and Traffic Lights}
Every cell has zero, one, or multiple outgoing cells, forming a directed graph. In the first case, the cell is a \emph{sink}, i.e., a street leaves the simulation area. Agents entering a sink will be randomly redistributed. The second case is most common and represents a regular street cell. The third case appears at intersections, where a cell is connected to the first cell of every outgoing\footnote{Streets in this simulation are one-way streets. Two-way streets consists of one incoming and one outgoing street.} street.

\begin{figure}
  \centering
  \subfloat[Example: Cells and intersections]{\includegraphics[width=0.5\columnwidth]{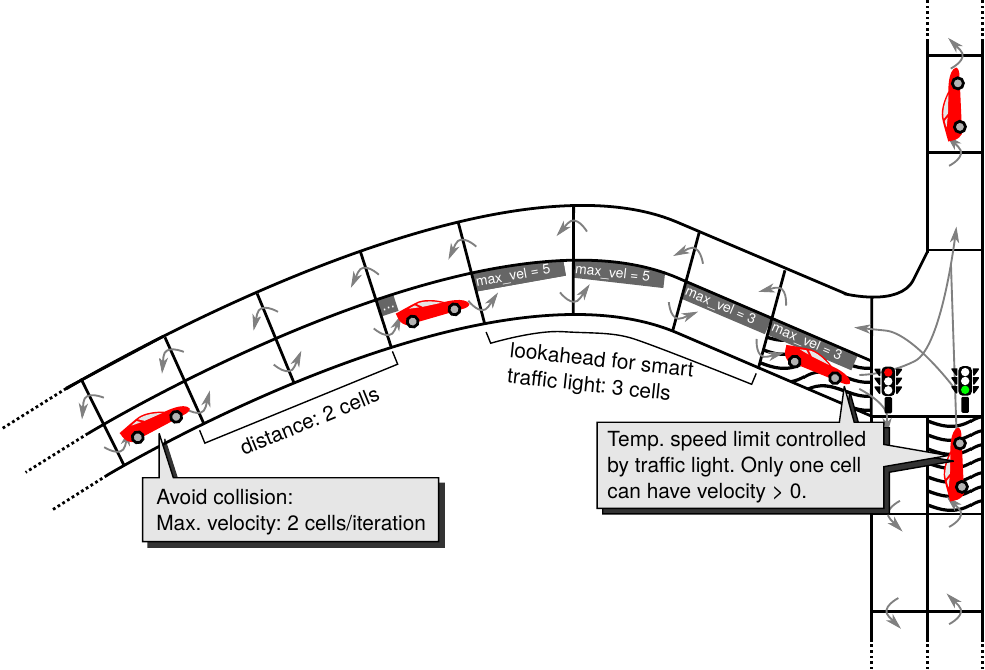}}\hfill
  \subfloat[Screenshot: Imported OpenStreetMap data]{\includegraphics[width=0.45\columnwidth]{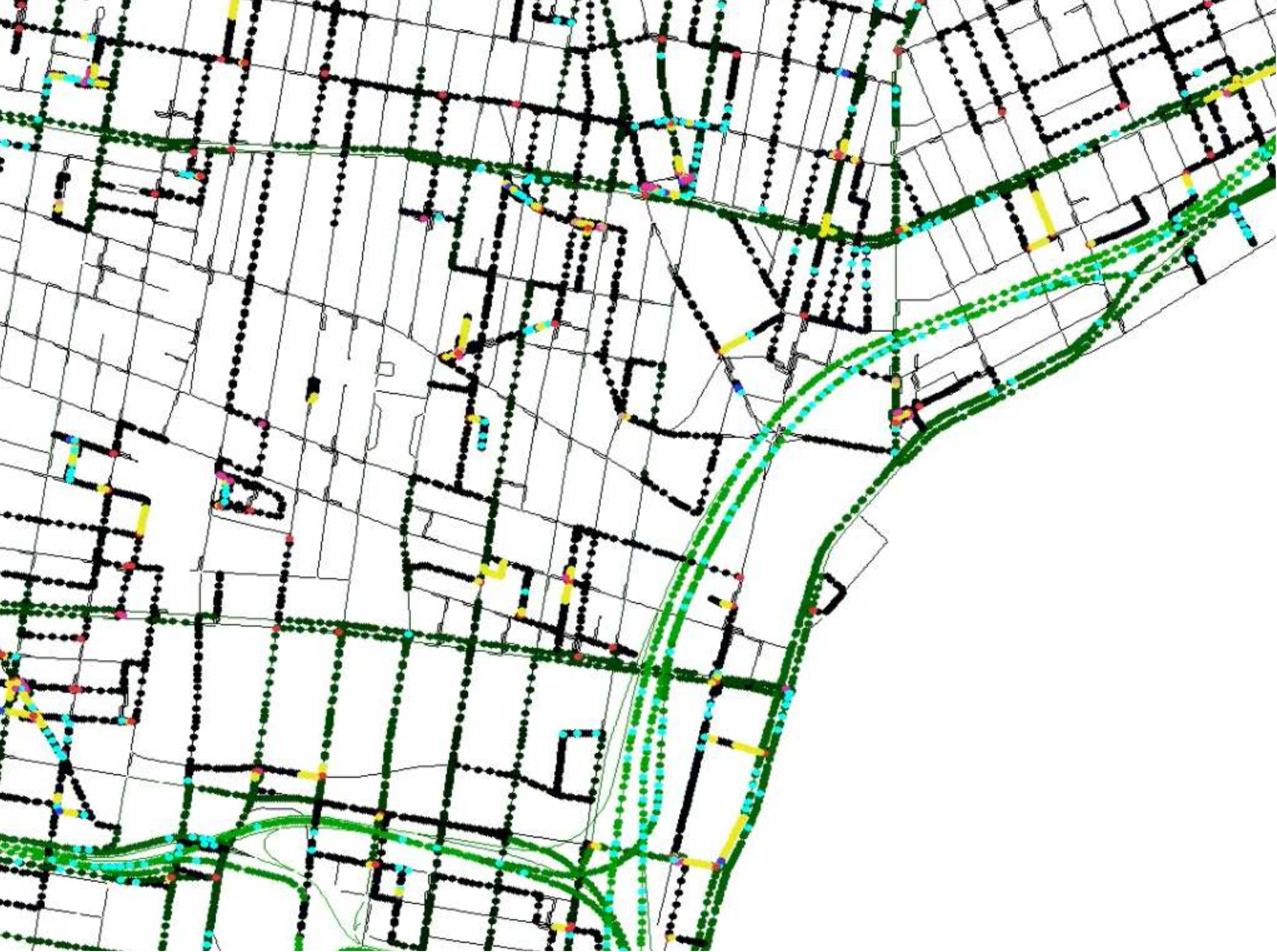}}
  \caption[\textsf{traffic}: Representation of street networks]{Representation of street networks in \textsf{traffic}}
  \label{fig:traffic_cells_and_inters}
  \vspace{30pt}

  \includegraphics[scale=0.75]{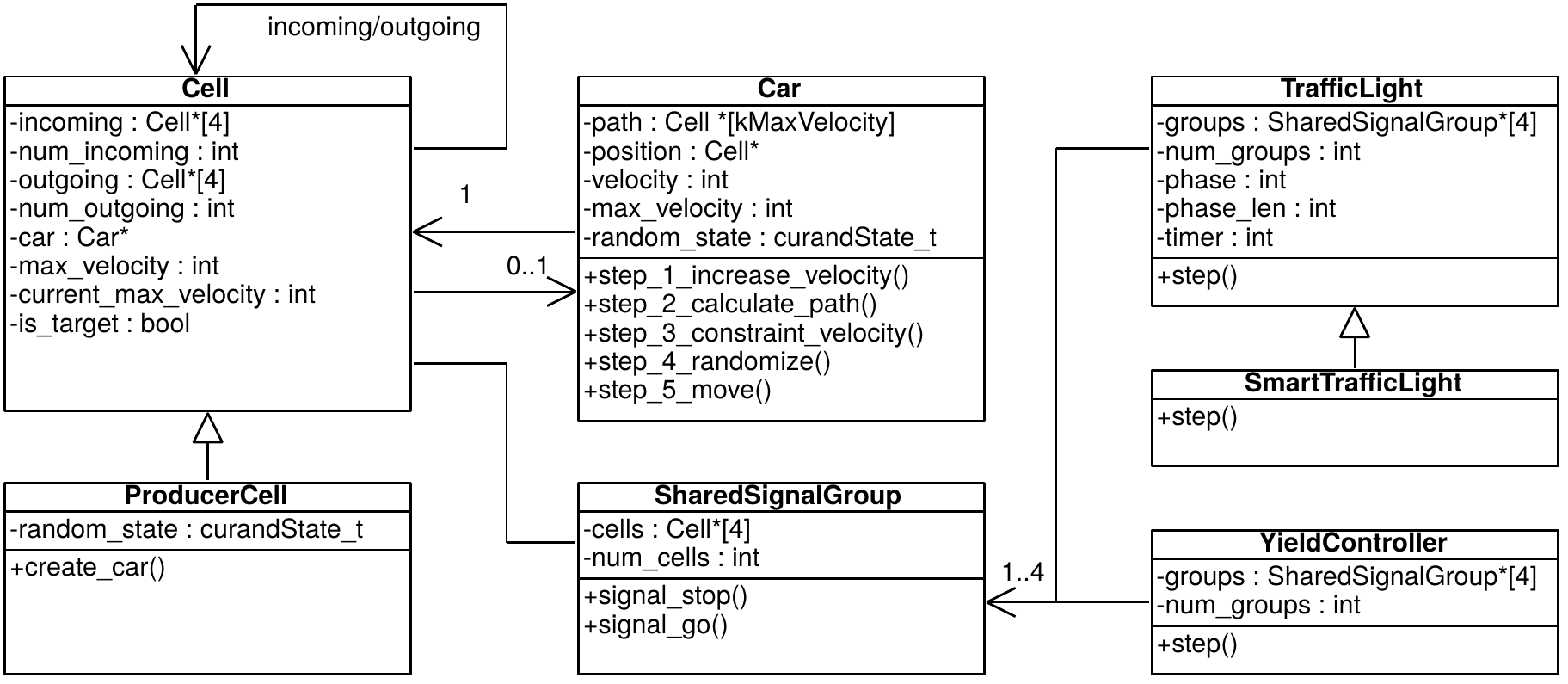}
  \caption[\textsf{traffic}: Data structure]{Data structure of \textsf{traffic}}
  \label{fig:traffic_architecture}
\end{figure}

It is important to ensure that only one car enters an outgoing street (i.e., first cell of the street) at an intersection in one iteration, even if multiple cars from different incoming streets are waiting. To that end, a \emph{traffic controller} can impose temporary speed limits on cells, e.g., a speed limit of zero, corresponding to a red light~\cite{doi:10.1142/S0129183197000904}. Traffic controllers set and remove speed limits for the last cells of all incoming streets such that only one incoming street has a green light at a time. We implemented three kinds of controllers.
\begin{itemize}
  \item A \emph{traffic light} imposes a temporary speed limit of zero on all incoming streets, except for one street which has a green \emph{phase} for a certain number of iterations (\emph{phase length}). Green phases are scheduled round-robin among all incoming streets.
  \item A \emph{smart traffic light} works like a normal traffic light but assigns a green phase to an incoming street immediately if this street is the only incoming street with a waiting car. Real traffic lights have sensors/cameras to provide such behavior. All traffic lights in the benchmark section are smart.
  \item A \emph{yield controller} corresponds to a yield traffic sign, which is often found at the end of merge lanes of highway entrances. Given $n$ incoming streets, it assigns a temporary speed limit of zero to all streets $i > s$ if street $s$ has a car, i.e., incoming streets/cells in the controller should be ordered by priority. 
\end{itemize}
The last two traffic controllers should ensure that traffic does not have to stop or slow down in front of an intersection. Thus, controllers must check all cells from which a car could cross an intersection in one iteration (not only the closest incoming cell), as indicated by the maximum allowed speed limit on a street (\emph{lookahead}). This is done with a graph traversal on back edges (\emph{incoming cells}), which terminates when a car was found within lookahead range.

\paragraph{Turn Lanes}
To allow the traffic to flow more smoothly, we generate turn lanes from each incoming street to each outgoing street at intersections. To implement a red traffic light signal for an incoming street, we now have to impose a speed limit on the last cell of each turn lane. Those cells are grouped in a \texttt{SharedSignalGroup} because they should all have the same traffic light signal\footnote{Traffic lights could be extended such that more than one street (for certain directions) has a green light at a time.}.


\subsection{Application Implementation}
This application consists of the following parallel do-all operations.


\begin{enumerate}
  \item \textbf{Advance Traffic Light State:} Increment a timer. Once it reaches the phase length, propagate a red signal to the currently green shared signal group, as indicated by \texttt{phase} (index into \texttt{groups}). Now, the next shared signal group receives a green light phase and we reset the counter. Smart traffic lights function similar but can change their phase even before the timer reaches the phase length, in case a car is waiting at only one shared signal group with a red signal. Parallel do-all: \texttt{TrafficLight::step}
  \item \textbf{Advance Yield Controller State:} Among all shared signal groups, find the first one that has a car that can cross within the next iteration. That signal group receives a green light and all other signal groups receive a red light.  Parallel do-all: \texttt{YieldController::step}
  \item \textbf{Nagel-Schreckenberg Iteration:} In the following steps, we denote the agent that is processed in a parallel do-all operation by $a$.
  \begin{enumerate}
    \item \textbf{Acceleration:} Increase the agent's velocity unless it is already driving at its maximum velocity: $v_a \gets \mathit{min}(v_a + 1, \mathit{v\_max}_a)$. \\ Parallel do-all: \texttt{Car::step\_1\_increase\_velocity}
    \item \textbf{Compute Path:} Determine the agent's path of movement of length $v_a$, i.e., the next $v_a$ many cells that it will pass through. A \emph{navigation strategy} determines the next cell at an intersection with multiple outgoing cells. Our current implementation uses random walk, biased towards large streets. Parallel do-all: \texttt{Car::step\_2\_calculate\_path}
    \item \textbf{Adjust Velocity:} Avoid collisions with other agents and enforce speed limits. To avoid collisions, reduce the speed $v_a$ to the largest possible value such that the first $v_a$ many cells on the calculated path are free. To follow speed limits, reduce $v_a$ to the largest possible value, such that $v_a \leq \mathit{v\_max}_c$ for each cell $c$ among the first $v_a$ many cells on the calculated path. \\ Parallel do-all: \texttt{Car::step\_3\_constraint\_velocity}
    \item \textbf{Randomization:} Reduce the agent's velocity by one unit with a probability of 20\%. Parallel do-all: \texttt{Car::step\_4\_randomize}
    \item \textbf{Update Position:} Move the agent from its current location by $v_a$ many cells according to the calculated path. Parallel do-all: \texttt{Car::step\_5\_move}
  \end{enumerate}
\end{enumerate}

In our \textsc{DynaSOAr} benchmarks, we, furthermore, distinguish between three kinds of cells: Regular cells, producer cells and sink cells. Producer cells create new agents with given probability if the cell is empty. Sink cells remove agents with a given probability. This is to simulate dynamic (de)allocation.

In a real traffic simulation, we would generate new agents in certain city areas (e.g., residential areas). Agents would move to certain waypoints according to a \emph{schedule} which can be generated from real-world traffic data. Agents would be removed from the simulation when they reach their final waypoint.



\paragraph{Step~3c: Adjust Velocity}
Listing~\ref{lst:traffic_avoid_crash} illustrates how \textsf{traffic} avoids collisions and speed limit violations in a Nagel-Schreckenberg iteration. Let us assume that we decided to move the agent at a certain \texttt{velocity} in Step~3a and precomputed the same number of cells in \texttt{path} as part of Step~3b. Now, we check every cell on the precomputed path. The speed limit of an agent has to be reduced in two cases.

\begin{enumerate}
  \item A cell on the path is occupied by another agent (Line~7). In this case, the speed limit has to be reduced such that the cell is not entered.
  \item A cell on the path has a speed limit of less than the agent's intended velocity (Line~13). Depending on which option would make more progress, we either reduce the speed limit to the cell's speed limit (Line~15) or decide to stop before entering the cell (Line~18). We make more progress in the latter case if the cell has such a low speed limit that we would not even reach the cell in this iteration. E.g., this is the case when approaching a red traffic light (speed limit zero).
\end{enumerate}

\begin{lstfloat}
\begin{lstlisting}[language=c++, caption={[\textsf{traffic}: Avoiding collisions and speed limit violations]Avoiding collisions and speed limit violations}, label={lst:traffic_avoid_crash}, morekeywords={__device__, nullptr}]
__device__ void Car::step_3_constraint_velocity() {
  for (int distance = 1; distance <= velocity_; ++distance) {
    // Invariant: Movement of up to distance - 1 many cells at velocity_ is allowed.
    Cell* next_cell = path_[distance - 1];

    // Avoid collision.
    if (next_cell->car_ != nullptr) {
      // Cannot enter cell.
      velocity_ = distance - 1;
      break;
    } // else: Can enter next cell.

    if (velocity_ > next_cell->current_max_velocity_) {
      // Car is too fast for this cell.
      if (next_cell->current_max_velocity_ > distance - 1) {
        // Even if we slow down, we would still make progress.
        velocity_ = next_cell->current_max_velocity_;
      } else {
        // Do not enter the next cell.
        velocity_ = distance - 1;
        break;
      }
    }
  }
}
\end{lstlisting}
\end{lstfloat}

Notice how the \texttt{Car::path} array is accessed. For each car, we access array elements sequentially, starting from index 0 up to index $\texttt{velocity} - 1$. We can optimize this access by storing this inner array in SOA layout.

\subsection{Object-oriented Traffic Simulations}
Previous work has demonstrated that traffic simulations can be expressed very well with object-oriented programming, because many real-world entities such as different kinds of vehicles, streets, intersections, traffic lights, etc. can be directly mapped to objects/classes~\cite{doi:10.1080/18128600808685689, 396833, traffic_phd_thesis_hel}. Traffic simulations are even used for teaching object-oriented software design to students~\cite{Proulx:1998:TSC:273133.273160}. 

\section{\textsf{wa-tor}: Fish and Sharks Simulation}
\label{sec:smmo_ex_wa_tor_sec}
\textsf{wa-tor} is an ecosystem of fish and sharks that occupy a 2D space in a predator-prey relationship~\cite{10.2307/24969495}. Such ecosystems can be described with the Lotka-Volterra equations~\cite{RevModPhys.43.231}. In this application, we simulate \textsf{wa-tor} with a cellular automaton (CA). Cellular automatons are well suited for GPU execution~\cite{GIBSON201511}. There are typically thousands of cells and all cells feature the same or similar computations, based on the state of neighboring cells.

\textsf{wa-tor} simulates a hypothetical, torus-shaped (2D space, wrapping around at the borders) planet made up of cells. Each cell can be occupied by up to one agent (fish or shark). Sharks are predators who are eating the fish (the prey). In each iteration, an agent can move to a neighboring cell. Fish can only move to empty cells, but sharks can move to empty cells or cells containing a fish, thereby consuming the fish. Sharks have an energy level that decreases with each iteration and increases when consuming a fish. Both fish and sharks reproduce after a certain number of iterations.

\subsection{Data Structure}
\begin{figure}
  \centering
  \includegraphics[scale=0.75]{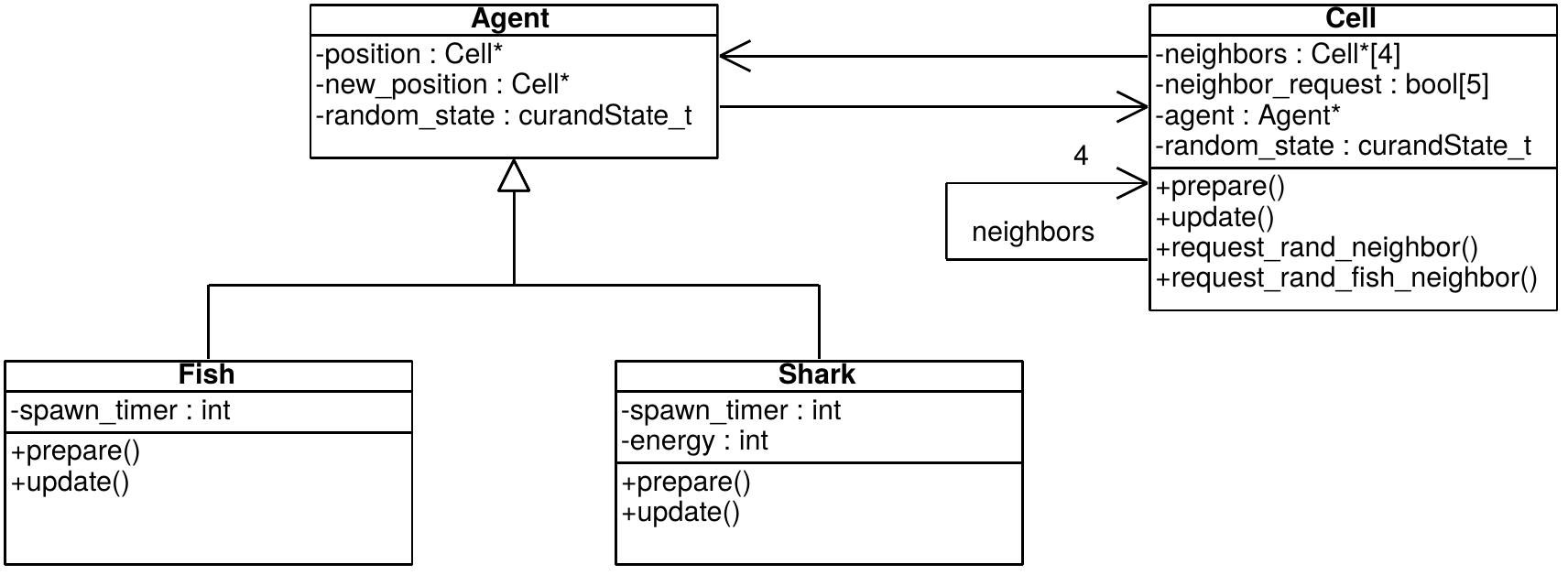}
  \caption[\textsf{wa-tor}: Data structure]{Data structure of \textsf{wa-tor}}
  \label{fig:data_s_wa_tor_smmoex}
\end{figure}
Figure~\ref{fig:data_s_wa_tor_smmoex} illustrates the data structure of \textsf{wa-tor}. There is an abstract class \texttt{Agent} with two subclasses \texttt{Fish} and \texttt{Shark}, and another class \texttt{Cell}. 

Agents store a pointer to their current position (\texttt{position}) and have a field \texttt{new\_position} which is used for determining the next cell to move to. Cells have pointers to all four neighboring cells (\texttt{neighbors}). Furthermore, there is an array \texttt{neighbor\_request} with a boolean flag per neighbor, indicating whether an agent from a neighboring cell is attemping to enter this cell. This data structure allows us to implement the movement of agents without expensive synchronization mechanisms such as atomic operations.

\subsection{Application Implementation}
\begin{figure}
  \centering
  \subfloat[Cell interaction]{\includegraphics[width=0.65\textwidth]{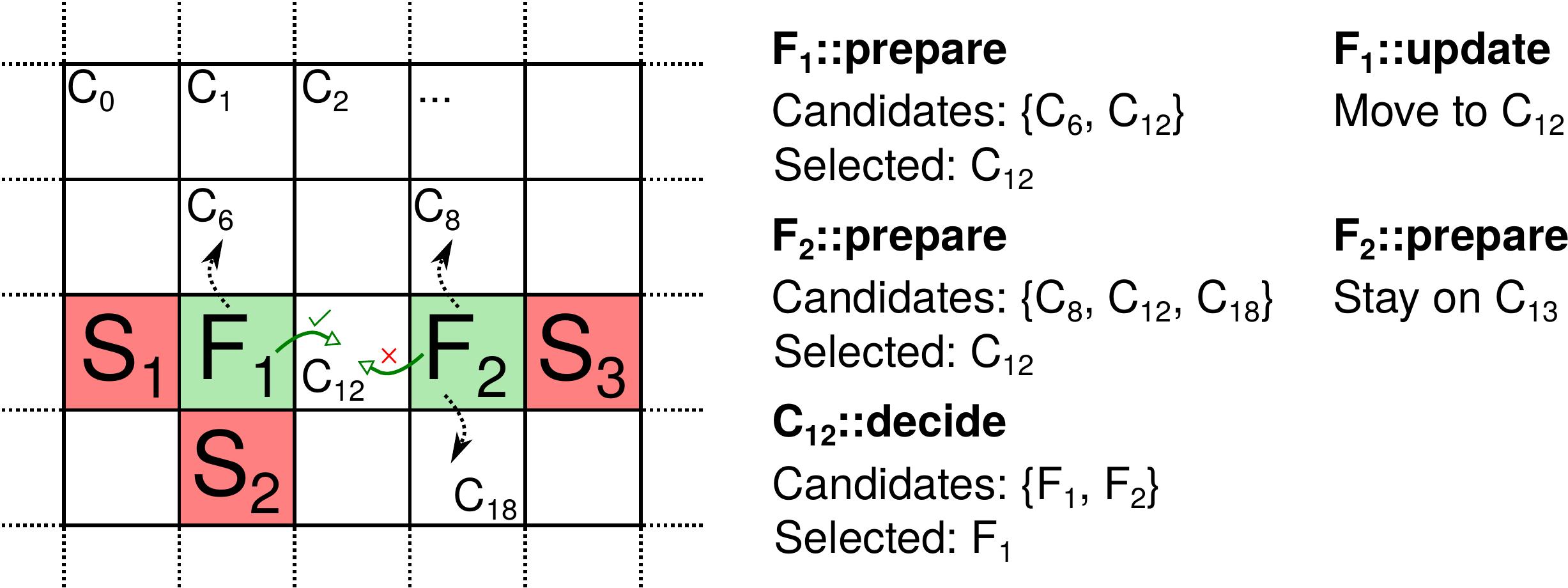}} \hfill
  \subfloat[Screenshot]{\includegraphics[width=0.24\textwidth]{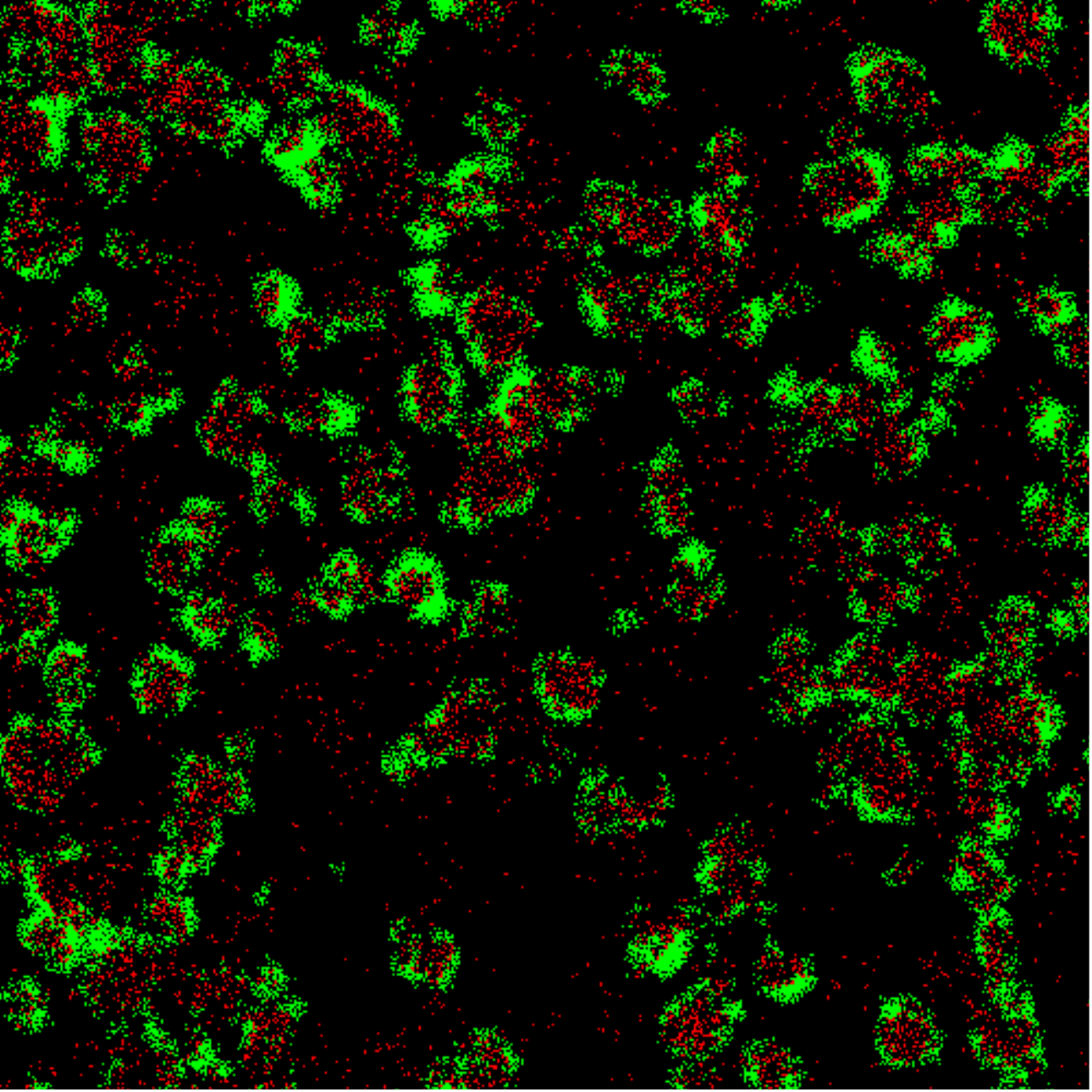}}
  \caption[\textsf{wa-tor}: Cell interaction and screenshot]{Cell interaction and screenshot of \textsf{wa-tor}. Green areas indicate fish and red areas indicate sharks.}
  \label{lst:wator_cell_interaction_smmoex}
\end{figure}

\textsf{wa-tor} is an iterative algorithm. An iteration consists of the following steps.
\begin{enumerate}
  \item \textbf{Reset Cells:} Initialize the array \texttt{neighbor\_requests} to \emph{false}. This array is used by agents of neighboring cells to indicate that they wish to enter this cell. Parallel do-all: \texttt{Cell::prepare}
  \item \textbf{[Fish] Select Outgoing Cell:} Among all four neighboring cells, select an empty cell at random and set the corresponding \texttt{neighbor\_request} to \emph{true}. This is the cell that the agent is planning to move to. Parallel do-all: \texttt{Fish::prepare}
  \item \textbf{Select Incoming Fish:} Among all agents that are trying to move to this cell, as indicated by the \texttt{neighbor\_request} array, select one at random. Parallel do-all: \texttt{Cell::decide}
  \item \textbf{[Fish] Move to New Cell:} Move to the selected cell if approval was granted. Leave a new \texttt{Fish} at the current location with a certain probability. Parallel do-all: \texttt{Fish::update}
  \item \textbf{Reset Cells:} Same as Step~1. Parallel do-all: \texttt{Cell::prepare}
  \item \textbf{[Shark] Select Outgoing Cell:} Same as Step~2, but sharks are also allowed to and prefer to move onto cells that contain a fish. Parallel do-all: \texttt{Shark::prepare}
  \item \textbf{Select Incoming Shark:} Same as Step~3. Parallel do-all: \texttt{Cell::decide}
  \item \textbf{[Shark] Move to New Cell:} Same as Step~4, but sharks consume a fish if present on the target cell. Parallel do-all: \texttt{Shark::update}
\end{enumerate}

Figure~\ref{lst:wator_cell_interaction_smmoex} illustrates the selection of cells by fish. In this figure, red blocks indicate sharks and green blocks indicate fish. $F_1$ scans its neighborhood for empty cells to move to\footnote{Fish and sharks cannot move diagonally.}. Possible candidates are $C_6$ and $C_{12}$. Among those, $F_1$ randomly selects $C_{12}$. However, $F_2$ also selects $C_{12}$ in the same iteration. Only one agent is allowed to move onto the cell. Among those two, \texttt{Cell::decide} randomly decides that $F_1$ is allowed to enter the cell. As a consequence, $F_1$ moves to $C_{12}$ and $F_2$ remains at its original position.

\begin{lstfloat}
\begin{lstlisting}[language=c++, caption={[\textsf{wa-tor}: Application logic]Simulation logic for \texttt{Fish}}, label={lst:wator_listing_step_1_to_4}, morekeywords={__device__, nullptr}]
__device__ void Cell::prepare() {
  for (int i = 0; i < 5; ++i) { neighbor_request_[i] = false; }
}

__device__ void Fish::prepare() {
  ++spawn_timer_;
  position_->request_random_free_neighbor();
}

__device__ void Cell::request_random_free_neighbor() {
  int candidates[4];  int num_candidates = 0;

  for (int i = 0; i < 4; ++i) {
    if (neighbors_[i]->agent_ == nullptr) { candidates[num_candidates++] = i; }
  }

  if (num_candidates == 0) {  // Stay on this cell.
    neighbor_request_[4] = true;
  } else {
    int selected_index = curand(&random_state) % num_candidates;  // cuRAND library
    int selected = candidates[selected_index];
    int neighbor_index = (selected + 2) % 4;
    neighbors_[selected]->neighbor_request_[neighbor_index] = true;
  }
}

__device__ void Cell::decide() {
  if (neighbor_request_[4]) {  // This cell has priority.
    agent_->new_position_= this;
  } else {  // Select random agent among requesting neighbors.
    int candidates[4];  int num_candidates = 0;

    for (int i = 0; i < 4; ++i) {
      if (neighbor_request_[i]) { candidates[num_candidates++] = i;  }
    }

    if (num_candidates > 0) {
      int selected_index = curand(&random_state_) % num_candidates;
      neighbors_[candidates[selected_index]]->agent_->new_position_ = this;
    }
  }
}

__device__ void Fish::update() {
  Cell* old_position = position_;
  if (old_position != new_position_) {
    position_ = new_position_;
    position_->agent_ = this;

    if (spawn_timer_ > kSpawnThreshold) {  // Spawn offspring.
      auto* new_fish = new(device_allocator) Fish();
      new_fish->position_ = old_position;
      old_position->agent_ = new_fish;
      spawn_timer_ = 0;
    } else { old_position->agent_ = nullptr; }
  }
}
\end{lstlisting}
\end{lstfloat}

Listing~\ref{lst:wator_listing_step_1_to_4}\textsc{a} shows how the simulation logic for \texttt{Fish} (Steps~1--4) is implemented. The simulation logic for \texttt{Shark}s is implemented in a similar way. Requests to move to a cell are stored in \texttt{neighbor\_request}. This array has five slots: One for each neighbor and one for the cell itself, indicating that an agent wishes to stay on its current cell. This is the case if all neighboring cells are occupied. After examining all requests, \texttt{Cell::decide} sets \texttt{new\_position} of the agent that is allowed to enter the cell. \texttt{Fish::update} moves to this cell and leaves a new \texttt{Fish} object at its old location every \texttt{kSpawnThreshold} iterations.

\subsection[Benefits of OOP and Dynamic Allocation]{Benefits of Object-oriented Programming and Dynamic Allocation}
Complex ecosystems such as predatory-prey ecosystems or population ecosystems are often modelled with object-oriented programming~\cite{JONES199431, doi:10.1177/003754979406200106, FERREIRA199521} because real or abstract entities can be mapped directly to objects/classes. Previous work describes in detail why object-oriented programming is suitable and an intuitive way of modelling such ecosystems~\cite{SILVERT199391}.

\textsf{wa-tor} is difficult to implement without dynamic object allocation. \texttt{Fish} and \texttt{Shark} objects are created and deleted all the time. However, every cell contains at most one \texttt{Fish} or \texttt{Shark} object at a time. To implement \textsf{wa-tor} with only static allocation (baseline), we merged all classes of \textsf{wa-tor} into a single class \texttt{Cell} (Figure~\ref{fig:data_s_wa_tor_smmoex_merged}). The field \texttt{agent\_type} indicates whether a cell contains an agent and if so, what the type of the agent is (fish or shark).

\begin{figure}
  \centering
  \includegraphics[scale=0.75]{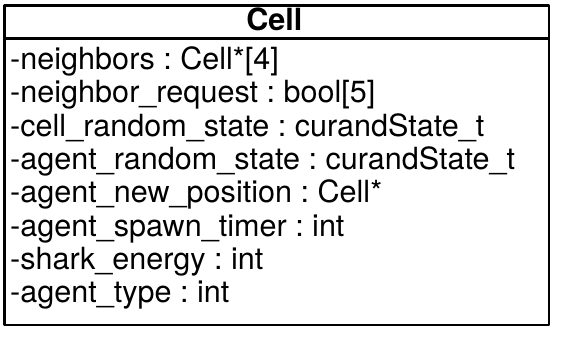}
  \caption[\textsf{wa-tor}: Data structure (merged agents into \texttt{Cell})]{Data structure of \textsf{wa-tor} (merged agents into \texttt{Cell}), without methods}
  \label{fig:data_s_wa_tor_smmoex_merged}
\end{figure}

This example highlights the importance of dynamic object allocation. Our baseline versions with only static allocation break abstractions of our object-oriented implementation because class \texttt{Cell} now contains fields that logically describe both cells and agents.

\section{\textsf{sugarscape}: Simulation of Population Dynamics}
Sugarscape is an agent-based social simulation, originally presented in Epstein and Axtell's book \emph{Growing Artificial Societies}~\cite{RePEc:mtp:titles:0262550253}. It is a celluar automaton that simulates the behavior and interaction of agents (male/female) on a 2D grid. Agents require a certain amount of sugar to survive an iteration (\emph{metabolism}). Cells grow and accumulate sugar over time and agents can harvest sugar by moving onto a cell.

Many variants of Sugarscape have been implemented in the past. Those variants can simulate complex social dynamics~\cite{DBLP:journals/corr/Kehoe15} such as trade, wars, diseases, etc. Our \textsf{sugarscape} implementation is rather simple and can simulate ageing and reproduction.

\begin{figure}
  \centering
  \includegraphics[scale=0.5]{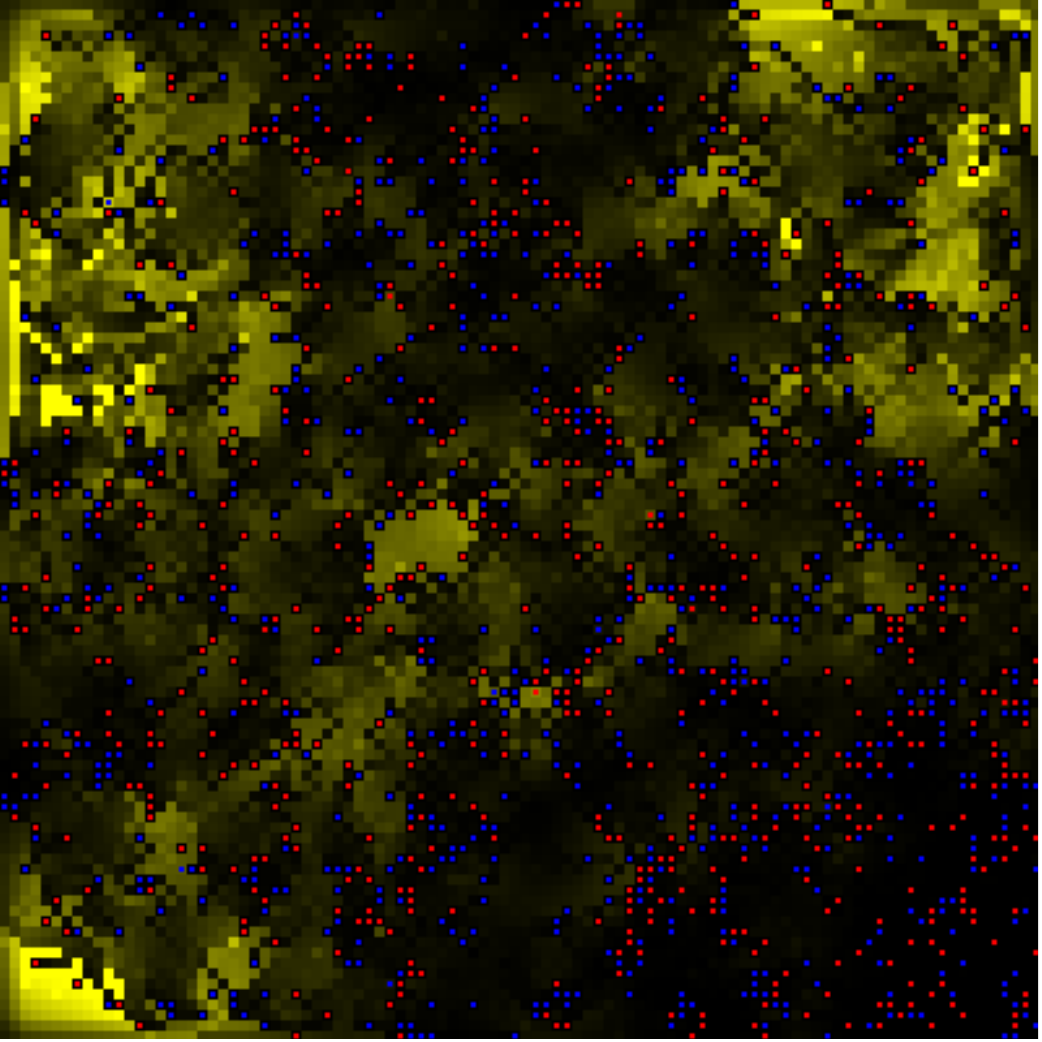}
    \caption[\textsf{sugarscape}: Screenshot]{Screenshot of \textsf{sugarscape}. Yellow areas indicate sugar levels. Sugar diffuses to neighboring cells over time. Male/female agents are colored in blue/red.}
\end{figure}

\subsection{Data Structure}
\textsf{sugarscape} consists of a 2D grid of cells with agents moving upon them, so the data structure (Figure~\ref{fig:data_s_sugarscape}) is similar to \textsf{wa-tor}. The main difference is that the grid structure is \emph{not} encoded with adjacency lists in class \texttt{Cell}. Instead, there is a global array \texttt{cells} which stores pointers to all cells from left to right, top to bottom. Either implementation, adjacency list or global array, is feasible.

\begin{figure}
  \centering
  \includegraphics[scale=0.75]{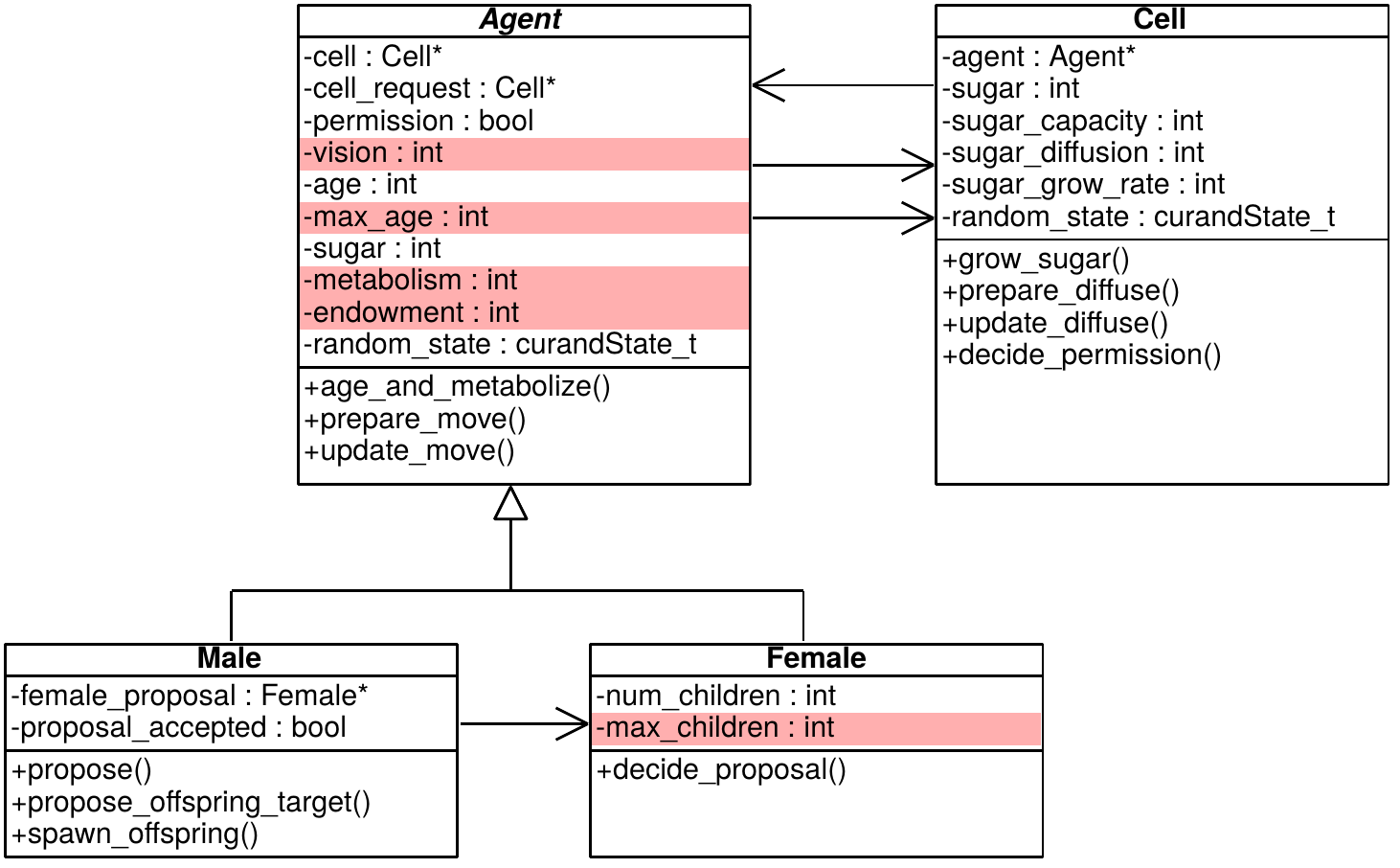}
  \caption{\textsf{sugarscape}: Data structure}
  \label{fig:data_s_sugarscape}
\end{figure}

Agents have a number of \emph{genetic properties}, as indicated by the red color in Figure~\ref{fig:data_s_sugarscape}. Those properties are passed down to offspring.

\subsection{Application Implementation}
This application consists of a large number of parallel do-all operations. More complex Sugarscape simulations would require even more do-all operations. The operations that implement the moving behavior of agents are similar to the respective operations of \textsf{wa-tor}.

\begin{enumerate}
  \item \textbf{Grow Sugar:} Increase the amount of sugar of this cell by \texttt{sugar\_grow\_rate}. Parallel do-all: \texttt{Cell::grow\_sugar}
  \item \textbf{Precompute Sugar Diffusion:} Compute the amount of sugar to be diffused to neighboring cells. Parallel do-all: \texttt{Cell::prepare\_diffuse}
  \item \textbf{Diffuse Sugar:} Take sugar from neighboring cells and add to this cell. Parallel do-all: \texttt{Cell::update\_diffuse}
  \item \textbf{Age and Metabolize:} Increase the age of the agent and reduce the sugar level of the agent by \texttt{metabolism\_rate}. If the age exceeds \texttt{max\_age} or if the sugar level drops below zero, the agent dies. Parallel do-all: \texttt{Agent::age\_and\_metabolize}
  \item \textbf{Select Outgoing Cell:} Among all neighboring cells with distance of up to \texttt{vision}, select the empty cell with the largest amount of sugar. Store a pointer to that cell in \texttt{cell\_request}. Parallel do-all: \texttt{Agent::prepare\_move}
  \item \textbf{Select Incoming Agent:} Among all agents that are trying to move to this cell, as indicated by the surrounding agents' \texttt{cell\_request}, select one at random. Parallel do-all: \texttt{Cell::decide\_permission}
  \item \textbf{Move to New Cell:} Move to the selected cell if approval was granted. Consume all sugar of that cell. Parallel do-all: \texttt{Agent::update\_move}
  \item \textbf{Male $\xrightarrow{\text{Propose}}$ Female:} If a \texttt{Male} is ready to spawn offspring, i.e., sugar level $>$ endowment: Among all neighboring \texttt{Female} agents with distance of up to \texttt{vision}, select the agent with the largest amount of sugar. Parallel do-all: \texttt{Male::propose}
  \item \textbf{Female $\xrightarrow{\text{Accept Proposal}}$ Male:} Among all neighboring, proposing \texttt{Male} agents with distance of up to \texttt{vision}, accept the agent with the largest amount of sugar. A \texttt{Female} cannot accept any more proposals once \texttt{num\_children} $\geq$ \texttt{max\_children}. Parallel do-all: \texttt{Female::decide\_proposal}
  \item \textbf{Choose Offspring Target:} If the proposal was accepted, select a neighboring cell with distance of up to \texttt{vision} at random. This is similar to Step~5. Parallel do-all: \texttt{Male::propose\_offspring\_target}
  \item \textbf{Select Agent:} Among all \texttt{Male} agents that are tyring to spawn offspring on this cell, select one agent at random. This is similar to Step~6. Parallel do-all: \texttt{Cell::decide\_permission}
  \item \textbf{Generate Offspring:} Spawn a new agent on the selected cell if approval was granted. The gender is selected randomly. The genetic properties of the new agent (vision, endowment, metabolism, maximum age) are average values of both parents. The new agent receives an initial sugar level of half the endowment of each parent. Parallel do-all: \texttt{Male::spawn\_offspring}
\end{enumerate}

Previous work describes how to implement Sugarscape efficiently on parallel architectures, in particular on GPUs~\cite{sugarscape_agents}. These Sugarscape implementations differ in details, but show that GPUs are well-suited for agent-based modelling~\cite{Aaby:2010:ESA:1808143.1808181, sugarcape_gpu_2} and simulating a massive number of agents.

\section{\textsf{gol}: Game of Life}
Game of Life is a cellular automaton due to John H. Conway. The game follows a simple set of rules, but can exhibit complex behavior and has even been shown to be Turing complete~\cite{Rendell:2015:TMU:2815663}.

\textsf{gol} consists of a 2D grid of cells. A cell can either be alive or dead. In every iteration, the new state of a cell is decided as follows: Living cells remain alive only if they are surrounded by two or three alive neighbors. Dead cells become alive if they are surrounded by three alive neighbors.

The most straightforward Game of Life implementation calculates the new state of \emph{every} cell depending on its neighborhood. \textsf{gol} follows a different approach. We only simulate cells that are alive or may become alive in the next iteration (\emph{alive-candidates}). Candidates are dead neighbors of alive cells. This approach has a lower expected runtime complexity because most cells are dead most of the time.

\subsection{Data Structure}
\begin{figure}
  \centering
  \includegraphics[width=0.75\textwidth]{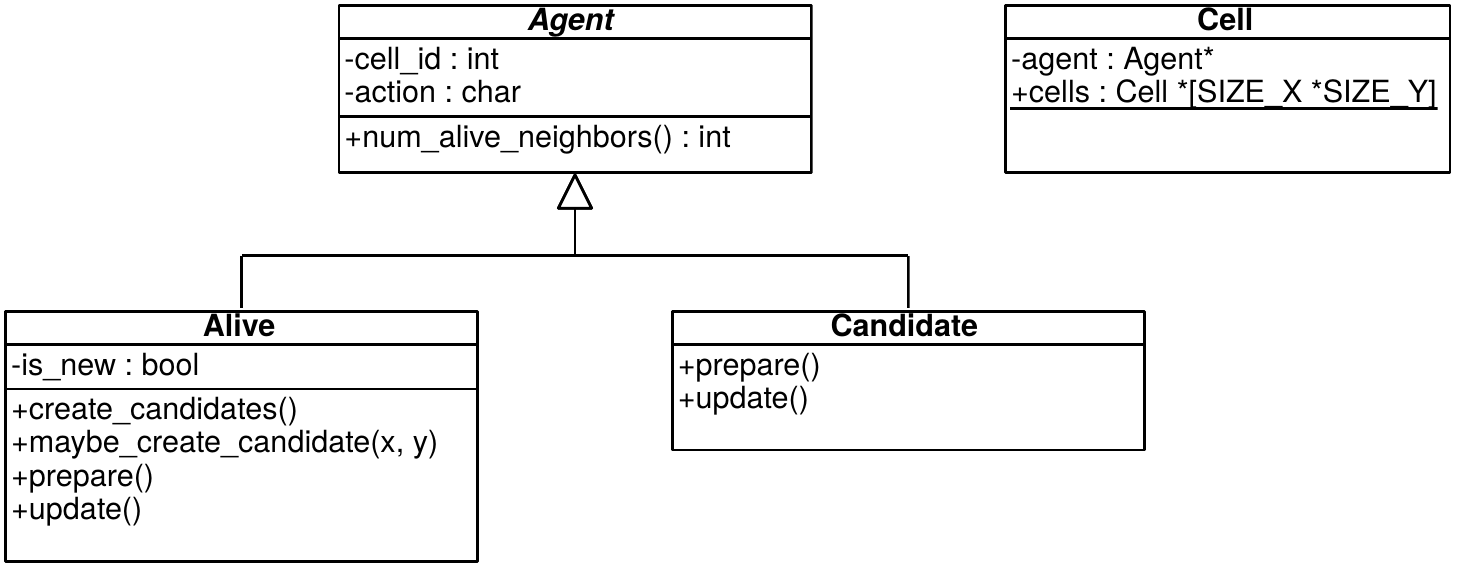}
  \caption[\textsf{gol}: Data structure]{Data structure of \textsf{gol}}
  \label{lst:data_s_gol}
\end{figure}
This application has four classes (Figure~\ref{lst:data_s_gol}). An abstract class \texttt{Agent} with two subclasses \texttt{Alive} and \texttt{Candidate}, and another class \texttt{Cell}. Cells are created at the beginning of the simulation and stored in a global array (static class variable) \texttt{Cell::cells}. A cell can contain an agent or be empty. Agents store references to their cell in the form of an integer index into the global cells array. This requires less memory than an 8-byte pointer and makes it easier to determine all neighboring cells without storing explicit adjencency lists as in \textsf{wa-tor}. 

The main problem of \textsf{gol} is space efficiency. A typical Game of Life implementation requires only 1 bit to encode the state of a cell. In \textsf{gol}, every cell requires at least 8~bytes in the form of an \texttt{Agent} pointer. Therefore, even though \textsf{gol} has a lower expected runtime complexity in terms of the number of processed cells, a naive implementation that spawns one GPU thread per cell and encodes the cell state with 1~bit/byte is usually much faster. However, our \textsf{gol} data structure and computation strategy can serve as a blueprint for other cellular automata that have to maintain a more complex state for each cell~\cite{10.1007/978-3-642-29219-4_84, ca_journal_appl, doi:10.1142/S0219525907001057}.

\subsection{Application Implementation}
\begin{figure}
  \centering
  \subfloat[Cell interaction: Creating new \texttt{Candidate}s]{\includegraphics[width=0.65\textwidth]{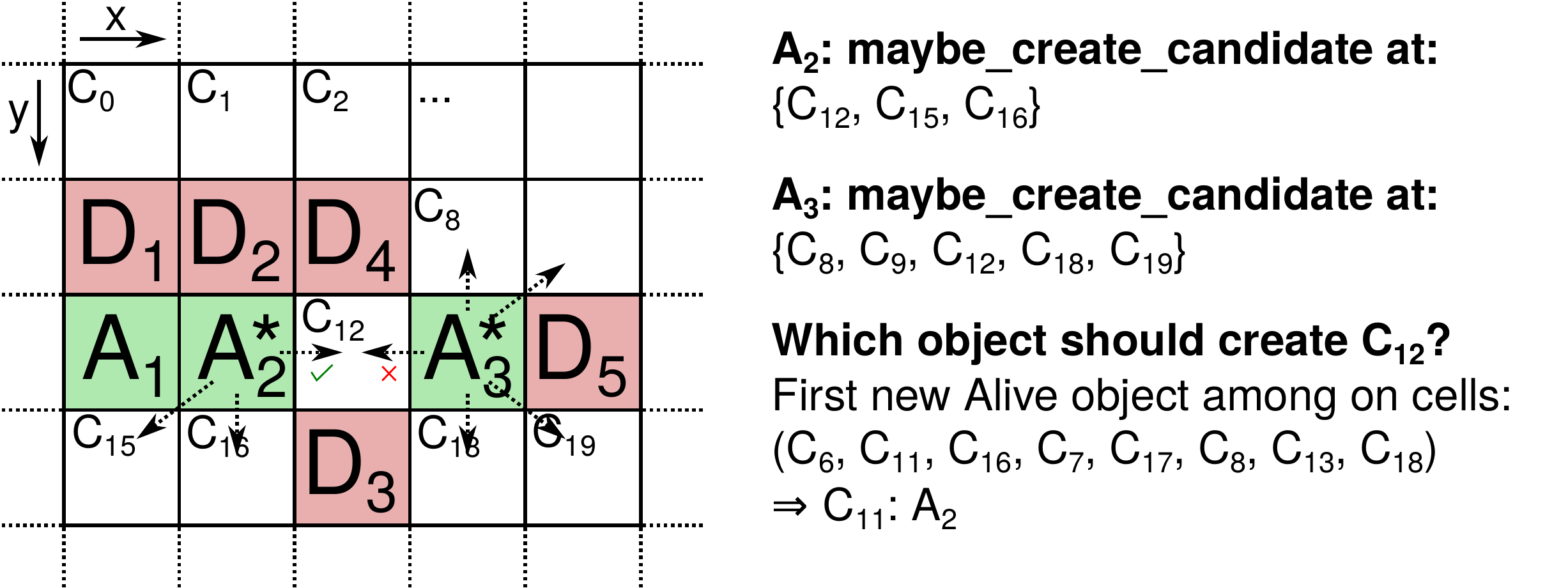}} \hfill
  \subfloat[Screenshot (UTM)]{\includegraphics[width=0.24\textwidth]{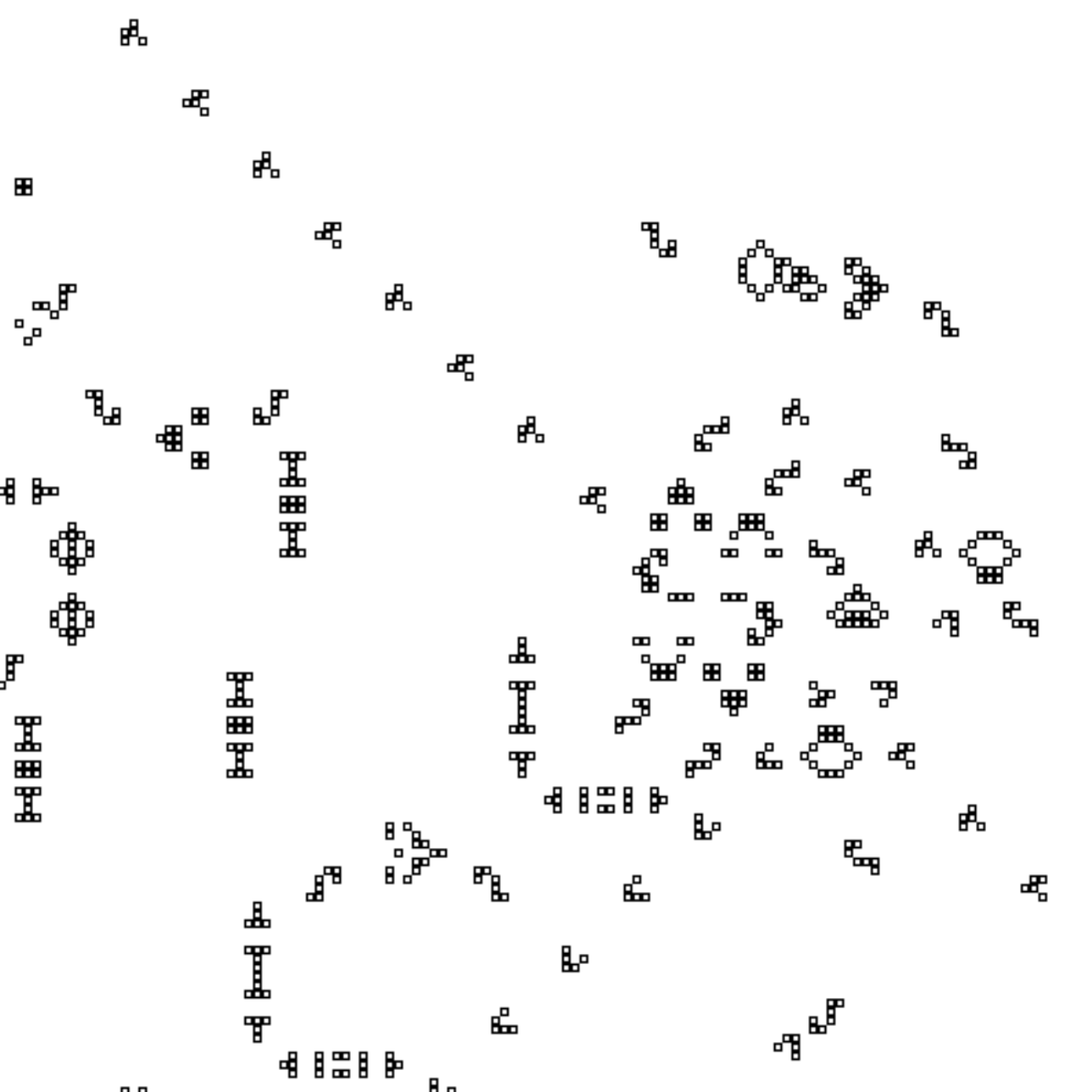}}
  \caption[\textsf{gol}: Cell interaction and screenshot]{Cell interaction and screenshot of \textsf{gol}. Green blocks = \texttt{Alive}, red blocks = \texttt{Candidate}, star = newly created.}
  \label{lst:gol_interaction_smmoex}

\vspace{10pt}
\begin{lstlisting}[language=c++, caption={[\textsf{gol}: Application logic]Creating \texttt{Candidate} objects on dead cells}, label={lst:create_candidate_object_smmo}, morekeywords={__device__, nullptr}, numbers=none]
__device__ void Alive::update() {
  if (is_new_) { create_candidates(); }
  else {
    if (action_ == kActionDie) {  // Replace with Candidate.
      Cell::cells[cell_id_]->agent_ = new(device_allocator) Candidate(cell_id_);
      destroy(device_allocator, this);
    }
  }
}

__device__ void Alive::create_candidates() {
  for (int dx = -1; dx < 2; ++dx) {
    for (int dy = -1; dy < 2; ++dy) {
      int nx = cell_id_ % SIZE_X + dx;
      int ny = cell_id_ / SIZE_X + dy;
      if (nx > -1 && nx < SIZE_X && ny > -1 && ny < SIZE_Y) {
        if (Cell::cells[ny*SIZE_X + nx]->agent_ == nullptr) {
          maybe_create_candidate(nx, ny);
        }
      }
    }
  }
}

__device__ void Alive::maybe_create_candidate(int x, int y) {
  // Check neighborhood of cell to determine who should create Candidate.
  for (int dx = -1; dx < 2; ++dx) {
    for (int dy = -1; dy < 2; ++dy) {
      int nx = x + dx;
      int ny = y + dy;
      if (nx > -1 && nx < SIZE_X && ny > -1 && ny < SIZE_Y) {
        Alive* alive = Cell::cells[ny*SIZE_X + nx]->agent_->cast<Alive>();
        if (alive != nullptr && alive->is_new_) {
          if (alive == this) {
            Cell::cells[y*SIZE_X + x]->agent_ =
                new(device_allocator) Candidate(y*SIZE_X + x);
          }  // else: Created by another Alive.
          return;
        }
      }
    }
  }
}
\end{lstlisting}
\end{figure}
Before entering the main application loop, \textsf{gol} loads the initial game state from a PBM (\emph{Portable BitMap}) file. Such files can be created with common image manipulation programs such as GIMP. We benchmarked \textsc{DynaSOAr} with Rendell's universal turing machine pattern~\cite{Rendell:2015:TMU:2815663}. Every loop iteration of \textsf{gol} consists of four parallel do-all operations.

\begin{enumerate}
  \item \textbf{Prepare Candidate Action:} Decide whether a candidate should be deleted, upgraded to an \texttt{Alive} cell or remain as is. Parallel do-all: \texttt{Candidate::prepare}
  \item \textbf{Prepare Alive Action:} Decide whether an alive cell should be deleted (downgraded to \texttt{Candidate}) or stay as is. Parallel do-all: \texttt{Alive::prepare}
  \item \textbf{Perform Candidate Action:} Perform the action determined in Step~1. Parallel do-all: \texttt{Candidate::update}
  \item \textbf{Perform Alive Action:} Perform the action determined in Step~2. If this is a newly created object (Step~3), create \texttt{Candidate}s on all surrounding cells. Requires special handling to ensure that no two objects create a \texttt{Candidate} on the same cell. Parallel do-all: \texttt{Alive::update}
\end{enumerate}

\paragraph{Step~4: Perform Alive Action}
This step is the most challenging part of the application (Listing~\ref{lst:create_candidate_object_smmo}). Newly created \texttt{Alive} objects must spawn \texttt{Candidate} objects around them, but we have to ensure that we do not create multiple objects per cell.

Figure~\ref{lst:gol_interaction_smmoex} illustrates this problem with an example. $A_2$ and $A_3$ are newly created \texttt{Alive} objects. These must create \texttt{Candidate} objects on all empty, surrounding cells. $C_{12}$ is a neighbor of both $A_2$ and $A_3$, but only one of them should create a \texttt{Candidate} at that location.

To that end, we specify an order among \texttt{Alive} objects that determines which object is creating a \texttt{Candidate} object: top to bottom, left to right. E.g., with respect to $C_{12}$, the \texttt{Alive} object on $C_6$ should create the candidate. However, that cell does not have a newly created \texttt{Alive} object. Therefore, we check $C_{11}$ next. This cell contains $A_2$, a newly created \texttt{Alive} object, so it is that object's responsibility to create a candidate on $C_{12}$. Every newly created \texttt{Alive} object can by itself scan the neighborhood of every surrounding, dead (empty) cell, to determine if it should create a candidate at that location.

\subsection{\textsf{generation}: Generational Cellular Automaton}
\begin{figure}
  \centering
  \includegraphics[scale=0.5]{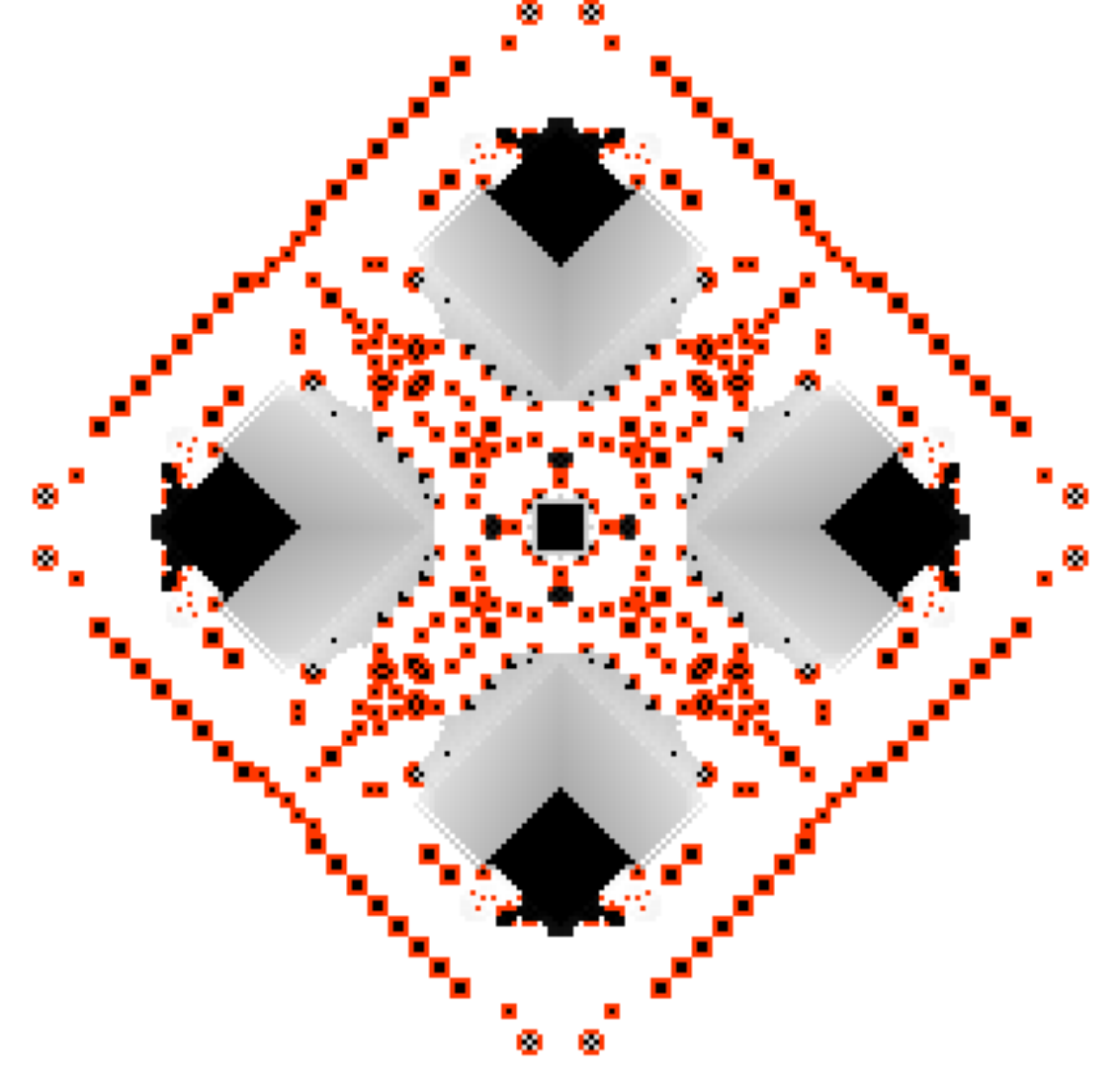}
  \caption[\textsf{generation}: Screenshot]{Screenshot of \textsf{generation}}
  \label{fig:smmo_screenshot_generation}
\end{figure}

\textsf{generation} is an extension of \textsf{gol}. When an \texttt{Alive} cell dies, it stays around for a constant number of iterations before it is replaced with a \texttt{Candidate} and can subsequently become alive again.

This application is based on the rule \emph{Burst}~\cite{gol_burst} (0235678/3468/255): An alive cell remains alive if it has 0, 2, 3, 5, 6, 7 or 8 alive, neighboring cells. Otherwise, the \texttt{Alive} object dies but blocks the cell for 255 more iterations. A dead cell becomes alive if it has 3, 4, 6 or 8 alive, neighboring cells.

Figure~\ref{fig:smmo_screenshot_generation} shows a screenshot of \textsf{generation}. The red areas indicate \texttt{Candidate}s. The black areas indicate \texttt{Alive} objects. \texttt{Alive} objects that are already dead but still block a cell are shown in gray.

\section{Conclusion}
In this chapter, we presented the design and implementation of nine applications from different domains with the object-oriented SMMO programming model. This chapter speaks to the importance and broadness of the SMMO model.

This chapter also highlights the benefits of object-oriented programming: In most applications, object-oriented abstractions greatly improve the code readability compared to a hand-written SOA data layout, mostly due to C++'s member-of operators. An object-oriented programming style also provides stronger type guarantees in C++, which can improve developer productivity. Some applications are more difficult to implement without dynamic object allocation, which we consider an essential feature of object-oriented programming. Finally, all applications exhibit an inherent object structure of abstract or real-world entities and, needless to say, we would like the application source code to reflect this structure.

\chapter{Conclusion}
\label{sec:thesis_concl}
GPU programming is challenging. For best performance, programmers have to write GPU programs in low-level C/C++ dialects and adopt a SIMD/GPU-specific programming style to avoid slowdowns due to inefficient memory access or control flow divergence. Even though high-level GPU programming languages and libraries exist, these systems often fail to deliver the performance of hand-tuned CUDA/OpenCL code.

Our goal is to make GPU programming available to a wider range of developers by providing better support for object-oriented programming (OOP) in low-level and high-level programming languages. Object-oriented programming is often seen as too inefficient for high-performance computing, but as we have shown in this thesis, object-oriented code can achieve competitive performance if properly optimized. Inefficient device memory access is often the biggest performance problem of object-oriented GPU code and we have demonstrated through various prototypes how to achieve good memory coalescing and cache utilization on GPUs.

\begin{itemize}
  \item \textsc{Ikra-Ruby}: We developed a Ruby library for array-based GPU computing in a high-level, object-oriented language. Kernel fusion of functional, parallel array operations such as \texttt{map}, \texttt{reduce} or \texttt{stencil} allows programmers to compose complex computations from small building blocks in a modular, object-oriented way, without losing performance compared to a single monolithic CUDA kernel. We also showed how to express kernel fusion as part of the static type inference process.
  \item \emph{Single-Method Multiple-Objects (SMMO)}: In the object-oriented SMMO programming model, a computation is expressed as running a method on all existing objects of a type. Such computations are well suited for SIMD parallelism and they achieve good performance on GPUs, because the expected amount of branch divergence is low. We demonstrated that a variety of applications and programming patterns from different domains can be expressed in SMMO, ranging from social/physical simulations over BFS graph traversals to dynamic tree updates/constructions.
  \item \textsc{Ikra-Cpp}: We developed a C++/CUDA library for SMMO applications. Most notably, \textsc{Ikra-Cpp} provides an embedded data layout DSL for the Structure of Arrays (SOA) data layout. The SOA data layout is a well-studied best practice for GPU programmers and results in a significant speedup of SMMO application code compared to a traditional AOS data layout. While standard C++/CUDA does not allow programmers to use object-oriented abstractions with custom data layouts such as SOA, \textsc{Ikra-Cpp} programmers can enjoy the benefits of object-oriented programming together with the performance improvements of the SOA data layout.
  \item \textsc{DynaSOAr}: Dynamic memory allocation is one of the corner stones of object-oriented programming. However, it is not supported well on GPUs. Most notably, existing dynamic memory allocators care only about data placement and miss key optimizations for efficient access of allocations. Building on top of \textsc{Ikra-Cpp}, \textsc{DynaSOAr} is a lock-free, dynamic memory allocator that optimizes the access of allocated memory with an SOA data layout and an efficient parallel do-all operation. Compared to other state-of-the-art allocators, \textsc{DynaSOAr} improves the performance of SMMO applications by a factor of up to 3x.
  \item \textsc{CompactGpu}: Unfortunate allocate-deallocate patterns can lead to increased memory fragmentation of dynamically allocated objects. On GPUs, fragmentation can have an adverse effect on memory coalescing and cache utilization. Therefore, we extended \textsc{DynaSOAr} with a memory defragmentation system \textsc{CompactGpu}, which compacts the heap by merging partly occupied memory blocks. In our benchmarks, \textsc{CompactGpu} could lower the memory consumption of SMMO applications by up to 14\% and improve the runtime performance of SMMO applications by up to 16\%.
\end{itemize}

With our three main prototypes \textsc{Ikra-Cpp}, \textsc{DynaSOAr} and \textsc{CompactGpu}, GPU programmers can achieve SMMO application performance that is close to, sometimes even faster than, non-OOP CUDA code. At the same time, they get all the benefits of object-oriented programming such as good abstraction, expressiveness, modularity and developer productivity.


We plan to improve \textsc{Ikra-Cpp} and \textsc{Ikra-Ruby} in the future. It is our vision that \textsc{Ikra-Cpp} will eventually become a part of \textsc{Ikra-Ruby}, such that programmers can develop SMMO applications in a high-level programming language.







\bibliographystyle{plainurlneww}
\bibliography{main}



\end{document}